%% file: main.tex
\documentclass[11pt]{article}
\usepackage[a4paper,margin=1in]{geometry}
\usepackage{fullpage}
\usepackage{bbm}
\def\showauthornotes{0}
\def\showdraftbox{0}

\usepackage{tikz}
\usetikzlibrary{positioning, arrows.meta}

\input{macros}

\newcommand{\eps}{\varepsilon}

\usepackage[
  backend=biber,
  backref=true,
  backrefstyle=none,
  date=year,
  doi=false,
  giveninits=true,
  hyperref=true,
  maxbibnames=10,
  sortcites=false,
  style=alphabetic,
  url=false, 
]{biblatex}

\bibliography{papers}

\title{
Acceleration Meets Inverse Maintenance: \\
Faster $\ell_{\infty}$-Regression}

\author{Deeksha Adil\\Institute for Theoretical Studies\\
ETH Zürich\\deeksha.adil@eth-its@ethz.ch \and Shunhua Jiang\\
Department of Computer Science\\
Columbia University\\ sj3005@columbia.edu \and Rasmus Kyng\\Department of Computer Science\\ETH Zürich\\kyng@inf.ethz.ch}
\date{}

\begin{document}

\maketitle
\begin{abstract}
We propose a randomized multiplicative weight update (MWU) algorithm for $\ell_{\infty}$ regression that runs in $\widetilde{O}\left(n^{2+1/22.5} \mathrm{poly}(1/\epsilon)\right)$ time when $\omega = 2+o(1)$, improving upon the previous best $\widetilde{O}\left(n^{2+1/18} \mathrm{poly}\log(1/\epsilon)\right)$ runtime in the low-accuracy regime. Our algorithm combines state-of-the-art inverse maintenance data structures with acceleration. In order to do so, we propose a novel acceleration scheme for MWU that exhibits \textit{stability} and \textit{robustness}, which are required for the efficient implementations of the inverse maintenance data structures.

We also design a faster \textit{deterministic} MWU algorithm that runs in $\widetilde{O}\left(n^{2+1/12}\mathrm{poly}(1/\epsilon)\right)$ time when $\omega = 2+o(1)$, improving upon the previous best $\widetilde{O}\left(n^{2+1/6} \mathrm{poly} \log(1/\epsilon)\right)$ runtime in the low-accuracy regime. We achieve this by showing a novel stability result that goes beyond previously known works based on interior point methods (IPMs).

Our work is the first to use acceleration and inverse maintenance together efficiently, finally making the two most important building blocks of modern structured convex optimization compatible.
\end{abstract}
\pagenumbering{arabic}
\newpage
\tableofcontents

\newpage
\input{introduction}

\input{tech_overview}

\input{FastMWU}

\input{organization}

\newpage
\printbibliography

\newpage
\appendix
\input{prelims}
\input{stability}
\input{data_structures}
\input{time}

\input{time_deterministic}

\input{MonotoneMWU}

\input{NewNonMonotoneAcceleration}

\input{NewRobustPrimalStep}

\input{MWUStability}
\input{Appendix}

\input{NonMonotoneAcceleration}

\newpage
\clearpage

\end{document}

%% file: macros.tex
\usepackage{amsmath,amssymb}
\usepackage{hyperref}
\usepackage{bbm,xspace,nicefrac,amsthm}
\usepackage{color,enumerate,bm,verbatim,textcomp,float}
\usepackage[small]{caption}
\usepackage{comment} 
\usepackage[T1]{fontenc}
\usepackage{algorithm} 
\usepackage[noend]{algpseudocode}
\usepackage{thmtools}
\usepackage{thm-restate}
\usepackage{fancybox}
\usepackage[shortlabels]{enumitem}
\usepackage{booktabs}
\usepackage{commath}
\usepackage{amsfonts}

\newcommand{\defeq}{\stackrel{\textup{def}}{=}}

\newtheorem{theorem}{Theorem}[section]
\newtheorem*{theorem*}{Theorem}

\newtheorem{lemma}[theorem]{Lemma}
\newtheorem{corollary}[theorem]{Corollary}
\newtheorem{fact}[theorem]{Fact}
\newtheorem{definition}[theorem]{Definition}

\def\abs#1{\left| #1 \right|}
\renewcommand{\norm}[1]{\ensuremath{\left\lVert #1 \right\rVert}}

\newcommand\rea{\mathbb{R}}

\newcommand{\marginlabel}[1]%
{\mbox{}\marginpar{\it{\raggedleft\hspace{0pt}#1}}}

\definecolor{Mygray}{gray}{0.8}

 \ifcsname ifcommentflag\endcsname\else
  \expandafter\let\csname ifcommentflag\expandafter\endcsname
                  \csname iffalse\endcsname
\fi

\ifnum\showauthornotes=1
\newcommand{\Authornote}[2]{{\sf\small\color{red}{[#1: #2]}}}
\newcommand{\Authoredit}[2]{{\sf\small\color{red}{[#1]}\color{blue}{#2}}}
\newcommand{\Authorcomment}[2]{{\sf \small\color{gray}{[#1: #2]}}}
\newcommand{\Authorfnote}[2]{\footnote{\color{red}{#1: #2}}}
\newcommand{\Authorfixme}[1]{\Authornote{#1}{\textbf{??}}}
\newcommand{\Authormarginmark}[1]{\marginpar{\textcolor{red}{\fbox{
#1:!}}}}
\else
\newcommand{\Authornote}[2]{}
\newcommand{\Authoredit}[2]{}
\newcommand{\Authorcomment}[2]{}
\newcommand{\Authorfnote}[2]{}
\newcommand{\Authorfixme}[1]{}
\newcommand{\Authormarginmark}[1]{}
\fi

\def\implies{\Rightarrow}

\newlength{\pgmtab}  
\newcommand{\Tmat}{\mathcal{T}_{\mathrm{mat}}}
\newcommand{\nnz}{\mathrm{nnz}}
\newcommand{\new}{\mathrm{new}}

\let\originalleft\left
\let\originalright\right
\renewcommand{\left}{\mathopen{}\mathclose\bgroup\originalleft}
  \renewcommand{\right}{\aftergroup\egroup\originalright}

\def\defeq{\stackrel{\mathrm{def}}{=}}

\newcommand\dd{\boldsymbol{\mathit{d}}}

\renewcommand\gg{\boldsymbol{\mathit{g}}}
\newcommand\hh{\boldsymbol{\mathit{h}}}

\newcommand\rr{\boldsymbol{\mathit{r}}}
\renewcommand\ss{\boldsymbol{\mathit{s}}}

\newcommand\uu{\boldsymbol{\mathit{u}}}
\newcommand\vv{\boldsymbol{\mathit{v}}}
\newcommand\ww{\boldsymbol{\mathit{w}}}
\newcommand\yy{\boldsymbol{\mathit{y}}}
\newcommand\zz{\boldsymbol{\mathit{z}}}
\newcommand\xx{\boldsymbol{\mathit{x}}}

\newcommand\rrbar{\overline{\boldsymbol{\mathit{r}}}}
\newcommand\xxtil{\widetilde{\boldsymbol{\mathit{x}}}}
\newcommand\yytil{\widetilde{\boldsymbol{\mathit{y}}}}

\renewcommand\AA{\boldsymbol{\mathit{A}}}
\newcommand\BB{\boldsymbol{\mathit{B}}}
\newcommand\CC{\boldsymbol{\mathit{C}}}

\newcommand\EE{\boldsymbol{\mathit{E}}}

\newcommand\II{\boldsymbol{\mathit{I}}}
\newcommand\JJ{\boldsymbol{\mathit{J}}}

\newcommand\NN{\boldsymbol{\mathit{N}}}
\newcommand\MM{\boldsymbol{\mathit{M}}}

\newcommand\RR{\boldsymbol{\mathit{R}}}
\renewcommand\SS{\boldsymbol{\mathit{S}}}
\newcommand\TT{\boldsymbol{\mathit{T}}}
\newcommand\UU{\boldsymbol{\mathit{U}}}

\newcommand\VV{\boldsymbol{\mathit{V}}}

\newcommand\RRbar{\boldsymbol{\overline{\mathit{R}}}}

\newcommand\xxhat{\widehat{\xx}}

\newcommand\Otil{\widetilde{O}}

\newcommand{\wh}{\widehat}
\newcommand{\wt}{\widetilde}
\newcommand{\ov}{\overline}
\newcommand\R{\mathbb{R}}
\DeclareMathOperator*{\E}{{\mathbb{E}}}
\DeclareMathOperator*{\Var}{{\bf {Var}}}
\DeclareMathOperator*{\poly}{{\mathrm{poly}}}
\newcommand{\heavy}{\text{heavy}}

 {
	\begin{enumerate}}{\end{enumerate}}

\ifnum\showdraftbox=1

\else

\fi

%% file: introduction.tex
\section{Introduction}

In this paper, we study the $\ell_{\infty}$-regression problem. Given $\epsilon>0$, a matrix $\CC\in \mathbb{R}^{n\times d}$ and vector $\dd\in \mathbb{R}^{n}$, $d\leq n$, we want to find $\xxtil \in \mathbb{R}^d$ such that,
\begin{equation}\label{eq:Prob}
  \|\CC\xxtil -\dd\|_{\infty}
     \leq (1+\epsilon) \min_{\xx\in \mathbb{R}^d} \|\CC\xx-\dd\|_{\infty}.
\end{equation}

Some of the popular approaches to obtaining fast algorithms for $\ell_{\infty}$-regression 
include using 
multiplicative weight update (MWU) routines~\cite{BN51,AHK12,christiano2011electrical,chin2013runtime,adil2019iterative,ene2019improved,adil2021unifying},
gradient descent~\cite{S13,kelner2014almost} and other ways to optimize a softmax function~\cite{carmon2020acceleration,sidford2018coordinate,adil2021unifying},
and using interior point methods (IPM)~\cite{K84,R88,nesterov1994interior}.
Interior point methods can find a high-accuracy solution, i.e., an $\epsilon$-approximate solution in $\Otil(\sqrt{n}\log (1/\epsilon))$\footnote{We use $\wt{O}(\cdot)$ to hide $\poly \log n$ factors, and we use $\Otil_{\epsilon}(\cdot)$ to additionally hide $\poly(\epsilon^{-1})$ factors.} linear system solves,
whereas most of the other methods are low accuracy solvers, i.e., their running time scales as $\operatorname{poly}(1/\epsilon)$.
Naively using gradient descent or MWU requires 
$O(\sqrt{n} \cdot \text{poly} (1/\epsilon))$ linear system solves.
Multiplicative weight update based approaches can be accelerated via a technique called {\it width reduction} to converge in $O(n^{1/3}\cdot\text{poly} (1/\epsilon))$ linear system solves~\cite{christiano2011electrical,chin2013runtime,adil2019iterative,ene2019improved,adil2021unifying,adil2022fast}.
Several acceleration techniques have also been developed to improve the iteration complexity of other low-accuracy regression algorithms~\cite{MS13,B18,carmon2020acceleration,sidford2018coordinate, adil2021unifying}.

To get an overall fast runtime, apart from improving the iteration complexity, a useful approach is to reduce the per-iteration cost.
This can be done using {\it inverse maintenance}, which reduces the cost via {\it lazy-update} schemes. 
Notions of inverse maintenance appear in the very first interior point methods, \cite{K84,NN89}, but the modern form was introduced by Vaidya \cite{V89}.
There have been many important developments in inverse maintenance algorithms since then, and state-of-the-art algorithms use both linear algebraic data structures and dimensionality reduction routines, such as sketching~\cite{bns19}. The improvements in runtimes of interior point methods including the state-of-the-art algorithms depend heavily on these developments in inverse maintenance routines~\cite{lee2015efficient,cohen2021solving,vdB20,jiang2021faster,lee2021tutorial}.

\subsection{Our Results}
For simplicity, in the discussion of our results and prior work on this problem, we focus on the case $\omega = 2+o(1)$ -- but our full technical theorems give results for all $\omega$.
In the low-accuracy regime of $\eps = 1/\operatorname{polylog(n)}$ 
the state-of-the-art running time for $\ell_\infty$-regression is
$\Otil(n^{2+1/18})$, obtained via the randomized algorithm of \cite{jiang2021faster}, and  $\Otil(n^{2+1/6})$ for deterministic algorithms via \cite{vdB20}.
Both these algorithms in fact obtain high-accuracy solutions, and they use inverse maintenance, but no acceleration.
In this work, we push the running time further in the low-accuracy regime by combining the state-of-the-art inverse
maintenance techniques of these results with new multiplicative weight methods which allow us to perform acceleration, yielding running times of $\Otil(n^{2+1/22.5} \poly(\epsilon^{-1}))$ with randomization and $\Otil(n^{2+1/12}\poly(\epsilon^{-1}))$ without. 

Our first result is a deterministic algorithm that combines acceleration and lazy inverse updates in a novel, more sophisticated way, and achieves a running time of $\Otil(n^{2+1/12}\poly(\epsilon^{-1}))$.
This improves on deterministic state-of-the-art $\Otil(n^{2+1/6}\operatorname{log}(\epsilon^{-1}))$ \cite{vdB20} in the low-accuracy regime.
The key to this result is a new notion of $\ell_3$-\emph{stability} which is tailored to the accelerated MWU.

\begin{theorem}[Informal statement of Theorem~\ref{thm:time_deterministic}]\label{thm:informal_deterministic}
There is a deterministic algorithm that solves Problem~\eqref{eq:Prob} in $\Otil(n^{2+1/12}\poly(\epsilon^{-1}))$ time when $\omega = 2 + o(1)$. This algorithm converges in $\Otil\left(n^{1/3} \poly(\epsilon^{-1})\right)$ iterations.
\end{theorem}
Our main result is our randomized algorithm with running time $\Otil(n^{2+1/22.5}\poly(\epsilon^{-1}))$. 
\begin{theorem}[Informal statement of Theorem~\ref{thm:time_combine}]\label{thm:informal_random}
There is a randomized algorithm that solves Problem~\eqref{eq:Prob} in $\Otil(n^{2+1/22.5}\poly(\epsilon^{-1}))$ time when $\omega = 2 + o(1)$. This algorithm converges in $\Otil\left(n^{1/2.5}\poly(\epsilon^{-1})\right)$ iterations.
\end{theorem}
To obtain this result, we introduce the first MWU which can combine all three key techniques for $\ell_\infty$-regression: (a) acceleration, (b) lazy inverse updates, and (c) sketching.

Thus, we give the optimization approach method which is able to efficiently combine these three key techniques of structured convex optimization. 
This is likely an essential building block toward $n^{2+o(1)}$ optimization for many objectives.
If, some day, acceleration is achieved for linear programming, an equivalent integration will be necessary for optimal algorithms in this context.
Before describing our new approach, we first review existing techniques for fast $\ell_\infty$-regression.

\subsection{Background: The Ingredients of Fast \texorpdfstring{$\ell_\infty$}{}-Regression Methods.}
Both MWUs and IPMs that solve $\ell_\infty$-regression methods rely on a sequence of calls to $\ell_2$-oracles, i.e. a subroutine that solves an $\ell_2$-minimization problem, or equivalently, solves a linear equation.
In order to solve the $\ell_\infty$-regression problem \eqref{eq:Prob}, a standard MWU approach repeatedly solves a sequence of $\ell_2$-oracle problems of the form
\begin{equation}\label{eq:L2Oracle}
    \xx^{(i)} = \arg\min_{\xx\in \mathbb{R}^d} 
    \sum_e r_e^{(i)} (\CC\xx-\dd)_e^2
\end{equation}
where the weights $\{r_e^{(i)}\}$ are chosen by the MWU depending on the magnitude of previous iterates.

\paragraph{Inverse maintenance via stability and robustness.}
The $\ell_2$-oracles of MWUs and IPMs can be implemented by applying the inverse of a matrix, and inverse maintenance can be used to solve the sequence of $\ell_2$-oracle calls faster than simply performing a full matrix inversion or linear equation solve on each call.
Two key phenomena drive inverse maintenance: \emph{stability} and \emph{robustness}.
\emph{Stability} is the property that the inputs to the $\ell_2$-oracle only change slowly. In the MWU case, this means the weights $\{r_e^{(i)}\}$ change slowly.
We say an optimizer is \emph{robust} if it can make progress using answers from $\ell_2$-oracles with somewhat inaccurate inputs.
The combination of stability and robustness is especially powerful.
Together, these properties ensure that we can delay making small coordinate updates to inputs until they build up to a large cumulative update, and that we only get few large cumulative updates, enabling the use of coordinate-sparse update techniques.
This approach of batching together small updates is known as \emph{lazy} inverse updating.
Obtaining further speed-ups using sketching also crucially relies on robustness. 
Because of robustness, we can afford to use sketching to estimate $\xx^{(i)}$, as long as our estimates allow sufficiently accurate updates to the weights $\{r_e^{(i)}\}$.

The IPM of \cite{cohen2021solving} first achieved a running time of $\Otil(n^{2+1/6}+n^\omega)$ by introducing a method with excellent stability and robustness, which in turn allowed them to implement a powerful inverse maintenance approach using lazy updates and sketching. 
Later, \cite{vdB20} showed that the same running time can be obtained deterministically using only lazy updates, and finally \cite{jiang2021faster} gave an improved running time of 
$\Otil(n^{2+1/18}+n^\omega)$ using both lazy updates and sketching.
The approach of \cite{jiang2021faster} can be thought of as a two-level inverse maintenance, and the use of the randomized sketching techniques is crucial for them to efficiently implement the {\it query} operation of this data structure. It remains open if there exists any deterministic IPM that can run faster than $\Otil(n^{2+1/6}+n^\omega)$.

\paragraph{Acceleration via width-reduction.}
In oracle-based optimization, there is a long history of developing \emph{accelerated} methods, which reduce the iteration count compared to more basic approaches.
This can be traced back to accelerated solvers for quadratic objectives \cite{L52,HS52} and first-order acceleration for gradient Lipschitz functions (\cite{N83} and earlier works by Nemirovski).
\textcite{christiano2011electrical} developed an acceleration method for multiplicative weight methods that reduces the iteration count for solving $\ell_\infty$ regression with $\ell_2$-oracles from $\Otil(\sqrt{n} \cdot \text{poly} (1/\epsilon))$ to $\Otil(n^{1/3} \cdot \text{poly} (1/\epsilon))$.
An alternative approach to acceleration for $\ell_\infty$-regression can be obtained via the methods of Monteiro and Svaiter \cite{MS13}, and has also been a major research topic, but is beyond the scope of our discussion.
For simplicity of our remaining discussion, we ignore $\epsilon$ dependencies.
A rough outline of the MWU acceleration approach of \textcite{christiano2011electrical} is as follows:
The MWU solves a sequence of $\ell_2$-oracle problems returning iterates 
$\xx^{(i)}$.
If we scale the problem so that $\|\CC\xx^{\star}-\dd\|_{\infty} \leq 1$,
then weights ensure that (a) in each iteration, 
$\|\CC\xx^{(i)}-\dd\|_{\infty} \lesssim \sqrt{n}$
and (b) after $T = \Otil(\sqrt{n})$ iterations, $\xxtil = \frac{1}{T}\sum_i \xx^{(i)}$ has $\|\CC\xxtil-\dd\|_{\infty}\leq 1+\epsilon$.
\cite{christiano2011electrical} made an important modification: if in some iteration we have $\|\CC\xx^{(i)}-\dd\|_{\infty}\geq \rho \approx n^{1/3}$, then instead of using $\xx^{(i)}$, we will adjust the weights $\{r_e^{(i)}\}$ in order to reduce the value of  $\|\CC\xx^{(i')}-\dd\|_{\infty}$ for future iterates $\xx^{(i')}$.
Using this method, an approximately optimal $\xxtil = \frac{1}{T}\sum_i \xx^{(i)}$ can be found in $T = \Otil(n^{1/3})$ iterations.
The parameter $\rho$ measures the $\ell_\infty$-norm $\|\CC\xx^{(i)}-\dd\|_{\infty}$ of each iterate, sometimes known as the \emph{width}, and the weight-adjustment steps of Christiano et al.~are hence known as \emph{width reduction steps}.
When the oracle width can be reduced in this way, we will say our method is \emph{width-reducible}.
This acceleration has never been developed for $\ell_{\infty}$-regression in the high-accuracy regime (i.e. running times that scale as $\operatorname{polylog}(1/\epsilon)$), and whether this is possible is one of the major open questions in convex optimization.

\paragraph{Weight monotonicity in MWUs: an obstacle to sketching.}
Many MWU methods are designed to have an important property, which we call \emph{weight monotonicity}.
Concretely, in \cite{christiano2011electrical} and many other MWUs, the oracle weights $\{\rr_e^{(i)}\}$ are only growing.
This often simplifies analyses greatly, and helps establish other properties including stability, robustness, and width-reducibility.
Referring back to our oracle queries introduced above in \eqref{eq:L2Oracle}, let us define
$\xxtil^{(i)} = \frac{1}{T}\sum_{j \leq i} \xx^{(j)}$.
Weight monotonicity arises because we choose the weights based on an overestimate of 
$|(\CC\xxtil^{(i)}-\dd)_e|$ given by $\gamma_i = \frac{1}{T}\sum_{j \leq i} |(\CC\xx^{(j)}-\dd)_e|$.
In particular, choosing $\rr_e^{(i)}= \exp(\alpha \gamma_i)$ for some scaling factor $\alpha$ 
will ensure the weights only grow.
As we will discuss later, \emph{weight monotonicity} seems inherently incompatible with sketching, and thus we will need to develop a non-monotone MWU.
Prior work by Madry \cite{M13,madry2016computing} introduced non-monotone weights in a highly specialized IPM for unit-capacity maximum flow.
This IPM of Madry has MWU-like properties and allows for some acceleration.
The method has other drawbacks including low stability and robustness, but nonetheless inspired some of our design choices.

\paragraph{Prior inverse maintenance with acceleration.}
We are aware of a single prior work which combined lazy inverse updates with an accelerated MWU to obtain a running time of 
$\Otil(n^{2+1/3}+n^\omega)$ for $\ell_{p}$-regression \cite{adil2019iterative}.
This approach is relatively naive, falling short of 
the $\Otil(n^{2+1/6}+n^\omega)$ running time which can be achieved using only lazy inverse updates.

\subsection{Discussion of Techniques}
The crucial algorithmic techniques we rely on for speeding up $\ell_\infty$-regression are 
(a) acceleration, (b) lazy inverse updates, and (c) sketching.
We can view each of these techniques as being enabled by different properties of the overall optimization approach.
Our approach to acceleration is enabled by \emph{width-reducibility}, while lazy updates require \emph{stability} and \emph{robustness}, and finally sketching requires \emph{robustness} and \emph{non-monotonicity}.
This means we need to develop an MWU which simultaneously exhibits all these properties, i.e. it must be 
stable, robust, non-monotone, and width-reducible. 
In Figure~\ref{fig:techniquesAndRequirements}, we summarize how our algorithmic techniques impose different requirements on our optimization approach.
Again, for simplicity, in the remaining discussion of our results and prior work on this problem, we focus on the case $\omega = 2+o(1)$.

We first discuss how to combine stability, robustness, and width-reducibility in a \emph{monotone} MWU, which leads to a comparatively simple, deterministic algorithm using acceleration and lazy inverse updates, but no sketching.

\begin{figure}
    \centering
\resizebox{.7\textwidth}{!}{
\begin{tikzpicture}[
    box/.style={rectangle, rounded corners, draw, minimum width=3cm, minimum height=1cm, align=center},
    arrow/.style={draw, -{Stealth[scale=1.2]}, thick}
    ]

\node[box] (top) {acceleration};
\node[box, below=2cm of top] (middle) {optimizer};
\node[box, below left=1.5cm and 2.5cm of middle] (left) {lazy updates};
\node[box, below right=1.5cm and 2.5cm of middle] (right) {sketching};

\draw[arrow] (top) -- node[pos=0.5, right=0.5cm, align=left]
{
width-reducibility
} (middle) ;
\draw[arrow] (middle) -- (top);
\draw[arrow] (middle) -- (left);
\draw[arrow] (left) -- node[pos=0.7, left=0.5cm, align=left]
{
stability \& robustness
} (middle) ;
\draw[arrow] (right) -- node[pos=0.7, right=0.5cm, align=left]
{
robustness \& non-monotonicty
} (middle) ;
\draw[arrow] (middle) -- (right); 

\end{tikzpicture}
}
\caption{Algorithmic techniques and their requirements on our optimizer.}
\label{fig:techniquesAndRequirements}
\end{figure}
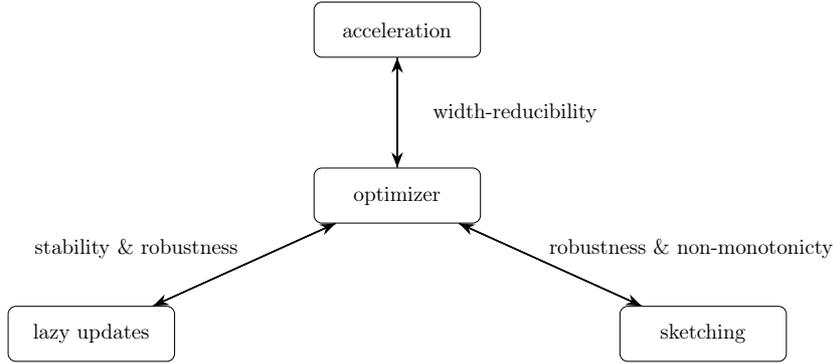

\paragraph{Stability and robustness of a monotone, width-reducible MWU.}
\cite{adil2019iterative} showed how to obtain stability, robustness, and width-reducibility together, with a monotone MWU.
However, this work only established a weak notion of stability 
and hence comparatively slow running time of $\Otil(n^{2+1/3})$.
In contrast, one can show that by directly using stability and robustness in a monotone accelerated MWU, 
one can adapt the data structure approach of \cite{vdB20} to achieve a running time of $\Otil(n^{2+1/9})$, yielding a faster MWU.

Our first result Theorem~\ref{thm:time_deterministic} 
is based on the observation that monotone MWU also enables a new, stronger notion of stability, which we call $\ell_3$-stability.
This allows us to further reduce the number of lazy updates we make and lets us achieve a deterministic running time of $\Otil(n^{2+1/12})$.

\paragraph{Non-monotone MWU - a key ingredient for sketching.}
As we described above, it is relatively easy to improve the running time of low-accuracy $\ell_\infty$-regression among deterministic algorithms, by designing  a monotone, robust, width-reducible MWU with a novel $\ell_3$-stability.

To further accelerate the algorithm by using a two-level inverse maintenance data structure, we need to use randomized sketching techniques to efficiently implement the query operation, which is required in every iteration of the MWU algorithm. 
Unfortunately, weight monotonicity is in conflict with sketching, because monotonicity arises from ignoring cancellations in $(\CC\xxtil^{(i)}-\dd)_e$ between different iterations.\footnote{Recall that the final output of our MWU is the last averaged iterate $\xxtil^{(T)}$.}
In contrast, when using sketching, we want to crucially rely on cancellation between different iterations, as we sometimes overestimate $(\CC\xx^{(i)}-\dd)_e$ and sometimes underestimate it, but get it right on average.
Because of this, we design an MWU with non-monotone weights.
This in turn makes width-reducibility, robustness, and stability much harder to obtain.

To allow us to work with non-monotone weights and still obtain acceleration, we introduce a more delicate width-reduction scheme, inspired by \cite{madry2016computing}.
We also provide a tighter analysis of the sketching technique (it was named coordinate-wise embedding by \cite{lsz19,jiang2021faster}) that upper bounds its total noise across different iterations using martingale concentration inequalities. This tighter analysis is necessary to control the overall error introduced by the sketching technique in our MWU algorithm. We believe this tighter analysis could also provide a simpler analysis for the IPM results of \cite{cohen2021solving,jiang2021faster}.

This new width-reduction approach in turn also requires us to estimate an $\ell_3$-norm associated with each iterate $\xx^{(i)}$, and to do this quickly, we need to employ new sketching tools.
To implement this approach, we also need an additional heavy-hitter sketch that allows us to identify which weights to adjust during width reduction.

\paragraph{Stability and robustness of a non-monotone, width-reducible MWU.}
Stability and robustness are crucial when we want to use lazy updates and sketching for inverse maintenance.
Standard techniques for acceleration by width-reduction are unstable in the context of \emph{non-monotone} MWU.
Thus, to combine stability, width-reduction, and non-monotonicity, we have to further change our width-reduction strategy. 

A central challenge is that width-reducibility is inherently in tension with the other properties.
To simultaneously achieve stability and width-reducibility, we introduce a new and rather different approach to width-reduction, which we call \emph{stable width-reduction}.
This approach is more conservative than existing methods, and uses smaller width-reduction steps to achieve stability. 

Combining width-reducibility with robustness is also difficult. Width-reduction relies on identifying too-large entries of the oracle outputs and making adjustments to the corresponding weights.
But, robustness requires us to operate with inaccurate weights. We want to allow for weights that are inaccurate up to a factor $(1\pm 1/\operatorname{polylog}(n))$, and this is enough to completely change which oracle outputs are too large.
In fact, we do not achieve general robust, but instead show that our method is robustness to (1) the errors induced by our specific lazy update scheme and (2) the errors induced by sketching.

\paragraph{Future perspectives.}
It remains open to design any algorithm for low-accuracy $\ell_{\infty}$ regression beyond $\Otil( n^{2+1/22.5})$ when $\omega = 2+o(1)$. We remark that if it were possible to use $\ell_3$-stability with the two-level data structure and an algorithm that converges in $n^{1/3}$ iterations, then we would achieve a runtime of $\Otil(n^{2+1/48})$. However, the current techniques for inverse maintenance and acceleration are not sufficient to achieve $\Otil(n^{2+o(1)} + n^\omega)$, which we believe would require substantially new techniques. On the other hand, even obtaining slight improvements in the runtime would require more robust acceleration and inverse maintenance frameworks which would be of independent interest.

In this paper, we analyze our algorithms in the RealRAM model. Establishing a similar analysis in finite precision arithmetic is an interesting open problem. Inverse maintenance-based IPM with finite precision arithmetic was studied by \cite{ghadiri2023bit}.

We have demonstrated that acceleration techniques for MWU can be efficiently combined with inverse maintenance methods. 
For linear programming, no similar acceleration techniques exist and it is a major open problem to design these or rule out the possibility in various computational models. 
If acceleration can be achieved for linear programming, deploying it in conjunction with inverse maintenance will likely require techniques similar to those we introduce in this work.

%% file: tech_overview.tex
\section{Technical Overview}

\subsection{Deterministic MWU Algorithm via One-Level Inverse Maintenance}
MWU methods reduce $\ell_\infty$-regression problems to a sequence of $\ell_2$-minimization problems, which can be solved by solving systems of linear equations -- or equivalently, applying the inverse of some matrix. More concretely, an MWU for finding approximate solutions to $\min_{\xx\in \mathbb{R}^d} \|\CC\xx-\dd\|_{\infty}$ requires us to repeatedly solve problems of the form 
\[
\min_{\Delta\in \mathbb{R}^d}\sum_e \rr^{(i)}_e(\CC\Delta-\dd)_e^2
\]
across iterations $i = 1,\ldots,T$.
The exact solution to these minimization problems is given by
\[
\Delta^{(i)} = (\CC^{\top}\RR^{(i)}\CC)^{-1}\CC^{\top}\RR^{(i)}\dd.
\]
The multiplicative weight update method iteratively updates the weights using $\Delta^{(i)}$ and ``penalizes'' the coordinates $e$ that have large $|\CC \Delta^{(i)} - \dd|_e$ by increasing their weights $\rr_e^{(i+1)}$ in the next iteration. In the end the method outputs $\xx =  \sum_{i=1}^T \Delta^{(i)} / T$ as the approximate $\ell_\infty$ minimizer. 

The cost of each iteration is dominated by the time required to solve the corresponding system of linear equations for $\Delta^{(i)}$ -- or equivalently, applying the inverse of some matrix.
If solving this sequence of systems of linear equations can be done faster than naively solving each system separately, then we can speed up the cost per iteration of the MWU algorithm, and hence make the algorithm faster. A similar problem of solving a sequence of systems of linear equations was studied for the IPM algorithms \cite{cohen2021solving,vdB20,jiang2021faster}, and they achieved speed-ups by using lazy updates with \emph{inverse maintenance} data structures. They could use lazy updates because the IPM algorithm satisfies a stability guarantee and a robustness guarantee. More precisely, (1) IPMs satisfy an $\ell_2$-stability guarantee that the $\ell_2$-norm of the relative changes between two iterations is bounded, i.e., $\|\frac{\rr^{(i+1)} - \rr^{(i)}}{\rr^{(i)}}\|_2^2 \leq O(1)$. (2) IPMs are still correct if the system of linear equations is solved with coordinate-wise approximate weights $\ov{\rr} \approx_{\delta} \rr$ for some $\delta>0$.

As it turns out, the \emph{monotone} MWU algorithm is also inherently stable and robust, even with acceleration. We can therefore use coordinate-wise approximate weights $\ov{\rr}^{(i)} \approx_{\delta} \rr^{(i)}$ in each iteration, and only update $\ov{\rr}^{(i)}_e$ when it differs from $\rr^{(i)}_e$ by more than $\delta$. This ensures that the approximate weights $\ov{\rr}^{(i)}$ undergoes low-rank updates. We present a robust version of the known {\it accelerated} multiplicative weights update method for $\ell_{\infty}$-regression from~\cite{christiano2011electrical,chin2013runtime} below, where when solving the system of linear equations for $\Delta^{(i)}$ we use the approximate weights $\ov{\rr}^{(i)}$.
\begin{algorithm}
\caption{Monotone Width Reduced MWU Algorithm}
\label{alg:MWU}
 \begin{algorithmic}[1]
 \Procedure{MWU-Solver}{$\epsilon, \CC,\dd$}
 \State $\ww^{(0,0)} \leftarrow 1_n, \quad \xx^{(0)} \leftarrow 0_d$
 \State $\tau \leftarrow \Theta\left(\frac{n^{\frac{1}{3}}}{\epsilon^{\frac{1}{3}}} \log\frac{n}{\Psi_0}\right), \alpha \leftarrow \Theta\left(n^{-\frac{1}{2} + \eta}\epsilon^{\frac{1}{3}} \left(\log\frac{n}{\Psi_0}\right)^{-1}\right)$, $\eta \gets \frac{1}{6}$
\State $T \leftarrow \alpha^{-1}\epsilon^{-2}\log n$ 
\State $i \leftarrow 0, k \leftarrow 0$
\While{$i < T$} 
\State $\rr^{(i,k)}_e \leftarrow \ww_e^{(i,k)} + \frac{\epsilon}{n} \|\ww^{(i,k)}\|_1$
\State $\rrbar^{(i,k)}\gets$ {\sc SelectVector}$(\rr^{(i,k)}, i+k, \delta)$ \Comment{$\rrbar \approx_{\delta}\rr$} \label{algline:MWU_select_vector}
\State $\Delta^{(i,k)} \leftarrow \arg\min_{\Delta \in \rea^d}\sum_e \rrbar^{(i,k)}_e(\CC\Delta-\dd)_e^2$ \Comment{$\Delta = (\CC^{\top}\RRbar^{(i,k)}\CC)^{-1} \CC^{\top} \RRbar^{(i,k)} \dd$}\label{algline:linsysMon}
\If{$\norm{\CC\Delta^{(i,k)}-\dd}_{\infty} \leq \tau$}\Comment{primal step}   \label{algline:CheckWidth}
\State $\ww^{(i+1,k)} \leftarrow \ww^{(i,k)}(1 + \epsilon\alpha |\CC\Delta^{(i,k)}-\dd|)$ \label{algline:UpdateWPrimal}
\State $\xx^{(i+1)} \leftarrow \xx^{(i)} +  \Delta^{(i,k)}$
\State $ i \leftarrow i+1$ 
\Else\label{lin:WidthReduceEdge}
\State For all coordinates $e$ with $|\CC\Delta^{(i,k)}-\dd|_e \geq \tau $\Comment{width reduction step}
\State $\quad \quad \ww_e^{(i,k+1)} \leftarrow (1+\epsilon)\ww^{(i,k)}_e + \frac{\epsilon^2}{n}\|\ww^{(i,k)}\|_1$ \label{algline:UpdateWWidth} 
\State $\quad \quad k \leftarrow k+1$ \label{algline:endAcc}
\EndIf
\EndWhile
 \State\Return $\xxhat = \frac{\xx^{(T)}}{T}$
 \EndProcedure 
\end{algorithmic}
\end{algorithm}
\begin{theorem}[\cite{chin2013runtime}]\label{thm:InfRegMainMonotone}
Let $0<\epsilon<1/2$ and $0\leq\delta\leq \epsilon/6$. Algorithm~\ref{alg:MWU} returns $\xxhat$ such that $\|\CC\xxhat-\dd\|_{\infty}\leq 1 +O(\epsilon)$ in $\widetilde{O}(n^{1/3}\epsilon^{-7/3})$ iterations. Each iteration solves a linear system as specified in Line~\ref{algline:linsysMon} of the algorithm.
\end{theorem}
In fact, we can prove that this algorithm satisfies an even stronger stability guarantee -- a quantitatively strong type of $\ell_3$-stability, namely
\[
\sum_{i=1}^T \left\|\frac{\rr^{(i+1)} - \rr^{(i)}}{\rr^{(i)}}\right\|_3^3 \leq O(n^{1/3}).
\]
The $\ell_3$-stability guarantee allows for the following lazy-update scheme: for every $\ell$, in every $2^{\ell}$ iterations perform an update of size $O(2^{3 \ell})$ to $\ov{\rr}^{(i)}$. 

Together with the one-level inverse maintenance of \textcite{bns19}, this improves upon the previous best deterministic algorithm for low-accuracy $\ell_{\infty}$ regression that runs in $O(n^{\omega} + n^{2 + 1/6})$. We present a simplified version of the data structure below, and the formal version tailored to our application is in Section~\ref{sec:InverseMaintenance}. 
\begin{lemma}[One-level inverse maintenance, (Informal) Theorem 4.1 of \cite{bns19}]
\label{lem:InverseMaintenanceOneLevelShort}
There is a data structure that supports the following two operations to maintain the inverse of an $n \times n$ matrix $M$:
\begin{itemize}
\item {\bf Reset}: Reset $M^{-1}$ to $(M+\Delta)^{-1}$, where $\Delta$ has $k_0$ non-zero entries. This operation can be done in $O(\Tmat(n, n, k_0))$ time.\footnote{$\Tmat(n,r,m)$ denotes the time complexity of multiplying an $n \times r$ matrix with an $r \times m$ matrix.}
\item {\bf Query}: Output the vector $(M+\Delta)^{-1} \cdot v$ using the maintained $M^{-1}$ and $M^{-1} v$, where $\Delta$ has at most $n^{a_0}$ non-zero entries. This operation can be done in $O(n^{\omega a_0} + n^{1 + a_0})$ time.
\end{itemize}
\end{lemma}

\paragraph{Runtime when $\omega = 2$.}
For simplicity, we only show the runtime of our algorithm when $\omega = 2$ in this section and omit polylogarithmic factors. Let us choose the parameter $a_0 = 3/4$, so that we perform a reset operation whenever we accumulate more than $n^{a_0} = n^{3/4}$ updates to $\ov{\rr}$. From our low-rank update scheme under the $\ell_3$ stability guarantee, this only happens in every $n^{1/4}$ iterations. So we perform a reset operation with cost $O(n^2)$ (since $\omega = 2$) in every $O(n^{1/4})$ iterations, and over the total $O(n^{1/3})$ iterations this gives a total reset time of $O(n^{2 - 1/4} \cdot n^{1/3}) = O(n^{2 + 1/12})$.

We perform a query operation in every iteration with cost $O(n^{2 a_0} + n^{1 + a_0}) = O(n^{1 + 3/4})$. Over all $O(n^{1/3})$ iterations this gives a total query time of $O(n^{1 + 3/4} \cdot n^{1/3}) = O(n^{2 + 1/12})$. Therefore, the total runtime is the sum of the reset time and the query time, which is $O(n^{2 + 1/12})$ as claimed in Theorem~\ref{thm:informal_deterministic}.

\subsection{Randomized MWU Algorithm via Two-Level Inverse Maintenance}
To further improve the runtime of the algorithm, we will use the following, more efficient two-level inverse maintenance data structure.
\begin{lemma}[Two-level inverse maintenance, (Informal) Theorem 4.2 of \cite{bns19}]
\label{lem:InverseMaintenanceTwoLevelShort}
There is a data structure that supports the following three operations to explicitly maintain the inverse of an $n \times n$ matrix $M$. The algorithm achieves the goal via explicitly maintaining the inverse of an $n \times n$ matrix $M_0$ and implicitly maintaining the inverse of another $n \times n$ matrix $M_1$ that differs from $M_0$ on at most $n^{a_0}$ entries, and the true matrix $M$ always differ from $M_1$ on at most $n^{a_1}$ entries where $a_1 \leq a_0$:
\begin{itemize}
\item {\bf Reset}: Reset $M_0^{-1}$ to $(M_0+\Delta_0)^{-1}$, where $\Delta_0$ has $k_0$ non-zero entries. This operation can be done in $\Tmat(n, n, k_0)$ time.
\item {\bf Partial reset}: Implicitly reset $M_1^{-1}$ to $(M_1+\Delta_1)^{-1}$, where $\Delta_1$ has $k_1$ non-zero entries. This operation can be done in $\Tmat(n, n^{a_0}, k_1)$ time.
\item {\bf Query}: Output $\ell$ entries of the vector $M^{-1} \cdot v$ using the maintained $M_0^{-1}$, $M_1^{-1}$ (implicitly). 
This operation can be done in $\Tmat(n^{a_0}, n^{a_1}, \max\{n^{a_1}, \ell\})$ time.
\end{itemize}
\end{lemma}
The total runtime of the above data structure is the sum of its reset, partial reset, and query times. Let us now compare the query times of this two-level data structure with the one-level version. Observe that, the query time of the one-level data structure is $n^{1+a_0}$ and that of the two-level data structure is better than $ n^{1+a_0}$ only if $\ell = o(n)$. In other words, we get an improvement via the two-level data structure only if we have an algorithm that does not require querying the entire maintained vector $M^{-1}v$.

So far, such an improvement via the two-level data structure has only been utilized, although in a complicated way, in the work of \textcite{jiang2021faster} where they give a fast algorithm for linear programming by using the data structure within the robust interior point method framework and querying a {\it sketch} of the vector at every iteration. It is still an open problem if one can achieve their runtime of $\approx n^{2+1/18}$ via a deterministic algorithm and it is conjectured that improving the runtime either requires an improved data structure or, a more sophisticated ``dimension reduction technique'' to work with the algorithm.

\paragraph{Sketching and non-monotone MWU.} 
Similar to~\cite{jiang2021faster}, in our work we also query a sketch of the maintained vector in every iteration. More precisely, in each iteration we use a random matrix $\SS \in \R^{n^{1/2+\eta} \times n}$ where $\eta$ is the acceleration that we get, i.e., the total number of iterations is $O(n^{1/2-\eta})$, and we compute an approximate step
$\SS^{\top} \cdot \SS \cdot (\CC^{\top} \Delta^{(i,k)} - \dd)$. Using the coordinate-wise embedding guarantee of the random matrix $\SS$, we can ensure that for each coordinate we have
\[
\Big( \SS^{\top} \SS (\CC^{\top} \Delta^{(i,k)} - \dd) \Big)_e \approx (\CC^{\top} \Delta^{(i,k)} - \dd)_e.
\]
We now require to change Line~\ref{algline:UpdateWPrimal} of Algorithm~\ref{alg:MWU} to update the weights by 
\[
\ww^{(i+1,k)} \leftarrow \ww^{(i,k)} \Big(1 + \epsilon\alpha \cdot \SS^{\top} \SS (\CC\Delta^{(i,k)}-\dd) \Big).
\]
Note that we lose monotonicity of the weights with this new primal step. We have to use this non-monotone update because the absolute values $|\SS^{\top} \SS (\CC^{\top} \Delta^{(i,k)} - \dd)|$ would result in an error that is around the standard deviation of the estimator in every update of $\ww^{(i,k)}$'s, and this would add up over iterates. Since the entire analysis of the MWU methods depends on tracking potentials which are functions of the weights, we would incur a large error. To circumvent this issue we require a version of the MWU method where the weights are not updated monotonically, and the random noise introduced by the sketching matrix $\SS$ can cancel out with each other across different coordinates $e$ and across different iterations $i$.

Monotonicity is crucial in accelerating MWU methods and it is non-trivial to achieve accelerated rates without it. A few works in graph algorithms have been successful in obtaining accelerated rates without monotonicity~\cite{madry2016computing,liu2020faster} for specific algorithms. In this paper, we extend the algorithm of \textcite{madry2016computing} to regression and obtain an algorithm with non-monotone updates that also converges in $n^{1/3}$ iterations and is robust (Refer to Appendix~\ref{sec:NonMonMWU} for the complete algorithm and analysis). 

Interior point methods directly control the solution quality of the last iterate.
In contrast, MWU algorithms only 
measure the quality of the average of the primal iterates $\Delta^{(i,k)}$. 
As a result, our bound on the final solution requires a new MWU analysis that can handle cancellations between iterates of the errors arising from using sketching. 
We achieve this by developing a tighter analysis that upper bounds the sum of the sketching error over multiple iterations:
\[
\sum_{i=0}^t \left( \Big(\SS^{\top} \SS (\CC\Delta^{(i)}-\dd)\Big)_e - (\CC\Delta^{(i)}-\dd)_e\right) \lesssim \frac{\sqrt{n t}}{\sqrt{b}}.
\]
We prove this bound using Freedman's concentration bound for martingales. We also believe this tighter analysis can simplify the sketching analysis for the previous IPM papers \cite{cohen2021solving, lsz19, jiang2021faster}.


\paragraph{Stability and robustness of non-monotone MWU.}
The non-monotone MWU with standard width reduction steps is neither stable nor satisfies a low-rank update per iteration. We propose a new width reduction step that satisfies a low-rank update scheme which is sufficient for our data structure. Our steps, however, do not satisfy $\ell_2$ stability, which is a sufficient condition for the low-rank update scheme. Instead of increasing all weights by a factor of $(1+\epsilon)$ as in Line~\ref{algline:UpdateWWidth} of Algorithm~\ref{alg:MWU}, our new width reduction step increases a carefully selected set of weights. As a result, we can ensure that whenever we increase a large set of weights, we also increase the potential by a lot, so this event doesn't happen very often. This helps ensure that weight updates from width-reduction steps occur on a similar ``schedule'' to weight updates from our primal update steps, and it allows us to efficiently handle both in the inverse maintenance data structure (Refer to Algorithm~\ref{alg:non_monotone_accel_robust} and Appendix~\ref{sec:Appendix_non_monotone_robust} for the complete algorithm and analysis).
To efficiently find the coordinates $e$ to perform width reduction on, we use an additional {\it heavy-hitter} data structure to identify these $\Delta_e$ exactly. We can only afford to find $n^{1/2+\eta}$ such coordinates in each iteration. This restriction on the number of coordinates restricts us to set $\eta$ to be $1/10$, and our final iteration complexity is $n^{1/2-\eta} = n^{2/5}$ instead of $n^{1/3}$. The non-monotone algorithm also requires estimating a weighted $\ell_3$-norm of $\wh{\Delta}^{(i,k)}$'s for which we use an additional sketch from \cite{wz13}.

Unlike the width reduction steps, the primal steps are stable, and they satisfy the $\ell_2$ stability,
\[
\left\|\frac{\rr^{(i+1)} - \rr^{(i)}}{\rr^{(i)}}\right\|_2^2 \leq O(n^{2 \eta}).
\]
Given the $\ell_2$ stability guarantee, we again use coordinate-wise approximate weights $\ov{\rr}^{(i)} \approx_{\delta} \rr^{(i)}$ in each primal step, and only update $\ov{\rr}^{(i)}_e$ to be $\rr^{(i)}_e$ if it differs from $\rr^{(i)}_e$ by more than $\delta$. This again guarantees a low-rank update scheme for the primal steps: for every $\ell$, in every $2^{\ell}$ iterations we only perform an update of size $O(2^{2 \ell} \cdot n^{2 \eta})$ to $\ov{\rr}^{(i)}$. 

It is non-trivial to show that the accelerated non-monotone MWU is robust under such coordinate-wise approximations to the weights. This is because we do not update the weights in every primal step, and we lazily update them in future iterations. We use an amortization argument to show that we can still gain enough changes in the required potentials even when we defer some updates to the future. However, this means our accelerated non-monotone MWU is only robust under the specific approximate weights $\ov{\rr}^{(i)}_e$ that are updated to be $\rr^{(i)}_e$ whenever it differs too much from $\rr^{(i)}_e$. We cannot guarantee robustness if in every iteration we choose an arbitrary coordinate-wise approximation unless we consider the unaccelerated algorithm, which was guaranteed in the IPM algorithms.

\paragraph{Runtime when $\omega = 2$.}
 Finally, we sketch the time complexity of our non-monotone MWU algorithm using sketching when $\omega = 2$.
 For simplicity, we omit polylogarithmic factors. 
 Using the two-level inverse maintenance data structure of Lemma~\ref{lem:InverseMaintenanceTwoLevelShort}, we perform a reset operation whenever we accumulate more than $n^{a_0}$ updates to $\ov{\rr}$, and by our low-rank update scheme under the $\ell_2$ stability guarantee, this only happens in every $n^{a_0/2-\eta}$ iterations. Similarly, we perform a partial reset operation whenever we accumulate more than $n^{a_1}$ updates to $\ov{\rr}$, and this only happens in every $n^{a_1/2-\eta}$ iterations. Finally, note that our query time is bounded by $n^{a_0+a_1}$ since we always ensure that we query for at most $\ell = O(n^{1/2+\eta})$ coordinates in each iteration. So our total runtime over $T = n^{1/2-\eta}$ iterations is
\begin{align*}
\underbrace{T \cdot \frac{n^2}{n^{a_0/2-\eta}}}_{\text{reset}} + \underbrace{T \cdot \frac{n^{1 + a_0}}{n^{a_1/2-\eta}}}_{\text{partial reset}} + \underbrace{T \cdot n^{a_0+a_1}}_{\text{reset}} 
= &~ n^{2.5 - a_0/2} + n^{1.5 + a_0 - a_1/2} + n^{0.5-\eta+a_0+a_1}.
\end{align*}
Choosing the parameters $a_0 = 1 - \frac{1 - 2 \eta}{9}$ and $a_1 = 1 - \frac{1 - 2 \eta}{3}$, we have that the total runtime is bounded by $O(n^{2+1/18-\eta/9})$. Since we achieve an acceleration of $\eta = 1/10$ and $n^{1/2-\eta} = n^{2/5}$ iterations, this gives the claimed $O(n^{2+1/18-\eta/9}) = O(n^{2+1/22.5})$ time complexity of Theorem~\ref{thm:informal_random}.

%% file: FastMWU.tex
\section{Fast Width-Reduced MWU Algorithms}\label{sec:MWUDirect}
In this section, we present the formal guarantees of our multiplicative weight update routines: a deterministic MWU algorithm with monotone weights (Algorithm~\ref{alg:MWU}) that is used in Theorem~\ref{thm:informal_deterministic}, and a randomized MWU algorithm with non-monotone weights and stable and robust steps  (Algorithm~\ref{alg:non_monotone_accel_robust}) that is used in Theorem~\ref{thm:informal_random}.

\subsection{Lazy Update Procedure}
We first present the \textsc{SelectVector} algorithm (Algorithm~\ref{alg:select_vector}) from \cite{lee2021tutorial} that computes a coordinate-wise approximate vector $\ov{\rr}$ of $\rr$ such that $\ov{\rr}$ undergoes small updates. 

We remark that the only difference between our algorithm and that of \cite{lee2021tutorial} is in Line~\ref{algline:update_S} where we only include a coordinate $e$ in $S$ if $\ww_e$ is not being updated by a width reduction step between primal iterations $i-2^{\ell}$ and $i$. This is due to a minor technicality of dealing with the two kinds of steps, primal and width reduction, in Algorithm~\ref{alg:non_monotone_accel_stab} and \ref{alg:non_monotone_accel_robust}. In all our algorithms, if we toggle a coordinate $e$ in a width reduction step, then we always update the ``lazy'' approximate vector $\rrbar_e$ to be the same as $\rr_e$, so the guarantees of the \textsc{SelectVector} algorithm still hold under this change in  Line~\ref{algline:update_S}.
\begin{algorithm}[H]
\caption{Compute a coordinate-wise approximate vector that undergoes small updates \cite{lee2021tutorial}}\label{alg:select_vector}
\begin{algorithmic}[1]
\Procedure{SelectVector}{$\rr^{(i)}, i, \delta$}
\State \Comment{This procedure stores all previous $\rr^{(0)}, \cdots, \rr^{(i-1)}$, and the $\ov{\rr}$ in the previous iteration}
\If{$i = 0$}
\State \Return $\ov{\rr} \gets \rr^{(0)}$
\EndIf
\State $S \leftarrow \emptyset$
\For{$\ell = 0, 1, \cdots, \log n$}
\If{$i \equiv 0 \mod 2^{\ell}$}
\If{$\ell = \log n$}
\State $S \gets [n]$ 
\Else
\State $S \gets S \cup \{e : |\ln(\frac{\rr_e^{(i)}}{\rr_e^{(i-2^\ell)}})| \geq \frac{\delta}{2 \log n} \text{ and } \textsc{LastWidth}(i,e) \leq i - 2^{\ell}\}$ \label{algline:update_S} 
\State \Comment{$\textsc{LastWidth}(i,e) \leq i$ is the last primal step during which a width reduction step updates $\ww_e$}
\EndIf
\EndIf
\EndFor
\State $\ov{\rr}_e \leftarrow \rr^{(i)}_e$ for all $e \in S$
\State \Return $\ov{\rr}$
\EndProcedure
\end{algorithmic}
\end{algorithm}

\subsection{Monotone Multiplicative Weights Update Algorithm}

We have already presented the convergence guarantees of Algorithm~\ref{alg:MWU} in Theorem~\ref{thm:InfRegMainMonotone}. We now add the stability guarantees that we use to prove the guarantees of our fast deterministic algorithm. The analysis of the algorithm is in Appendix~\ref{sec:MWUMonotone} and the stability guarantees are in Appendix~\ref{sec:stability_algo_1}.

\begin{restatable}[Stability bound of $\ell_3$ norm over all primal iterations]{lemma}{MonotonePrimal}
\label{lem:StabilityL3AllIterResistance}
Let $k_i$ denote the number of width reduction steps taken by the algorithm when the $i^{th}$ primal step is being executed. Then over all $T$ primal steps of Algorithm~\ref{alg:MWU}, we have
\[
\sum_{i=0}^{T-1} \sum_{e\in S_i} \left(\frac{\rr_e^{\left(i + 1,k_i\right)}-\rr_e^{\left(i,k_i \right)}}{\rr_e^{\left(i,k_i \right)}} \right)^3 \leq \Otil(\alpha^2 n) = \Otil(n^{1/3}\epsilon^{2/3}).
\]
Here $S_i$ is the set of coordinates $e$ at primal iteration $i$ such that $\rr_e^{(i+1,k_i)}\geq \rr_e^{(i,k_i)}(1+3\epsilon\alpha)$\footnote{We note that it is sufficient to consider these sets $S_i$'s since any change that is smaller than the ones captured here can happen only $\Otil(1)$ times.}.
\end{restatable}

\begin{restatable}[Stability bound of $\ell_3$ norm over all width reduction iterations]{lemma}{MonotoneWidth}
\label{lem:StabilityL3AllWidth}
Let $i_k$ denote the number of primal steps taken before the execution of the $k^{th}$ width reduction step. Then, over all $K$ width reduction steps of Algorithm~\ref{alg:MWU}, we have
\[
\sum_{k=0}^{K-1} \left(\frac{\rr_e^{\left(i_k, k+1\right)}-\rr_e^{\left(i_k, k \right)}}{\rr_e^{\left(i_k, k \right)}} \right)^3 \leq \Otil(n^{1/3}).
\]
\end{restatable}

\subsection{Algorithm with Non-Monotone Weights, Stability and Robustness}

We now give our main algorithm which can be used with our two-level inverse maintenance data structure. Algorithm~\ref{alg:non_monotone_accel_robust} updates the weights in a non-monotone way, and additionally has stable primal and width reduction steps. It is also compatible with sketching as required by the data structure. We can prove the following guarantees.

\begin{theorem}\label{thm:RobustAccMWU}
For $\eta\leq 1/10$, with probability $1-1/n^3$, Algorithm~\ref{alg:non_monotone_accel_robust} with input ($\begin{bmatrix}
    \CC\\ -\CC
\end{bmatrix}, \begin{bmatrix}
    \dd\\ -\dd
\end{bmatrix}, \epsilon$) finds $\xxhat\in \mathbb{R}^n$ such that $\|\CC\xxhat-\dd\|_{\infty} \leq 1+O(\epsilon)$ in at most $\Otil\left(n^{1/2-\eta}\epsilon^{-4}\right)$ iterations. Furthermore, the algorithm satisfies the following extra guarantees:
\begin{enumerate}
    \item In the width reduction step of the algorithm, the algorithm only requires to find at most $\Otil\left(n^{1/2+\eta}\right)$ large coordinates per iteration.
    \item The algorithm satisfies the following low-rank update scheme: There are at most $\frac{T+K}{2^{\ell}}$ number of iterations where $\ov{\rr}$ receives an update of rank $\Otil_{\epsilon}(n^{1/5} 2^{2 \ell})$.
\end{enumerate}
\end{theorem}

In order to get Algorithm~\ref{alg:non_monotone_accel_robust} we begin by extending the graph based algorithms of~\cite{madry2016computing} to $\ell_{\infty}$-regression. A direct extension (Algorithm~\ref{alg:non_monotone_accel_opt}, analysis included in Appendix~\ref{sec:NonMonMWU}) does not have stable width reduction steps. We therefore design a new set of width reduction steps (Algorithm~\ref{alg:non_monotone_accel_stab}, Appendix~\ref{sec:NonMonStab}), which necessitates a new analysis for bounding the number of such steps. We then additionally add sketching to the primal steps to get the final algorithm which is analyzed in Appendix~\ref{sec:NonMonRob}.

\begin{algorithm}
\caption{Accelerated MWU algorithm with non-monotone weights and stable and robust steps}\label{alg:non_monotone_accel_robust}
\begin{algorithmic}[1]
\Procedure{MWU-NonMonotoneRobust}{$\wt{\CC}, \wt{\dd}, \epsilon$}
\State $\ww^{(0,0)} \gets 1_{2n}, \quad \ov{\rr}^{(0,0)} \gets \rr^{(0,0)} \gets (1 + \epsilon) 1_{2n}, \quad \xx^{(0)} \gets 0_d$
\State $\alpha \gets \widetilde{\Theta}(n^{-1/2+\eta}\epsilon)$
\State $\tau \gets \widetilde{\Theta}(n^{1/2+\eta}\epsilon^{-4}), \quad \rho \gets \widetilde{\Theta}(n^{1/2-3\eta}\epsilon^{-2})$
\State $T \gets \alpha^{-1}\epsilon^{-2}\ln n$
\State $i,k = 0$
\State $b \gets \widetilde{\Theta}(n^{1/2+\eta} \epsilon^{-2})$
\State Let $\SS^{(0)}, \SS^{(1)}, \cdots, \SS^{(T-1)} \in \R^{b \times 2n}$ be random matrices as described in Lemma~\ref{lem:CE_one_sketch}.
\While{$i<T$}\label{algline:PrimalRobust}
\State $\Delta^{(i,k)} \leftarrow (\wt{\CC}^{\top} \overline{\RR}^{(i,k)} \wt{\CC})^{-1} \wt{\CC}^{\top} \overline{\RR}^{(i,k)} \wt{\dd}$ \Comment{$\Delta^{(i,k)} = \arg\min_{\Delta} \sum_e \rrbar^{(i,k)}_e (\wt{\CC}\Delta-\wt{\dd})_e^2$} \label{algline:linsysRobust}
\State $\uu^{(i,k)} \gets \wt{\CC} \Delta^{(i,k)} - \wt{\dd}$ \label{algline:u_Robust}
\State $\wh{\uu}^{(i,k)} \gets (\RRbar^{(i,k)})^{-1/2} \cdot (\SS^{(i)})^{\top} \SS^{(i)} \cdot (\RRbar^{(i,k)})^{1/2} \uu^{(i,k)}$\label{algline:hat_u_Robust}
\State $\Psi(\rrbar^{(i,k)}) \gets \sum_e \rrbar^{(i,k)}_e (\uu^{(i,k)}_e)^2$
\If{$\sum_e \rrbar_e^{(i,k)} |\uu^{(i,k)}_e|^3\leq C_3 \rho \Psi(\rrbar^{(i,k)}) $}\Comment{primal step}\label{algline:CheckPrimalRobust}
\State $\ww^{(i+1,k)}\gets \ww^{(i,k)}\left(1 + \epsilon \overrightarrow{\alpha}^{(i,k)} \wh{\uu}^{(i,k)}\right)$,
$\overrightarrow{\alpha}^{(i,k)}_e = \begin{cases}
\alpha \cdot (1 + \epsilon\alpha \wh{\uu}^{(i,k)}_e) & \text{ if } \wh{\uu}^{(i,k)}_e \geq 0\\
\alpha / (1 - \epsilon\alpha \wh{\uu}^{(i,k)}_e) & \text{ else }
\end{cases}$ \label{algline:updateWRobust}
\State $\rr^{(i+1,k)} \leftarrow \ww^{(i+1,k)} + \frac{\epsilon}{2n} \sum_{e} \ww^{(i+1,k)}_e$
\State $\rrbar^{(i+1,k)}\gets$ {\sc SelectVector}$(\rr^{(i+1,k)}, i+1, \delta)$ \Comment{Algorithm~\ref{alg:select_vector}}
\State $\xx^{(i+1)} \gets \xx^{(i)} + \Delta^{(i,k)}$ \label{algline:updateXRobust}
\State $i \gets i + 1$
\Else{ {\bf if} $\sum_e \rrbar_e^{(i,k)} |\uu^{(i,k)}_e|^3\geq C_3^{-1} \rho\Psi(\rrbar^{(i,k)})$ {\bf then}}\Comment{width reduction step}
\State Let $S$ be the set of coordinates $e$ such that $|\uu^{(i,k)}_e|\geq \rho/(2C_3)$ \label{algline:setSRobust}
\State $H \subseteq S$ be maximal subset such that $\sum_{e\in H}\rrbar_e^{(i,k)}\leq \tau^{-1}\Psi(\rrbar^{(i,k)}) $
\If{$H \neq S$}
\State Pick any $\bar{e}\in S\setminus H$.
\State For all $e\in H \cup \{\ov{e}\} $, $\ww_e^{(i,k+1)}\gets (1+\epsilon)\ww_e^{(i,k)} + \frac{\epsilon^2}{2n}\Phi(\ww^{(i,k)})$ 
\State $\rr^{(i,k+1)}\gets \ww^{(i,k+1)} + \frac{\epsilon}{2n}\Phi(\ww^{(i,k+1)})$
\State For all $e\in H \cup \{\ov{e}\} $, $\ov{\rr}_e^{(i,k+1)} \gets \rr_e^{(i,k+1)}$ 
\Else
\For{$\zeta = \rho, 2 \rho, 4 \rho, \cdots, 2^{c_\rho} \rho$}
\State \Comment{$c_{\rho}$ is defined to be the smallest integer $c$ that satisfies $2^c \rho \geq \sqrt{n/\epsilon}$}
\State Define the set $H_\zeta = \{e \in H \mid |\wt{\CC} \Delta^{(i,k)} - \wt{\dd}|_e \in [\zeta, 2 \zeta)\}$.
\State If $\sum_{e \in H_{\zeta}} \ov{\rr}^{(i,k)}_e |\wt{\CC} \Delta^{(i,k)} - \wt{\dd}|_e^3 \geq \frac{\rho \Psi(\ov{\rr}^{(i,k)})}{\log (\frac{n}{\epsilon \rho})}$, set $\zeta^* \gets \zeta$, and break. 
\EndFor
\State For all $e \in H_{\zeta^*}$, $\ww_e^{(i,k+1)}\gets (1+\epsilon)\ww_e^{(i,k)} + \frac{\epsilon^2}{2n}\Phi(\ww^{(i,k)})$
\State $\rr^{(i,k+1)}\gets \ww^{(i,k+1)} + \frac{\epsilon}{2n}\Phi(\ww^{(i,k+1)})$
\State For all $e\in H_{\zeta^*}$, $\ov{\rr}_e^{(i,k+1)} \gets \rr_e^{(i,k+1)}$ 
\EndIf
\State $k \gets k+1$
\EndIf
\EndWhile
\State \Return $\xx^{(T)}/T$
\EndProcedure
\end{algorithmic}
\end{algorithm}

%% file: organization.tex
\section*{Organization of Appendix}
The appendix contains detailed proofs of our algorithms. In Appendix~\ref{sec:preli} we give some preliminaries and basic results we use for our proofs. In Appendix~\ref{sec:low_rank} we prove the guarantees of our low-rank update scheme under $\ell_2$ and $\ell_3$ stability, given by the subroutine $\textsc{SelectVector}$ in Algorithm~\ref{alg:non_monotone_accel_robust}. Further in Appendix~\ref{sec:data_structures} we provide the guarantees for all the data structures required to implement our final algorithms. In Appendix~\ref{sec:time_random} we show how to implement our randomized algorithm (Algorithm~\ref{alg:non_monotone_accel_robust}) using the data structures from Appendix~\ref{sec:data_structures}, and prove Theorem~\ref{thm:informal_random}. In Appendix~\ref{sec:time_deterministic} we show how to implement our deterministic MWU algorithm (Algorithm~\ref{alg:MWU}) using the data structures from Appendix~\ref{sec:data_structures}, and prove Theorem~\ref{thm:informal_deterministic}. In Appendix~\ref{sec:MWUMonotone} and~\ref{sec:Appendix_non_monotone_robust} we prove the guarantees of Algorithm~\ref{alg:MWU} and Algorithm~\ref{alg:non_monotone_accel_robust} respectively. In Appendix~\ref{sec:StabilityMWU} we prove the stability guarantees of Algorithms~\ref{alg:MWU} and~\ref{alg:non_monotone_accel_robust}. In Appendix~\ref{sec:missing_proofs} we provide the missing proofs for standard results from Appendix~\ref{sec:preli} and~\ref{sec:data_structures}. Finally, in Appendix~\ref{sec:NonMonMWU} (optional) we provide the analysis of a non-monotone MWU that is robust but not stable and it has $O_{\epsilon}(n^{1/3})$ iterations. This analysis is included for completeness, and is not required to prove our results.

%% file: prelims.tex
\section{Preliminaries}\label{sec:preli}
\paragraph{Basic notations.} For any vectors $\xx$ and $\yy$ with non negative entries, and $\delta>0$ we use $\xx\approx_{\delta}\yy$ to imply that for all coordinates $i$, we have $e^{-\delta} \yy_i \leq \xx_i \leq e^{\delta} \yy_i$. We use $\wt{O}(\cdot)$ and $\wt{\Theta}(\cdot)$ to hide $\poly \log n$ factors, and we use $\Otil_{\epsilon}(\cdot)$ and $\wt{\Theta}_{\epsilon}(\cdot)$ to additionally hide $\poly(\epsilon^{-1})$ factors.

Given any two vectors $\xx, \yy \in \mathbb{R}^n$, we use $\xx \cdot \yy \in \mathbb{R}^n$ to denote the coordinate-wise multiplication of the two vectors, i.e., its $i$-th entry is $\xx_i \cdot \yy_i$. Similarly, we also use other scalar operations on vectors to denote coordinate-wise operations.

For any vector $\rr \in \R^n$, we use the capital letter $\RR \in \R^{n \times n}$ to denote a diagonal matrix whose diagonal entries are $\rr$.

\paragraph{Potential functions.} In this paper we consider a fixed problem $\min_x \|\CC \xx - \dd\|_{\infty}$ and assume that this has optimum objective value $1$. We define the following two potential functions for weights $\ww$ and $\rr$ such that $\rr = \ww + \frac{\epsilon}{n}\|\ww\|_1$:
\begin{equation}\label{eq:defPhi}
    \Phi(\ww) \defeq \norm{\ww}_1
\end{equation}
\begin{equation}\label{eq:defPsi}
    \Psi(\rr)\defeq \min_{\Delta\in \rea^d }\sum_e \rr_e (\CC\Delta-\dd)^2_e.
\end{equation}

The two potentials satisfy the following three lemmas. Their proofs are standard, and we defer them to Section~\ref{sec:missing_proofs}.

The first lemma shows how the two potentials are related.
\begin{lemma}\label{lem:PsiPhi}
    Let $\ww \geq 0,\rr, \rrbar$, such that $\forall e, \rr_e = \ww_e + \frac{\epsilon}{n}\|\ww\|_1$, and $\rrbar\approx_{\delta}\rr$. Then, $\Psi(\rrbar) \approx_{\delta} \Psi(\rr)$, and $\Psi(\rrbar) \leq e^{\epsilon + \delta} \cdot \Phi(\ww)$.
\end{lemma}


In our algorithms we have the following lower bound on the initial $\Psi$ potential.
\begin{lemma}\label{lem:lower_bound_Psi_0}
If $\ww^{(0,0)} = 1$ and $\rr^{(0,0)} = \ww^{(0,0)} + \frac{\epsilon}{n} \cdot \Phi(\ww^{(0,0)})$, then we have
$\Psi(\rr^{(0,0)})\geq \Psi_0 \defeq \min\{1, \dd^{\top}\left(\II - \CC^{\top}(\CC^{\top}\CC)^{-1}\CC\right)\dd\}$. 
\end{lemma}

The last lemma provides a lower bound on the $\Psi$ potential after a small perturbation to the weights $\rr$.
\begin{restatable}{lemma}{PsiChange}\label{lem:PsiChange}
    Let $\Psi(\rr) =\min_{\Delta}\sum_e \rr_e (\CC\Delta-\dd)_e^2 $. For any $\rr',\rr\geq 0$ that satisfies $|\rr_e'-\rr_e| \leq \rr'_e$ for all $e$, we have
    \[
    \Psi(\rr') \geq \Psi(\rr)+ \sum_e \left(\frac{\rr'_e-\rr_e}{\rr'_e}\right)\rr_e (\CC\widetilde{\Delta}-\dd)_e^2,
    \]
    where $\widetilde{\Delta} := \arg\min_{\Delta}\sum_e \rr_e (\CC\Delta-\dd)_e^2$.
\end{restatable}

\paragraph{Primal iterate and width iterate of MWU algorithms.} We will use $i$ to denote primal iterates and $k$ to denote the width reduction iterations in our multiplicative weight update (mwu) algorithms. We use $\wh{\uu}$ to denote the vector $\uu$ after applying sketching, and $\ov{\uu}$ to denote an approximation to the vector $\uu$. For the (mwu) algorithms, we would use $i_k$ to denote the number of primal steps executed when the $k^{th}$ width step is being taken, i.e., the $k^{th}$ width step is from $(i_k, k)$ to $(i_k, k+1)$, and we use $k_i$ to denote the number of width reduction steps executed when the $i^{th}$ primal step is taken, i.e., the $i^{th}$ primal step is from $(i, k_i)$ to $(i+1, k_i)$.

For any primal step $i$ and any coordinate $e$, we define $\textsc{LastWidth}(i,e)$ to be the largest $i' \leq i$ such that the algorithm executed a width reduction step from $(i',k)$ to $(i',k+1)$ during which the weight of $e$ is updated, i.e., $\ww_e^{(i',k+1)} \neq \ww_e^{(i',k)}$.

\paragraph{Fast matrix multiplication.}
We use $\Tmat(n,r,m)$ to denote the time complexity required to compute the product of an $n \times r$ matrix with an $r \times m$ matrix.

In our proofs we will frequently use the following fact. See e.g. \cite{bcs97} for the basic properties of fast matrix multiplication exponents.
\begin{fact}
$\Tmat(n,r,m) = O(\Tmat(n,m,r)) = O(\Tmat(m,n,r))$.
\end{fact}

\begin{definition}[Fast matrix multiplication exponent]
For any $\beta$, define a function $\omega_{\beta}(x)$ to be the minimum value such that $\Tmat(n,n^x,n^{\beta}) = n^{\omega_{\beta}(x) + o(1)}$. 

With an abuse of notation we also define the function $\omega(x) = \omega_1(x)$, and define the value $\omega = \omega(1)$.

We also define $\alpha_* \in \R_+$ to be the dual exponent of matrix multiplication, i.e., $\omega(\alpha_*) = 2$.\footnote{It's common in the literature to use $\alpha$ to denote the dual exponent of matrix multiplication. We use $\alpha_*$ here because we will use $\alpha$ to denote the ``step size'' of accelerated MWU.}
\end{definition}

We will also use the following fact about convexity. We present a proof (deferred to Section~\ref{sec:missing_proofs}) that generalizes the proof of Lemma~3.6 of \cite{jklps20}.
\begin{fact}[Convexity]\label{fact:FMM_convex}
For any $\beta$, $\omega_{\beta}(x)$ is convex in $x$.
\end{fact}

\begin{fact}[Upper bound of $\Tmat(n,n,r)$]\label{fact:upper_bound_Tmat}
For any $r \leq n$, 
\[
\Tmat(n,n,r) \leq n^{2+o(1)} + r^{\frac{\omega-2}{1-\alpha}} n^{2-\frac{\alpha (\omega-2)}{1-\alpha} + o(1)}.
\]
\end{fact}

%% file: stability.tex
\section{Low Rank Update Scheme under Stability Guarantees}\label{sec:low_rank}

\subsection{Low Rank Update under \texorpdfstring{$\ell_2$}{} Stability}
Algorithm~\ref{alg:select_vector} is the same as \cite{lee2021tutorial}, and it satisfies the following lemma.
\begin{lemma}[Low-rank update scheme under $\ell_2$ stability, Lemma~19 of \cite{lee2021tutorial}]\label{lem:LowRankL2}
If we have the guarantee
\[
\sum_{e} \ln\left(\frac{\rr_e^{\left(i + 1\right)}}{\rr_e^{\left(i \right)}} \right)^2 \leq \zeta,
\]
then the Algorithm~\ref{alg:select_vector} outputs a vector $\ov{\rr} \approx_{\delta} \rr^{(i)}$ in each iteration, and the approximate vector $\ov{\rr}$ undergoes a update of size $O((\frac{\log n}{\delta})^2 \cdot \zeta \cdot 2^{2\ell})$ in every $2^\ell$ iterations for every $\ell \in [0:\log T]$.
\end{lemma}

Next we show that the robust $\ell_2$ stability guarantee also generates a $\delta$-approximate sequence with low-rank updates.
\begin{lemma}[Low-rank update scheme under robust $\ell_2$ stability]\label{lem:LowRankL2Robust}
If the sequence $\rr_e^{(0)}, \cdots, \rr_e^{(T)}$ satisfies the following guarantee: There exists another sequence $\wt{\rr}_e^{(0)}, \cdots, \wt{\rr}_e^{(T)}$ such that
\begin{enumerate}
    \item 
    \[
    \sum_e \ln\left(\frac{\wt{\rr}_e^{(i+1)}}{\rr_e^{(i)}}\right)^2 \leq \zeta, ~~~ \forall i \in [0:T], 
    \]
    \item $\forall t \leq t' \in [T]$, $\forall e$, with probability $1 - 1/n^4$,
    \[
    \left|\sum_{i = t'-t}^{t'} \ln\left(\frac{\wt{\rr}_e^{(i)}}{\rr_e^{(i)}}\right)\right| \leq \frac{\delta}{10 \log n}. 
    \]
\end{enumerate}
Then the Algorithm~\ref{alg:select_vector} outputs a vector $\ov{\rr} \approx_{\delta} \rr^{(i)}$ in each iteration, and the approximate vector $\ov{\rr}$ undergoes an update of size $O((\frac{\log n}{\delta})^2 \cdot \zeta \cdot 2^{2\ell})$ in every $2^\ell$ iterations for every $\ell \in [0:\log T]$.
\end{lemma}
\begin{proof}
Consider a fixed iteration $i$. We first show that in the $i$-th iteration $e^{-\delta} \rr^{(i)}_e \leq \ov{\rr}_e \leq e^{\delta} \rr^{(i)}_e$ for any $e \in [n]$. Let $i'$ be the iteration when $\ov{\rr}_e$ was last updated. We can write $i' = i_{0} < i_{1} < i_{2}< \cdots < i_{s} = i$ such that $i_{j+1} - i_{j}$ is a power of $2$ and $i_{j+1} - i_{j}$ divides $i_{j+1}$, and $|s| \leq 2 \log n$.
Hence, we have that
\begin{align*}
\frac{\rr_e^{(i)}}{\ov{\rr}_e} = &~ \frac{\rr_e^{(i_s)}}{\rr_e^{(i_0)}} 
= \prod_{j=0}^{s-1} \frac{\rr_e^{(i_{j+1})}}{\rr_e^{(i_j)}}
= \exp\Big(\sum_{j=0}^{s-1} \ln(\frac{\rr_e^{(i_{j+1})}}{\rr_e^{(i_j)}})\Big) \leq \exp(\delta),
\end{align*}
where in the fourth step we used that since $\ov{\rr}_{e}$ is not updated since step $i'$, we have $|\ln(\frac{\rr_e^{(i_{j+1})}}{\rr_e^{(i_j)}})| \leq \frac{\delta}{2 \log n}$. Similarly we also have $\frac{\rr_e^{(i)}}{\ov{\rr}_e} \geq \exp(-\delta)$.

Next we bound the size of the update after every $2^{\ell}$ iterations. Let $i$ be any iteration where $i \equiv 0 \mod 2^{\ell}$. We denote the set that is being updated as $S_{\ell} := \{e : |\ln \frac{\rr_e^{(i)}}{\rr_e^{(i-2^\ell)}} | \geq \frac{\delta}{2 \log n}\}$. Wlog assume that $\ln \frac{\rr_e^{(i)}}{\rr_e^{(i-2^\ell)}} \geq 0$. Using the second property of the sequence that $\left|\sum_{j = i-2^{\ell}+1}^{i} \ln(\frac{\wt{\rr}_e^{(j)}}{\rr_e^{(j)}})\right| \leq \frac{\delta}{10 \log n}$, we have
\begin{align*}
\frac{\rr_e^{(i)}}{\rr_e^{(i-2^\ell)}} = &~ \prod_{j=i-2^{\ell}}^{i-1} \frac{\rr_e^{(j+1)}}{\rr_e^{(j)}}
= \prod_{j=i-2^{\ell}}^{i-1} \frac{\wt{\rr}_e^{(j+1)}}{\rr_e^{(j)}} \cdot \prod_{j=i-2^{\ell}+1}^{i} \frac{\rr_e^{(j)}}{\wt{\rr}_e^{(j)}} \geq \prod_{j=i-2^{\ell}}^{i-1} \frac{\wt{\rr}_e^{(j+1)}}{\rr_e^{(j)}} \cdot \exp(-\frac{\delta}{10 \log n}),
\end{align*}
So for any $e \in S_{\ell}$, we have 
\begin{align*}
\sum_{j=i-2^{\ell}}^{i-1} \ln\left(\frac{\wt{\rr}_e^{(j+1)}}{\rr_e^{(j)}} \right) \geq \ln \left(\frac{\rr_e^{(i)}}{\rr_e^{(i-2^\ell)}}\right) - \frac{\delta}{10 \log n} \geq \frac{\delta}{5 \log n}.
\end{align*}
So we have
\begin{align*}
|S_{\ell}| \cdot \frac{\delta^2}{(5 \log n)^2} 
\leq \sum_{e \in S_{\ell}} \left(\sum_{j=i-2^{\ell}}^{i-1} \ln\left(\frac{\wt{\rr}_e^{(j+1)}}{\rr_e^{(j)}} \right)\right)^2 
\leq 2^{\ell} \cdot \sum_{e \in S_{\ell}} \sum_{j=i-2^{\ell}}^{i-1} \ln\left(\frac{\wt{\rr}_e^{(j+1)}}{\rr_e^{(j)}} \right)^2
\leq 2^{2 \ell} \cdot \zeta,
\end{align*}
where the last step follows from the first property of the sequence.

So we have $|S_{\ell}| \leq O(2^{2\ell}(\log n/\delta)^{2} \zeta)$.
\end{proof}

\subsection{Low Rank Update under \texorpdfstring{$\ell_3$}{} Stability}
In this section we prove the low-rank update guarantee under $\ell_3$ stability, which holds for MWU with monotone weights, and we only use it in our deterministic algorithm.
\subsubsection{Decomposition of Iterations}
\begin{lemma}[Decomposition of iterations]\label{lem:decomposition_iteration}
If the weights satisfy that 
\[
\sum_{i=1}^T \sum_{e} \left|\frac{\rr_e^{\left(i + 1\right)}}{\rr_e^{\left(i \right)}} - 1 \right|^3 \leq \zeta,
\]
then we can decompose the $T$ iterations into $\log T + 1$ disjoint sets:
\begin{align*}
B_j := &~ \left\{ i \in [T] ~\bigg|~ \frac{\zeta}{2^{j+1}} < \sum_{e} \left|\frac{\rr_e^{\left(i + 1\right)}}{\rr_e^{\left(i \right)}} - 1 \right|^3 \leq \frac{\zeta}{2^{j}} \right\},~~ \forall j \in [0:\log T-1], \\
B_{\log T} := &~ \left\{ i \in [T] ~\bigg|~ \sum_{e} \left|\frac{\rr_e^{\left(i + 1\right)}}{\rr_e^{\left(i \right)}} - 1 \right|^3 \leq \frac{\zeta}{T} \right\},
\end{align*}
and these sets satisfy that $\cup_{j=0}^{\log T} B_j = [T]$, and $|B_j| \leq 2^{j+1}$ for all $j \in [\log T]$.
\end{lemma}
\begin{proof}
Since $\sum_{i=1}^T \sum_{e} \left|\frac{\rr_e^{\left(i + 1\right)}}{\rr_e^{\left(i \right)}} - 1 \right|^3 \leq \zeta$, we have that for any $i \in [T]$, $0 \leq \sum_{e} \left|\frac{\rr_e^{\left(i + 1\right)}}{\rr_e^{\left(i \right)}} - 1 \right|^3 \leq \zeta$, so each $i \in [T]$ must fall into exactly one set $B_j$.

For any $j \in [0:\log T - 1]$, by the definition of $B_j$ we have
\begin{align*}
\sum_{i \in B_j} \sum_{e} \left|\frac{\rr_e^{\left(i + 1\right)}}{\rr_e^{\left(i \right)}} - 1 \right|^3 \geq |B_j| \cdot \frac{\zeta}{2^{j+1}}.
\end{align*}
Combining with our assumption that $\sum_{i=1}^T \sum_{e} \left|\frac{\rr_e^{\left(i + 1\right)}}{\rr_e^{\left(i \right)}} - 1 \right|^3 \leq \zeta$, we have that
\begin{align*}
|B_j| \cdot \frac{\zeta}{2^{j+1}} \leq \zeta \implies |B_j| \leq 2^{j+1}.
\end{align*}
Finally, note that we trivially have $|B_{\log T}| \leq T < 2^{\log T +1}$.
\end{proof}

\subsubsection{Low Rank Update Scheme under \texorpdfstring{$\ell_3$}{} Stability}
\begin{algorithm}[H]
\caption{Low rank update in the $t$-th iteration}
\label{alg:select_vector_L3}
 \begin{algorithmic}[1]
\Procedure{SelectVectorL3}{$\rr^{(t)}$}
\For{all $j \in [0: \log T]$}
\For{all $\ell \in [0: \log T]$}
\If{$i \in B_j$ and $i$ is the $k$-th element in $B_j$ where $k \equiv 0 \pmod {2^{\ell}}$}
\State $I \leftarrow \left\{ e ~\Big|~
\sum_{k' = k - 2^{\ell}}^{k} \left|\frac{\rr_e^{\left(B_j[k'] + 1\right)}}{\rr_e^{\left(B_j[k'] \right)}} - 1 \right| \geq \frac{\delta}{10 \log^2 n}\right\}$
\EndIf
\State Update the weights for all $e \in I$ to be $\rrbar^{(t)}_e \leftarrow \rr^{(t)}_e$
\EndFor
\EndFor
\EndProcedure 
\end{algorithmic}
\end{algorithm}

\begin{lemma}[Low-rank update scheme under $\ell_3$ stability]\label{lem:LowRankL3}
Assume that the weights are monotonically increasing and satisfy 
\[
\sum_{i=1}^T \sum_{e} \left|\frac{\rr_e^{\left(i + 1\right)}}{\rr_e^{\left(i \right)}} - 1 \right|^3 \leq \zeta. 
\]
Define the sets $B_0, \cdots, B_{\log T} \subseteq [T]$ as Lemma~\ref{lem:decomposition_iteration}, and for any $j$ let $B_j[1], \cdots, B_j[|B_j|]$ denote the elements in $B_j$ in increasing order. 

For any $\delta \leq 0.1$, Algorithm~\ref{alg:select_vector_L3} maintains a vector $\rrbar \approx_{\delta} \rr$ where $\rrbar$ undergoes the following updates: for any $j \in [0: \log T]$, for any $\ell \in [0: \log |B_j|]$, $\rrbar$ receives an update of size $O(\zeta \cdot 2^{3 \ell-j} \cdot \frac{\log^6 n}{\delta^3})$ in iterations $B_j[2^{\ell}], B_j[2 \cdot 2^{\ell}], B_j[3 \cdot 2^{\ell}], \cdots, B_j[\lfloor \frac{|B_j|}{2^{\ell}} \rfloor \cdot 2^{\ell}]$.
\end{lemma}
\begin{proof}
For any $j \in [0: \log T]$, and for any $\ell \in [0: \log |B_j|]$, in any iteration $k$ that equals to an integer times $2^{\ell}$, Algorithm~\ref{alg:select_vector_L3} performs an update for all coordinates in set $I$, where
\begin{align*}
I = \left\{ e ~\Big|~
\sum_{k' = k - 2^{\ell}}^{k} \left|\frac{\rr_e^{\left(B_j[k'] + 1\right)}}{\rr_e^{\left(B_j[k'] \right)}} - 1 \right| \geq \frac{\delta}{10 \log^2 n}\right\}.
\end{align*}
We first bound the size of the set $I$. We have
\begin{align*}
|I| \cdot (\frac{\delta}{10 \log^2 n})^3 \leq &~ \sum_e \left( \sum_{k' = k - 2^{\ell}}^{k} \left|\frac{\rr_e^{\left(B_j[k'] + 1\right)}}{\rr_e^{\left(B_j[k'] \right)}} - 1 \right| \right)^3 \\
\leq &~ \sum_e 2^{2 \ell} \cdot \sum_{k' = k - 2^{\ell}}^{k} \left|\frac{\rr_e^{\left(B_j[k'] + 1\right)}}{\rr_e^{\left(B_j[k'] \right)}} - 1 \right|^3 \\
\leq &~ \frac{2^{3 \ell} \cdot \zeta}{2^j},
\end{align*}
where the second step follows from $(\sum_{i=1}^n |a_i|)^3 \leq n^2 \cdot \sum_{i=1}^n |a_i|^3$ for any sequence $a_i$, the third step follows from $\sum_{e} \left|\frac{\rr_e^{\left(i + 1\right)}}{\rr_e^{\left(i \right)}} - 1 \right|^3 \leq \frac{\zeta}{2^{j}}$ for all $i \in B_j$.

So we have
\begin{align*}
|I| \leq \zeta \cdot 2^{3 \ell-j} \cdot (\frac{10 \log^2 n}{\delta})^3.
\end{align*}

Next we prove that the vector $\rrbar$ maintained in Algorithm~\ref{alg:select_vector_L3} satisfies $|\frac{\rrbar^{(i)}_e}{\rr_e^{(i)}} - 1| \leq \delta$ for all coordinates $e$ and in all iterations $i$. Fix a coordinate $e$ and an iteration $i$, and let $i_0$ be the last iteration that $\rrbar_e$ was updated. For any $j \in [0:\log T]$, let $B_j[k], B_j[k+1], \cdots, B_j[k+t]$ denote the iterations in $[i_0, i]$ that fall into $B_j$, and note that $t = |[i_0,i] \cap B_j| \leq |B_j| \leq 2^{j+1}$. We can write $k = k_0 < k_1 < k_2 < \cdots < k_s = k+t$ where each $k_{\ell+1} - k_{\ell}$ is a power of $2$ and $s \leq 2 \log t \leq 2 (j+1)$. Since $\rrbar_e$ is not updated in any iterations $B_j[k_1], \cdots, B_j[k_s]$, we have that for any $\ell \in [s]$,
\begin{align*}
\sum_{k' = k_{\ell-1}}^{k_{\ell}} \left|\frac{\rr_e^{\left(B_j[k'] + 1\right)}}{\rr_e^{\left(B_j[k'] \right)}} - 1 \right| < \frac{\delta}{10 \log^2 n},
\end{align*}
so we have
\begin{align*}
\sum_{\ell=1}^{s} \sum_{k' = k_{\ell-1}}^{k_{\ell}} \left|\frac{\rr_e^{\left(B_j[k'] + 1\right)}}{\rr_e^{\left(B_j[k'] \right)}} - 1 \right| < \frac{s \delta}{10 \log^2 n} \leq \frac{2(j+1) \delta}{10 \log^2 n} \leq \frac{\delta}{\log n}.
\end{align*}
Since the same argument holds for all $j \in [0: \log T]$, and each iteration in $[i_0, i]$ falls into exactly one $B_j$, we have that
\begin{align*}
\sum_{i'=i_0}^{i-1} \left|\frac{\rr_e^{(i'+1)}}{\rr_e^{(i')}} - 1 \right| < \frac{\delta}{\log n} \cdot \log T \leq \delta.
\end{align*}
Using the above inequality, and note that the weights are always increasing, we have
\begin{align*}
\frac{\rr_e^{(i)}}{\rrbar^{(i)}_e} = \frac{\rr_e^{(i)}}{\rr_e^{(i_0)}}
= &~ \prod_{i'=i_0}^{i-1} \left( 1 + \left|\frac{\rr_e^{(i'+1)}}{\rr_e^{(i')}} - 1 \right| \right) \\
\leq &~ \exp\left( \sum_{i'=i_0}^{i-1} \left|\frac{\rr_e^{(i'+1)}}{\rr_e^{(i')}} - 1 \right| \right) \leq \exp(\delta).
\end{align*}
Finally note that we also have $\frac{\rr_e^{(i)}}{\rrbar_e} \geq 1$ since the weights are always increasing.
\end{proof}

\begin{corollary}\label{cor:LowRankL3}
For any $j \in [0:\log T]$ and any $t$, the total number of coordinates that are updated in iterations $B_j[k], \cdots, B_j[k+t]$ is $O\left(\zeta \cdot 2^{3 \log t-j} \cdot \frac{\log^6 n}{\delta^3} \right)$.
\end{corollary}
\begin{proof}
Using Lemma~\ref{lem:LowRankL3} we have the total number of updates in iterations $B_j[k], \cdots, B_j[k+t]$ is
\begin{align*}
\sum_{\ell=0}^{\log t} O\left(\zeta \cdot 2^{3 \ell - j} \cdot \frac{\log^6 n}{\delta^3} \right) \cdot \frac{t}{2^{\ell}} \leq O\left(\zeta \cdot 2^{3 \log t-j} \cdot \frac{\log^6 n}{\delta^3} \right).
\end{align*}
\end{proof}

%% file: data_structures.tex
\section{Data Structures}\label{sec:data_structures}
In this section, we would present all the data structures we use for the various tasks in Algorithm~\ref{alg:non_monotone_accel_robust}. 
\subsection{Inverse Maintenance Data Structure}\label{sec:InverseMaintenance}
In this section we present the formal versions of Lemmas~\ref{lem:InverseMaintenanceOneLevelShort} and \ref{lem:InverseMaintenanceTwoLevelShort}. These are the inverse maintenance data structures of \cite{bns19}, and we have included a version of their results which is tailored to our notations and analysis. For completeness, we include the proofs of the following two lemmas in Section~\ref{sec:missing_proofs}.
\begin{lemma}[One-level inverse maintenance, Theorem 4.1 of \cite{bns19}]
\label{lem:InverseMaintenanceOneLevel}
There exists a data structure that initially has a matrix $\MM^{(0)} \in \R^{n \times n}$, and in each iteration it receives an update $\Delta^{(t)} \in \R^{n \times n}$ to update the matrix to $\MM^{(t)} = \MM^{(t-1)} + \Delta^{(t)}$. The data structure maintains an iteration counter $t_0$ and it maintains the inverse $\NN = (\MM^{(t_0)})^{-1}$ internally. Let $k = \nnz(\Delta^{(t_0+1)}) + \cdots + \nnz(\Delta^{(t)})$ denote the total size of the updates until the current iteration. The runtime for each operation of the data structure is as follows:
\begin{itemize}
\item {\bf Initialize$(\MM^{(0)})$}: Initially set $t_0 = 0$ and $\NN = (\MM^{(0)})^{-1}$. This operation takes $O(n^{\omega})$ time.
\item {\bf Update$(\Delta^{(t)})$}: The data structure receives the $t$-th update. This operation takes $O(\nnz(\Delta^{(t)}))$ time.
\item {\bf Reset$()$}: Reset $t_0 = t$ and $\NN = (\MM^{(t_0)})^{-1}$.  This operation takes $O(\Tmat(n, n, k))$ time.
\item {\bf Query$(J_r, J_c \subseteq [n])$}: Output the submatrix $\big((\MM^{(t)})^{-1} \big)_{J_r, J_c}$ that has $|J_r| = \ell_r$ rows and $|J_c| = \ell_c$ columns. This operation takes $O\big(k^{\omega} + \Tmat(\ell_r, k, \ell_c) \big)$ time.
\end{itemize}
\end{lemma}

\begin{lemma}[Two-level inverse maintenance, Theorem 4.2 of \cite{bns19}]
\label{lem:InverseMaintenanceTwoLevel}
There exists a data structure that initially has a matrix $\MM^{(0)} \in \R^{n \times n}$, and in each iteration it receives an update $\Delta^{(t)} \in \R^{n \times n}$ to update the matrix to $\MM^{(t)} = \MM^{(t-1)} + \Delta^{(t)}$. The data structure maintains two iteration counters $t_0 \leq t_1$, and it also maintains $k_0 := \nnz(\Delta^{(t_0+1)}) + \cdots + \nnz(\Delta^{(t)})$ and $k_1 := \nnz(\Delta^{(t_1+1)}) + \cdots + \nnz(\Delta^{(t)})$. Let $J \subseteq [n]$ denote the indexes of the non-zero columns of $\Delta^{(t_0+1)} + \cdots + \Delta^{(t_1)}$. For any $t' \leq t$, define the transformation matrix
\[
\TT^{(t',t)} := \II + (\MM^{(t')})^{-1} \cdot (\MM^{(t)} - \MM^{(t')}) \in \R^{n \times n}.
\]
The data structure maintains $\NN = (\MM^{(t_0)})^{-1} \in \R^{n \times n}$, and $\BB = (\TT^{(t_0,t_1)}_{J,J})^{-1}$ that has size at most $k_0 \times k_0$, and $\EE = (\TT^{(t_0,t_1)}_{J,J})^{-1} \cdot \NN_{J,:}$. The runtime for each operation of the data structure is as follows:
\begin{itemize}
\item {\bf Initialize$(\MM^{(0)})$}: Initially set $t_0 = t_1 = 0$, $\BB = 0$, $\EE = 0$, and $\NN = (\MM^{(0)})^{-1}$. This operation takes $O(n^{\omega})$ time.
\item {\bf Update$(\Delta^{(t)})$}: The data structure receives the $t$-th update. This operation takes $O(\nnz(\Delta^{(t)}))$ time.
\item {\bf Reset$()$}: Reset $t_0 = t$ and $\NN = (\MM^{(t_0)})^{-1}$.  This operation takes $O(\Tmat(n, n, k_0))$ time.
\item {\bf PartialReset$()$}: Reset $t_1 = t$, reset $J \subseteq [n]$ to be the indexes of the non-zero columns of $\Delta^{(t_0+1)}+ \cdots + \Delta^{(t)}$, and reset $\BB = (\TT^{(t_0,t)}_{J,J})^{-1}$ and $\EE = (\TT^{(t_0,t)}_{J,J})^{-1} \cdot \NN_{J,:}$. This operation takes $O(\Tmat(n, k_0, k_1))$ time.

\item {\bf Query$(J_r, J_c \subseteq [n])$}: Output the submatrix $\big((\MM^{(t)})^{-1} \big)_{J_r, J_c}$ that has $|J_r| = \ell_r$ rows and $|J_c| = \ell_c$ columns. This operation takes $O\big(\Tmat(k_0, k_1, \max\{k_1,\ell_c\}) + \Tmat(k_0, \ell_r, \ell_c)\big)$ time.
\end{itemize}
\end{lemma}

We can maintain any matrix formula using the inverse maintenance data structure, as shown in \cite{b21}.
\begin{theorem}[Matrix formula as inverse, Theorem~3.1 of \cite{b21}]\label{lem:matrix_formula}
Given any formula $f$ with input matrices $\AA_1 \in \R^{n_1 \times m_1}, \cdots, \AA_d \in \R^{n_d \times m_d}$, where the formula $f$ consists of only matrix addition, subtraction, multiplication, and inversion,
define $n := \sum_{i = 1}^d n_i + m_i$.

Then there exists a symbolic block matrix $\NN$ of size at most $n\times n$, and sets $I,J \subset [n]$, such that for all matrices $\AA_1,...,\AA_p$ for which $f(\AA_1,...,\AA_d)$ is executable, $(\NN(\AA_1,...,\AA_d)^{-1})_{I,J} = f(\AA_1,...,\AA_d)$.

Constructing $\NN$ from $f$ can be done in $O(n^2)$ time.
\end{theorem}

\subsection{Implicit Inverse Maintenance}
In our algorithm, we also require a data structure that allows us to update $\xx^{(i+1)} \gets \xx^{(i)} + \Delta^{(i,k)}$ in each primal step \emph{implicitly} since we don't have the time budget to query the entire vector $\Delta^{(i,k)}$, and we only query the final vector $\xx^{(T)}$ in the end. To solve this problem we present an implicit inverse maintenance data structure, and its proof can be found in Section~\ref{sec:missing_proofs}. Similar techniques were developed in Section~I of \cite{jiang2021faster} to maintain feasibility. 

\begin{lemma}[Implicit two-level inverse maintenance]\label{lem:implicit_inverse_maintenance}
There exists a data structure that initially has a matrix $\MM^{(0)} \in \R^{n \times n}$ and a vector $\vv \in \R^{n}$, and in each iteration it receives an update $\Delta^{(t)} \in \R^{n \times n}$ to update the matrix to $\MM^{(t)} = \MM^{(t-1)} + \Delta^{(t)}$. The goal of our algorithm is to support queries that output the sum of inverse vector products $\sum_{i=0}^t (\MM^{(i)})^{-1} \cdot \vv$ occasionally.

The data structure maintains two iteration counters $t_0 \leq t_1$. Let $k_0 := \nnz(\Delta^{(t_0+1)}) + \cdots + \nnz(\Delta^{(t)})$ and $k_1 := \nnz(\Delta^{(t_1+1)}) + \cdots + \nnz(\Delta^{(t)})$. Similar as Lemma~\ref{lem:InverseMaintenanceTwoLevel}, the data structure maintains $J \subseteq [n]$ that consists of the indexes of the non-zero columns of $\Delta^{(t_0+1)} + \cdots + \Delta^{(t_1)}$, $\NN = (\MM^{(t_0)})^{-1} \in \R^{n \times n}$, $\BB = (\TT^{(t_0,t_1)}_{J,J})^{-1}$ that has size at most $k_0 \times k_0$, and $\EE = (\TT^{(t_0,t_1)}_{J,J})^{-1} \cdot \NN_{J,:}$. The data structure also maintains three vector $\uu_0, \uu_1$, and $\uu_2$ that satisfy the invariant:
\[
\sum_{i=0}^t (\MM^{(i)})^{-1} \vv = \uu_0 + \NN \cdot \uu_1 + \begin{bmatrix}
\NN_{J,:} \cdot \uu_2 \\ 0
\end{bmatrix}.
\]
The runtime for each operation of the data structure is as follows:
\begin{itemize}
\item {\bf Initialize$(\MM^{(0)}, \vv)$}: Initially set $t_0 = t_1 = 0$, $\BB = 0$, $\EE = 0$, $\NN = (\MM^{(0)})^{-1}$, $\uu_0 = (\MM^{(0)})^{-1} \vv$, and $\uu_1 = \uu_2 = 0$. This operation takes $O(n^{\omega})$ time.
\item {\bf Update$(\Delta^{(t)})$}: The data structure receives the $t$-th update and update $\uu_0, \uu_1$, and $\uu_2$. This operation takes $O(\Tmat(k_0, k_1, k_1) + n)$ time.
\item {\bf Reset$()$}: Reset $t_0 = t$ and $\NN = (\MM^{(t_0)})^{-1}$.  This operation takes $O(\Tmat(n, n, k_0))$ time.
\item {\bf PartialReset$()$}: Reset $t_1 = t$, reset $J \subseteq [n]$ to be the indexes of the non-zero columns of $\Delta^{(t_0+1)}+ \cdots + \Delta^{(t)}$, and reset $\BB = (\TT^{(t_0,t)}_{J,J})^{-1}$ and $\EE = (\TT^{(t_0,t)}_{J,J})^{-1} \cdot \NN_{J,:}$. This operation takes $O(\Tmat(n, k_0, k_1))$ time.

\item {\bf QuerySum$()$}: Output $\sum_{i=0}^t (\MM^{(i)})^{-1} \vv$. This operation takes $O(n^2)$ time.
\end{itemize}

\end{lemma}

\subsection{\texorpdfstring{$\ell_3$}{} and \texorpdfstring{$\ell_2$}{}-Norm Estimations}
We will also use the following $\ell_3$ norm estimation lemma from \cite{wz13} to estimate the quantity on Line~\ref{algline:CheckPrimalRobust} of Algorithm~\ref{alg:non_monotone_accel_robust}.
\begin{lemma}[$\ell_3$ norm estimation, Theorem~1 of \cite{wz13}]\label{lem:l3_norm_estimation}
There exists a distribution $\Pi$ of matrices of size $O(n^{1/3} \log^3 n) \times n$ such that for any vector $\xx \in \R^n$, with probability $0.99$ we have that a random matrix $\UU \sim \Pi$ satisfies
\[
C_3^{-1/3} \|\xx\|_3 \leq \|\UU \xx\|_{\infty} \leq C_3^{1/3} \|\xx\|_3,
\]
where $C_3 > 1$ is a constant.
\end{lemma}
We remark that we can easily boost the success probability of the above theorem to $1-1/n^4$ by using $O(\log n)$ copies and take the median of the estimates. 

We will also use the standard JL lemma to estimate the $\Psi$ potential which can be written as a $\ell_2$ norm. 
\begin{lemma}[Johnson-Lindenstrauss Lemma \cite{jl84}]\label{lem:JL}
There exists a function \textsc{JL}$(n,m,\epsilon, \delta)$ that returns a random matrix $\JJ \in \R^{k \times n}$ where $k = O(\epsilon^{-2} \log(m/\delta))$, and $\JJ$ satisfies that for any fixed $m$-element subset $V \subset \R^n$,
\begin{align*}
    \Pr\big[\forall \vv \in V, ~ (1 - \epsilon) \|\vv\|_2 \leq \|\JJ \vv\|_2 \leq (1 + \epsilon) \|\vv\|_2\big] \geq 1 - \delta.
\end{align*}
Furthermore, the function \textsc{JL} runs in $O(kn)$ time.
\end{lemma}

\subsection{\texorpdfstring{$\ell_2$}{} Heavy Hitter}
We use a heavy-hitter data structure to get a list of all the large coordinates on which we wish to perform width reduction in Algorithm~\ref{alg:non_monotone_accel_robust}.
\begin{lemma}[$\ell_2$ heavy hitter, \cite{knpw11,p13}]\label{lem:heavy_hitter}
Given any $n$, $\epsilon$, and $\delta$, there exists a random matrix $\Phi \in \R^{O(\epsilon^{-2} \log(\delta^{-1}) \log n) \times n}$, and a decoding function \textsc{Decode}, such that given a vector $\yy = \Phi \cdot \xx$ for some $\xx \in \R^n$, \textsc{Decode}$(\yy)$ outputs a list $L \subseteq [n]$ of size $|L| = O(\epsilon^{-2})$, where with probability $1-\delta$ the list $L$ includes all $i \in [n]$ that satisfies
    \begin{align*}
        |\xx_i| \geq \epsilon \cdot \|\xx\|_2.
    \end{align*}
Furthermore, \textsc{Decode}$(\yy)$ runs in $O(\epsilon^{-2} \log(\delta^{-1}) \log n)$ time.
\end{lemma}

%% file: time.tex
\section{Time Complexity of the Randomized Algorithm Using Fast Data Structures}\label{sec:time_random}
\subsection{Implementing MWU Using Fast Data Structures}
In this section, we give an algorithm, Algorithm~\ref{alg:combine_algo}, that implements Algorithm~\ref{alg:non_monotone_accel_robust} using the data structures stated from  Section~\ref{sec:data_structures}.

\begin{algorithm}
\caption{Implementing MWU Algorithm Using Fast Data Structures}\label{alg:combine_algo}
\begin{algorithmic}[1]
\Procedure{MWU-NonMonotoneRobust}{$\CC, \dd, \epsilon, a_0, a_1$} \Comment{Assume all variables are global}
\State $\alpha \gets \widetilde{\Theta}(n^{-1/2+\eta}\epsilon)$
\State $\tau \gets \widetilde{\Theta}(n^{1/2+\eta}\epsilon^{-4}), \quad \rho \gets \widetilde{\Theta}(n^{1/2-3\eta}\epsilon^{-2})$
\State $T \gets \alpha^{-1}\epsilon^{-2}\ln n$
\State $i,k = 0$
\State $b \gets \widetilde{\Theta}(n^{1/2+\eta} \epsilon^{-3})$
\State $\ww^{(0,0)} \gets 1_n, \quad \xx^{(0)} \gets 0_d$
\State Let $\SS^{(0)}, \SS^{(1)}, \cdots, \SS^{(T-1)} \in \R^{b \times n}$ be random matrices as described in Lemma~\ref{lem:CE_one_sketch}.
\State {\color{blue} Initialize data structures $\textsc{DS}_{\textsc{Inv}}$, $\textsc{DS}_{\textsc{Norm}}$, $\textsc{DS}_{\textsc{ImplicitInv}}$, $\textsc{DS}_{\textsc{HeavyHitters}}$ \Comment{Algorithm~\ref{alg:DS_Inv}, \ref{alg:DS_Norm}, \ref{alg:DS_Implicit_Inv}, \ref{alg:DS_Heavy_Hitters}}} \label{algline:initialize_ds}
\While{$i<T$}
\State $\rr^{(i,k)} \leftarrow \ww^{(i,k)} + \frac{\epsilon}{m} \sum_{e} \ww^{(i,k)}_e$
\State $\ov{\rr}^{(i,k)} \leftarrow \textsc{SelectVector}(\rr^{(i,k)})$
\State {\color{blue} $\wh{\uu}^{(i,k)} \leftarrow (\RR^{(i,k)})^{-1/2} (\SS^{(i)})^{\top} \cdot \textsc{DS}_\textsc{Inv}.\textsc{UpdateQuery}(\ov{\rr}^{(i,k)}, i)$} \label{algline:hat_u_ds}
\State {\color{blue} $\Psi,\xi \gets \textsc{DS}_{\textsc{Norm}}(\ov{\rr}^{(i,k)}, i+k)$ \Comment{$\Psi \approx_{\epsilon} \sum_e \rr^{(i,k)}_e (\uu^{(i,k)}_e)^2$, $\xi \approx_{C_3} \sum_e \rr_e^{(i,k)} |\uu^{(i,k)}_e|^3$}} \label{algline:norm_ds}
\If{$\xi \leq \rho \Psi$}
\State $\overrightarrow{\alpha}^{(i,k)}_e = \begin{cases}
\alpha \cdot (1 + \epsilon\alpha \wh{\uu}^{(i,k)}_e) & \text{ if } \wh{\uu}^{(i,k)}_e \geq 0\\
\alpha / (1 - \epsilon\alpha \wh{\uu}^{(i,k)}_e) & \text{ else }
\end{cases}$
\State $\ww^{(i+1,k)}\gets \ww^{(i,k)}\left(1 + \epsilon \overrightarrow{\alpha}^{(i,k)} \wh{\uu}^{(i,k)}\right)$  
\State {\color{blue} $\textsc{DS}_\textsc{ImplicitInv}.\textsc{Update}(\ov{\rr}^{(i,k)})$ \Comment{Implicitly update $\xx^{(i+1)} = \xx^{(i)} + \Delta^{(i,k)}$} } \label{algline:update_x_ds}
\State $i \gets i + 1$
\Else{ {\bf if} $\xi > \rho \Psi$ {\bf then}}
\State {\color{blue} $L, \uu^{(i,k)}_{L} \gets \textsc{DS}_{\textsc{HeavyHitters}}.\textsc{UpdateQuery}(\ov{\rr}^{(i,k)})$} \label{algline:heavy_hitter_ds}
\State Let $S$ be the set of coordinates $e$ such that $|\uu^{(i,k)}_e|\geq \rho/(2C_3)$ \label{algline:setSRobust_ds}
\State $H \subseteq S$ be maximal subset such that $\sum_{e\in H}\rrbar_e^{(i,k)}\leq \tau^{-1}\Psi(\rrbar^{(i,k)}) $
\If{$H \neq S$}
\State Pick any $\bar{e}\in S\setminus H$.
\State For all $e\in H \cup \{\ov{e}\} $, $\ww_e^{(i,k+1)}\gets (1+\epsilon)\ww_e^{(i,k)} + \frac{\epsilon^2}{2n}\Phi(\ww^{(i,k)})$ 
\State $\rr^{(i,k+1)}\gets \ww^{(i,k+1)} + \frac{\epsilon}{2n}\Phi(\ww^{(i,k+1)})$
\State For all $e\in H \cup \{\ov{e}\} $, $\ov{\rr}_e^{(i,k+1)} \gets \rr_e^{(i,k+1)}$ 
\Else
\For{$\zeta = \rho, 2 \rho, 4 \rho, \cdots, 2^{c_\rho} \rho$}
\State \Comment{$c_{\rho}$ is defined to be the smallest integer $c$ that satisfies $2^c \rho \geq \sqrt{n/\epsilon}$}
\State Define the set $H_\zeta = \{e \in H \mid |\wt{\CC} \Delta^{(i,k)} - \wt{\dd}|_e \in [\zeta, 2 \zeta)\}$.
\State If $\sum_{e \in H_{\zeta}} \ov{\rr}^{(i,k)}_e |\wt{\CC} \Delta^{(i,k)} - \wt{\dd}|_e^3 \geq \frac{\rho \Psi(\ov{\rr}^{(i,k)})}{\log (\frac{n}{\epsilon \rho})}$, set $\zeta^* \gets \zeta$, and break. 
\EndFor
\State For all $e \in H_{\zeta^*}$, $\ww_e^{(i,k+1)}\gets (1+\epsilon)\ww_e^{(i,k)} + \frac{\epsilon^2}{2n}\Phi(\ww^{(i,k)})$
\State $\rr^{(i,k+1)}\gets \ww^{(i,k+1)} + \frac{\epsilon}{2n}\Phi(\ww^{(i,k+1)})$
\State For all $e\in H_{\zeta^*}$, $\ov{\rr}_e^{(i,k+1)} \gets \rr_e^{(i,k+1)}$ 
\EndIf
\State $k \gets k+1$
\EndIf
\EndWhile
\State {\color{blue} $\xx^{(T)} \gets \textsc{DS}_{\textsc{ImplicitInv}}.\textsc{Query}()$} \label{algline:query_x_ds}
\State \Return $\xx^{(T)}/T$
\EndProcedure
\end{algorithmic}
\end{algorithm}

\begin{algorithm}
\caption{Inverse maintenance data structure $\textsc{DS}_{\textsc{Inv}}$ to compute $\wh{\uu}$}\label{alg:DS_Inv}
\begin{algorithmic}[1]
\Procedure{Initialize}{ }
\State $\SS \gets [(\SS^{(0)})^{\top}, (\SS^{(1)})^{\top}, \cdots, (\SS^{(T-1)})^{\top}]^{\top} \in \R^{bT \times n}$
\State $\ov{\rr} \leftarrow \rr^{(0)}$
\State Let $\NN$ be the matrix given by Lemma~\ref{lem:matrix_formula} that encodes the matrix formula
\[
f(\RR, \SS, \CC, \dd) = \SS \cdot \RR^{1/2} \Big(\CC (\CC^{\top} \RR \CC)^{-1} \CC^{\top} \RR - \II\Big) \dd,
\]
i.e., there exist index sets $I,J$ such that $(\NN^{-1})_{I,J} = f(\RR, \SS, \CC, \dd)$. Also let $I_0, I_1, \cdots, I_{T-1} \subset I$ each of size $b$, denote the indexes of the rows corresponding to $\SS^{(0)}, \SS^{(1)}, \cdots, \SS^{(T-1)}$.
\State $\textsc{DS}.\textsc{Initialize}(\NN)$ where $\textsc{DS}$ is the two-level inverse maintenance data structure of Lemma~\ref{lem:InverseMaintenanceTwoLevel}.
\EndProcedure
\Procedure{UpdateQuery}{$\ov{\rr}^{\new}$, $i$}
\State $\textsc{DS}.\textsc{Update}(\Delta)$, where $\Delta = \ov{\RR}^{\new} - \ov{\RR}$ \label{algline:DS_inv_update}
\State $\ov{\rr} \gets \ov{\rr}^{\new}$
\If{$\textsc{DS}.k_0 \geq n^{a_0}$}
\State $\textsc{DS}.\textsc{Reset}()$ \label{algline:DS_inv_reset}
\ElsIf{$\textsc{DS}.k_1 \geq n^{a_1}$}
\State $\textsc{DS}.\textsc{PartialReset}()$ \label{algline:DS_inv_partial_reset}
\EndIf
\State \Return $\textsc{DS}.\textsc{Query}(I_i, J)$ \Comment{$|I_i| = b$ and $|J| = 1$}\label{algline:DS_inv_return}
\EndProcedure
\end{algorithmic}
\end{algorithm}

\begin{algorithm}
\caption{Data structures $\textsc{DS}_{\textsc{Norm}}$ to approximately compute $\ell_2$ and $\ell_3$ norms}\label{alg:DS_Norm}
\begin{algorithmic}[1]
\Procedure{Initialize}{ }
\State Let $\JJ^{(0)}, \cdots, \JJ^{(T+K)} \in \R^{O(\epsilon^{-2} \log(n)) \times n}$ be random JL matrices as described in Lemma~\ref{lem:JL}.
\State $\JJ \gets [(\JJ^{(0)})^{\top}, (\JJ^{(1)})^{\top}, \cdots, (\JJ^{(T+K)})^{\top}]^{\top} \in \R^{O(\epsilon^{-2} \log(n) (T+K)) \times n}$
\State Let $\UU^{(0)}, \cdots, \UU^{(T+K)} \in \R^{O(n^{1/3} \log^3(n)) \times n}$ be random matrices as described in Lemma~\ref{lem:l3_norm_estimation}.
\State $\UU \gets [(\UU^{(0)})^{\top}, (\UU^{(1)})^{\top}, \cdots, (\UU^{(T+K)})^{\top}]^{\top} \in \R^{O(n^{1/3} \log^3(n) (T+K)) \times n}$
\State $\ov{\rr} \leftarrow \rr^{(0)}$
\State Let $\NN_{\ell_2}$ and $\NN_{\ell_3}$ be the matrices given by Lemma~\ref{lem:matrix_formula} that encodes the matrix formulas
\begin{align*}
f_{\ell_2}(\RR, \JJ, \CC, \dd) = &~ \JJ \RR^{1/2} (\CC (\CC^{\top} \RR \CC)^{-1} \CC^{\top} \RR - \II) \dd, \\
f_{\ell_3}(\RR, \UU, \CC, \dd) = &~ \UU \RR^{1/3} (\CC (\CC^{\top} \RR \CC)^{-1} \CC^{\top} \RR - \II) \dd,
\end{align*}
i.e., there exist index sets $I_{\ell_2},J_{\ell_2},I_{\ell_3},J_{\ell_3}$ such that $(\NN_{\ell_2}^{-1})_{I_{\ell_2},J_{\ell_2}} = f_{\ell_2}(\RR, \JJ, \CC, \dd)$ and $(\NN_{\ell_3}^{-1})_{I_{\ell_3},J_{\ell_3}} = f_{\ell_3}(\RR, \UU, \CC, \dd)$. Also let $I_{\ell_2,0}, \cdots, I_{\ell_2,T+K} \subset I_{\ell_2}$ denote the index sets of the rows corresponding to $\JJ^{(0)}, \cdots, \JJ^{(T+K)}$, and let $I_{\ell_3,0}, \cdots, I_{\ell_3,T+K} \subset I_{\ell_3}$ denote the rows corresponding to $\UU^{(0)}, \cdots, \UU^{(T+K)}$.
\State $\textsc{DS}_{\ell_2}.\textsc{Initialize}(\NN_{\ell_2})$ and $\textsc{DS}_{\ell_3}.\textsc{Initialize}(\NN_{\ell_3})$ by Lemma~\ref{lem:InverseMaintenanceTwoLevel}.
\EndProcedure
\Procedure{UpdateQuery}{$\ov{\rr}^{\new}$, $i$}
\State $\textsc{DS}_{\ell_2}.\textsc{Update}(\Delta)$ and $\textsc{DS}_{\ell_3}.\textsc{Update}(\Delta)$, where $\Delta = \ov{\RR}^{\new} - \ov{\RR}$
\State $\ov{\rr} \gets \ov{\rr}^{\new}$
\If{$\textsc{DS}_{\ell_2}.k_0 \geq n^{a_0}$}
\State $\textsc{DS}_{\ell_2}.\textsc{Reset}()$ and $\textsc{DS}_{\ell_3}.\textsc{Reset}()$
\ElsIf{$\textsc{DS}_{\ell_2}.k_1 \geq n^{a_1}$}
\State $\textsc{DS}_{\ell_2}.\textsc{PartialReset}()$ and $\textsc{DS}_{\ell_3}.\textsc{PartialReset}()$
\EndIf
\State \Return $(\|\textsc{DS}_{\ell_2}.\textsc{Query}(I_{\ell_2,i}, J_{\ell_2})\|_2^2, ~\|\textsc{DS}_{\ell_3}.\textsc{Query}(I_{\ell_3,i}, J_{\ell_3})\|_{\infty}^3)$ \label{algline:DS_norm_return}
\EndProcedure
\end{algorithmic}
\end{algorithm}

\begin{algorithm}
\caption{Implicit inverse maintenance data structure $\textsc{DS}_{\textsc{ImplicitInv}}$ to compute $\Delta$ and to update $\xx$}\label{alg:DS_Implicit_Inv}
\begin{algorithmic}[1]
\Procedure{Initialize}{ }
\State $\ov{\rr} \leftarrow \rr^{(0)}$
\State Let $\NN$ be the matrix given by Lemma~\ref{lem:matrix_formula} that encodes the matrix formula $f(\RR, \CC) = (\CC^{\top} \RR \CC)^{-1} \CC^{\top} \RR$, i.e., there exist index sets $I,J$ such that $(\NN^{-1})_{I,J} = f(\RR, \CC)$.
\State $\textsc{DS}.\textsc{Initialize}(\NN, \dd')$ where $\textsc{DS}$ is the implicit inverse maintenance data structure of Lemma~\ref{lem:implicit_inverse_maintenance}, and $\dd'$ has the same size as $\NN$, and it equals to $\dd$ in $J$, and its other coordinates are all zero. \label{algline:DS_implicit_inv_initialize}
\EndProcedure
\Procedure{Update}{$\ov{\rr}^{\new}$}
\State $\textsc{DS}.\textsc{Update}(\Delta)$, where $\Delta = \ov{\RR}^{\new} - \ov{\RR}$
\State $\ov{\rr} \gets \ov{\rr}^{\new}$
\If{$\textsc{DS}.k_0 \geq n^{a_0}$}
\State $\textsc{DS}.\textsc{Reset}()$
\ElsIf{$\textsc{DS}.k_1 \geq n^{a_1}$}
\State $\textsc{DS}.\textsc{PartialReset}()$
\EndIf
\EndProcedure
\Procedure{QuerySum}{ }
\State \Return $\textsc{DS}.\textsc{QuerySum}()$
\EndProcedure
\end{algorithmic}
\end{algorithm}

\begin{algorithm}
\caption{Heavy hitter data structure $\textsc{DS}_{\textsc{HeavyHitters}}$ to compute the heavy entries of $\uu$}\label{alg:DS_Heavy_Hitters}
\begin{algorithmic}[1]
\Procedure{Initialize}{ }
\State $\epsilon_{\heavy} \gets \frac{\rho \sqrt{\epsilon}}{2 C_3 \sqrt{n}}$
\State $\ov{\rr} \leftarrow \rr^{(0)}$
\State Let $\Phi \in \R^{O(\epsilon_{\heavy}^{-2} \log^2 n) \times n}$ be the random matrix as described in Lemma~\ref{lem:heavy_hitter}.
\State Let $\NN$ and $\NN_{\Phi}$ be the matrix given by Lemma~\ref{lem:matrix_formula} that encodes the matrix formulas
\begin{align*}
f(\RR, \CC, \dd) = &~ \RR^{1/2} (\CC (\CC^{\top} \RR \CC)^{-1} \CC^{\top} \RR - \II) \dd, \\
f_{\Phi}(\Phi, \RR, \CC, \dd) = &~ \Phi \cdot \RR^{1/2} (\CC (\CC^{\top} \RR \CC)^{-1} \CC^{\top} \RR - \II) \dd,
\end{align*}
i.e., there exist index sets $I,J, I_{\Phi}, J_{\Phi}$ such that $(\NN^{-1})_{I,J} = f(\RR, \CC, \dd)$ and $(\NN_{\Phi}^{-1})_{I_{\Phi},J_{\Phi}} = f_{\Phi}(\Phi, \RR, \CC, \dd)$.
\State $\textsc{DS}.\textsc{Initialize}(\NN)$, $\textsc{DS}_{\Phi}.\textsc{Initialize}(\NN_{\Phi})$, where $\textsc{DS}$ and $\textsc{DS}_{\Phi}$ are both the inverse maintenance data structure of Lemma~\ref{lem:InverseMaintenanceTwoLevel}.
\EndProcedure
\Procedure{UpdateQuery}{$\ov{\rr}^{\new}$}
\State $\textsc{DS}.\textsc{Update}(\Delta)$ and $\textsc{DS}_{\Phi}.\textsc{Update}(\Delta)$, where $\Delta = \ov{\RR}^{\new} - \ov{\RR}$
\State $\ov{\rr} \gets \ov{\rr}^{\new}$
\If{$\textsc{DS}.k_0 \geq n^{a_0}$}
\State $\textsc{DS}.\textsc{Reset}()$ and $\textsc{DS}_{\Phi}.\textsc{Reset}()$
\ElsIf{$\textsc{DS}.k_1 \geq n^{a_1}$}
\State $\textsc{DS}.\textsc{PartialReset}()$ and $\textsc{DS}_{\Phi}.\textsc{PartialReset}()$
\EndIf
\State $\yy \gets \textsc{DS}_{\Phi}.\textsc{Query}(I_{\Phi},J_{\Phi})$ \label{algline:DS_heavy_hitters_y}
\State $L \gets \textsc{Decode}(\yy)$, where $\textsc{Decode}()$ is the decoding algorithm of Lemma~\ref{lem:heavy_hitter}. We can view $L \subset [n]$ as a subset of $I$. \label{algline:DS_heavy_hitters_L}
\State \Return $(L, \textsc{DS}.\textsc{Query}(L, J))$ \label{algline:DS_heavy_hitters_return}
\EndProcedure
\end{algorithmic}
\end{algorithm}

\subsection{Correctness of Algorithm}
\begin{lemma}[Correctness of Algorithm~\ref{alg:combine_algo}]
The output of Algorithm~\ref{alg:combine_algo} is the same as that of Algorithm~\ref{alg:non_monotone_accel_robust}.
\end{lemma}
\begin{proof}
Algorithm~\ref{alg:combine_algo} implements Algorithm~\ref{alg:non_monotone_accel_robust} by using the data structures $\textsc{DS}_{\textsc{Inv}}$, $\textsc{DS}_{\textsc{ImplicitInv}}$, $\textsc{DS}_{\textsc{Norm}}$, $\textsc{DS}_{\textsc{HeavyHitters}}$. So it suffices to prove that all these data structures are correct.

{\bf Compute $\wh{\uu}$ by $\textsc{DS}_{\textsc{Inv}}$.} We first prove that on Line~\ref{algline:hat_u_ds} of Algorithm~\ref{alg:combine_algo}, the computed vector $\wh{\uu}^{(i,k)} \leftarrow (\ov{\RR}^{(i,k)})^{-1/2} (\SS^{(i)})^{\top} \cdot \textsc{DS}_\textsc{Inv}.\textsc{UpdateQuery}(\ov{\rr}^{(i,k)}, i)$ satisfies
\[
\wh{\uu}^{(i,k)} = (\ov{\RR}^{(i,k)})^{-1/2} \cdot (\SS^{(i)})^{\top} \SS^{(i)} \cdot (\ov{\RR}^{(i,k)})^{1/2} \Big(\CC (\CC^{\top} \ov{\RR}^{(i,k)} \CC)^{-1} \CC^{\top} \ov{\RR}^{(i,k)} - \II\Big) \dd,
\]
as required by Line~\ref{algline:hat_u_Robust} of Algorithm~\ref{alg:non_monotone_accel_robust}. 

The data structure $\textsc{DS}_{\textsc{Inv}}$ (Algorithm~\ref{alg:DS_Inv}) uses the two-level inverse maintenance data structure of Lemma~\ref{lem:InverseMaintenanceTwoLevel} to maintain the inverse of matrix $\NN$ that by Lemma~\ref{lem:matrix_formula} encodes the matrix formula
\[
f(\ov{\RR}, \SS, \CC, \dd) = \begin{bmatrix}
\SS^{(0)} \\ \vdots \\ \SS^{(T-1)}
\end{bmatrix} \cdot \ov{\RR}^{1/2} \Big(\CC (\CC^{\top} \ov{\RR} \CC)^{-1} \CC^{\top} \ov{\RR} - \II\Big) \dd.
\]
$\textsc{DS}_{\textsc{Inv}}$ maintains that its internal variable $\ov{\RR} = \ov{\RR}^{(i,k)}$ in each iteration, since we update it on Line~\ref{algline:DS_inv_update} of Algorithm~\ref{alg:DS_Inv}. The output to $\textsc{DS}_\textsc{Inv}.\textsc{UpdateQuery}(\ov{\rr}^{(i,k)}, i)$ is (see Line~\ref{algline:DS_inv_return} of Algorithm~\ref{alg:DS_Inv})
\[
(\NN^{-1})_{I_i,J} = \SS^{(i)} \cdot (\ov{\RR}^{(i,k)})^{1/2} \Big(\CC (\CC^{\top} \ov{\RR}^{(i,k)} \CC)^{-1} \CC^{\top} \ov{\RR}^{(i,k)} - \II\Big) \dd.
\]
So we have that the $\wh{\uu}^{(i,k)}$ is computed as required.

{\bf Compute $\xx$ by $\textsc{DS}_{\textsc{ImplicitInv}}$.} Next we prove that Line~\ref{algline:update_x_ds} of Algorithm~\ref{alg:combine_algo} implicitly updates $\xx^{(i+1)} \gets \xx^{(i)} + \Delta^{(i,k)}$, as required by Line~\ref{algline:updateXRobust} of Algorithm~\ref{alg:non_monotone_accel_robust}, and that Line~\ref{algline:query_x_ds} of Algorithm~\ref{alg:combine_algo} outputs the correct $\xx^{(T)}$.

The data structure $\textsc{DS}_{\textsc{ImplicitInv}}$ (Algorithm~\ref{alg:DS_Implicit_Inv}) uses the implicit inverse maintenance data structure of Lemma~\ref{lem:implicit_inverse_maintenance} to maintain the inverse of matrix $\NN$ that by Lemma~\ref{lem:matrix_formula} encodes the matrix formula
\[
f(\ov{\RR}, \CC) = (\CC^{\top} \ov{\RR} \CC)^{-1} \CC^{\top} \ov{\RR}.
\]
We also initialize this data structure with the vector $\dd'$ (Line~\ref{algline:DS_implicit_inv_initialize} of Algorithm~\ref{alg:DS_Implicit_Inv}), and by Lemma~\ref{lem:implicit_inverse_maintenance} the algorithm maintains 
\begin{align*}
\sum_{i=0}^{(T-1)} \Big((\NN^{(i)})^{-1} \cdot \dd'\Big)_{I,:} = &~ \sum_{i=0}^{(T-1)} f(\ov{\RR}^{(i,k)}, \CC) \cdot \dd \\
= &~ \sum_{i=0}^{(T-1)} (\CC^{\top} \ov{\RR}^{(i,k)} \CC)^{-1} \CC^{\top} \ov{\RR}^{(i,k)} \cdot \dd,
\end{align*}
and this is exactly $\xx^{(T)} = \sum_{i=0}^{(T-1)} \Delta^{(i,k)}$ where $\Delta^{(i,k)} = (\CC^{\top} \ov{\RR}^{(i,k)} \CC)^{-1} \CC^{\top} \ov{\RR}^{(i,k)} \dd$ is what we need to compute (see Line~\ref{algline:linsysRobust} of Algorithm~\ref{alg:non_monotone_accel_robust}).

{\bf Compute norms by $\textsc{DS}_{\textsc{Norm}}$.} Next we show that Line~\ref{algline:norm_ds} of Algorithm~\ref{alg:combine_algo} computes
\[
\Psi \approx_{\epsilon} \sum_e \ov{\rr}^{(i,k)}_e (\uu^{(i,k)}_e)^2, ~~\text{and}~~ \xi \approx_{C_3} \sum_e \ov{\rr}_e^{(i,k)} |\uu^{(i,k)}_e|^3.
\]
Similar to the proof for $\textsc{DS}_{\textsc{Inv}}$, in each iteration $\textsc{DS}_{\textsc{Norm}}$ (see Algorithm~\ref{alg:DS_Norm}) maintains
\begin{align*}
&~ \JJ \cdot (\ov{\RR}^{(i,k)})^{1/2} (\CC (\CC^{\top} \ov{\RR}^{(i,k)} \CC)^{-1} \CC^{\top} \ov{\RR}^{(i,k)} - \II) \dd, \\
~ \text{and} &~ \UU \cdot (\ov{\RR}^{(i,k)})^{1/3} (\CC (\CC^{\top} \ov{\RR}^{(i,k)} \CC)^{-1} \CC^{\top} \ov{\RR}^{(i,k)} - \II) \dd.
\end{align*}
And on Line~\ref{algline:DS_norm_return} of Algorithm~\ref{alg:DS_Norm} it outputs
\begin{align*}
\Psi = \|\JJ^{(i)} \cdot (\ov{\RR}^{(i,k)})^{1/2} \uu^{(i,k)} \|_2^2, ~
\text{and}~\xi = \|\UU^{(i)} \cdot (\ov{\RR}^{(i,k)})^{1/3} \uu^{(i,k)}\|_{\infty}^3,
\end{align*}
where $\uu^{(i,k)} = (\CC (\CC^{\top} \ov{\RR}^{(i,k)} \CC)^{-1} \CC^{\top} \ov{\RR}^{(i,k)} - \II) \dd$ is as required by Line~\ref{algline:u_Robust} of Algorithm~\ref{alg:non_monotone_accel_robust}.

Then by Lemma~\ref{lem:JL} and \ref{lem:l3_norm_estimation}, and since we use a new random matrix $\JJ^{(i)}$ and $\UU^{(i)}$ in each iteration, we have that with probability $1 - 1/n^e$, for all iterations we have
\begin{align*}
(1 - \epsilon) \|(\ov{\RR}^{(i,k)})^{1/2} \uu^{(i,k)} \|_2^2 \leq \Psi & \leq (1 + \epsilon) \|(\ov{\RR}^{(i,k)})^{1/2} \uu^{(i,k)} \|_2^2, \\
C_3^{-1} \|(\ov{\RR}^{(i,k)})^{1/3} \uu^{(i,k)} \|_3^3 \leq \xi & \leq C_3 \|(\ov{\RR}^{(i,k)})^{1/3} \uu^{(i,k)} \|_3^3.
\end{align*}

{\bf Compute set $S$ for the width reduction step by $\textsc{DS}_{\textsc{HeavyHitters}}$.} Finally we prove that Line~\ref{algline:heavy_hitter_ds} of Algorithm~\ref{alg:combine_algo} computes the set $S = \{e : |\uu^{(i,k)}_e|\geq \rho/(2C_3)\}$ and the values of $\uu^{(i,k)}_e$ for all $e \in S$, as by Line~\ref{algline:setSRobust} of Algorithm~\ref{alg:non_monotone_accel_robust}. 

Similar to the proof for $\textsc{DS}_{\textsc{Inv}}$, in each iteration $\textsc{DS}_{\textsc{HeavyHitters}}$ (see Algorithm~\ref{alg:DS_Heavy_Hitters}) maintains
\begin{align*}
&~ \Phi \cdot (\ov{\RR}^{(i,k)})^{1/2} (\CC (\CC^{\top} \ov{\RR}^{(i,k)} \CC)^{-1} \CC^{\top} \ov{\RR}^{(i,k)} - \II) \dd, \\
~ \text{and} &~ (\ov{\RR}^{(i,k)})^{1/2} (\CC (\CC^{\top} \ov{\RR}^{(i,k)} \CC)^{-1} \CC^{\top} \ov{\RR}^{(i,k)} - \II) \dd.
\end{align*}
And on Line~\ref{algline:DS_heavy_hitters_y} of Algorithm~\ref{alg:DS_Heavy_Hitters} it computes
\[
y = \Phi \cdot (\ov{\RR}^{(i,k)})^{1/2} \uu^{(i,k)},
\]
where $\uu^{(i,k)} = (\CC (\CC^{\top} \ov{\RR}^{(i,k)} \CC)^{-1} \CC^{\top} \ov{\RR}^{(i,k)} - \II) \dd$ is as required by Line~\ref{algline:u_Robust} of Algorithm~\ref{alg:non_monotone_accel_robust}.

Then on Line~\ref{algline:DS_heavy_hitters_L} of Algorithm~\ref{alg:DS_Heavy_Hitters} it decodes $\yy$ and compute the set $L$, and by Lemma~\ref{lem:heavy_hitter}, with probability $1 - 1/n^4$, $L$ includes all $e \in [n]$ that satisfies
\begin{align*}
    |(\ov{\rr}^{(i,k)})_e^{1/2} \uu^{(i,k)}_e| \geq \epsilon_{\heavy} \cdot \|(\ov{\rr}^{(i,k)})^{1/2} \uu^{(i,k)}\|_2.
\end{align*}
Note that for any $e \in S$, we have
\begin{align*}
|\uu^{(i,k)}_e|\geq \rho/(2C_3) 
\implies &~ \rr^{(i,k)}_e \cdot (\uu^{(i,k)}_e)^2 \geq \frac{\rho^2}{4 C_3^2} \cdot \rr^{(i,k)}_e \\
\implies &~ \rr^{(i,k)}_e \cdot (\uu^{(i,k)}_e)^2 \geq \frac{\rho^2 \epsilon}{4 C_3^2 n} \cdot \Psi(\rr^{(i,k)}) \\
\implies &~ \rr^{(i,k)}_e \cdot (\uu^{(i,k)}_e)^2 \geq \epsilon_{\heavy}^2 \cdot \|(\ov{\rr}^{(i,k)})^{1/2} \uu^{(i,k)}\|_2^2,
\end{align*}
where the second step follows from $\rr^{(i,k)}_e \geq \frac{\epsilon}{n} \Psi(\rr^{(i,k)})$, and the third follows from $\epsilon_{\heavy} = \frac{\rho \sqrt{\epsilon}}{2 C_3 \sqrt{n}}$. This means we have
\[
S \subseteq L,
\]
and it suffices to enumerate all $e \in L$ to check if $|\uu^{(i,k)}_e|\geq \rho/(2C_3)$ and compute $S$.

In the output on Line~\ref{algline:DS_heavy_hitters_return} of Algorithm~\ref{alg:DS_Heavy_Hitters}, we also output $\textsc{Query}(L, J)$ which computes $\uu^{(i,k)}_L$ exactly.

Finally, note that we can re-use the random matrix $\Phi$ because the set $S$ and $\uu_e^{(i,k)}$ for $e \in S$ are computed exactly, so the next iteration does not depend on the randomness of $\Phi$.
\end{proof}

\subsection{Time Complexity under \texorpdfstring{$\ell_2$}{} Stability}\label{sec:RuntimeL2}
In this section we bound the time complexity of Algorithm~\ref{alg:combine_algo}.
\begin{theorem}[Time complexity of Algorithm~\ref{alg:combine_algo}]\label{thm:time_combine}
For any parameters that satisfy $a_0 \leq \alpha_*$ and $a_1 \leq a_0 \cdot \alpha_*$, the time complexity of Algorithm~\ref{alg:combine_algo} is
\begin{align*}
\wt{O}\left(n^{\omega} + n^{2.5-a_0/2} + n^{1.5+a_0-a_1/2} + n^{1/2-\eta + a_0 + (\omega-1) a_1} \right) \cdot \poly(\epsilon^{-1}).
\end{align*}
In particular, when $\omega = 2 + o(1)$, this time complexity is bounded by
\[
\wt{O}\left(n^{2+1/22.5} \right) \poly(\epsilon^{-1}).
\]
\end{theorem}
\begin{proof}
The dominating steps of Algorithm~\ref{alg:combine_algo} are the operations involving the data structures, since all other operations can be computed in $O(n \log n)$ time per iteration. We will focus on bounding the runtimes of the data structure operations.

\paragraph{Initialization} The \textsc{Initialize} operations of the four data structures $\textsc{DS}_{\textsc{Inv}}$, $\textsc{DS}_{\textsc{ImplicitInv}}$, $\textsc{DS}_{\textsc{Norm}}$, $\textsc{DS}_{\textsc{HeavyHitters}}$ are each called once on Line~\ref{algline:initialize_ds} of Algorithm~\ref{alg:combine_algo}.
\begin{itemize}
\item $\textsc{DS}_{\textsc{Inv}}$: By Lemma~\ref{lem:matrix_formula}, the matrix $\NN$ of Algorithm~\ref{alg:DS_Inv} has size $N \times N$ where
\[
N = \max\{bT, n\} = \widetilde{\Theta}(n \epsilon^{-5} \ln^5 n),
\]
since $T = \alpha^{-1}\epsilon^{-2}\ln n$,  $b = \frac{n \alpha \log^4 n}{\epsilon^3}$. By Lemma~\ref{lem:InverseMaintenanceTwoLevel}, the initialization time of the two level inverse maintenance data structure is 
\[
O(N^{\omega}) = \widetilde{\Theta}(n^{\omega} \epsilon^{-5 \omega}).
\]
\item $\textsc{DS}_{\textsc{ImplicitInv}}$: By Lemma~\ref{lem:matrix_formula}, the matrix $\NN$ of Algorithm~\ref{alg:DS_Implicit_Inv} has size $\Otil_{\epsilon}(n) \times \Otil_{\epsilon}(n)$. By Lemma~\ref{lem:implicit_inverse_maintenance}, the initialization time of the implicit inverse maintenance data structure is $\Otil_{\epsilon}(n^{\omega})$.
\item $\textsc{DS}_{\textsc{Norm}}$: By Lemma~\ref{lem:matrix_formula}, the matrices $\NN_{\ell_2}$ and $\NN_{\ell_3}$ of Algorithm~\ref{alg:DS_Norm} have size $N_{\ell_2} \times N_{\ell_2}$ and $N_{\ell_3} \times N_{\ell_3}$, where
\begin{align*}
N_{\ell_2} = &~ \max\{O(\epsilon^{-2} \log(n) (T+K)), n\} = \wt{\Theta}_{\epsilon}(\max\{n^{1/2-\eta} , n\}), \\
N_{\ell_3} = &~ \max\{O(n^{1/3} \log^3(n) (T+K)), n\} = \wt{\Theta}_{\epsilon}(\max\{n^{5/6 - \eta}, n\}).
\end{align*} 
From Lemma~\ref{lem:InverseMaintenanceTwoLevel}, the initialization time of the two inverse maintenance data structures is
\[
O(N_{\ell_2}^{\omega} + N_{\ell_3}^{\omega}) = \wt{\Theta}_{\epsilon}(\max\{n^{(1/2-\eta) \omega}, n^{(5/6 - \eta) \omega}, n^{\omega}\}).
\]
\item $\textsc{DS}_{\textsc{HeavyHitters}}$: By Lemma~\ref{lem:matrix_formula}, the matrix $\NN$ of Algorithm~\ref{alg:DS_Norm} has size $O(n) \times O(n)$, and $\NN_{\Phi}$ has size $N_{\Phi} \times N_{\Phi}$, where
\[
N_{\Phi} = \max\{O(\epsilon_{\heavy}^{-2} \log^2 n), n\} = \wt{\Theta}_{\epsilon}(\max\{n^{6 \eta}, n\})
\]
since $\epsilon_{\heavy} = \frac{\rho \sqrt{\epsilon}}{2 C_3 \sqrt{n}}$ and $\rho = \widetilde{\Theta}(n^{1/2-3\eta}\epsilon^{-2})$. By Lemma~\ref{lem:InverseMaintenanceTwoLevel}, the initialization time of the two inverse maintenance data structures is
\[
O(n^{\omega} + N_{\Phi}^{\omega}) = \wt{\Theta}_{\epsilon}(\max\{n^{6 \eta \omega} , n^{\omega}\}).
\]
\end{itemize}
Summing up all these terms, and since we set $\eta = 1/10$, the total initialization time is
\[
\wt{O}_{\epsilon}(n^{\omega}).
\]

\paragraph{Reset} Since Algorithm~\ref{alg:combine_algo} implements Algorithm~\ref{alg:non_monotone_accel_robust}, it satisfies the low-rank update scheme of Theorem~\ref{thm:RobustAccMWU}, so we have that the sequence $\ov{\rr}^{(i,k)}$ undergoes at most $\frac{T+K}{2^{\ell}}$ 
number of updates of size 
\begin{equation}\label{eq:low_rank_update_size}
\Otil_{\epsilon}\left(\left(\frac{\log n}{\delta}\right)^2 \cdot n^{2 \eta} \cdot 2^{2 \ell}\right)
\end{equation}
for every $\ell \in [0:\log T]$. 

Note that the four data structures $\textsc{DS}_{\textsc{Inv}}$, $\textsc{DS}_{\textsc{ImplicitInv}}$, $\textsc{DS}_{\textsc{Norm}}$, $\textsc{DS}_{\textsc{HeavyHitters}}$ all follow the same reset and partial reset scheme. And since the matrix maintained by $\textsc{DS}_{\textsc{Inv}}$ has the largest size $N = \widetilde{\Theta}(n \epsilon^{-5} )$, it suffices to bound the reset and partial reset time of $\textsc{DS}_{\textsc{Inv}}$.

As stated on Line~\ref{algline:DS_inv_reset} of Algorithm~\ref{alg:DS_Inv}, we only perform the \textsc{Reset} operation in the $i$-th iteration if $k_0 \geq n^{a_0}$, where $k_0 = \nnz(\ov{\rr}^{(i,k)} - \ov{\rr}_0)$, and $\ov{\rr}_0$ is the variable maintained by the data structure that was updated in the last \textsc{Reset} (see Lemma~\ref{lem:InverseMaintenanceTwoLevel}). By the low rank update size of Eq.~\eqref{eq:low_rank_update_size}, for any $a \in [a_0,1]$, we only accumulate updates of size $n^{a}$ for at most $\frac{T+K}{n^{a/2 - \eta}} \cdot \frac{\log n}{\delta}$ number of times. For convenience for any $a$ we define a parameter $\ell = \log (n^{a/2 - \eta} \cdot \frac{\delta}{\log n})$ such that we get a update of size $n^a = \left(\frac{\log n}{\delta}\right)^2 \cdot n^{2 \eta} \cdot 2^{2 \ell}$ for at most $(T+K)/2^{\ell}$ times. Let $\ell_0 = \log (n^{a_0/2 - \eta} \cdot \frac{\delta}{\log n})$, and by Lemma~\ref{lem:InverseMaintenanceTwoLevel} the total reset time over all iterations is
\begin{align*}
\Otil_{\epsilon}\left(\sum_{\ell = \ell_0}^{\log T} \frac{T+K}{2^{\ell}} \cdot \Tmat\left(N, N, \left(\frac{\log n}{\delta}\right)^2 \cdot n^{2 \eta} \cdot 2^{2 \ell}\right)\right). 
\end{align*}
Define $x(\ell) = \log_N((\frac{\log n}{\delta})^2 \cdot n^{2 \eta} \cdot 2^{2 \ell}) = \log_N((\frac{\log n}{\delta})^2 \cdot n^{2 \eta}) + \frac{2 \ell}{\log N}$, then since we defined $\omega(x)$ such that $\Tmat(n,n,n^x) = n^{\omega(x) + o(1)}$, we have the above time complexity equals to
\begin{align*}
\Otil_{\epsilon}\left((T+K) \cdot \sum_{\ell = \ell_0}^{\log T} N^{\omega(x(\ell)) - \frac{\ell}{\log N}}\right). 
\end{align*}
By Fact~\ref{fact:FMM_convex} we know that $\omega(x)$ is convex, and hence the function $f(\ell) = \omega(x(\ell)) - \frac{\ell}{\log N}$ is also convex. So we have that the summation is upper bounded by the terms $\ell=\ell_0$ and $\ell=\log T$. And so the above time complexity is bounded by
\begin{align*}
&~ \Otil_{\epsilon}\left(\frac{T+K}{2^{\ell_0}} \cdot \Tmat\Big(N, N, \left(\frac{\log n}{\delta}\right)^2 \cdot n^{2 \eta} \cdot 2^{2 \ell_0}\Big) + \Tmat\Big(N, N, \left(\frac{\log n}{\delta}\right)^2 \cdot n^{2 \eta} \cdot T^2 \Big)\right) \\ 
\leq &~ \wt{O}\left(n^{1/2-a_0/2} \cdot \Tmat\Big(n, n, n^{a_0} \Big) + \Tmat\Big(n, n, n\Big)\right) \cdot \poly(\epsilon^{-1}),
\end{align*}
since $\ell_0 = \log (n^{a_0/2 - \eta} \cdot \frac{\delta}{\log n})$, $N = \widetilde{\Theta}_{\epsilon}(n )$, $T = \alpha^{-1}\epsilon^{-2}\ln n$, $\alpha = \widetilde{\Theta}_{\epsilon}(n^{-1/2+\eta})$. 

\paragraph{Partial reset} Similar to the \textsc{Reset} operation, it suffices to bound the partial reset time of $\textsc{DS}_{\textsc{Inv}}$. As stated on Line~\ref{algline:DS_inv_partial_reset} of Algorithm~\ref{alg:DS_Inv}, we only perform the \textsc{PartialReset} operation in the $i$-th iteration if $k_1 \geq n^{a_1}$, where $k_1 = \nnz(\ov{\rr}^{(i,k)} - \ov{\rr}_1)$, and $\ov{\rr}_1$ is the variable maintained by the data structure that was updated in the last \textsc{PartialReset} (see Lemma~\ref{lem:InverseMaintenanceTwoLevel}).  Let $\ell_1 = \log (n^{a_1/2 - \eta} \cdot \frac{\delta}{\log n})$ such that we perform a \textsc{PartialReset} of size $n^{a_1}$ for at most $\frac{T+K}{2^{\ell_1}}$ times, and by Lemma~\ref{lem:InverseMaintenanceTwoLevel} the total partial reset time over all iterations is
\begin{align*}
\Otil_{\epsilon}\left(\sum_{\ell = \ell_1}^{\log T} \frac{T+K}{2^{\ell}} \cdot \Tmat\Big(N, n^{a_0}, \left(\frac{\log n}{\delta}\right)^2 \cdot n^{2 \eta} \cdot 2^{2 \ell}\Big)\right). 
\end{align*}
Using the same argument as for reset operation and using the convexity of $\omega_{a_0}(x)$ (Fact~\ref{fact:FMM_convex}), we have that the summation is upper bounded by the terms $\ell=\ell_1$ and $\ell = \log T$. And so the above time complexity is bounded by
\begin{align*}
&~ \Otil_{\epsilon}\left(\frac{T+K}{2^{\ell_1}} \cdot \Tmat\Big(N, n^{a_0}, \left(\frac{\log n}{\delta}\right)^2 \cdot n^{2 \eta} \cdot 2^{2 \ell_1}\Big) + \Tmat\Big( N, n^{a_0}, \left(\frac{\log n}{\delta}\right)^2 \cdot n^{2 \eta} \cdot T^2\Big)\right) \\
\leq &~ \wt{O}\left( n^{1/2-a_1/2} \cdot \Tmat(n, n^{a_0}, n^{a_1}) + \Tmat(n, n^{a_0}, n) \right) \cdot \poly(\epsilon^{-1})
\end{align*}
since $\ell_1 = \log (n^{a_1/2 - \eta} \cdot \frac{\delta}{\log n})$, $N = \widetilde{\Theta}_{\epsilon}(n )$, $T = \alpha^{-1}\epsilon^{-2}\ln n$, $\alpha = \widetilde{\Theta}_{\epsilon}(n^{-1/2+\eta})$. 

\paragraph{Query.} Next we bound the runtime of the query operations of the four data structures.
\begin{itemize}
\item $\textsc{DS}_{\textsc{Inv}}$: In $\textsc{DS}_{\textsc{Inv}}$ (Algorithm~\ref{alg:DS_Inv}), the query operation is called with sets $|I_i| = b$ and $|J| = 1$, and by Lemma~\ref{lem:InverseMaintenanceTwoLevel} its runtime per iteration is
\begin{align*}
O\big(\Tmat(n^{a_0}, n^{a_1}, n^{a_1}) + n^{a_0} \cdot b\big) = \wt{O}_{\epsilon}\big(\Tmat(n^{a_0}, n^{a_1}, n^{a_1}) + n^{a_0 + 1/2 + \eta} \cdot \big)
\end{align*}
since $b = \wt{\Theta}_{\epsilon}(n^{1/2+\eta} )$.
\item $\textsc{DS}_{\textsc{ImplicitInv}}$: In Algorithm~\ref{alg:combine_algo}, the \textsc{QuerySum} operation of $\textsc{DS}_{\textsc{ImplicitInv}}$ (Algorithm~\ref{alg:DS_Implicit_Inv}) is only called once in the algorithm, and by Lemma~\ref{lem:implicit_inverse_maintenance} its runtime over all iterations is $O(n^2)$.
\item $\textsc{DS}_{\textsc{Norm}}$: In $\textsc{DS}_{\textsc{Norm}}$ (Algorithm~\ref{alg:DS_Norm}), the query operation is called with sets $|I_{\ell_2,i}| = O(\epsilon^{-2} \log(n))$ and $|J_{\ell_2}| = 1$, and with sets $|I_{\ell_3,i}| = O(n^{1/3} \log^3(n))$ and $|J_{\ell_3}| = 1$, and by Lemma~\ref{lem:InverseMaintenanceTwoLevel} its runtime per iteration is
\begin{align*}
O\Big(\Tmat(n^{a_0}, n^{a_1}, n^{a_1}) + n^{a_0} \cdot \big(\epsilon^{-2} \log(n) + n^{1/3} \log^3(n)\big)\Big).
\end{align*}
\item $\textsc{DS}_{\textsc{HeavyHitters}}$: In $\textsc{DS}_{\textsc{HeavyHitters}}$ (Algorithm~\ref{alg:DS_Heavy_Hitters}), the query operation is called with sets $|I_{\Phi}| = O(\epsilon_{\heavy}^{-2} \log^2 n)$ and $|J_{\Phi}| = 1$, and with sets $|L| = O(\epsilon_{\heavy}^{-2})$ and $|J| = 1$, and by Lemma~\ref{lem:InverseMaintenanceTwoLevel} its runtime per iteration is
\begin{align*}
O\big(\Tmat(n^{a_0}, n^{a_1}, n^{a_1}) + n^{a_0} \cdot \epsilon_{\heavy}^{-2} \log^2 n\big) 
= \wt{O}_{\epsilon}\big(\Tmat(n^{a_0}, n^{a_1}, n^{a_1}) + n^{a_0 + 6 \eta} \cdot \big),
\end{align*}
since $\epsilon_{\heavy} = \frac{\rho \sqrt{\epsilon}}{2 C_3 \sqrt{n}}$ and $\rho = \widetilde{\Theta}_{\epsilon}(n^{1/2-3\eta})$.
\end{itemize}
Combining these four query time, we have that over all iterations, the total query time is
\begin{align*}
&~ \wt{O}\left(T \cdot \Big(\Tmat(n^{a_0}, n^{a_1}, n^{a_1}) + n^{a_0 + 1/2 + \eta} + n^{a_0 + 1/3} + n^{a_0 + 6 \eta} \Big) \right) \cdot \poly(\epsilon^{-1}) \\
= &~ \wt{O}\left(n^{1/2-\eta} \cdot \Tmat(n^{a_0}, n^{a_1}, n^{a_1}) + n^{1 + a_0} \right) \cdot \poly(\epsilon^{-1}),
\end{align*}
since $T = \alpha^{-1}\epsilon^{-2}\ln n$, $\alpha = \widetilde{\Theta}_{\epsilon}(n^{-1/2+\eta})$, and we set $\eta = 1/10$.

\paragraph{Total runtime.} Combining the time complexities of initialization, reset, partial reset, and query, we have that the total time complexity is bounded by $\poly(\epsilon^{-1})$ times
\begin{align*}
\wt{O}\left(n^{\omega} + n^{1/2-a_0/2} \cdot \Tmat(n,n,n^{a_0}) + n^{1/2-a_1/2} \cdot \Tmat(n,n^{a_0},n^{a_1}) + n^{1/2-\eta} \cdot \Tmat(n^{a_0}, n^{a_1}, n^{a_1}) \right).
\end{align*}

When $a_0 \leq \alpha_*$ and $a_1 \leq a_0 \cdot \alpha_*$, this becomes
\begin{align*}
\wt{O}\left(n^{\omega} + n^{2.5-a_0/2} + n^{1.5+a_0-a_1/2} + n^{1/2-\eta + a_0 + (\omega-1) a_1} \right) \cdot \poly(\epsilon^{-1}).
\end{align*}
\end{proof}

\paragraph{Time complexity when $\omega = 2$.}
We remark that when $\omega = 2$, we can choose the optimal trade-off $a_0 = 1 - \frac{1 - 2 \eta}{9}$ and $a_1 = 1 - \frac{1 - 2 \eta}{3}$, and the time complexity becomes
\begin{align*}
\wt{O}\left(n^{2 + (1 - 2\eta) / 18} \right) \cdot \poly(\epsilon^{-1}).
\end{align*}

Since we choose $\eta = 1/10$, this become 
\begin{align*}
\wt{O}\left(n^{2 + 1 / 22.5} \right) \cdot \poly(\epsilon^{-1}).
\end{align*}

%% file: time_deterministic.tex
\section{Time Complexity of the Deterministic Algorithm Using Fast Data Structures}\label{sec:time_deterministic}
In this section we show that we can implement Algorithm~\ref{alg:MWU} by using the one-level inverse maintenance data structure (Lemma~\ref{lem:InverseMaintenanceOneLevel}). This algorithm is deterministic since we don't use any randomized techniques.

\begin{theorem}[Time complexity of Algorithm~\ref{alg:MWU} using one-level inverse maintenance]\label{thm:time_deterministic}
For any parameters $\ell_0, \cdots, \ell_{\log T} \in [0,\log T]$, the time complexity of Algorithm~\ref{alg:MWU} when using one-level inverse maintenance is
\begin{align*}
\wt{O}\left(n^{\omega} + n^{5/3} \cdot \sum_{j=0}^{\log T} 2^{3 \ell_j - j} + \sum_{j=0}^{\log T} \Tmat(n,n, n^{1/3} \cdot 2^{3 \ell_j - j}) \cdot 2^{j - \ell_j}\right) \poly(\epsilon^{-1}).
\end{align*}

In particular, when $\omega = 2 + o(1)$, this time complexity is bounded by
\[
\wt{O}\left(n^{2+1/12} \right) \poly(\epsilon^{-1}).
\]
\end{theorem}

\begin{proof}
By Lemma~\ref{lem:matrix_formula} there exists a matrix $\NN(\rrbar) \in \rea^{O(n) \times O(n)}$ such that $\Delta^{(i,k)}$ and $\CC \Delta^{(i,k)} - \dd$ can be read off from a column of $\NN(\rrbar^{(i,k)})^{-1}$, and it undergoes coordinate updates to $\rrbar$. We maintain this inverse using the one-level data structure of Lemma~\ref{lem:InverseMaintenanceOneLevel}. 
We use the one-level inverse maintenance data structure to query for $\Delta^{(i,k)}$ and $\uu^{(i,k)}$ exactly in each iteration, and the correctness is straightforward.

Note that the $\rr^{(i,k)}$'s in Algorithm~\ref{alg:MWU} follows the update scheme of Lemma~\ref{lem:LowRankL3} with $\zeta = \wt{O}(n^{1/3} \epsilon^{2/3})$. We maintain a vector $\rrbar^{(i,k)} \approx_{\delta} \rr^{(i,k)}$ for Algorithm~\ref{alg:MWU}, and we update $\rrbar^{(i,k)}$ in a slightly different way than Line~\ref{algline:MWU_select_vector} of Algorithm~\ref{alg:MWU}: We update it by \textsc{SelectVectorL3} (Algorithm~\ref{alg:select_vector_L3}) instead of \textsc{SelectVector}.

Let $\ell_0, \ell_1, \cdots, \ell_{\log T} \in [0,\log T]$ be parameters to be determined later. For any $j \in [0: \log T]$, we perform an update operation of the inverse maintenance data structure of Lemma~\ref{lem:InverseMaintenanceOneLevel} in all iterations $B_j[k]$ where $k \equiv 0 \pmod {2^{\ell_j}}$. Note that by Corollary~\ref{cor:LowRankL3} this update has size $k_j = O\left(\zeta \cdot 2^{3 \ell_j - j} \cdot \frac{\log^6 n}{\delta^3} \right)$. Next we compute the time complexity of all query and update operations.
 and note that it follows the update scheme of Lemma~\ref{lem:LowRankL3}.

{\bf Query.} The total number of updated coordinates when performing a query operation is at most
\begin{align*}
k = \sum_{j=0}^{\log T} k_j = \sum_{j=0}^{\log T} O\left(\zeta \cdot 2^{3 \ell_j - j} \cdot \frac{\log^6 n}{\delta^3} \right)
\end{align*}
Note that we can assume $k \leq n$ since otherwise the query is trivial. So the query time in each iteration is
\begin{align*}
O(k^{\omega} + nk) = n \cdot \sum_{j=0}^{\log T} \wt{O}\left(n^{1/3} \cdot 2^{3 \ell_j - j} \right) \cdot \poly(\epsilon^{-1}).
\end{align*}

{\bf Update.} For any $j \in [0: \log T]$, the update takes time $\Tmat(n,n,k_j)$, and in total it is performed $\frac{|B_j|}{2^{\ell_j}}$ times. So the amortized cost of all updates is
\begin{align*}
\frac{1}{T} \cdot \sum_{j=0}^{\log T} \Tmat(n,n,k_j) \cdot \frac{|B_j|}{2^{\ell_j}} \leq &~ \wt{O}\left(\frac{1}{T} \cdot \sum_{j=0}^{\log T} \Tmat(n,n, n^{1/3} \cdot 2^{3 \ell_j - j}) \cdot 2^{j - \ell_j} \right)  \cdot \poly(\epsilon^{-1}), 
\end{align*}
where we used Lemma~\ref{lem:decomposition_iteration} that $|B_j| \leq 2^{j+1}$.

{\bf Total runtimes.} Combining the query and update operation and the $\wt{O}(n^{\omega}) \poly(\epsilon^{-1})$ initialization time, and since $T = \wt{O}(n^{1/3}) \poly(\epsilon^{-1})$, the total running time is
\begin{align*}
&~ \wt{O}\left(n^{\omega} + n^{5/3} \cdot \sum_{j=0}^{\log T} 2^{3 \ell_j - j} + \sum_{j=0}^{\log T} \Tmat(n,n, n^{1/3} \cdot 2^{3 \ell_j - j}) \cdot 2^{j - \ell_j}\right) \poly(\epsilon^{-1}) \\
\leq &~ \wt{O}\left(n^{\omega} + n^{5/3} \cdot \sum_{j=0}^{\log T} 2^{3 \ell_j - j} + \sum_{j=0}^{\log T} \Big( n^{\frac{\omega-2}{3(1-\alpha)}} \cdot 2^{(3 \ell_j - j)\frac{\omega-2}{1-\alpha}} n^{2-\frac{\alpha(\omega-2)}{1-\alpha}} + n^2 \Big)\cdot 2^{j - \ell_j} \right) \poly(\epsilon^{-1}) \\
= &~ \wt{O}\left(n^{\omega} + n^{5/3} \cdot \sum_{j=0}^{\log T} 2^{3 \ell_j - j} + \sum_{j=0}^{\log T} \Big( n^{\frac{\omega - 3 \alpha \omega + 4}{3(1-\alpha)}} \cdot 2^{\frac{3\omega+\alpha-7}{1-\alpha} \ell_j + \frac{3-\omega-\alpha}{1-\alpha} j} + n^2 \cdot 2^{j-\ell_j}\Big)\right) \poly(\epsilon^{-1}),
\end{align*}
where the second step follows from Fact~\ref{fact:upper_bound_Tmat}.

We choose $\ell_j$ such that $n^2 \cdot 2^{j-\ell_j} = \max\{n^{5/3} \cdot 2^{3 \ell_j - j}, n^{\frac{\omega - 3 \alpha \omega + 4}{3(1-\alpha)}} \cdot 2^{\frac{3\omega+\alpha-7}{1-\alpha} \ell_j + \frac{3-\omega-\alpha}{1-\alpha} j}\}$, i.e.,
\begin{align*}
2^{\ell_j} = &~ n^{1/12} 2^{j/2} ~~\text{ or } ~~ n^{\alpha/3-1/9} 2^{j/3}. 
\end{align*}

So we have the total runtime is the minimum of the following two equations:
\begin{align*}
&~ \wt{O}\left(n^{\omega} + n^{5/3} \cdot \sum_{j=0}^{\log T} n^{1/4} 2^{j/2} + \sum_{j=0}^{\log T} \Big( n^{\frac{\omega - 3 \alpha \omega + 4}{3(1-\alpha)}} \cdot n^{\frac{3\omega+\alpha-7}{12(1-\alpha)}} 2^{\frac{3\omega+\alpha-7}{2(1-\alpha)} j} \cdot 2^{\frac{3-\omega-\alpha}{1-\alpha} j} \Big)\right) \poly(\epsilon^{-1}) \\
= &~ \wt{O}\left(n^{\omega} + n^{5/3} \cdot \sum_{j=0}^{\log T} n^{1/4} 2^{j/2} + \sum_{j=0}^{\log T} \Big( n^{\frac{7 \omega - 12 \alpha \omega + 9 + \alpha}{12(1-\alpha)} } \cdot  2^{\frac{\omega-\alpha-1}{2(1-\alpha)} j} \Big)\right) \poly(\epsilon^{-1}) \\
= &~ \wt{O}\left(n^{\omega} + n^{25/12} +  n^{\frac{9 \omega - 12 \alpha \omega + 7 - \alpha}{12(1-\alpha)}} \right) \poly(\epsilon^{-1}),
\end{align*}
and
\begin{align*}
&~ \wt{O}\left(n^{\omega} + n^{5/3} \cdot \sum_{j=0}^{\log T} n^{\alpha-1/3}  + \sum_{j=0}^{\log T} \Big( n^2 \cdot 2^{2j/3} \cdot n^{-\alpha/3+1/9} \Big)\right) \poly(\epsilon^{-1}) \\
= &~ \wt{O}\left(n^{\omega} + n^{4/3 + \alpha} + n^{2+1/3-\alpha/3} \right) \poly(\epsilon^{-1}) \\
= &~ \wt{O}\left(n^{\omega} + n^{2+1/3-\alpha/3} \right) \poly(\epsilon^{-1}).
\end{align*}

In conclusion, we have that the total runtime is
\begin{align*}
\wt{O}\left(n^{\omega} + \min\left\{n^{2+1/12} +  n^{\omega - \frac{3\omega+\alpha-7}{12(1-\alpha)}}, ~~ n^{2+1/3-\alpha/3}\right\} \right) \poly(\epsilon^{-1}).
\end{align*}

Note that when $\omega = 2$, we have $\alpha = 1$, and the algorithm runs in $\wt{O}(n^{\omega} + n^{2+1/12}) \cdot \poly(\epsilon^{-1})$ time.

\end{proof}

%% file: MonotoneMWU.tex
\section{Guarantees of Algorithm \ref{alg:MWU}: MWU with Monotone Weights}\label{sec:MWUMonotone}

The algorithm in this section has been analysed previously as mentioned in the main text. We include a proof similar to that of \cite{adil2022fast} for $\ell_{\infty}$-regression here for completeness. The analysis of Algorithm~\ref{alg:MWU} is based on tracking two potential functions that were defined in Eq.~\eqref{eq:defPhi} and \eqref{eq:defPsi}: 
\begin{equation*}
    \Phi\left(\ww^{( i,k )} \right) \defeq \norm{\ww^{(i,k)}}_1
\end{equation*}
\begin{equation*}
    \Psi(\rrbar^{(i,k)})\defeq \min_{\Delta\in \rea^d }\sum_e \rrbar^{(i,k)}_e (\CC\Delta-\dd)^2_e.
\end{equation*}
Also recall that by Lemma~\ref{lem:PsiPhi} we always have $\Psi(\rrbar^{(i,k)}) \leq e^{\epsilon+\delta} \cdot \Phi(\ww^{(i,k)})$. 
We will show how these potential functions change with a primal step (Line \ref{algline:CheckWidth}) and a width reduction step (Line \ref{lin:WidthReduceEdge}) in Algorithm~\ref{alg:MWU}. 
Finally, to prove our runtime bound, we will first show that if the total number of width reduction steps $K$ is not too large, then $\Phi$ is bounded. We then prove that the number of width reduction steps cannot be too large by using the relation between $\Phi$ and $\Psi$ and their respective changes throughout the algorithm.

\subsection*{Convergence Analysis}

\subsubsection*{Change in $\Phi$}

\begin{restatable}{lemma}{ChangePhiMon}
  \label{lem:ChangePhiMon}
  After $i$ primal steps, and $k$ width-reduction steps,
  the potential $\Phi$ is bounded as follows:
  \begin{align*}
 \Phi\left(\ww^{(i,k)}\right) \leq \Phi(\ww^{(0,0)})\left(1+e^{\epsilon + \delta} \cdot \epsilon\alpha\right)^i \left(1+ e^{\epsilon + \delta} \cdot \frac{\epsilon}{\tau} \right)^k .
 \end{align*}
\end{restatable}
\begin{proof}
  We prove this claim by induction. Initially, $i = k = 0,$
  and $\Phi\left(\ww^{(0,0)}\right) = n,$ and thus, the claim holds trivially. Assume that the claim holds for some $i,k \ge 0$. Denote $\ww = \ww^{(i,k)}$ and $\Delta = \Delta^{(i,k)}$.
\paragraph*{Primal Step.} If the next step is a \emph{primal} step, 

\begin{align*}
\Phi\left(\ww^{( i+1,k)} \right) = &\norm{ \ww+ \epsilon\alpha |\CC\Delta-\dd|\ww}_1 = \|\ww\|_1 + \epsilon\alpha \sum_e \ww_e|\CC\Delta-\dd|_e.
\end{align*}
We next bound $\sum_e |\CC\Delta-\dd|_e\ww_e$.
Using Cauchy-Schwarz inequality,
\begin{equation}\label{eq:CS}
 \sum_e \ww_e |\CC\Delta-\dd|_e  \leq \sqrt{\sum_e \ww_e \sum_e \rr_e(\CC\Delta-\dd)^2_e} \leq \sqrt{e^{\delta} \cdot \Phi(\ww)\Psi(\rrbar)} \leq e^{\epsilon + \delta} \cdot \Phi(\ww),
\end{equation}
where the last inequality follows from Lemma~\ref{lem:PsiPhi}.
We thus have,
\begin{align*}
 \Phi\left(\ww^{( i+1,k)} \right) \leq \Phi(\ww^{(i,k)})\cdot \left(1 + e^{\epsilon + \delta} \cdot \epsilon\alpha\right).
\end{align*}

 \paragraph*{Width Reduction Step.}
Let $\Delta = \Delta^{(i,k)}$ denote the solution returned in Line~\ref{algline:linsysMon} of Algorithm~\ref{alg:MWU}, and let $H$ denote the set of indices $j \in [n]$ such that $|\CC\Delta-\dd|_j \geq \tau $, i.e., the set of indices on which the algorithm performs width reduction.
 We have the following:
\begin{align*}
\Phi(\ww^{(i,k+1)}) & = \sum_{j \notin H}  \ww_j^{(i,k)} + \sum_{j \in H}(\ww_j^{(i,k)} + \epsilon \rr_j^{(i,k)})   =\Phi  + \epsilon \sum_{j \in H}  \rr_j^{(i,k)} \\
&\leq \Phi + \frac{\epsilon}{\tau} \sum_j \rr_j^{(i,k)}|\CC\Delta-\dd|_j \\
&\leq \Phi(\ww^{(i,k)})\left(1 + \frac{\epsilon}{\tau} \cdot e^{\epsilon + \delta} \right).
\end{align*}
The last inequality follows from a reasoning similar to Equation~\eqref{eq:CS}.
\end{proof}

\subsubsection*{Change in $\Psi$}

We now prove how the $\Psi$ potential changes throughout the algorithm.
\begin{restatable}{lemma}{ChangePsiMon}
\label{lem:ChangePsiMon}
After $i$ primal steps and $k$ width reduction steps, if $\delta \leq \epsilon/6$,
\[
\Psi(\rrbar^{(i,k)})\geq \Psi(\rr^{(0,0)})\left(1  +\frac{\epsilon^2\tau^2}{4  n}\right)^k.
\]

\end{restatable}
\begin{proof}
First note that the primal steps only increases $\ww$, and so it only increases $\Psi(\rrbar)$, so it suffices to only consider the width reduction steps.

In this proof from simplicity we use $\rr$ and $\rrbar$ to denote $\rr^{(i,k)}$ and $\rrbar^{(i,k)}$, and $\rr'$ and $\rrbar'$ to denote $\rr^{(i,k+1)}$ and $\rrbar^{(i,k+1)}$. We let $H$ denote the set of indices $j \in [n]$ such that $|\CC\Delta-\dd|_j \geq \tau $, i.e., the set of indices on which the algorithm performs width reduction.

We start by noting the following from Lemma~\ref{lem:PsiChange} that for $\rrbar, \rrbar' \geq 0$
\[
\Psi(\rrbar')\geq \Psi(\rrbar) + \sum_e \left(1-\frac{\rrbar_e}{\rrbar_e'}\right)\rrbar_e(\CC\widetilde{\Delta}-\dd)_e^2,
\]
where $\widetilde{\Delta}$ is the solution of $\Psi(\rrbar)$.
For any coordinate $e \in H$, 
\[
\frac{\rr^{(i,k+1)}_e-\rr^{(i,k)}_e}{\rr_e^{(i,k+1)}} \geq \frac{\ww^{(i,k+1)}_e-\ww^{(i,k)}_e}{\rr_e^{(i,k+1)}} = \epsilon \frac{\rr^{(i,k)}_e}{\rr^{(i,k+1)}_e  }\geq \frac{\epsilon}{1+2\epsilon}.
\]
In the above we used, $\ww_e^{(i,k+1)} = \ww_e^{(i,k)} + \epsilon \rr_e^{(i,k)}$, and  $\rr_e^{(i,k+1)}\leq (1+2\epsilon)\rr_e^{(i,k)}$. 
Now,
\[
\frac{\rrbar'_e - \rrbar_e}{\rrbar'_e}\geq \frac{e^{-\delta} \rr'_e - e^{\delta} \rr_e}{e^{\delta} \rr_e'} = \frac{\rr'_e - \rr_e}{\rr_e'} - \frac{e^{\delta} - e^{-\delta}}{e^{\delta}} \geq \frac{\epsilon}{1+2\epsilon} - 2 \delta.
\]
Since $\delta \leq \epsilon/6$, the above becomes,
\[
\frac{\rrbar'_e - \rrbar_e}{\rrbar'_e}\geq \frac{\epsilon}{3}.
\]

We also know that for all $e \in H$, $|\CC\Delta-\dd|_e\geq \tau$, and since $\rrbar_e \geq e^{-\delta} \rr_e \geq e^{-\delta} \frac{\epsilon}{n} \Phi(\ww)$,
\[
\Psi(\rrbar')\geq \Psi(\rrbar) +  \frac{\epsilon}{3} \cdot e^{-\delta} \frac{\epsilon}{n}\Phi(\ww) \cdot \tau^2 \geq \Psi(\rrbar) + \frac{\epsilon^2 \tau^2 e^{-\epsilon - 2 \delta}}{3 n}\Psi(\rrbar).
\]
Overall, after $k$ such steps and using that $\delta \leq \epsilon/6\leq 1/60$, we get,
\[
\Psi(\rrbar^{(i,k)})\geq \Psi(\rr^{(0,0)})\left(1  +\frac{\epsilon^2\tau^2}{4 n}\right)^k.
\]
\end{proof}

\subsubsection*{Proof of Theorem~\ref{thm:InfRegMainMonotone}}
\begin{proof}
Let $\xxhat = \frac{\xx^{(T)}}{T}$ be the solution returned by Algorithm \ref{alg:MWU}. We first bound the objective value at $\xxhat$. Suppose the algorithm terminates in $T = \alpha^{-1}\epsilon^{-2}\log n$ primal steps and $K \leq \tau/\epsilon^2$ width reduction steps. We can assume this without loss of generality since otherwise we can just halt the algorithm after executing $\tau/\epsilon^2$ width reduction steps. In the end we will show that the algorithm executes $<\tau/\epsilon^2$ width reduction steps. 

We can now apply Lemma \ref{lem:ChangePhiMon} to get,
\[
\Phi\left(\ww^{(T,K)}\right) \le  n e^{2(1+\epsilon)\epsilon\alpha T + 2(1+\epsilon)\epsilon\tau^{-1}K}\leq ne^{2\frac{(1+\epsilon)\log n}{\epsilon}+(1+\epsilon)} \leq ne^{\frac{2(1+\epsilon)}{\epsilon}(1 +\log n)} = n^{O\left(\frac{1}{\epsilon} \right)}.
\]
We next observe from the weight and $\xx$ update steps in our algorithm that, $\ww^{(T,K)} \geq |\xx^{(T)}|/T$.
\[
\ww_e^{(T,K)} \geq  \Pi_{i\geq 0}\left(1+\epsilon\alpha|\CC\Delta^{(i,k)}-\dd|_e\right)\geq \exp\left((1-\epsilon)\epsilon\alpha \sum_{i\geq 0}|\CC\Delta^{(i,k)}-\dd|_e\right).
\]

In the last inequality we used that $\|\CC\Delta^{(i,k)}-\dd\|_{\infty}\leq \tau$ and $\alpha\tau \leq 1$ along with $\exp(\epsilon(1-\epsilon)x) \leq(1+\epsilon x)$ for $0<\epsilon <1/2$ and $0\leq x\leq 1$. We also know that $\ww_e^{(T,K)}\leq \Phi(\ww^{(T,K)}) \leq ne^{(1+\delta)(1+\epsilon)\epsilon\alpha T + (1+\delta)(1+\epsilon)/\epsilon }$
Now, 
\[
\|\CC\xxhat-\dd\|_{\infty}\leq \frac{1}{T}\|\sum_{t=1}^T|\CC\Delta^{(t,k)}-\dd|\|_{\infty} \leq \frac{(1+\epsilon)(1+\delta)}{1-\epsilon} + \frac{(1+\epsilon)(1+\delta)}{\epsilon \alpha T (1-\epsilon)} + \frac{\log n}{(1-\epsilon)\epsilon \alpha T} \leq 1+10\epsilon.
\]
We have shown that if the number of width reduction steps is bounded by $K$ then our algorithm returns the required solution. We will next prove that we cannot have more than $K$ width reduction steps.

Suppose to the contrary, the algorithm takes a width reduction step starting from step $(i,k)$ where $i < T$ and  $k = \tau/\epsilon^2$. Since the conditions for Lemma~\ref{lem:ChangePhiMon} hold for all preceding steps, we must have $\Phi\left(\ww^{(i,k)}\right) \leq n^{O\left(\frac{1}{\epsilon} \right)}$ which combined with Lemma \ref{lem:PsiPhi} implies $\Psi \leq (1+\epsilon)n^{O\left(\frac{1}{\epsilon} \right)}$. Let $L = \dd^{\top}\left(\II - \CC^{\top}(\CC^{\top}\CC)^{-1}\CC\right)\dd$ denote a lower bound on $\Psi(\rr^{(0,0)})$. Using this bound, from lemma \ref{lem:ChangePsiMon},

\[
 {\Psi\left({\rrbar^{(i,k+1)}}\right)} \geq {\Psi\left({\rr^{(0,0)}}\right)}\left(1  +  \frac{\epsilon^2 \tau^2}{4(1+\epsilon)(1+\delta) n}\right)^k\geq  L\left(1  +  \frac{\epsilon^2 \tau^2}{4(1+\epsilon)(1+\delta) n}\right)^k .
\]
Therefore,
\[
n^{O\left(\frac{1}{\epsilon} \right)}\geq  L \left(1  +  \frac{\epsilon^2 \tau^2}{(1+\epsilon)(1+\delta) n}\right)^k.
\]
which is a contradiction if $K> \tau/\epsilon^2$ for the set value of $\tau =  \Theta\left(\frac{n^{\frac{1}{3}}}{\epsilon^{\frac{1}{3}}} \log\frac{n}{L}\right)$. We can thus conclude that we can never have more than $K = \tau/\epsilon^2$ width reduction steps, thus concluding the correctness of the returned solution. The total number of iterations is at most,
\[
T + K \leq \alpha^{-1}\epsilon^{-2}\log n+  \tau \epsilon^{-2} = \Theta\left(n^{1/3}\epsilon^{-7/3}\log^2\frac{n}{L}\right).
\]
\end{proof}

%% file: NewNonMonotoneAcceleration.tex
\section{Guarantees of Algorithm~\ref{alg:non_monotone_accel_robust}: Robust Primal Step and Stable Width Reduction Step}\label{sec:Appendix_non_monotone_robust}

\subsection{Starting Point: Non-Monotone MWU Algorithm}
We first present the starting point from where we developed our final algorithm. We extend the algorithm in~\cite{madry2016computing} to general $\ell_{\infty}$-regression. This gives an algorithm which converges in $\wt{O}(n^{1/3} \poly(\epsilon^{-1}))$ iterations and has weights that are not monotonically increasing. The algorithm is presented as Algorithm~\ref{alg:non_monotone_accel_opt} and we have moved the analysis to the end in Appendix~\ref{sec:NonMonMWU} since we do not really use this algorithm directly. The analysis is just kept for completeness. In particular we prove:

\begin{restatable}{theorem}{NonMonOpt}\label{thm:NonMonotoneMWUAcc}
Let $\eta \leq 1/6$. There is an algorithm that does not update the weights monotonically (Algorithm~\ref{alg:non_monotone_accel_opt}, and with input $\begin{bmatrix}
    \CC\\ -\CC
\end{bmatrix}, \begin{bmatrix}
    \dd\\ -\dd
\end{bmatrix}, \epsilon$) returns $\xxhat$ such that $\|\CC\xxhat-\dd\|_{\infty}\leq 1+O(\epsilon)$ in at most $\widetilde{O}\left((n^{1/2-\eta} + n^{2\eta})\epsilon^{-7/3}\right)$ iterations. Each iteration solves a system of linear equations.
\end{restatable}
\begin{algorithm}
\caption{Accelerated MWU algorithm with non-monotone weights}\label{alg:non_monotone_accel_opt}
\begin{algorithmic}[1]
\Procedure{MWU-NonMonotone}{$\wt{\CC}, \wt{\dd}, \epsilon$}
\State $\ww^{(0,0)} \gets 1_n, \quad \xx^{(0,0)} \gets 0_d$
\State $\alpha \gets \widetilde{\Theta}\left( n^{-1/2+\eta}\epsilon^{1/3}\right)$, $\alpha_+ \gets \alpha$, $\alpha_- \gets \alpha/(1+2\epsilon)$
\State $\tau \gets \widetilde{\Theta}\left(n^{1-4\eta}\epsilon^{-1/3}\right), \quad \rho \gets \widetilde{\Theta} \left(n^{1/2-3\eta}\right)$
\State $T \gets \alpha^{-1}\ln \frac{n}{\Psi_0}/\epsilon^2$ 
\State $i,k = 0$
\While{$i<T$}\label{algline:Primalopt}
\State $\rr^{(i,k)} \leftarrow \ww^{(i,k)} + \frac{\epsilon}{2n} \sum_{e} \ww^{(i,k)}_e$ 
\State $\Delta^{(i,k)} \leftarrow \arg\min_{\Delta} \sum_e \rr^{(i,k)}_e (\wt{\CC}\Delta-\wt{\dd})_e^2$ \Comment{$\Delta = (\wt{\CC}^{\top}\RR^{(i,k)}\wt{\CC})^{-1} \wt{\CC}^{\top} \RR^{(i,k)} \wt{\dd}$} \label{algline:linsysNM}
\State $\Psi(\rr^{(i,k)}) \gets \sum_e \rr^{(i,k)}_e (\wt{\CC}\Delta^{(i,k)}-\wt{\dd})_e^2$
\If{$\sum_e \rr_e^{(i,k)} |\wt{\CC}\Delta^{(i,k)}-\wt{\dd}|_e^3\leq 2\rho \Psi(\rr^{(i,k)}) $}\Comment{primal step}\label{algline:CheckPrimalNM}
\State $\overrightarrow{\alpha}^{(i,k)} _e = \begin{cases}
\alpha_+ & \text{ if } (\wt{\CC} \Delta^{(i,k)} - \wt{\dd})_e \geq 0\\
\alpha_- & \text{ else }
\end{cases}$
\State $\ww^{(i+1,k)}\gets \ww^{(i,k)}\left(1 + \epsilon \overrightarrow{\alpha}^{(i,k)} (\wt{\CC}\Delta^{(i,k)}-\wt{\dd})\right)$  
\State $\xx^{(i+1)} \gets \xx^{(i)} + \Delta^{(i,k)}$
\State $i \gets i + 1$
\Else\label{algline:Width}\Comment{width reduction step}
\State Let $S$ be the set of coordinates $e$ such that $|\wt{\CC}\Delta^{(i,k)}-\wt{\dd}|_{e}\geq \rho$
\State $H \subseteq S$ be maximal subset such that $\sum_{e\in H}\rr^{(i,k)}_e\leq \tau^{-1}\Psi(\rr^{(i,k)}) $ \label{algline:sum_r_in_H}
\State \indent For all $e\in H $, $\ww_e^{(i,k+1)}\gets (1+\epsilon)\ww_e^{(i,k)} + \frac{\epsilon^2}{n}\Phi(\ww^{(i,k)})$
\State \indent If $H \neq S$, for one $\bar{e}\in S\setminus H$, let $\gamma = \min\{1,\frac{\tau^{-1}}{ \rr^{(i,k)}_{\bar{e}}}\Psi(\rr^{(i,k)}) \}$
\State \indent \indent $\ww_{\bar{e}}^{(i,k+1)}\gets (1+\epsilon\gamma)\ww_{\bar{e}}^{(i,k)} + \frac{\epsilon^2 \gamma}{n}\Phi(\ww^{(i,k)})$ \label{algline:CaseBigR}
\State $k \gets k+1$
\EndIf
\EndWhile
\State \Return $\xx^{(T)}/T$
\EndProcedure
\end{algorithmic}
\end{algorithm}

Let us consider the width reduction steps. Since $\rr_e\geq \frac{\epsilon}{2n}\Phi(\ww)\geq \frac{\epsilon}{2n}\Psi(\ww)$, from the condition of the width steps we get that,
\[
\frac{\epsilon}{2n}\Psi(\rr) |H| \leq\tau^{-1}\Psi(\rr) \Rightarrow |H|\leq O(n^{4\eta}).
\]
The above calculations imply that for $\eta = 1/6$, the algorithm in the worst case can update $\approx n^{2/3}$ coordinates in every iteration. There are a total of $\approx n^{1/3}$ iterations. At an intuitive level, if we change the value of $\rr$ so many coordinates per iteration, and that too by a factor of $\epsilon$, then the entire idea of doing a lazy update is moot. Therefore, we require a new kind of width reduction steps that can schedule these steps in a way amenable to a lazy update schedule, which is what we do in the next section.

\subsection{Warm Up: Stable Width Reduction Step}\label{sec:NonMonStab}

As a warm-up, we first analyze the required stable algorithm, but do not introduce any sketching. In the next section, we will consider the changes to the analysis from this section which occur due to the introduction of sketching. We encourage the reader to go through this section in order to understand how our final algorithm works.

\begin{algorithm}
\caption{Accelerated MWU algorithm with non-monotone weights and stable steps}\label{alg:non_monotone_accel_stab}
\begin{algorithmic}[1]
\Procedure{MWU-NonMonotoneStable}{$\wt{\CC}, \wt{\dd}, \epsilon$}
\State $\ww^{(0,0)} \gets 1_{2n}, \quad \ov{\rr}^{(0,0)} \gets \rr^{(0,0)} \gets (1 + \epsilon) 1_{2n}, \quad \xx^{(0)} \gets 0_d$
\State $\alpha \gets n^{-1/2+\eta} \cdot \epsilon \cdot \log(n)^{-4/3} \log(\frac{n}{\Psi_0})^{-1/3} / 10$, $\alpha_+ \gets \alpha$, $\alpha_- \gets \alpha/(1+2\epsilon)$
\State $\tau \gets n^{1/2-\eta} \cdot \epsilon^{-4} \cdot \log(n)^8 \log(\frac{n}{\Psi_0})^2, \quad \rho \gets n^{1/2-3\eta} \cdot \epsilon^{-2} \cdot \log(n)^4 \log(\frac{n}{\Psi_0}), \quad \delta \gets \frac{\epsilon}{100}$
\State $T \gets \alpha^{-1}\epsilon^{-2}\log n $
\State $i,k = 0$
\While{$i<T$}\label{algline:Primalstab}
\State $\Delta^{(i,k)} \leftarrow \arg\min_{\Delta} \sum_e \rrbar^{(i,k)}_e (\wt{\CC}\Delta-\wt{\dd})_e^2$ \Comment{$\rrbar\approx_{\delta}\rr$}\label{algline:linsysStab}
\State $\Psi(\rrbar^{(i,k)}) \gets \sum_e \rrbar^{(i,k)}_e (\wt{\CC}\Delta^{(i,k)}-\wt{\dd})_e^2$
\If{$\sum_e \rrbar_e^{(i,k)} |\wt{\CC}\Delta^{(i,k)}-\wt{\dd}|_e^3\leq 2\rho\Psi(\rrbar^{(i,k)}) $}\Comment{primal step}
\State $\overrightarrow{\alpha}^{(i,k)}_e = \begin{cases}
\alpha_+ & \text{ if } (\wt{\CC} \Delta^{(i,k)} - \wt{\dd})_e \geq 0\\
\alpha_- & \text{ else }
\end{cases}$
\State $\ww^{(i+1,k)}\gets \ww^{(i,k)}\left(1 + \epsilon \overrightarrow{\alpha}^{(i,k)} (\wt{\CC}\Delta^{(i,k)}-\wt{\dd})\right)$ 
\State $\rr^{(i+1,k)} \leftarrow \ww^{(i+1,k)} + \frac{\epsilon}{2n} \sum_{e} \ww^{(i+1,k)}_e$
\State $\rrbar^{(i+1,k)}\gets$ {\sc SelectVector}$(\rr^{(i+1,k)}, i+1, \delta)$ \Comment{Algorithm~\ref{alg:select_vector}}
\State $\xx^{(i+1)} \gets \xx^{(i)} + \Delta^{(i,k)}$
\State $i \gets i + 1$
\Else\Comment{width reduction step}
\State Let $S$ be the set of coordinates $e$ such that $|\wt{\CC}\Delta^{(i,k)}-\wt{\dd}|_{e}\geq \rho$
\State $H \subseteq S$ be maximal subset such that $\sum_{e\in H}\rrbar_e^{(i,k)}\leq \tau^{-1}\Psi(\rrbar^{(i,k)}) $
\If{$H \neq S$}\label{algline:if_H_neq_S}
\State Pick any $\bar{e}\in S\setminus H$.
\State For all $e\in H \cup \{\ov{e}\} $, $\ww_e^{(i,k+1)}\gets (1+\epsilon)\ww_e^{(i,k)} + \frac{\epsilon^2}{2n}\Phi(\ww^{(i,k)})$ \label{algline:width_update_H_neq_S}
\State $\rr^{(i,k+1)}\gets \ww^{(i,k+1)} + \frac{\epsilon}{2n}\Phi(\ww^{(i,k+1)})$
\State For all $e\in H \cup \{\ov{e}\} $, $\ov{\rr}_e^{(i,k+1)} \gets \rr_e^{(i,k+1)}$ \label{algline:width_update_rrbar_H_neq_S}
\Else
\For{$\zeta = \rho, 2 \rho, 4 \rho, \cdots, 2^{c_\rho} \rho$} 
\State \Comment{$c_{\rho}$ is defined to be the smallest integer $c$ that satisfies $2^c \rho \geq \sqrt{n/\epsilon}$}
\State Define the set $H_\zeta = \{e \in H \mid |\wt{\CC} \Delta^{(i,k)} - \wt{\dd}|_e \in [\zeta, 2 \zeta)\}$.
\State If $\sum_{e \in H_{\zeta}} \ov{\rr}^{(i,k)}_e |\wt{\CC} \Delta^{(i,k)} - \wt{\dd}|_e^3 \geq \frac{\rho \Psi(\ov{\rr}^{(i,k)})}{\log (\frac{n}{\epsilon \rho})}$, set $\zeta^* \gets \zeta$, and break.\label{algline:zeta*}
\EndFor
\State For all $e \in H_{\zeta^*}$, $\ww_e^{(i,k+1)}\gets (1+\epsilon)\ww_e^{(i,k)} + \frac{\epsilon^2}{2n}\Phi(\ww^{(i,k)})$\label{algline:update_e_in_zeta*}
\State $\rr^{(i,k+1)}\gets \ww^{(i,k+1)} + \frac{\epsilon}{2n}\Phi(\ww^{(i,k+1)})$
\State For all $e\in H_{\zeta^*}$, $\ov{\rr}_e^{(i,k+1)} \gets \rr_e^{(i,k+1)}$ \label{algline:width_update_rrbar_H_eq_S}
\EndIf
\State $k \gets k+1$
\EndIf
\EndWhile
\State \Return $\xx^{(T)}/T$
\EndProcedure
\end{algorithmic}
\end{algorithm}

In this section we define $\wt{\CC} = \begin{bmatrix}
    \CC\\ -\CC
\end{bmatrix}$ and $\wt{\dd} = \begin{bmatrix}
    \dd\\ -\dd
\end{bmatrix}$. Our goal is to prove the following result. We will focus on the convergence and runtime in this section and defer the stability guarantees of the algorithm to Appendix~\ref{sec:StabilityMWU}.

\begin{restatable}{theorem}{StableAccMWU}\label{thm:StableAccMWU}
For $\eta\leq 1/10$, Algorithm~\ref{alg:non_monotone_accel_stab} with input ($\begin{bmatrix}
    \CC\\ -\CC
\end{bmatrix}, \begin{bmatrix}
    \dd\\ -\dd
\end{bmatrix}, \epsilon$) finds $\xxhat\in \mathbb{R}^d$ such that $\|\CC\xxhat-\dd\|_{\infty} \leq 1+O(\epsilon)$ in at most $T + K \leq \Otil(n^{1/2-\eta}\epsilon^{-4})$ iterations. Furthermore, the algorithm satisfies the following additional guarantees:
\begin{enumerate}
    \item In the width reduction step of the algorithm, the algorithm only requires to find at most $\Otil\left(n^{1/2+\eta}\right)$ large coordinates per iteration.
    \item The algorithm satisfies the following low-rank update scheme: There are at most $\frac{T+K}{2^{\ell}}$ number of iterations where $\ov{\rr}$ receives an update of rank $\Otil_{\epsilon}(n^{1/5} 2^{2 \ell})$.
\end{enumerate}
\end{restatable}

We first define some notation and basic properties.
\subsubsection{Definitions and Basic Properties}
\paragraph{Potentials.}
Recall that we had defined the two potentials of interest as, 
\begin{align*}
\Phi(\ww) = \|\ww\|_1, ~~~ \Psi(\rr) = \min_{\Delta\in \R^d }\sum_{e=1}^{2n} \rr_e (\CC\Delta-\dd)^2_e.
\end{align*}
Also recall that by Lemma~\ref{lem:PsiPhi}, for $\rr = \ww + \frac{\epsilon}{2n}\|\ww\|_1$, and $\rrbar\approx_{\delta}\rr$, we always have $\Psi(\rrbar) \approx_{\delta} \Psi(\rr)$, and $\Psi(\rrbar) \leq e^{\epsilon+\delta} \Phi(\ww)$.

\paragraph{Lazy updates for primal steps.}
Recall that in Algorithm~\ref{alg:select_vector}, for any primal step $i$ and any $e \in [2n]$, we defined $\textsc{LastWidth}(i,e)$ to be the largest $i' \leq i$ such that the algorithm executed a width reduction step from $(i',k)$ to $(i',k+1)$ during which the weight of $e$ is updated, i.e., $\ww_e^{(i',k+1)} \neq \ww_e^{(i',k)}$. 

For any primal step $i$ and any $e \in [2n]$, we also define $\textsc{Last}(i,e)$ to be the largest $i'' \leq i$ such that $\ov{\rr}_e$ was updated by $\textsc{SelectVector}$ in primal step $i''-1$, i.e., $\ov{\rr}_e^{(i'', k_{i''-1})} \neq \ov{\rr}_e^{(i''-1, k_{i''-1})}$.

For any primal step $i$, we define $S_i \subseteq [2n]$ to be the set of coordinates that are being updated by $\textsc{SelectVector}$ in the $i$-th primal step.

Finally, for any primal step $i$ and any $e \in S_i$, we also define $\ell_{i,e}$ to be the smallest integer $\ell$ such that $i \equiv 0 \pmod{2^{\ell}} $ and $\left|\ln\left(\frac{\rr_e^{(i)}}{\rr_e^{(i-2^\ell)}}\right)\right| \geq \frac{\delta}{2 \log n}$. 

\paragraph{Remark:} We note that $\rr = \ww + \frac{\epsilon}{2n}\Phi(\ww)$. Now in our lazy update scheme, $\rr$ can also change due to changes in $\Phi(\ww)$. We claim that we do not need to consider the changes in $\rr$ due to the change in $\Phi$ in the lazy update scheme. This is because, the change in $\Phi$ contributes enough only when the change is $\approx n$ which can happen only $\Otil_{\epsilon}(1)$ times (Refer to Lemma~\ref{lem:ChangePhiStab}), or once every $\Otil_{\epsilon}(\alpha^{-1})$ iterations, and in such cases, the algorithm can reset the values of $\rr$ for all coordinates.

\paragraph{Sizes of width reduction steps.}
We say a width reduction step has size $s$ if it updates the weight of $s$ coordinates, and we denote the size of the $k$-th width reduction step as $\textsc{Size}(k)$. For example, if in the $k$-th width reduction step on Line~\ref{algline:if_H_neq_S} of Algorithm~\ref{alg:non_monotone_accel_stab} we have that $H \neq S$, then we have $\textsc{Size}(k) = |H|+1$ by Line~\ref{algline:width_update_H_neq_S} of Algorithm~\ref{alg:non_monotone_accel_stab}.

\subsubsection{Change in \texorpdfstring{$\Phi$}{Phi}}

\begin{lemma}[Change in $\Phi$ for Algorithm~\ref{alg:non_monotone_accel_stab}]\label{lem:ChangePhiStab}
After $i$ primal steps, and $k$ width-reduction steps, if 
$\alpha \rho^{1/3}\leq \frac{\epsilon^{1/3}}{10 n^{1/3}}$, the potential $\Phi$ is bounded as follows:
\begin{align*}
\Phi\left(\ww^{(i,k)}\right) \leq \Phi(\ww^{(0,0)}) \cdot \left(1+\epsilon\alpha_+ e^{\epsilon+\delta} \right)^i \cdot \left( 1 + \epsilon e^{\epsilon+2\delta} \cdot (\tau^{-1} + \rho^{-2}) \right)^k.
\end{align*}
Furthermore, after every primal step, the potential can decrease by at most,
 \[
 \Phi\left(\ww^{(i+1,k)}\right) \geq  \Phi\left(\ww^{(i,k)}\right)\left(1-\epsilon\alpha_+ e^{\epsilon+\delta}\right).  
 \]
\end{lemma}
\begin{proof}
First note that we always have $\ww_e^{(i,k)} > 0$ from the same proof as Lemma~\ref{lem:positiveWopt}.

\noindent {\bf Primal Step.} If the next step is a \emph{primal} step, then 
\begin{align}\label{eq:helpPhi1opt}
\Phi\left(\ww^{( i+1,k)} \right) = &~ \norm{ \ww^{( i,k)} + \epsilon \overrightarrow{\alpha}^{(i,k)} (\wt{\CC}\Delta^{(i,k)}-\wt{\dd})\ww^{(i,k)}}_1 \notag \\
= &~ \|\ww^{(i,k)}\|_1 + \epsilon \sum_e \ww_e^{(i,k)} \overrightarrow{\alpha}_e^{(i,k)} (\wt{\CC}\Delta^{(i,k)}-\wt{\dd})_e.
\end{align}
We first bound $\sum_e \ww_e^{(i,k)} \cdot \overrightarrow{\alpha}_e^{(i,k)} \cdot (\wt{\CC} \Delta^{(i,k)} - \wt{\dd})_e$. Using Cauchy-Schwarz inequality, we have
\begin{align*}
\sum_e \ww_e^{(i,k)} \cdot \overrightarrow{\alpha}_e^{(i,k)} \cdot |\wt{\CC} \Delta^{(i,k)} - \wt{\dd}|_e \leq &~ \sqrt{\Big(\sum_e \ww_e^{(i,k)} (\overrightarrow{\alpha}_e^{(i,k)})^2\Big) \cdot \Big(\sum_e \ww_e^{(i,k)} \cdot (\wt{\CC} \Delta^{(i,k)} - \wt{\dd})_e^2\Big)} \\
\leq &~ \alpha_+ \cdot \sqrt{e^{\delta} \cdot \Phi(\ww^{(i,k)}) \cdot \Psi(\rrbar^{(i,k)})} \\
\leq &~ \alpha_+ \cdot e^{\epsilon+\delta} \cdot \Phi(\ww^{(i,k)}),
\end{align*}
where the second step follows from $\alpha_- < \alpha _+$ and $\ww_e^{(i,k)} \leq e^{\delta} \cdot \rrbar_e^{(i,k)}$, the third step follows from Lemma~\ref{lem:PsiPhi} that $\Psi(\rrbar^{(i,k)}) \leq e^{\epsilon+\delta} \cdot \Phi(\ww^{(i,k)})$.

Now, from Equation~\eqref{eq:helpPhi1opt}, and $\ww_e^{(i,k)} > 0$, we have that $\Phi(\ww^{(i+1,k)})$ is between $\Phi(\ww^{(i,k)}) - \epsilon \sum_e \ww_e^{(i,k)} \overrightarrow{\alpha}_e^{(i,k)} |\wt{\CC}\Delta^{(i,k)}-\wt{\dd}|_e$ and $\Phi(\ww^{(i,k)}) + \epsilon \sum_e \ww_e^{(i,k)} \overrightarrow{\alpha}_e^{(i,k)} |\wt{\CC}\Delta^{(i,k)}-\wt{\dd}|_e$. 

Therefore, we get our bounds,
\[
\Phi\left(\ww^{( i,k)} \right)(1-\epsilon\alpha_+ e^{\epsilon+\delta}) \leq  \Phi\left(\ww^{( i+1,k)} \right) \leq  \Phi\left(\ww^{( i,k)} \right)(1+\epsilon\alpha_+ e^{\epsilon+\delta}).
\]

\noindent {\bf Width Reduction Step.}

When $H\neq S$, we have the following: 
\begin{align*}
 \Phi(\ww^{(i,k+1)}) 
& = \sum_{j \notin H\cup \{\bar{e}\}}  \ww_j^{(i,k)} + \sum_{j \in H \cup\{\ov{e}\}} \left( (1+\epsilon)\ww_j^{(i,k)} +\frac{\epsilon^2 }{2n}\Phi(\ww^{(i,k)})\right) \\
& = \Phi(\ww^{(i,k)}) + \epsilon \sum_{j \in H \cup \{\ov{e}\}} \rr_j^{(i,k)} \\
& \leq \Phi(\ww^{(i,k)})+ \epsilon e^{\delta} \sum_{j \in H}  \rrbar_j^{(i,k)} + \epsilon e^{\delta} \rrbar^{(i,k)}_{\bar{e}}\\
& \leq\Phi(\ww^{(i,k)}) + \epsilon e^{\delta} \tau^{-1} \Psi(\rrbar^{(i,k)}) + \epsilon e^{\delta} \rho^{-2} \Psi(\rrbar^{(i,k)}) \\
& \leq \Phi(\ww^{(i,k)}) + \epsilon e^{\epsilon+2\delta} (\tau^{-1} + \rho^{-2}) \Phi(\ww^{(i,k)}) 
\end{align*}
where the fourth step follows from $\sum_{e\in H}\rrbar_e^{(i,k)}\leq \tau^{-1}\Psi(\rrbar^{(i,k)})$ by the definition of $H$, and that $\ov{\rr}_{\ov{e}}^{(i,k)} \leq \frac{\Psi(\ov{\rr}^{(i,k)})}{(\wt{\CC} \Delta^{(i,k)} - \wt{\dd})_{\ov{e}}^2} \leq \frac{\Psi(\ov{\rr}^{(i,k)})}{\rho^2}$ since $\ov{e} \in S$ and so it satisfies $|\wt{\CC} \Delta^{(i,k)} - \wt{\dd}|_{\ov{e}} \geq \rho$, and the last step follows from $\Psi(\ov{\rr}^{(i,k)}) \leq e^{\epsilon+\delta} \cdot \Phi(\ww^{(i,k)})$ by Lemma~\ref{lem:PsiPhi}.

Now, when $H = S$, we only update weights in a set $H_{\zeta^*} \subseteq H$ (see Line~\ref{algline:update_e_in_zeta*} in Algorithm~\ref{alg:non_monotone_accel_stab}), 
\begin{align*}
\Phi(\ww^{(i,k+1)}) & = \sum_{j \notin H_{\zeta^*}}  \ww_j^{(i,k)} + \sum_{j \in H_{\zeta^*}} \left( (1+ \epsilon )\ww_j^{(i,k)} +\frac{ \epsilon^2 }{2n}\Phi(\ww^{(i,k)})\right) \\
& = \Phi(\ww^{(i,k)}) + \epsilon \sum_{j \in H_{\zeta^*}} \ov{\rr}_j^{(i,k)} \\
& \leq \Phi(\ww^{(i,k)}) + \epsilon \tau^{-1} \Psi(\ov{\rr}^{(i,k)}) \\
& \leq \Phi(\ww^{(i,k)}) + \epsilon e^{\epsilon+\delta} \cdot \tau^{-1} \Phi(\ww^{(i,k)}),
\end{align*}
where the third step follows from $\sum_{e\in H_{\zeta^*}}\rrbar_e^{(i,k)} \leq \sum_{e\in H}\rrbar_e^{(i,k)}\leq \tau^{-1}\Psi(\rrbar^{(i,k)}) $ by the definition of $H$, and the fourth step follows from $\Psi(\rrbar^{(i,k)}) \leq e^{\epsilon+\delta} \cdot \Phi(\ww^{(i,k)})$.

In both cases we have 
\[
\Phi(\ww^{(i,k+1)})\leq \Phi(\ww^{(i,k)})\left(1 + \epsilon e^{\epsilon+2\delta} \cdot (\tau^{-1} + \rho^{-2})\right).
\]
Also note that for a width reduction step, $\Phi(\ww^{(i,k+1)}) \geq \Phi(\ww^{(i,k)})$.
\end{proof}

\subsubsection{Change in \texorpdfstring{$\Psi$}{Psi}}
\begin{lemma}[Bound on $|H|$ in width reduction steps]\label{lem:size_H}
For any width reduction step, the size of $H$ satisfies $|H| \leq \frac{n}{\tau \epsilon} \cdot e^{\epsilon+2\delta}$.
\end{lemma}
\begin{proof}
Consider a width reduction step which updates $(i,k)$ to $(i,k+1)$. 

Since $\sum_{e\in H}\rrbar_e^{(i,k)}\leq \tau^{-1}\Psi(\rrbar^{(i,k)})$ and $\ov{\rr}_e^{(i,k)} \geq e^{-\delta} \rr_e^{(i,k)} \geq e^{-\delta} \frac{\epsilon}{2n} \Phi(\ww^{(i,k)}) \geq e^{-\epsilon-2\delta} \frac{\epsilon}{2n} \Psi(\ov{\rr}^{(i,k)})$, we have
\begin{align*}
&~ |H| \cdot e^{-\epsilon-2\delta} \cdot \frac{\epsilon}{2n} \Psi(\ov{\rr}^{(i,k)}) \leq \sum_{e\in H}\rrbar_e^{(i,k)} \leq \tau^{-1}\Psi(\rrbar^{(i,k)}) \\
\implies &~ |H| \leq \frac{n}{\tau \cdot \epsilon} \cdot e^{\epsilon+2\delta}. \qedhere
\end{align*} 
\end{proof} 

Let $L$ be the largest power of $2$ such that $L \leq \frac{1}{100 (\log^4 n) \epsilon \alpha \rho}$. Note that we have $L = \Theta(\frac{1}{(\log^4 n) \epsilon \alpha \rho})$. We have the following lemma.
\begin{lemma}[Change in $\Psi$ for Algorithm~\ref{alg:non_monotone_accel_stab}]\label{lem:ChangePsiStab}
For any integer $c \geq 0$, after $L$ primal steps from $(c-1)L$ to $cL$, if $\rho^2 \tau^{-1} \geq 0.1$, the potential $\Psi$ is bounded as follows:

\begin{align*}
\Psi(\ov{\rr}^{(cL,k_{cL})}) 
\geq &~ \Psi(\ov{\rr}^{((c-1)L,k_{(c-1)L})}) \cdot \Big(1 -  \wt{O}(\epsilon \alpha \rho L) \Big) \cdot \prod_{k=k_{(c-1)L}}^{k_{cL}} \left( 1 + O\Big(\frac{\epsilon^{4/3} \rho^{2/3} \cdot \textsc{Size}(k)^{1/3}}{n^{1/3} \cdot \log^{2/3}(\frac{n}{\epsilon \rho})}\Big) \right).
\end{align*}
\end{lemma}
\begin{proof}
\noindent {\bf Width Reduction Step.}
We first consider the width reduction steps. Note that if the width step updates a coordinate $e$, then it always updates $\ov{\rr}^{(i,k+1)}_{\ov{e}}$ to its exact value $\rr^{(i,k+1)}_{\ov{e}}$ (see Line~\ref{algline:width_update_rrbar_H_neq_S} and \ref{algline:width_update_rrbar_H_eq_S} in Algorithm~\ref{alg:non_monotone_accel_stab}. So we always have $\ov{\rr}^{(i,k+1)}_{e} = \rr^{(i,k+1)}_{e}$ if $e$ is being updated by the $k$-th width reduction step, and otherwise we have $\ov{\rr}^{(i,k+1)}_{e} = \ov{\rr}^{(i,k)}_{e}$.

Next we prove an upper bound on $\rr^{(i,k+1)}_e$. If the width reduction step updates the weight of a coordinate $e$ to be $\ww_e^{(i,k+1)}\gets (1+\epsilon)\ww_e^{(i,k)} + \frac{\epsilon^2}{2n} \Phi(\ww^{(i,k)})$, then we have 
\begin{align}\label{eq:Psi_width_re}
\rr^{(i,k+1)}_e = &~ \ww^{(i,k+1)}_e + \frac{\epsilon}{2n} \Phi(\ww^{(i,k+1)}) \notag \\
= &~ (1 + \epsilon) \ww^{(i,k)}_e + \frac{\epsilon^2}{2n} \Phi(\ww^{(i,k)}) + \frac{\epsilon}{2n} \Phi(\ww^{(i,k+1)}) \notag \\
= &~ (1 + \epsilon) \rr^{(i,k)}_e + \frac{\epsilon}{2n} \Phi(\ww^{(i,k+1)}) - \frac{\epsilon}{2n} \Phi(\ww^{(i,k)}) \notag \\
\leq &~ (1 + \epsilon) \rr^{(i,k)}_e + \frac{\epsilon}{2n} 2\epsilon \cdot \Phi(\ww^{(i,k+1)}) \notag \\
\leq &~ (1 + 3 \epsilon) \rr^{(i,k)}_e,
\end{align}
where the fourth step follows from $\Phi(\ww^{(i,k+1)}) \leq \Phi(\ww^{(i,k)}) \cdot (1 +  \epsilon e^{\epsilon+2\delta} (\tau^{-1} + \rho^{-2})) \leq \Phi(\ww^{(i,k)}) \cdot (1 + 2 \epsilon)$ by Lemma~\ref{lem:ChangePhiStab}. Note that we also have $\rr^{(i,k+1)}_e \geq \rr^{(i,k)}_e$ since we are only increasing the weights. Also,
\begin{align}\label{eq:Psi_width_re_diff}
\frac{\ov{\rr}^{(i,k+1)}_e - \ov{\rr}^{(i,k)}_e}{\ov{\rr}^{(i,k+1)}_e} \geq &~ \frac{\rr^{(i,k+1)}_e - (1 + 0.1 \epsilon) \rr^{(i,k)}_e}{\rr^{(i,k+1)}_e} \notag \\
\geq &~ \frac{\ww^{(i,k+1)}_e - \ww^{(i,k)}_e - 0.1 \epsilon \rr^{(i,k)}_e}{\rr^{(i,k+1)}_e} \notag \\
= &~ \frac{(1+\epsilon)\ww_e^{(i,k)} + \frac{\epsilon^2}{2n}\Phi(\ww^{(i,k)}) - \ww_e^{(i,k)} - 0.1 \epsilon \rr^{(i,k)}_e}{\rr^{(i,k+1)}_e} \notag \\
= &~ \frac{0.9 \epsilon \rr^{(i,k)}_e}{\rr^{(i,k+1)}_e}
\geq \frac{0.9 \epsilon}{1 + 3 \epsilon}.
\end{align}
where the first step follows from $\ov{\rr}^{(i,k+1)}_e = \rr^{(i,k+1)}_e$, and $\ov{\rr}^{(i,k)}_e \leq e^{\delta} \rr^{(i,k)}_e \leq (1 + 0.1 \epsilon) \rr^{(i,k)}_e$ since $\delta = \epsilon/100$, and the last step follows from Eq.~\eqref{eq:Psi_width_re}.

Next we consider the two cases of $H \neq S$ and $H = S$ separately.

{\bf When $H\neq S$.} Using Lemma~\ref{lem:PsiChange} we have the following: 
\begin{align*}
\Psi(\ov{\rr}^{(i,k+1)}) \geq &~ \Psi(\ov{\rr}^{(i,k)})+ \sum_e \left(\frac{\ov{\rr}^{(i,k+1)}_e-\ov{\rr}^{(i,k)}_e}{\ov{\rr}^{(i,k+1)}_e}\right)\ov{\rr}^{(i,k)}_e (\wt{\CC} \Delta^{(i,k)}-\wt{\dd})_e^2 \\
\geq &~ \Psi(\ov{\rr}^{(i,k)}) + \frac{0.9 \epsilon}{1 + 3 \epsilon} \cdot \sum_{e \in H \cup \{\ov{e}\}} \ov{\rr}^{(i,k)}_e (\wt{\CC} \Delta^{(i,k)}-\wt{\dd})_e^2 \\
\geq &~ \Psi(\ov{\rr}^{(i,k)}) + \frac{0.9 \epsilon}{1 + 3 \epsilon} \cdot \rho^2 \cdot \sum_{e \in H \cup \{\ov{e}\}} \ov{\rr}^{(i,k)}_e \\
\geq &~ \Psi(\ov{\rr}^{(i,k)}) + \frac{0.9 \epsilon}{1 + 3 \epsilon} \cdot \rho^2 \cdot \tau^{-1} \cdot \Psi(\rrbar^{(i,k)}) \\
= &~ \Psi(\ov{\rr}^{(i,k)}) \cdot \left( 1 + O(\epsilon \rho^{2} \tau^{-1})\right)
\geq \Psi(\ov{\rr}^{(i,k)}) \cdot \left( 1 + O\left( \frac{\epsilon^{4/3} \rho^{2/3} \cdot |H|^{1/3}}{n^{1/3}} \right) \right),
\end{align*}
where the second step follows from when $H \neq S$ the width reduction step only updates edges in $H \cup \{\ov{e}\}$ and Eq.~\eqref{eq:Psi_width_re_diff}, the third step follows from every $e \in S$ satisfies $|\wt{\CC}\Delta^{(i,k)}-\wt{\dd}|_{e}\geq \rho$ and $H \cup \{\ov{e}\} \subseteq S$, the fourth step follows from $H \subseteq S$ is a maximal subset such that $\sum_{e\in H}\rrbar_e^{(i,k)}\leq \tau^{-1}\Psi(\rrbar^{(i,k)})$ so we must have $\sum_{e\in H \cup \{\ov{e}\}}\rrbar_e^{(i,k)}\geq \tau^{-1}\Psi(\rrbar^{(i,k)})$, and the last step follows from Lemma~\ref{lem:size_H} that $|H| \leq \frac{n}{\tau \epsilon} \cdot e^{\epsilon + 2 \delta}$ and our assumption $\rho^2 \tau^{-1} \geq 0.1$ and so $\rho^2 \tau^{-1} \geq 0.2 \cdot \rho^{2/3} \tau^{-1/3}$.

{\bf When $H = S$.} In this case any $e \notin H = S$ must satisfy $|\wt{\CC}\Delta^{(i,k)}-\wt{\dd}|_e < \rho$, so we have
\[
\sum_{e \notin H}\ov{\rr}^{(i,k)}_e|\wt{\CC}\Delta^{(i,k)}-\wt{\dd}|_e^3 \leq \max_{e\notin H}\{|\wt{\CC}\Delta^{(i,k)}-\wt{\dd}|_e\}\sum_{e \notin H}\ov{\rr}^{(i,k)}_e (\wt{\CC}\Delta^{(i,k)}-\wt{\dd})_e^2 \leq \rho\Psi(\ov{\rr}^{(i,k)}).
\]
Since this is a width reduction step, we know that $\sum_{e }\ov{\rr}^{(i,k)}_e|\wt{\CC}\Delta^{(i,k)}-\wt{\dd}|_e^3 \geq 2\rho\Psi(\ov{\rr}^{(i,k)})$, and therefore combining this and the inequality above, we must have
\[
\sum_{e\in H }\ov{\rr}^{(i,k)}_e|\wt{\CC}\Delta^{(i,k)}-\wt{\dd}|_e^3 \geq \rho\Psi(\ov{\rr}^{(i,k)}).
\]

Recall that for any $\zeta = \rho, 2 \rho, 4 \rho, \cdots,  2^{c_\rho} \rho$ where $2^{c_{\rho}} \rho \geq \sqrt{n/\epsilon}$, the algorithm defines the set $H_{\zeta} = \{e \in H \mid |\wt{\CC} \Delta^{(i,k)} - \wt{\dd}|_e \in [\zeta, 2 \zeta)\}$. Note that since $\Psi(\ov{\rr}^{(i,k)}) = \sum_e \ov{\rr}^{(i,k)}_e (\wt{\CC} \Delta^{(i,k)} - \wt{\dd})_e^2$, for any $e$ we have 
\[
|\wt{\CC} \Delta^{(i,k)} - \wt{\dd}|_e \leq \left(\frac{\Psi(\ov{\rr}^{(i,k)})}{\ov{\rr}^{(i,k)}_e}\right)^{1/2} \leq \sqrt{\frac{e^{\epsilon+2\delta} n}{\epsilon}} < 2 \sqrt{\frac{n}{\epsilon}},
\]
where the second step follows from $\ov{\rr}^{(i,k)}_e \geq e^{-\delta} \frac{\epsilon}{n} \Phi(\ww^{(i,k)}) \geq e^{-\epsilon-2\delta} \frac{\epsilon}{n} \Psi(\ov{\rr}^{(i,k)})$ using Lemma~\ref{lem:PsiPhi}. So we have
\[
\bigsqcup_{\zeta = \rho}^{2^{c_{\rho}} \rho} H_{\zeta} = H.
\]
There are at most $\log(\frac{n}{\epsilon \rho})$ such $\zeta$'s, so this implies that there must exist a $\zeta^*$ that satisfies $\sum_{e \in H_{\zeta^*}} \ov{\rr}^{(i,k)}_e |\wt{\CC} \Delta^{(i,k)} - \wt{\dd}|_e^3 \geq \frac{\rho \Psi(\ov{\rr}^{(i,k)})}{\log(\frac{n}{\epsilon \rho})}$, and we can find it on Line~\ref{algline:zeta*} in Algorithm~\ref{alg:non_monotone_accel_stab}. 

Again using Lemma~\ref{lem:PsiChange} we have that
\begin{align}\label{eq:Psi_width_H_equal_S}
\Psi(\ov{\rr}^{(i,k+1)}) \geq &~ \Psi(\ov{\rr}^{(i,k)})+ \sum_e \left(\frac{\ov{\rr}^{(i,k+1)}_e-\ov{\rr}^{(i,k)}_e}{\ov{\rr}^{(i,k+1)}_e}\right)\ov{\rr}^{(i,k)}_e (\wt{\CC} \Delta^{(i,k)}-\wt{\dd})_e^2 \notag \\
\geq &~ \Psi(\ov{\rr}^{(i,k)}) + \frac{0.9 \epsilon}{1 + 3 \epsilon} \cdot \sum_{e \in H_{\zeta^*}} \ov{\rr}^{(i,k)}_e (\wt{\CC} \Delta^{(i,k)}-\wt{\dd})_e^2,
\end{align}
where the second step follows from when $H = S$ the width reduction step only updates edges in $H_{\zeta^*}$ and Eq.~\eqref{eq:Psi_width_re_diff},

Next we prove two lower bounds of $\sum_{e \in H_{\zeta^*}} \ov{\rr}^{(i,k)}_e (\wt{\CC} \Delta^{(i,k)}-\wt{\dd})_e^2$. On the one hand, we have
\begin{align*}
\sum_{e \in H_{\zeta^*}} \ov{\rr}^{(i,k)}_e (\wt{\CC} \Delta^{(i,k)}-\wt{\dd})_e^2 \geq &~ |H_{\zeta^*}| \cdot e^{-\epsilon - 2\delta} \frac{\epsilon}{n} \Psi(\ov{\rr}^{(i,k)}) \cdot (\zeta^*)^2,
\end{align*}
which follows from $|\wt{\CC} \Delta^{(i,k)} - \wt{\dd}|_e \geq \zeta^*$ for all $e \in H_{\zeta^*}$, and that $\ov{\rr}^{(i,k)}_e \geq e^{-\delta} \cdot \frac{\epsilon}{n} \Phi(\ww^{(i,k)}) \geq e^{-\epsilon-2\delta} \cdot \frac{\epsilon}{n} \Psi(\ov{\rr}^{(i,k)})$. On the other hand, we also have
\begin{align*}
\sum_{e \in H_{\zeta^*}} \ov{\rr}^{(i,k)}_e (\wt{\CC} \Delta^{(i,k)}-\wt{\dd})_e^2 \geq &~ \frac{1}{\max_{e \in H_{\zeta^*}} |\wt{\CC} \Delta^{(i,k)}-\wt{\dd}|_e} \sum_{e \in H_{\zeta^*}} \ov{\rr}^{(i,k)}_e |\wt{\CC} \Delta^{(i,k)}-\wt{\dd}|_e^3 \\
\geq &~ \frac{1}{2 \zeta^*} \cdot \frac{\rho \Psi(\ov{\rr}^{(i,k)})}{\log(\frac{n}{\epsilon \rho})},
\end{align*}
where the second step follows from $|\wt{\CC} \Delta^{(i,k)} - \wt{\dd}|_e \leq 2 \zeta^*$ for all $e \in H_{\zeta^*}$, and $\sum_{e \in H_{\zeta^*}} \ov{\rr}^{(i,k)}_e |\wt{\CC} \Delta^{(i,k)} - \wt{\dd}|_e^3 \geq \frac{\rho \Psi(\ov{\rr}^{(i,k)})}{\log(\frac{n}{\epsilon \rho})}$. 

Plugging these two lower bounds into Eq.~\eqref{eq:Psi_width_H_equal_S} we have
\begin{align*}
\Psi(\ov{\rr}^{(i,k+1)}) \geq &~ \Psi(\ov{\rr}^{(i,k)}) \cdot \Big(1 + \frac{0.9 \epsilon}{1 + 3 \epsilon} \cdot \max\Big\{e^{-\epsilon-2\delta} \cdot \frac{\epsilon \cdot |H_{\zeta^*}|\cdot (\zeta^*)^2}{n},~~ \frac{1}{2 \log(\frac{n}{\epsilon \rho})} \cdot \frac{\rho}{\zeta^*}\Big\} \Big) \\
\geq &~ \Psi(\ov{\rr}^{(i,k)}) \cdot \left(1 + O\left(\epsilon^{4/3} \cdot \frac{|H_{\zeta^*}|^{1/3} \cdot \rho^{2/3}}{n^{1/3} \cdot \log^{2/3}(\frac{n}{\epsilon \rho})} \right) \right),
\end{align*}
where the second step follows from $a+b \geq a^{1/3} b^{2/3}$.

Finally, note that when $H \neq S$ the size of the width reduction step is $\textsc{Size}(k) = |H| + 1$, and when $H = S$ the size of the width reduction step is $\textsc{Size}(k) = |H_{\zeta^*}|$. So combining these two cases we can conclude that we always have
\[
\Psi(\ov{\rr}^{(i,k+1)}) \geq \Psi(\ov{\rr}^{(i,k)}) \cdot \left(1 + O\Big( \frac{\epsilon^{4/3} \cdot \rho^{2/3} \cdot \textsc{Size}(k)^{1/3}}{n^{1/3} \cdot \log^{2/3}(\frac{n}{\epsilon \rho})}\Big) \right). 
\]

\noindent {\bf Primal step.} Next we prove that $\Psi$ doesn't decrease by too much after any primal step $i$.

Since $\ww_e^{(i+1,k_i)} = \ww_e^{(i,k_i)}(1+\epsilon\overrightarrow{\alpha}_e^{(i,k_i)} (\wt{\CC}\Delta^{(i,k_i)}-\wt{\dd})_e)$, we have
\begin{align}\label{eq:Psi_primal_re}
\rr_e^{(i+1,k_i)} & = \ww_e^{(i+1,k_i)} + \frac{\epsilon}{n}\Phi(\ww^{(i+1,k_i)}) \notag \\
&\geq \ww_e^{(i,k_i)} -  \epsilon\alpha_+ \ww^{(i,k_i)}_e |\wt{\CC}\Delta^{(i,k_i)}-\wt{\dd}|_e + \frac{\epsilon}{n} \Phi(\ww^{(i,k_i)}) (1-\epsilon\alpha_+ e^{\epsilon+\delta}) \notag \\
& \geq \ww^{(i,k_i)}_e (1-\epsilon) + \frac{\epsilon}{n} \Phi(\ww^{(i,k_i)}) (1-\epsilon) = \rr^{(i,k_i)}_e (1-\epsilon),
\end{align}
where the second step follows from $\overrightarrow{\alpha}^{(i,k_i)} \leq \alpha_+$ and $\Phi(\ww^{(i+1,k_i)}) \geq \Phi(\ww^{(i,k_i)}) \cdot (1-\epsilon\alpha_+ e^{\epsilon+\delta})$ by Lemma~\ref{lem:ChangePhiStab}, the third step follows from $\alpha_+ |\wt{\CC} \Delta^{(i,k_i)} - \wt{\dd}| \leq 0.1$ from the same proof as that of Lemma~\ref{lem:positiveWopt}. 
Similarly we also have
\begin{equation}\label{eq:Psi_primal_re_UB}
\rr_e^{(i+1,k_i)} \leq (1 + \epsilon) \rr_e^{(i,k_i)}.
\end{equation}

We also have,
\begin{align}\label{eq:Psi_primal_re_diff}
    \left|\frac{\rr_e^{(i+1,k_i)}-\rr_e^{(i,k_i)}}{\rr_e^{(i+1,k_i)}}\right| & \leq \frac{|\ww_e^{(i+1,k_i)}-\ww_e^{(i,k_i)}| + \frac{\epsilon}{n}|\Phi(\ww^{(i+1,k_i)})-\Phi(\ww^{(i,k_i)})|}{\rr_e^{(i,k_i)} (1-\epsilon)} \notag \\
    & \leq \frac{\epsilon \alpha_+ \ww_e^{(i,k_i)} |\wt{\CC}\Delta^{(i,k_i)}-\wt{\dd}|_e + \frac{\epsilon}{n}\cdot e^{\epsilon+\delta} \epsilon\alpha_+ \Phi(\ww^{(i,k_i)})}{\rr^{(i,k_i)}_e (1-\epsilon)} \notag \\
    & \leq \frac{\epsilon \alpha_+ \rr^{(i,k_i)}_e|\wt{\CC}\Delta^{(i,k_i)}-\wt{\dd}|_e + e^{\epsilon+\delta} \epsilon\alpha_+\rr^{(i,k_i)}_e}{\rr^{(i,k_i)}_e (1-\epsilon)} \notag \\
    & \leq (1+2\epsilon)\epsilon\alpha_+ |\wt{\CC}\Delta^{(i,k_i)}-\wt{\dd}|_e + (1+4\epsilon)\epsilon\alpha_+. 
\end{align}

Next we consider the two cases that could happen to the coordinate $e$ in the $i$-th primal step. From now on, when it's clear from the context, we will use $\rr_e^{(i)}$ to refer to $\rr_e^{(i,k_{i-1})}$, and similarly $\rrbar_e^{(i)}$ to refer to $\rrbar_e^{(i,k_{i-1})}$, so that this is consistent with the notations used in \textsc{SelectVector}.
\begin{enumerate}
\item If \textsc{SelectVector} doesn't update $\ov{\rr}_e$ on the $i$-th iteration, then we have $\ov{\rr}^{(i+1,k_i)}_e = \ov{\rr}^{(i,k_i)}_e$. 
\item If \textsc{SelectVector} does update $\ov{\rr}_e$ on the $i$-th iteration, i.e., $e \in S_i$, then defining $j_{i,e} := \max\{\textsc{LastWidth}(i,e), \textsc{Last}(i,e)\}$, i.e., $j_{i,e} \leq i$ is the last primal iterate during which the algorithm updates $\ww_e$. Define $\ell_{i,e}$ to be the smallest integer $\ell$ such that $i+1 \equiv 0 \pmod{2^{\ell}}$ and $|\ln(\frac{\rr_e^{(i+1)}}{\rr_e^{(i+1-2^\ell)}})| \geq \frac{\delta}{2 \log n}$. 
By definition we have $\ov{\rr}^{(i+1,k_i)}_e = \rr^{(i+1,k_i)}_e$, and $\ov{\rr}^{(i,k_i)}_e = \rr^{(j_{i,e},k_{j_{i,e}})}_e$.

Further note that in this case we have the following properties:
\begin{itemize}
\item For all $j \in [j_{i,e}+1, i]$, the value of $\rrbar_e$ remains the same for all width reduction steps between the $(j-1)^{th}$ and $j^{th}$ primal steps, i.e., $\rrbar^{(j,k_{j-1})}_e = \rrbar^{(j,k_{j})}_e$. This is because there is no width reduction step that updates the weight of $e$ in these iterations.
\item For all $j \in [j_{i,e}+1, i]$, $\rr_e^{(j,k_j)} \approx_{\delta} \ov{\rr}_e^{(j,k_j)} = \rr_e^{(j_{i,e},k_{j_{i,e}})}$. This is because the value of $\ov{\rr}_e$ remains the same from primal iterations $(j_{i,e}+1)$ to $i$, and it is always a $\delta$-approximation of the true value of $\rr_e$.
\item If $i+1-2^{\ell_{i,e}} > j_{i,e}$, then
\begin{align*}
|\rr_e^{(i+1-2^{\ell_{i,e}})} - \rr_e^{(j_{i,e},k_{j_{i,e}})}| \leq &~ \delta \rr_e^{(i+1-2^{\ell_{i,e}})} \\
\leq &~ 2 \log n \cdot |\ln(\frac{\rr_e^{(i+1)}}{\rr_e^{(i+1-2^\ell_{i,e})}})| \cdot \rr_e^{(i+1-2^{\ell_{i,e}})} \\
\leq &~ 5 \log n \cdot |\rr_e^{(i+1)} - \rr_e^{(i+1-2^{\ell_{i,e}})}|
\end{align*}
where the first step follows from $\rr_e^{(i+1-2^{\ell_{i,e}})} \approx_{\delta} \ov{\rr}_e^{(i+1-2^{\ell_{i,e}})} = \rr_e^{(j_{i,e},k_{j_{i,e}})}$ since $i+1-2^{\ell_{i,e}} > j_{i,e}$, the second step follows from $|\ln(\frac{\rr_e^{(i+1)}}{\rr_e^{(i+1-2^\ell_{i,e})}})| \geq \frac{\delta}{2 \log n}$, the third step follows from $|\ln(\frac{\rr_e^{(i+1)}}{\rr_e^{(i+1-2^\ell_{i,e})}})| \leq 2 |\frac{\rr_e^{(i+1)}}{\rr_e^{(i+1-2^\ell_{i,e})}} - 1|$ since $|\ln(\frac{\rr_e^{(i+1)}}{\rr_e^{(i+1-2^\ell_{i,e})}})| \leq 0.1$, and this is because $\rr_e^{(i,k_i)} \approx_{\delta} \rr_e^{(j_{i,e}, k_{j_{i,e}})} \approx_{\delta} \rr_e^{(i+1-2^{\ell_{i,e}})} $ by the second bullet and $\rr_e^{(i+1,k_i)} \approx_{\epsilon} \rr_e^{(i,k_i)}$ by Eq.~\eqref{eq:Psi_primal_re} and \eqref{eq:Psi_primal_re_UB}.
\end{itemize}
\end{enumerate}
Because of the third bullet point above, if $j_{i,e} < i+1-2^{\ell_{i,e}}$ then we have
\begin{align*}
|\ov{\rr}^{(i+1,k_i)}_e - \ov{\rr}^{(i ,k_i)}_e| = &~ |\rr_e^{(i+1,k_i)} - \rr_e^{(j_{i,e},k_{j_{i,e}})}| \\
\leq &~ |\rr^{(i+1,k_i)}_e - \rr^{(i+1-2^{\ell_{i,e}})}_e| + |\rr^{(i+1-2^{\ell_{i,e}})}_e - \rr^{(j_{i,e},k_{j_{i,e}})}_e| \\
\leq &~ 10 \log n\cdot |\rr^{(i+1,k_i)}_e - \rr^{(i+1-2^{\ell_{i,e}})}_e|.
\end{align*}
From now on we can without loss of generality assume that $j_{i,e} \geq i+1-2^{\ell_{i,e}}$, since otherwise we can upper bound $|\ov{\rr}^{(i+1,k_i)}_e - \ov{\rr}^{(i ,k_i)}_e|$ by $\wt{O}(|\rr^{(i+1,k_i)}_e - \rr^{(i+1-2^{\ell_{i,e}})}_e|)$ instead of $|\rr_e^{(i+1,k_i)} - \rr_e^{(j_{i,e},k_{j_{i,e}})}|$.
Now,
\begin{align*}
\Psi(\ov{\rr}^{(i+1,k_i)}) 
\substack{(i)\\ \geq} &~ \Psi(\ov{\rr}^{(i,k_i)}) - \sum_e \Big|\frac{\rrbar_e^{(i+1,k_i)}-\ov{\rr}^{(i,k_i)}_e}{\ov{\rr}_e^{(i+1,k_i)}} \Big| \ov{\rr}^{(i,k_i)}_e (\wt{\CC}\Delta^{(i,k_i)}-\wt{\dd})_e^2\\
\substack{(ii)\\ \geq}&~ \Psi(\ov{\rr}^{(i,k_i)}) - \sum_{e \in S_i} \Big|\frac{\rr_e^{(i+1,k_i)} - \rr_e^{(j_{i,e},k_{j_{i,e}})}}{\rr_e^{(i+1,k_i)}} \Big| \ov{\rr}^{(i,k_i)}_e (\wt{\CC}\Delta^{(i,k_i)}-\wt{\dd})_e^2 \\
\substack{(iii)\\ \geq} &~ \Psi(\ov{\rr}^{(i,k_i)}) - \sum_{e \in S_i} \Big|\frac{\sum_{j=j_{i,e}}^{i} (\rr_e^{(j+1,k_j)} - \rr_e^{(j,k_j)}) }{\rr_e^{(i+1,k_i)}} \Big| \ov{\rr}^{(i,k_i)}_e (\wt{\CC}\Delta^{(i,k_i)}-\wt{\dd})_e^2 \\
\substack{(iv)\\ \geq} &~ \Psi(\ov{\rr}^{(i,k_i)}) - \frac{e^{2\delta}}{(1-\epsilon)} \cdot \sum_{e \in S_i} \sum_{j=j_{i,e}}^{i} \Big|\frac{ (\rr_e^{(j+1,k_j)} - \rr_e^{(j,k_j)}) }{\rr_e^{(j+1,k_j)}} \Big| \ov{\rr}^{(i,k_i)}_e (\wt{\CC}\Delta^{(i,k_i)}-\wt{\dd})_e^2 \\
\substack{(v)\\ \geq} &~ \Psi(\ov{\rr}^{(i,k_i)}) - O(1) \cdot \sum_{e \in S_i} \sum_{j=j_{i,e}}^{i} \Big(  \epsilon\alpha |\wt{\CC}\Delta^{(j,k_j)}-\wt{\dd}|_e + \epsilon\alpha \Big) \ov{\rr}^{(i,k_i)}_e (\wt{\CC}\Delta^{(i,k_i)}-\wt{\dd})_e^2 \\
\substack{(vi)\\ \geq} &~ \Psi(\ov{\rr}^{(i,k_i)}) - O(1) \cdot \epsilon \alpha \cdot \sum_{e \in S_i} (i + 1 - j_{i,e}) \cdot \ov{\rr}^{(i,k_i)}_e (\wt{\CC}\Delta^{(i,k_i)}-\wt{\dd})_e^2 \\
&~ - O(1) \cdot \epsilon \alpha \cdot \sum_{e \in S_i} (i + 1 - j_{i,e}) \cdot \ov{\rr}^{(i,k_i)}_e |\wt{\CC}\Delta^{(i,k_i)}-\wt{\dd}|_e^3\\
&~ - O(1) \cdot \epsilon \alpha \cdot \sum_{e \in S_i} \sum_{j=j_{i,e}}^{i-1} \ov{\rr}^{(j,k_j)}_e |\wt{\CC}\Delta^{(j,k_j)}-\wt{\dd}|_e^3,
\end{align*}
where $(i)$ follows from Lemma~\ref{lem:PsiChange}, $(ii)$ follows from $\ov{\rr}^{(i+1,k_i)}_e = \rr^{(i+1,k_i)}_e$, and $\ov{\rr}^{(i,k_i)}_e = \rr^{(j_{i,e},k_{j_{i,e}})}_e$ for $e \in S_i$, $(iv)$ follows from $\rr_e^{(j+1,k_j)} \approx_{\delta} \rr_e^{(j_{i,e},k_{j_{i,e}})}$ and $\rr_e^{(i,k_i)} \approx_{\delta} \rr_e^{(j_{i,e},k_{j_{i,e}})}$ as we argued above, and since we also have $\rr_e^{(i+1,k_i)} \geq (1-\epsilon) \rr_e^{(i,k)}$ from Eq.~\eqref{eq:Psi_primal_re}, combining these we have $\rr_e^{(j+1,k_j)} \leq e^{\delta} \rr_e^{(j_{i,e},k_{j_{i,e}})} \leq e^{2\delta} \rr_e^{(i,k_i)} \leq \frac{e^{2\delta}}{(1-\epsilon)} \rr_e^{(i+1,k)}$. Step $(v)$ follows from Eq.~\eqref{eq:Psi_primal_re_diff}, $(vi)$ follows from AM-GM inequality that $|\wt{\CC}\Delta^{(j,k_j)}-\wt{\dd}|_e (\wt{\CC}\Delta^{(i,k_i)}-\wt{\dd})_e^2 \leq \frac{1}{3} \cdot |\wt{\CC}\Delta^{(j,k_j)}-\wt{\dd}|_e^3 + \frac{2}{3} \cdot |\wt{\CC}\Delta^{(i,k_i)}-\wt{\dd}|_e^3$ and that $\rr_e^{(j,k_j)} \approx_{\delta} \rr_e^{(j_{i,e},k_{j_{i,e}})} \approx_{\delta} \ov{\rr}_e^{(i,k_i)}$ for all $e \in S_i$ and $j \in [j_{i,e}, i]$.

With an abuse of notation, let $\ell_i$ denote the largest integer such that $i+1 \equiv 0 \pmod{2^{\ell_i}}$. Since we assumed that $j_{i,e} \geq i+1-2^{\ell_{i,e}}$ for all $e \in S_i$, we have $i+1 - j_{i,e} \leq 2^{\ell_{i,e}} \leq 2^{\ell_i}$. Also note that in primal steps we have $\sum_e \rrbar_e^{(i,k)} |\wt{\CC}\Delta^{(i,k)}-\wt{\dd}|_e^3\leq 2\rho\Psi(\rrbar^{(i,k)})$, so the above equation becomes, for some $C_1 = \wt{O}(1)$
\begin{align}\label{eq:Psi_primal_one_step}
\Psi(\ov{\rr}^{(i+1,k_i)}) \geq \Big( 1 - C_1 \cdot \epsilon \alpha \rho \cdot 2^{\ell_i} \Big) \cdot \Psi(\ov{\rr}^{(i,k_i)}) - C_1 \cdot \epsilon \alpha \cdot \sum_{e \in S_i} \sum_{j=j_{i,e}}^{i-1} \ov{\rr}^{(j,k_j)}_e |\wt{\CC}\Delta^{(j,k_j)}-\wt{\dd}|_e^3.
\end{align}

Recall that we defined $L = \Theta(\frac{1}{(\log^4 n) \cdot \epsilon \alpha \rho})$ to be a power of $2$. Next we consider the iterations between $(c-1) L$ and $cL$ using Eq.~\eqref{eq:Psi_primal_one_step}. We have
\begin{align*}
&~ \Psi(\ov{\rr}^{(cL, k_{cL-1})}) \\
\geq &~ \Big(1 -  C_1 \cdot \epsilon \alpha \rho \cdot 2^{\ell_{cL-1}}\Big) \cdot \Psi(\ov{\rr}^{(cL-1, k_{cL-1})}) - C_1 \cdot \epsilon \alpha \cdot \sum_{e \in S_{cL-1}} \sum_{j=j_{cL-1,e}}^{cL-2} \ov{\rr}^{(j,k_j)}_e |\wt{\CC}\Delta^{(j,k_j)}-\wt{\dd}|_e^3 \\
\geq &~ \Big(1 -  C_1 \cdot \epsilon \alpha \rho \cdot (2^{\ell_{cL-1}} + 2^{\ell_{cL-2}}) \Big) \cdot \frac{\Psi(\ov{\rr}^{(cL-1, k_{cL-1})})}{\Psi(\ov{\rr}^{(cL-1, k_{cL-2})})} \cdot \Psi(\ov{\rr}^{cL-2, k_{cL-2}}) \\
&~ - \Big(1 -  C_1 \cdot \epsilon \alpha \rho \cdot 2^{\ell_{cL-1}}\Big) \cdot \frac{\Psi(\ov{\rr}^{(cL-1, k_{cL-1})})}{\Psi(\ov{\rr}^{(cL-1, k_{cL-2})})} \cdot C_1 \cdot \epsilon \alpha \cdot \sum_{e \in S_{cL-2}} \sum_{j=j_{cL-2,e}}^{cL-3} \ov{\rr}^{(j,k_j)}_e |\wt{\CC}\Delta^{(j,k_j)}-\wt{\dd}|_e^3  \\
&~ - C_1 \cdot \epsilon \alpha \cdot \sum_{e \in S_{cL-1}} \sum_{j=j_{cL-1,e}}^{cL-2} \ov{\rr}^{(j,k_j)}_e |\wt{\CC}\Delta^{(j,k_j)}-\wt{\dd}|_e^3 \\
\geq &~ \Big(1 -  C_1 \cdot \epsilon \alpha \rho \cdot (2^{\ell_{cL-1}} + 2^{\ell_{cL-2}} + 2) \Big) \cdot \frac{\Psi(\ov{\rr}^{(cL-1, k_{cL-1})})}{\Psi(\ov{\rr}^{(cL-1, k_{cL-2})})} \cdot \Psi(\ov{\rr}^{cL-2, k_{cL-2}}) \\
&~ - \Big(1 -  C_1 \cdot \epsilon \alpha \rho \cdot 2^{\ell_{cL-1}}\Big) \cdot \frac{\Psi(\ov{\rr}^{(cL-1, k_{cL-1})})}{\Psi(\ov{\rr}^{(cL-1, k_{cL-2})})} \cdot C_1 \cdot \epsilon \alpha \cdot \sum_{e \in S_{cL-2}} \sum_{j=j_{cL-2,e}}^{cL-3} \ov{\rr}^{(j,k_j)}_e |\wt{\CC}\Delta^{(j,k_j)}-\wt{\dd}|_e^3  \\
&~ - C_1 \cdot \epsilon \alpha \cdot \sum_{e \in S_{cL-1}} \sum_{j=j_{cL-1,e}}^{cL-3} \ov{\rr}^{(j,k_j)}_e |\wt{\CC}\Delta^{(j,k_j)}-\wt{\dd}|_e^3 \\
\geq &~ \cdots \\
\geq &~ \Big(1 -  C_1 \cdot \epsilon \alpha \rho \cdot (2L + \sum_{i=(c-1)L+1}^{cL} 2^{\ell_{i-1}}) \Big) \cdot \prod_{i=(c-1)L}^{cL} \frac{\Psi(\ov{\rr}^{(i-1, k_{i-1})})}{\Psi(\ov{\rr}^{(i-1, k_{i-2})})} \cdot \Psi(\ov{\rr}^{(c-1)L, k_{(c-1)L}}),
\end{align*}
where the second step follows from $\Psi(\ov{\rr}^{cL-1,k_{cL-2}}) \geq \Big(1 -  C_1 \cdot \epsilon \alpha \rho \cdot 2^{\ell_{cL-2}}\Big) \cdot \Psi(\ov{\rr}^{cL-2, k_{cL-2}}) - C_1 \cdot \epsilon \alpha \cdot \sum_{e \in S_{cL-2}} \sum_{j=j_{cL-2,e}}^{cL-3} \ov{\rr}^{(j,k_j)}_e |\wt{\CC}\Delta^{(j,k_j)}-\wt{\dd}|_e^3$, the third step follows from taking out all terms $\ov{\rr}^{(j,k_j)}_e |\wt{\CC}\Delta^{(j,k_j)}-\wt{\dd}|_e^3$ for $j = cL-2$ and that $\sum_e \rrbar_e^{(j,k_{j})} |\wt{\CC}\Delta^{(j,k_{j})}-\wt{\dd}|_e^3\leq 2\rho\Psi(\rrbar^{(j,k_{j})})$, and the last two steps follow from repeat this process for $L$ times, and noting that for any $cL-t$, an coordinate $e$ can only be in one $S_{cL-t'}$ where $t' < t$ and $j_{cL-t',e} \leq cL-t$, and also noting that we proved that we can wlog assume $j_{cL-t,e} \geq cL-t+1-2^{\ell_{cL-t,e}}$ for all $t\leq L$, and this is then $\geq (c-1)L$ since $L$ is a power of $2$. 

Finally, note that $\prod_{i=(c-1)L}^{cL} \frac{\Psi(\ov{\rr}^{(i-1, k_{i-1})})}{\Psi(\ov{\rr}^{(i-1, k_{i-2})})}$ is exactly the increase that we get from the width reduction steps, and also note that we have $\sum_{i=(c-1)L+1}^{cL} 2^{\ell_{i-1}} \leq L \cdot \log n$ since by definition $\ell_{i-1}$ is the largest integer $\ell$ such that $i \equiv 0 \pmod{2^{\ell}}$. And this gives the claimed lower bound of this lemma. 
\end{proof}

\subsubsection{Putting Everything Together: Analysis of Algorithm}
Next we analyze the iteration complexity and the error of Algorithm~\ref{alg:non_monotone_accel_stab}. We first bound the total number of width reduction steps. In the following lemma we denote the hidden factors in Lemma~\ref{lem:ChangePsiStab} as $C_2 \leq O(\log^3 n)$ and $C_3 \geq O(1)$ such that
\begin{align*}
\Psi(\ov{\rr}^{(cL,k_{cL})}) 
\geq &~ \Psi(\ov{\rr}^{((c-1)L,k_{(c-1)L})}) \cdot \Big(1 -  C_2 \epsilon \alpha \rho L \Big) \cdot \prod_{k=k_{(c-1)L}}^{k_{cL}} \left( 1 + C_3 \frac{\epsilon^{4/3} \rho^{2/3} \cdot \textsc{Size}(k)^{1/3}}{n^{1/3} \cdot \log^{2/3}(\frac{n}{\epsilon \rho})} \right).
\end{align*}
\begin{lemma}[Number of width reduction steps]\label{lem:number_of_width_reduction_steps}
The total number of width reduction steps of Algorithm~\ref{alg:non_monotone_accel_stab} is at most $O\left(\frac{n^{1/3} \rho^{1/3} \log^5 n}{\epsilon^{10/3}} \cdot \log(\frac{n}{\Psi_0}) \right)$ for large enough $n$. 
\end{lemma}
\begin{proof}
Let $K$ denote the total number of width reduction steps. We first assume that we halt the algorithm if there are more than $K' = 10^4 C_3^{-1} \cdot \frac{n^{1/3} \rho^{1/3} \log^6 n}{\epsilon^{10/3}} \cdot \log(\frac{n}{\Psi_0})$ width reduction steps. We will then prove that $K \leq 9000 C_3^{-1} \cdot \frac{n^{1/3} \rho^{1/3} \log^5 n}{\epsilon^{10/3}} \cdot \log(\frac{n}{\Psi_0}) < K'$, which means we can make this assumption without changing the algorithm. Under this assumption, and using Lemma~\ref{lem:ChangePhiStab}, we have that during the algorithm, we always have
\begin{align}\label{eq:upper_bound_Phi}
\Phi\left(\ww^{(i,k)}\right) \leq &~ \Phi(\ww^{(0,0)}) \cdot \left(1+\epsilon\alpha e^{\epsilon+\delta} \right)^T \cdot \left( 1 + \epsilon e^{\epsilon+2\delta} (\tau^{-1} + \rho^{-2})) \right)^{K'} \notag \\
\leq &~ n \cdot \exp(2 \epsilon\alpha \cdot \alpha^{-1} \epsilon^{-2} \log n) \cdot \exp(2 \epsilon (\tau^{-1} + \rho^{-2}) 10^4 C_3^{-1} \cdot \frac{n^{1/3} \rho^{1/3} \log^6 n}{\epsilon^{10/3}} \cdot \log(\frac{n}{\Psi_0})) \notag \\
\leq &~ n^{2 /\epsilon + 10^4 C_3^{-1}/\epsilon} \leq n^{3 \log n / \epsilon}, 
\end{align}
where in the second and third steps we used the parameters of Algorithm~\ref{alg:non_monotone_accel_stab} that $T = \alpha^{-1} \epsilon^{-2} \log n$, $\tau = n^{1/2-\eta} \cdot \epsilon^{-4} \cdot \log(n)^8 \log(\frac{n}{\Psi_0})^2$, $\rho = n^{1/2-3\eta} \cdot \epsilon^{-2} \cdot \log(n)^4 \log(\frac{n}{\Psi_0})$, and $\eta \leq 1/10$, in the last step we assume $n$ is large enough such that $n \geq 10^4 C_3^{-1}$.

Consider any integer $c \geq 1$. We next bound the number of width reduction steps between primal steps $(c-1)L$ and $cL$, and in this proof we denote this number as $K_c$.

Using Lemma~\ref{lem:ChangePsiStab}, we have
\begin{align*}
\left( 1 + C_3 \frac{\epsilon^{4/3} \cdot \rho^{2/3}}{n^{1/3} \log^{2/3}(\frac{n}{\epsilon\rho})} \right)^{K_c} \leq &~ \frac{2 \Psi\left(\ov{\rr}^{(cL,k_{cL})}\right)}{\Psi(\ov{\rr}^{((c-1)L,k_{(c-1)L})})}  
\leq  \frac{4 \Psi\left(\ov{\rr}^{(cL,k_{cL})}\right)}{\Psi_0} 
\leq \frac{10 \cdot n^{3 \log n / \epsilon}}{\Psi_0} \\
\implies K_c \leq &~ \frac{3 \cdot n^{1/3} \log n}{C_3 \cdot \epsilon^{4/3} \cdot \rho^{2/3}} \cdot \frac{3 \log n}{\epsilon} \cdot \log(\frac{n}{\Psi_0}),
\end{align*}
where the first step follows from Lemma~\ref{lem:ChangePsiStab} and that we always have $\textsc{Size}(k) \geq 1$ for all $k$, and that the $C_2 \epsilon \alpha \rho \cdot L$ factor of Lemma~\ref{lem:ChangePsiStab} is upper bounded by $1/2$ since $L \leq \frac{1}{100 (\log^4 n) \epsilon \alpha \rho}$, the second step follows from the same proof as Lemma~\ref{lem:lower_bound_Psi} that $\Psi(\ov{\rr}^{((c-1)L,k_{(c-1)L})})\geq \frac{1}{1+2\epsilon}\cdot\Psi(\ov{\rr}^{(0,0)})$ and Lemma~\ref{lem:lower_bound_Psi_0} that $\Psi(\ov{\rr}^{(0,0)}) \geq \Psi_0$, and the third step follows from $\Psi(\rrbar^{(cL,k_{cL})}) \leq e^{\epsilon + \delta} \Phi(\ww^{(cL,k_{cL})}) \leq e^{\epsilon + \delta} n^{3 \log n / \epsilon}$ by Lemma~\ref{lem:PsiPhi} and Eq.~\eqref{eq:upper_bound_Phi}.

Note that the above bound holds for any integer $c$. Since there are in total $T = \alpha^{-1}\epsilon^{-2}\log n $ number of primal steps, and since $L \leq \frac{1}{1000 \epsilon \alpha \rho \log^4 n}$, the total number of width reduction steps is upper bounded by
\begin{align*}
\frac{T}{L} \cdot K_c \leq &~ \frac{1000 \rho \log^3 n}{\epsilon} \cdot \frac{3 \cdot n^{1/3} \log n}{C_3 \cdot \epsilon^{4/3} \cdot \rho^{2/3}} \cdot \frac{3 \log n}{\epsilon} \cdot \log(\frac{n}{\Psi_0}) \\
\leq &~ 9000 C_3^{-1} \cdot \frac{n^{1/3} \rho^{1/3} \log^5 n}{\epsilon^{10/3}} \cdot \log(\frac{n}{\Psi_0}). \qedhere
\end{align*}
\end{proof}

Next we bound the error of Algorithm~\ref{alg:non_monotone_accel_stab}.
\begin{lemma}[Error of Algorithm~\ref{alg:non_monotone_accel_stab}]\label{lem:error_non_monotone_acc_stab}
Algorithm~\ref{alg:non_monotone_accel_stab} outputs a vector $\xxhat \in \R^d$ such that $\|\CC \xxhat - \dd\|_{\infty} \leq 1 + O(\epsilon)$.
\end{lemma}
\begin{proof}
First note that all the requirements on the parameters of the lemmas in Section~\ref{sec:NonMonStab} are satisfied by the parameters of Algorithm~\ref{alg:non_monotone_accel_stab} for $\eta \leq 1/10$.

Let $\xxhat = \frac{\xx}{T}$ be the solution returned by Algorithm \ref{alg:non_monotone_accel_stab}. We will bound the objective value at $\xxhat$. The algorithm has $T = \alpha^{-1}\epsilon^{-2}\log n$ primal steps, and by Lemma~\ref{lem:number_of_width_reduction_steps} we know that it has at most $K \leq O\left(\frac{n^{1/3} \rho^{1/3} \log^5 n}{\epsilon^{10/3}} \cdot \log(\frac{n}{\Psi_0}) \right)$ width reduction steps. We can now apply Lemma \ref{lem:ChangePhiStab} to get,
\[
\Phi\left(\ww^{(T,K)}\right) \le  n\cdot  e^{O(\epsilon\alpha T+\epsilon(\tau^{-1}+\rho^{-2}) K)} \leq  n^{O\left(\frac{1}{\epsilon} \right)},
\]
where the second step follows from the parameters of Algorithm~\ref{alg:non_monotone_accel_stab} that $\alpha = n^{-1/2+\eta} \cdot \epsilon \cdot \log(n)^{-4/3} \log(\frac{n}{\Psi_0})^{-1/3} / 10$, $\tau = n^{1/2-\eta} \cdot \epsilon^{-4} \cdot \log(n)^8 \log(\frac{n}{\Psi_0})^2$, and $\rho = n^{1/2-3\eta} \cdot \epsilon^{-2} \cdot \log(n)^4 \log(\frac{n}{\Psi_0})$.

We bound the $\ell_{\infty}$ norm of $\CC \xxhat - \dd = \frac{1}{T} \cdot \sum_{i=0}^{T-1} (\CC \Delta^{(i,k_i)} - \dd)$ using the upper bound of the potential. Since $\wt{\CC} = \begin{bmatrix}
    \CC \\-\CC
\end{bmatrix}$, and $\wt{\dd}= \begin{bmatrix}
    \dd \\-\dd
\end{bmatrix}$, we have that the weights $\ww \in \mathbb{R}^{2n}$. Therefore, for $\ww_+\in \mathbb{R}^n$ and $\ww_{-}\in \mathbb{R}^n$, we can write $\ww^{(i,k)} = \begin{bmatrix}
    \ww_+^{(i,k)}\\ \ww_{-}^{(i,k)}
\end{bmatrix}$, and we have that $\Phi(\ww) = \sum_{e\in [n]}{\ww_+}_e + {\ww_{-}}_e$. We can similarly define $\ov{\rr}_+$ and $\ov{\rr}_{-}$ such that $\ov{\rr} = \begin{bmatrix}
    \ov{\rr}_+\\ \ov{\rr}_{-}
\end{bmatrix}$. Since $\Delta^{(i,k)}$ is obtained by solving,
\[
\Delta^{(i,k)} = \arg\min_{\Delta}\sum_{e\in [2n]}\ov{\rr}^{(i,k)}_e (\wt{\CC}\Delta-\wt{\dd})_e^2 = \sum_{e\in [n]}(\ov{\rr}_+^{(i,k)} + \ov{\rr}_{-})^{(i,k)}_e (\CC\Delta-\dd)_e^2, 
\]
the update rule $\ww^{(i+1,k)} = \ww^{(i,k)} \cdot \big(1 + \epsilon \overrightarrow{\alpha}^{(i,k)} (\wt{\CC} \Delta^{(i,k)} - \wt{\dd})\big)$ implies that in every primal step,
\begin{align*}
(\ww_+)_e^{(i+1,k)}=&~ (\ww_+)_e^{(i,k)} \cdot \big(1 + \epsilon \overrightarrow{\alpha}^{(i,k)} (\CC \Delta^{(i,k)} - \dd)\big), \\
(\ww_-)_e^{(i+1,k)}=&~ (\ww_-)_e^{(i,k)} \cdot \big(1 - \epsilon\overrightarrow{\alpha}^{(i,k)} (\CC \Delta^{(i,k)} - \dd)\big).
\end{align*}
Now,
\begin{align*}
(\ww_+)^{(T,K)}_e = &~ (\ww_+)^{(0,0)}_e \cdot \prod_{i=0}^{T-1} \Big(1 + \epsilon \overrightarrow{\alpha}^{(i,k_i)}_e (\CC \Delta^{(i,k_i)} - \dd)_e\Big) \\
= &~ \prod_{i: (\CC \Delta^{(i,k_i)} - \dd)_e \geq 0} (1 + \epsilon \alpha_+ (\CC \Delta^{(i,k_i)} - \dd)_e) \cdot \prod_{i: (\CC \Delta^{(i,k_i)} - \dd)_e < 0} (1 + \epsilon \alpha_- (\CC \Delta^{(i,k_i)} - \dd)_e) \\
\geq &~ \exp\left(\epsilon(1-\epsilon)\alpha \cdot \sum_{i=0}^{T-1} (\CC \Delta^{(i,k_i)} - \dd)_e\right),
\end{align*}
where the second step follows from $\ww_+^{(0,0)} = 1_{n}$, and $\overrightarrow{\alpha}_e^{(i,k_i)} = \alpha_+$ if $(\CC \Delta^{(i,k_i)} - \dd)_e \geq 0$ and $\overrightarrow{\alpha}_e^{(i,k_i)} = \alpha_-$ otherwise, the third step follows from $1 + \epsilon x \geq \exp(\epsilon (1-\epsilon) x)$ for all $0 \leq x \leq 1$ and $1 + \epsilon x \geq \exp(\epsilon (1 + \epsilon) x)$ for all $-1 \leq x \leq 0$, and we have that $|\overrightarrow{\alpha}^{(i,k)} \cdot (\CC \Delta^{(i,k)} - \dd)| \leq  \frac{1}{10}$ by Lemma~\ref{lem:positiveWopt}. Similarly, we also get,
\[
(\ww_{-})^{(T,K)}_e \geq \exp\left(\epsilon(1-\epsilon)\alpha \cdot \sum_{i=0}^{T-1} - (\CC \Delta^{(i,k_i)} - \dd)_e\right).
\]

This implies that
\begin{align}\label{eq:finalBound}
\left|\sum_{i=0}^{T-1} (\CC \Delta^{(i,k_i)} - \dd)_e\right| \leq \frac{\ln\left((\ww_+)^{(T,K)}_e +(\ww_{-})_e^{(T,K)}\right)}{\epsilon(1-\epsilon)\alpha} \leq \frac{\ln(\Phi(\ww^{(T,K)}))}{\epsilon(1-\epsilon)\alpha}.
\end{align}

So we have
\begin{align*}
\|\CC \xxhat - \dd\|_{\infty} = &~ \frac{1}{T} \max_e \left|\sum_{i=0}^{T-1} (\CC \Delta^{(i,k_i)} - \dd)_e\right| \\
\leq &~ \frac{\ln(\Phi(\ww^{(T,K)}))}{\alpha T} \\
\leq &~ \frac{\ln n + (1+\epsilon)\epsilon \alpha T + (1+\epsilon)}{\epsilon(1-\epsilon)\alpha T} \\
\leq &~ 1 + 10\epsilon. \qedhere
\end{align*}
\end{proof}

%% file: NewRobustPrimalStep.tex
\subsection{Guarantees of Algorithm~\ref{alg:non_monotone_accel_robust}: Robust Primal Step}\label{sec:NonMonRob}

In this section, we finally present the analysis of Algorithm~\ref{alg:non_monotone_accel_robust} which additionally approximates the primal steps via a {\it sketch}. We will again use the two potentials as defined in Eq.~\eqref{eq:defPhi} and \eqref{eq:defPsi}.

In the next section, we will first present the properties of the sketching matrices that we need to use. This constitutes the major part of the section. The remaining would directly build on the analysis of the previous section.

\subsubsection{Sketching Bounds}
\begin{lemma}[Coordinate-wise embedding, Lemma~E.5 of \cite{lsz19}]\label{lem:CE_one_sketch}
Let $\SS \in \R^{b \times n}$ be sampled from distribution $\Pi$ such that each entry is $+\frac{1}{\sqrt{b}}$ with probability $1/2$ and $-\frac{1}{\sqrt{b}}$ with probability $1/2$.
For any fixed vectors $\gg, \hh\in \R^n$, the following properties hold:
\begin{align*}
1. & \E_{\SS \sim \Pi}[ \gg^{\top} \SS^\top \SS \hh ] = \gg^{\top} \hh, \\
2. & \E_{\SS \sim \Pi}[ (\gg^{\top} \SS^\top \SS \hh )^2 ] \leq (\gg^{\top} \hh)^2 + \frac{C_1}{b} \|\gg\|_2^2 \| \hh \|_2^2, \\
3. & \Pr_{\SS \sim \Pi}\Big[ | \gg^{\top} \SS^\top \SS \hh - \gg^{\top} \hh | \leq \frac{ C_2 }{ \sqrt{b} } \|\gg\|_2 \| \hh \|_2 \Big] \geq 1 - 1 / n^4.
\end{align*}
where $C_1 = O(1)$, $C_2 = O(\log n)$.
\end{lemma}

\begin{lemma}[Bounds for the vector $\wh{\uu}$]\label{lem:hat_u_bounds} 
For all $i \in [0:T]$, the vector $\wh{\uu}^{(i,k)} = (\RRbar^{(i,k)})^{-1/2} \cdot (\SS^{(i)})^{\top} \SS^{(i)} \cdot (\RRbar^{(i,k)})^{1/2} \uu^{(i,k)}$ satisfies the following properties:
\begin{enumerate}
\item {\bf Expectation.} $\E_{\SS^{(i)}}[\wh{\uu}^{(i,k)} \mid \SS^{(0)}, \cdots, \SS^{(i-1)}] = \uu^{(i,k)}$.
\item {\bf Variance.} For any vector $\gg \in \R^n$ that is independent of $\SS^{(i)}$, 
\[
\Var_{\SS^{(i)}}[\sum_e \gg_e \wh{\uu}^{(i,k)}_e \mid \SS^{(0)}, \cdots, \SS^{(i-1)}] \leq \frac{C_1}{b} \cdot \gg^{\top} (\RRbar^{(i,k)})^{-1} \gg \cdot (\uu^{(i,k)})^{\top} \RRbar^{(i,k)} \uu^{(i,k)}.
\]
In particular, if $\ww^{(i,k)} \geq 0$, then this implies that:
\begin{itemize}
\item $\Var_{\SS^{(i)}}\left[\wh{\uu}^{(i,k)}_e \mid \SS^{(0)}, \cdots, \SS^{(i-1)} \right] \leq \frac{C_1 n}{\epsilon b}$.
\item $\E_{\SS^{(i)}}[\wh{\Psi}(\rrbar^{(i,k)}, \SS^{(i)}) \mid \SS^{(0)}, \cdots, \SS^{(i-1)}] \leq (1 + \frac{C_1 \cdot n}{b}) \cdot \Psi(\rrbar^{(i,k)})$.
\item $\Var_{\SS^{(i)}}[\sum_e \ww^{(i,k)}_e \wh{\uu}^{(i,k)}_e \mid \SS^{(0)}, \cdots, \SS^{(i-1)}] \leq \frac{C_1}{b} \cdot \Phi(\ww^{(i,k)}) \cdot \Psi(\rrbar^{(i,k)})$.
\end{itemize}
\item {\bf Coordinate-wise absolute value.} For any vector $\gg \in \R^n$ that is independent of $\SS^{(i)}$, and when conditioned on any $\SS^{(0)}, \cdots, \SS^{(i-1)}$, we have
\[
\Pr_{\SS^{(i)}}\left[|\sum_e \gg_e (\wh{\uu}^{(i,k)}_e - \uu^{(i,k)}_e)| \leq \frac{C_2}{\sqrt{b}} \cdot (\gg^{\top} (\RRbar^{(i,k)})^{-1} \gg)^{1/2} \cdot \Psi(\rrbar^{(i,k)})^{1/2} \right] \geq 1 - 1/n^4.
\]
In particular, this implies that when conditioned on any $\SS^{(0)}, \cdots, \SS^{(i-1)}$:
\begin{itemize}
\item For any $e \in [n]$, if $\ww_e^{(i,k)} \geq 0$, then
\[
\Pr_{\SS^{(i)}}\left[|\wh{\uu}^{(i,k)}_e - \uu^{(i,k)}_e| \leq \frac{C_2}{\sqrt{b}} \cdot \frac{\sqrt{n}}{\sqrt{\epsilon}} \cdot \Phi(\ww^{(i,k)})^{-1/2} \cdot \Psi(\rrbar^{(i,k)})^{1/2} \right] \geq 1 - 1/n^4.
\]
\item If $\ww^{(i,k)} \geq 0$, then 
\begin{align*}
\Pr_{\SS^{(i)}}\left[| \sum_e \ww_e^{(i,k)} \cdot (\wh{\uu}^{(i,k)}_e - \uu^{(i,k)}_e)| \leq \frac{C_2}{\sqrt{b}} \cdot \Phi(\ww^{(i,k)})^{1/2} \cdot \Psi(\rrbar^{(i,k)})^{1/2} \right] \geq &~ 1 - 1/n^4, \\
\Pr_{\SS^{(i)}}\left[| \sum_e \rrbar_e^{(i,k)} \cdot (\wh{\uu}^{(i,k)}_e - \uu^{(i,k)}_e)| \leq \frac{C_2 (1 + \epsilon)}{\sqrt{b}} \cdot \Phi(\ww^{(i,k)})^{1/2} \cdot \Psi(\rrbar^{(i,k)})^{1/2} \right] \geq &~ 1 - 1/n^4.
\end{align*}
\end{itemize}
\item {\bf Symmetry.} When conditioned on any $\SS^{(0)}, \cdots, \SS^{(i-1)}$, for any $i \in [T]$ and any $e \in [n]$, the distribution of $\wh{\uu}^{(i,k)}_e - \uu^{(i,k)}_e$ is symmetric around zero, i.e., for any $z \in \R$, 
\[
\Pr_{\SS^{(i)}}[\wh{\uu}^{(i,k)}_e - \uu^{(i,k)}_e = z] = \Pr_{\SS^{(i)}}[\wh{\uu}^{(i,k)}_e - \uu^{(i,k)}_e = -z].
\]
\end{enumerate}
\end{lemma}
\begin{proof}
In this proof we will use the following fact from Algorithm~\ref{alg:non_monotone_accel_robust}: for any $i \in [0:T]$, $\uu^{(i,k)}$, $\ww^{(i,k)}$, and $\rrbar^{(i,k)}$ are random variables that depend on $\SS^{(0)}, \cdots, \SS^{(i-1)}$, and $\wh{\uu}^{(i,k)}$ is a random variable that depends on $\SS^{(0)}, \cdots, \SS^{(i-1)}$, and $\SS^{(i)}$.

{\bf Part 1 (Expectation).} For any fixed $\SS^{(0)}, \cdots, \SS^{(i-1)}$, we have
\begin{align*}
\E_{\SS^{(i)}}[\wh{\uu}^{(i,k)}] = &~ \E_{\SS^{(i)}}[(\RRbar^{(i,k)})^{-1/2} \cdot (\SS^{(i)})^{\top} \SS^{(i)} \cdot (\RRbar^{(i,k)})^{1/2} \uu^{(i,k)}] \\
= &~ (\RRbar^{(i,k)})^{-1/2} \cdot (\RRbar^{(i,k)})^{1/2} \uu^{(i,k)} \\
= &~ \uu^{(i,k)},
\end{align*}
where the second step follows from Part~1 of Lemma~\ref{lem:CE_one_sketch}.

{\bf Part 2 (Variance).} Consider any fixed $\SS^{(0)}, \cdots, \SS^{(i-1)}$. For any vector $\gg \in \R^n$ that is independent of $\SS^{(i)}$, we have
\begin{align*}
\Var_{\SS^{(i)}}[\sum_e \gg_e \wh{\uu}^{(i,k)}_e] = &~ \E_{\SS^{(i)}}[(\sum_e \gg_e \wh{\uu}^{(i,k)}_e)^2] - (\sum_e \gg_e \uu^{(i,k)}_e)^2.
\end{align*}
We also have 
\begin{align*}
\E_{\SS^{(i)}}[(\sum_e \gg_e \wh{\uu}^{(i,k)}_e)^2] = &~ \E_{\SS^{(i)}}[(\gg^{\top} \wh{\uu}^{(i,k)})^2] \\
= &~ \E_{\SS^{(i)}}\left[\left(\gg^{\top} (\RRbar^{(i,k)})^{-1/2} \cdot (\SS^{(i)})^{\top} \SS^{(i)} \cdot (\RRbar^{(i,k)})^{1/2} \uu^{(i,k)}\right)^2\right] \\
\leq &~ (\gg^{\top} \uu^{(i,k)})^2 + \frac{C_1}{b} \cdot \gg^{\top} (\RRbar^{(i,k)})^{-1} \gg \cdot (\uu^{(i,k)})^{\top} \RRbar^{(i,k)} \uu^{(i,k)},
\end{align*}
where the third step follows from Part~2 of Lemma~\ref{lem:CE_one_sketch}.

So we have
\begin{align*}
\Var_{\SS^{(i)}}[\sum_e \gg_e \wh{\uu}^{(i,k)}_e] \leq \frac{C_1}{b} \cdot \gg^{\top} (\RRbar^{(i,k)})^{-1} \gg \cdot (\uu^{(i,k)})^{\top} \RRbar^{(i,k)} \uu^{(i,k)}.
\end{align*}

In particular, if $\ww^{(i,k)} \geq 0$, then this implies the following bounds:
\begin{itemize}
\item For any $e \in [n]$, let $\gg = 1_e$, then we have 
\begin{align*}
\Var_{\SS^{(i)}}\left[\wh{\uu}^{(i,k)}_e\right] \leq &~ \frac{C_1}{b} \cdot \frac{\Psi(\rrbar^{(i,k)})}{\rrbar_e^{(i)}} \leq \frac{C_1 n}{\epsilon b},
\end{align*}
where the second step follows from $(1+\delta)\rrbar\geq \rr^{(i,k)} = \ww^{(i,k)} + \frac{\epsilon}{n} \Phi(\ww^{(i,k)}) \geq \frac{\epsilon}{n} \Phi(\ww^{(i,k)})$ when $\ww^{(i,k)} \geq 0$.
\item Let $g = \ww^{(i,k)}$. Then we have
\begin{align*}
\Var_{\SS^{(i)}}[\sum_e \ww^{(i,k)}_e \wh{\uu}^{(i,k)}_e] \leq &~ \frac{C_1}{b} \cdot (\ww^{(i,k)})^{\top} (\RRbar^{(i,k)})^{-1} \ww^{(i,k)} \cdot (\uu^{(i,k)})^{\top} \RRbar^{(i,k)} \uu^{(i,k)} \\
\leq &~ \frac{C_1}{b} \cdot (\sum_e \frac{(\ww^{(i,k)}_e)^2}{\rrbar^{(i,k)}_e}) \cdot \Psi(\rrbar^{(i,k)}) \\
\leq &~ \frac{C_1}{b} \cdot \Phi(\ww^{(i,k)}) \cdot \Psi(\rrbar^{(i,k)}),
\end{align*}
where the second step follows from $\Psi(\rrbar^{(i,k)}) = (\uu^{(i,k)})^{\top} \RRbar^{(i,k)} \uu^{(i,k)}$, the third step follows from $(1+\delta)\rrbar \geq \rr^{(i,k)} = \ww^{(i,k)} + \frac{\epsilon}{n} \Phi(\ww^{(i,k)}) \geq \ww^{(i,k)}$ when $\ww^{(i,k)} \geq 0$, and $\sum_e \ww^{(i,k)}_e = \Phi(\ww^{(i,k)})$.
\end{itemize}

{\bf Part 3 (Coordinate-wise absolute value).} 
Consider any fixed $\SS^{(0)}, \cdots, \SS^{(i-1)}$. For any $\gg \in \R^n$ that is independent of $\SS^{(i)}$, using Part~3 of Lemma~\ref{lem:CE_one_sketch} we have that with probability at least $ 1 - 1/n^4$ over the randomness of $\SS^{(i)}$,
\begin{align*}
|\sum_e \gg_e (\wh{\uu}^{(i,k)}_e - \uu^{(i,k)}_e)| = &~ \left|\gg^{\top} (\RRbar^{(i,k)})^{-1/2}  \cdot \left( (\SS^{(i)})^{\top} \SS^{(i)} \cdot (\RRbar^{(i,k)})^{1/2} \uu^{(i,k)} - (\RRbar^{(i,k)})^{1/2} \uu^{(i,k)} \right)\right| \\
\leq &~ \frac{C_2}{\sqrt{b}} \|(\RRbar^{(i,k)})^{-1/2} \gg\|_2 \cdot \|(\RRbar^{(i,k)})^{1/2} \uu^{(i,k)}\|_2 \\
= &~ \frac{C_2}{\sqrt{b}} \cdot (\gg^{\top} (\RRbar^{(i,k)})^{-1} \gg)^{1/2} \cdot \Psi(\rrbar^{(i,k)})^{1/2},
\end{align*}
where the third step follows from $\Psi(\rrbar^{(i,k)}) = \sum_{e'} \rrbar^{(i,k)}_{e'} (\uu^{(i,k)}_{e'})^2$.

In particular, if $\ww^{(i,k)} \geq 0$, then this implies the following bounds:
\begin{itemize}
\item For any $e \in [n]$, let $\gg = 1_e$, we have that with probability at least $1 - 1/n^4$ over the randomness of $\SS^{(i)}$,
\begin{align*}
\left|\wh{\uu}^{(i,k)}_e - \uu^{(i,k)}_e \right| \leq &~ \frac{C_2}{\sqrt{b}} (\rrbar_e^{(i,k)})^{-1/2} \cdot \Psi(\rrbar^{(i,k)})^{1/2} \\
\leq &~ \frac{C_2}{\sqrt{b}} \cdot \frac{\sqrt{n}}{\sqrt{\epsilon}} \cdot \Phi(\rrbar^{(i,k)})^{-1/2} \cdot \Psi(\rrbar^{(i,k)})^{1/2} 
\end{align*}
where the second step follows from $(1+\delta)\rrbar\geq\rr_e^{(i,k)} = \ww_e^{(i,k)} + \frac{\epsilon}{n} \Phi(\ww^{(i,k)}) \geq \frac{\epsilon}{n} \Phi(\ww^{(i,k)})$ since $\ww_e^{(i,k)} \geq 0$. 
\item Let $\gg = \ww^{(i,k)}$, we have that with probability at least $1 - 1/n^4$ over the randomness of $\SS^{(i)}$,
\begin{align*}
| \sum_e \ww_e^{(i,k)} \cdot (\wh{\uu}^{(i,k)}_e - \uu^{(i,k)}_e)| \leq &~ \frac{C_2}{\sqrt{b}} (\sum_e \frac{(\ww_e^{(i,k)})^2}{\rrbar_e^{(i,k)}})^{1/2} \cdot \Psi(\rrbar^{(i,k)})^{1/2} \\
\leq &~ \frac{C_2}{\sqrt{b}} \Phi(\ww^{(i,k)})^{1/2} \Psi(\rrbar^{(i,k)})^{1/2}
\end{align*}
where the second step follows from $(1+\delta)\rrbar\geq\rr^{(i,k)} = \ww^{(i,k)} + \frac{\epsilon}{n} \Phi(\ww^{(i,k)}) \geq \ww^{(i,k)} \geq 0$ since $\ww^{(i,k)} \geq 0$ and hence $\Psi(\rrbar^{(i,k)}) \geq 0$.

Similarly let $\gg = \rrbar^{(i,k)}$ we have
\begin{align*}
| \sum_e \rrbar_e^{(i,k)} \cdot (\wh{\uu}^{(i,k)}_e - \uu^{(i,k)}_e)| \leq &~ \frac{C_2}{\sqrt{b}} (\sum_e \frac{(\rrbar_e^{(i,k)})^2}{\rrbar_e^{(i,k)}})^{1/2} \cdot \Psi(\rrbar^{(i,k)})^{1/2} \\
\leq &~ \frac{C_2 (1 + \epsilon)}{\sqrt{b}} \Phi(\ww^{(i,k)})^{1/2} \Psi(\rrbar^{(i,k)})^{1/2}
\end{align*}
where the second step follows from $\sum_e \rr_e^{(i,k)} = \Phi(\ww^{(i,k)}) + n \cdot \frac{\epsilon}{n} \Phi(\ww^{(i,k)})$.
\end{itemize}

{\bf Part 4 (Symmetry).} It suffices to prove that for any vectors $\gg, \hh\in \R^n$, the distribution of $\gg^{\top} (\SS^{\top} \SS - I) \hh$ is symmetric around zero when $\SS \in \R^{b \times n}$ is sampled from distribution $\Pi$ such that each entry is $+\frac{1}{\sqrt{b}}$ with probability $1/2$ and $-\frac{1}{\sqrt{b}}$ with probability $-1/2$.

Let $\ss_1, \cdots, \ss_n \in \R^b$ denote the columns of $\SS$. We have
\begin{align*}
\gg^{\top} (\SS^{\top} \SS - I) \hh = &~ \sum_{i=1}^n \sum_{j\neq i} \gg_i \hh_j \langle \ss_i, \ss_j \rangle \\
= &~ \sum_{i=1}^{n-1} \Big\langle \ss_i, \sum_{j=i+1}^n (\gg_i \hh_j + \gg_j \hh_i )\ss_j \Big\rangle.
\end{align*}
Define $x_i = \Big\langle \ss_i, \sum_{j=i+1}^n (\gg_i \hh_j + \gg_j \hh_i )\ss_j \Big\rangle$ for all $i \in [n-1]$. Observe that when fixing any vectors $\ss_{i+1} = \ss^*_{i+1}, \cdots, \ss_{n} = \ss^*_{n}$, since $\ss_i$ is independent of them and each entry of $\ss_i$ is $+\frac{1}{\sqrt{b}}$ or $-\frac{1}{\sqrt{b}}$ with $1/2$ probability, we have
\begin{align*}
&~ \Pr\Big[\Big\langle \ss_i, \sum_{j=i+1}^n (\gg_i \hh_j + \gg_j \hh_i )\ss_j \Big\rangle = z ~\Big|~ \ss_{i+1} = \ss^*_{i+1}, \cdots, \ss_{n} = \ss^*_{n} \Big] \\
= &~ \Pr\Big[\Big\langle \ss_i, \sum_{j=i+1}^n (\gg_i \hh_j + \gg_j \hh_i )\ss_j \Big\rangle = -z ~\Big|~ \ss_{i+1} = -\ss^*_{i+1}, \cdots, \ss_{n} = -\ss^*_{n} \Big] 
\end{align*}
So we can prove the claim by induction: First note that the term $x_{n-1} = \Big\langle \ss_{n-1}, (\gg_i \hh_j + \gg_j \hh_i )\ss_n \Big\rangle$ is symmetric around zero by the definition of $\SS$. Suppose that for some $i^*$ we have $\Pr[\sum_{i=i^*}^{n-1} x_i = z] = \Pr[\sum_{i=i^*}^{n-1} x_i = -z]$, then we have for any $z \in \R$,
\begin{align*}
\Pr\Big[\sum_{i=i^*-1}^{n-1} x_i = z\Big]
= &~ \sum_{t \in R} \Pr\Big[\sum_{i=i^*}^{n-1} x_i = t\Big] \cdot \Pr\Big[x_{i^*-1} = z-t ~\Big|~ \sum_{i=i^*}^{n-1} x_i = t\Big] \\
= &~ \sum_{t \in R} \Pr\Big[\sum_{i=i^*}^{n-1} x_i = -t\Big] \cdot \Pr\Big[x_{i^*-1} = -z+t ~\Big|~ \sum_{i=i^*}^{n-1} x_i = -t\Big] \\
= &~ \Pr\Big[\sum_{i=i^*-1}^{n-1} x_i = -z\Big]
\end{align*}
where the second step follows from $\Pr\Big[\sum_{i=i^*}^{n-1} x_i = t\Big] = \Pr\Big[\sum_{i=i^*}^{n-1} x_i = -t\Big]$ by induction hypothesis, and $\Pr\Big[x_{i^*-1} = z-t ~\Big|~ \sum_{i=i^*}^{n-1} x_i = t\Big] = \Pr\Big[x_{i^*-1} = -(z-t) ~\Big|~ \sum_{i=i^*}^{n-1} x_i = -t\Big]$ by our previous observation.
\end{proof}

The variance and coordinate-wise bounds of the previous lemma only hold when the weights are non negative. Next we prove that we can maintain this non negativity as long as $|\wh{\uu}^{(i,k)} - \uu^{(i,k)}|$ is bounded.

\begin{lemma}[Positivity of the weights]\label{lem:positivity_weights}
Let $k_i$ denote the number of width reduction steps taken by the algorithm when the $i^{th}$ primal step is being executed. For all $i \in [0:T]$, in the $i$-th primal iteration of Algorithm~\ref{alg:non_monotone_accel_robust}, if
\[
\ww^{(i,k_i)} \geq 0, ~~~\text{and}~~~ |\wh{\uu}^{(i,k_i)} - \uu^{(i,k_i)}| \leq \frac{100 C_2 \sqrt{n}}{\sqrt{b \epsilon}},
\]
then we have
\[
\|\uu^{(i,k_i)}\|_{\infty} \leq 2C_3^{1/3}\frac{n^{1/2-\eta}}{\epsilon^{1/3}}, ~~~ |\overrightarrow{\alpha}^{(i,k_i)} \wh{\uu}^{(i,k_i)}| \leq 1 / 10, ~~~\text{and}~~~\ww^{(i+1,k_{i})} \geq 0.
\]
\end{lemma}
\begin{proof}
For simplicity, we will let $k$ denote $k_i$. Observe that, when we do a primal step i.e., the condition on Line~\ref{algline:CheckPrimalRobust} is true, since for all $e$, $\rrbar^{(i,k)}\geq e^{-\delta}\rr^{(i,k)}_e \geq \frac{e^{-\delta}\epsilon}{2n}\Phi(\ww^{(i,k)})\geq \frac{e^{-\delta}\epsilon}{2n}\Psi(\rrbar^{(i,k)})$,
    \[
    \frac{e^{-\delta}\epsilon}{2n}\Psi(\rrbar^{(i,k)})\|\uu^{(i,k)}\|_3^3 \leq C_3 \rho \Psi(\rrbar^{(i,k)}).
    \]
    Using the value of $\rho$, this implies that,
    \[
    \|\uu^{(i,k)}\|_{\infty}
\leq \|\uu^{(i,k)}\|_3 \leq \frac{6C_3^{1/3}n^{1/2-\eta}}{\epsilon^{1/3}}.
    \]
Then using our other assumption that $|\wh{\uu}^{(i,k)} - \uu^{(i,k)}| \leq \frac{100 C_2 \sqrt{n}}{\sqrt{b \epsilon}}$, we have
\begin{equation*}
|\wh{\uu}^{(i,k)}_e| \leq |\uu^{(i,k)}_e| + |\wh{\uu}^{(i,k)}_e - \uu^{(i,k)}_e| \leq \frac{6C_3^{1/3} n^{1/2-\eta}}{\epsilon^{1/3}} + \frac{100 C_2 \sqrt{n}}{\sqrt{b \epsilon}}.
\end{equation*}
We also have
\begin{align*}
|\overrightarrow{\alpha}^{(i,k)} \wh{\uu}^{(i,k)}| \leq \alpha  \cdot \Big( \frac{6C_3^{1/3}n^{1/2-\eta}}{\epsilon^{1/3}} + \frac{100 C_2 \sqrt{n}}{\sqrt{b \epsilon}} \Big) \leq 1/10.
\end{align*}

If $\wh{\uu}^{(i,k)}_e \geq 0$, then we directly have $\ww^{(i+1,k)}_e \geq \ww^{(i,k)}_e \geq 0$.
If $\wh{\uu}^{(i,k)}_e < 0$, then
\begin{align*}
\ww^{(i+1,k)}_e = &~ \ww^{(i,k)}_e \cdot \big(1 + \epsilon \overrightarrow{\alpha}^{(i,k)}_e \wh{\uu}^{(i,k)}_e\big) \\
\geq &~ \ww^{(i,k)}_e \cdot \big(1 - \epsilon/10\big) > 0. \qedhere
\end{align*}
\end{proof}

Next we use the above basic properties of the approximate vectors $\wh{\uu}^{(i,k)}$'s to prove a concentration property of the sum of the $\wh{\uu}^{(i,k)}$'s over all iterations. Our proof crucially uses the following concentration inequality of martingales:

\begin{lemma}[Freedman's inequality, \cite{freedman1975tail}]\label{thm:freedman}
Consider a martingale $Y_0, Y_1, \cdots, Y_n$ with difference sequence $X_1, X_2, \cdots, X_n$, i.e., $Y_0 = 0$, and for all $i \in [n]$, $Y_i = Y_{i-1} + X_i$ and $\E_{i-1}[Y_i] = Y_{i-1}$. Suppose $|X_i| \leq R$ almost surely for all $i \in [n]$. Define the predictable quadratic variation process of the martingale as $W_i = \sum_{j=1}^{i} \E_{j-1}[X_j^2]$, for all $i \in [n]$. Then for all $u \geq 0$, $\sigma^2 > 0$,
\[
\Pr\left[ \exists i \in [n]: |Y_i| \geq u \text{ and } W_i \leq \sigma^2 \right] \leq 2 \exp\left(-\frac{u^2/2}{\sigma^2 + R u / 3}\right).
\]
\end{lemma}

For any $i \in [0:T]$, let $k_i$ denote the value of the width reduction step counter $k$ when $i$ is being incremented, and in the next lemma for simplicity of notations we will use the superscript $^{(i)}$ for the variables with superscript $^{(i,k_i)}$.
\begin{lemma}[Bounds of sum of $\wh{\uu}$ over all rounds]\label{lem:sum_hat_u}
Let $k_i$ denote the number of width reduction steps taken by the algorithm when the $i^{th}$ primal step is being executed. Let $a_0, \cdots, a_T \in \R$ be an arbitrary sequence such that each $a_i$ only depends on $\SS^{(0)}, \cdots, \SS^{(i-1)}$ and each $|a_i| \leq C_a$. Then, for all $i \in [0:T-1]$, the vector $\wh{\uu}^{(i,k_i)} = (\RRbar^{(i,k_i)})^{-1/2} \cdot (\SS^{(i)})^{\top} \SS^{(i)} \cdot (\RRbar^{(i,k_i)})^{1/2} \uu^{(i,k_i)}$ satisfies the following properties: 
\begin{align*}
& \E_{\SS^{(0)},\cdots,\SS^{(T-1)}}\left[\sum_{i=0}^{T-1} a_i \cdot (\wh{\uu}^{(i,k_i)}_e - \uu^{(i,k_i)}_e)\right] = 0, \\
& \Pr_{\SS^{(0)},\cdots,\SS^{(T-1)}}\left[ \left|\sum_{i=0}^{T-1} a_i \cdot (\wh{\uu}^{(i,k_i)}_e - \uu^{(i,k_i)}_e)\right| \leq \frac{10 C_a (C_1 + C_2) \log n \cdot \sqrt{nT}}{\sqrt{b \epsilon}} \right] \geq 1 - 1/n^3.
\end{align*}
\end{lemma}
\begin{proof}
Consider a fixed $e \in [n]$. Define the following truncated sequence: for $\tau^{(i)} = \wh{\uu}_e^{(i,k_i)} - \uu_e^{(i,k_i)}$,  
\[
\ov{\tau}^{(i)} = 
\begin{cases}
\tau^{(i)} & \text{if } |\tau^{(i')}| \leq \frac{C_2 \sqrt{n}}{\sqrt{b \epsilon}} \text{ for all } i' \leq i\\
0 & \text{otherwise.}
\end{cases}
\]
Also define $y^{(0)} = 0$ and $y^{(i+1)} = y^{(i)} + a_i \cdot \ov{\tau}^{(i)}$. We use the notation $\E_{i}[\cdot] = \E[\cdot \mid \SS^{(0)}, \cdots, \SS^{(i)}]$ to denote the expectation conditioned on $\SS^{(0)}, \cdots, \SS^{(i)}$.

From Part 1 of Lemma~\ref{lem:hat_u_bounds} we have $\E_{i-1}[\tau^{(i)}] = 0$, and since $a_i$ only depends on $\SS^{(0)}, \cdots, \SS^{(i-1)}$, we also have $\E_{i-1}[a_i \cdot \tau^{(i)}] = 0$, and so we have $\E_{\SS^{(0)},\cdots,\SS^{(T-1)}}[\sum_{i=0}^{T-1} a_i \cdot (\wh{\uu}^{(i,k_i)}_e - \uu^{(i,k_i)}_e)] = 0$. From Part 4 of Lemma~\ref{lem:hat_u_bounds} we have that conditioned on $\SS^{(0)}, \cdots, \SS^{(i)}$, $\Pr[\tau^{(i)} \geq \frac{C_2 \sqrt{n}}{\sqrt{b \epsilon}}] = \Pr[\tau^{(i)} \leq -\frac{C_2 \sqrt{n}}{\sqrt{b \epsilon}}]$, and therefore $\E_{i-1}[\ov{\tau}^{(i)}] = 0$, and $\E_{i-1}[a_i \cdot \ov{\tau}^{(i)}] = 0$. So we have that the sequence $y^{(0)}, \cdots, y^{(T)}$ is a martingale.

Next we bound the quadratic variation. For any $i \in [T]$, if there exist $i' \leq i$ such that $|\tau^{(i')}| > \frac{C_2 \sqrt{n}}{\sqrt{b \epsilon}}$, then we have $(\ov{\tau}^{(i)})^2 = 0$. Otherwise by Lemma~\ref{lem:positivity_weights} we have $\ww^{(i,k_i)} \geq 0$, and since $\E_{i-1}[\wh{\uu}_e^{(i,k_i)}] = \uu_e^{(i,k_i)}$, we have
\begin{align*}
\E_{i-1}[(\tau^{(i)})^2] = &~ \Var_{\SS^{(i)}}[\wh{\uu}_e^{(i,k_i)} \mid \SS^{(0)}, \cdots, \SS^{(i-1)}] 
\leq \frac{C_1 n}{\epsilon b}
\end{align*}
where the last step follows from Part 2 of Lemma~\ref{lem:hat_u_bounds}. Combining these two cases we have $\E_{i-1}[(\ov{\tau}^{(i)})^2] \leq \frac{C_1 n}{\epsilon b}$, and so we have $\E_{i-1}[(a_i \ov{\tau}^{(i)})^2] \leq C_a^2 \cdot \frac{C_1 n}{\epsilon b}$. 
Define $W_i = \sum_{j=1}^i \E_{j-1}[(a_i \ov{\tau}^{(j)})^2]$, and we have
\begin{align*}
W_i \leq \frac{C_a^2 C_1 n i}{\epsilon b}.
\end{align*}

Using Freedman's inequality for our martingale $y^{(0)}, \cdots, y^{(T)}$ with parameters $R = \frac{C_a C_2 \sqrt{n}}{\sqrt{b \epsilon}}$, $\sigma^2 = \frac{C_a^2 C_1 n T}{\epsilon b}$, and $u = \frac{10 C_a (C_1 + C_2) \log n \cdot \sqrt{nT}}{\sqrt{b \epsilon}}$, we have 
\begin{align*}
\Pr\left[ |y^{(T)}| \geq u \right] =
&~ \Pr\left[ |y^{(T)}| \geq u \text{ and } W_T \leq \sigma^2 \right] \\
\leq &~ 2 \exp\left(-\frac{u^2/2}{\sigma^2 + R u / 3}\right) \leq 1/n^4.
\end{align*}
Finally, note that $y^{(T)} = \sum_{i=0}^{T-1} a_i \cdot (\wh{\uu}_e^{(i,k_i)} - \uu_e^{(i,k_i)})$ if we have $|\wh{\uu}_e^{(i,k_i)} - \uu_e^{(i,k_i)}| \leq \frac{C_2 \sqrt{n}}{\sqrt{b \epsilon}}$ for all $i$. By Part~3 of Lemma~\ref{lem:hat_u_bounds} and union bound over all iterations, we have that this happens with probability at least $1-1/n^3$. Combining this with the above equation, we have
\[
\Pr_{\SS^{(0)},\cdots,\SS^{(T-1)}}\left[ \left|\sum_{i=0}^{T-1} a_i \cdot (\wh{\uu}^{(i,k_i)}_e - \uu^{(i,k_i)}_e)\right| \leq \frac{10 C_a (C_1 + C_2) \log n \cdot \sqrt{nT}}{\sqrt{b \epsilon}} \right] \geq 1 - 1/n^3.
\]
\end{proof}

\subsubsection{Analysis of Algorithm~\ref{alg:non_monotone_accel_robust}}
Next we analyze the guarantee and iteration complexity of Algorithm~\ref{alg:non_monotone_accel_robust}, we again first prove the change of the potentials $\Phi$ (Def. Eq.~\eqref{eq:defPhi}) and $\Psi$ (Def. Eq.~\eqref{eq:defPsi} as in the previous sections. 

From the argument in the proof of Lemma~\ref{lem:sum_hat_u}, we have that with probability at least $1-1/n^3$, we have $|\wh{\uu}_e^{(i,k_i)} - \uu_e^{(i,k_i)}| \leq \frac{C_2 \sqrt{n}}{\sqrt{b \epsilon}}$ for all $i$, and then by Lemma~\ref{lem:positivity_weights} we have $\ww^{(i,k_i)} \geq 0$ for all $i$. In this section, we assume that we are conditioned on this event, and the failure of this event corresponds to the failure of our algorithm, which happens with probability at most $1/n^3$, as stated in Theorem~\ref{thm:RobustAccMWU}.

\subsubsection*{Change in $\Phi$}
\begin{restatable}{lemma}{ChangePhiRobust}
  \label{lem:ChangePhiRobust}
  With probability at least $1-1/n^3$, after $i$ primal steps, and $k$ width-reduction steps, if $\alpha \rho^{1/3}\leq \frac{\epsilon^{1/3}}{10 n^{1/3}}$,
  the potential $\Phi$ is bounded as follows:
  \begin{multline*}
\Phi\left(\ww^{(i,k)}\right) \leq \left(\Phi(\ww^{(0,0)})\right)\left( 1 + e^{\epsilon+\delta}\epsilon \alpha  \cdot \left(1 + \frac{C_2}{\sqrt{b}}\right) + 2e^{\epsilon+2\delta} \epsilon^2 \alpha^2 \left(1 + \frac{C_2^2 \cdot n}{b \epsilon}\right) \right)^i\\ \left( 1 + \epsilon e^{\epsilon+2\delta} \cdot (\tau^{-1} + \rho^{-2}) \right)^k.
 \end{multline*}
 Furthermore, after every primal step, the potential can decrease by at most,
 \[
 \Phi\left(\ww^{(i+1,k)}\right) \geq  \Phi\left(\ww^{(i,k)}\right)\left( 1 - e^{\epsilon+\delta}\epsilon \alpha  \cdot \left(1 + \frac{C_2}{\sqrt{b}}\right) - 2e^{\epsilon+2\delta} \epsilon^2 \alpha^2 \left(1 + \frac{C_2^2 \cdot n}{b \epsilon}\right) \right).  
 \]
\end{restatable}
\begin{proof}
  We prove this claim by induction. Initially, $i = k = 0,$ and $\Phi\left(\ww^{(0,0)}\right) =2 n,$ and thus, the claim holds trivially. Assume that the claim holds for some $i,k \ge 0.$
We will use $\Phi$ as an abbreviated notation for $\Phi\left(\ww^{(i,k)}\right)$ below and $\ww$ to denote $\ww^{(i,k)}$
\paragraph*{Primal Step.} 
Since we update the weights to be $\ww^{(i+1,k)} = \ww^{(i,k)} \cdot \big(1 + \epsilon\overrightarrow{\alpha}^{(i,k)} (\CC \Delta^{(i,k)} - \dd)\big)$, we have
\begin{align}\label{eq:potential_increase}
\Phi(\ww^{(i+1,k)}) = &~ \Phi(\ww^{(i,k)}) + \epsilon \cdot \sum_e \ww^{(i,k)}_e \cdot \overrightarrow{\alpha}^{(i,k)}_e \cdot \wh{\uu}^{(i,k)}_e \notag \\
\leq &~ \Phi(\ww^{(i,k)}) + \epsilon\alpha \cdot \sum_e \ww^{(i,k)}_e \cdot \wh{\uu}^{(i,k)}_e + \epsilon^2\alpha^2 \cdot \sum_{e} \ww^{(i,k)}_e \cdot (\wh{\uu}^{(i,k)}_e)^2,
\end{align}
where the second step follows from our definition $\overrightarrow{\alpha}^{(i,k)}_e = \begin{cases}
\alpha \cdot (1 + \epsilon\alpha \wh{\uu}^{(i,k)}_e) & \text{ if } \wh{\uu}^{(i,k)}_e \geq 0\\
\alpha / (1 - \epsilon\alpha \wh{\uu}^{(i,k)}_e) & \text{ else }
\end{cases}$.

Next we bound the two terms in Eq.~\eqref{eq:potential_increase} separately. For the first term $\sum_e \ww_e^{(i,k)} \cdot \wh{\uu}_e^{(i,k)},$ using Cauchy-Schwarz inequality, 
\begin{align*}
\sum_e \ww_e^{(i,k)} \cdot \uu_e^{(i,k)} \leq &~ \sqrt{\Big(\sum_e \ww_e^{(i,k)} \Big) \Big(\sum_e \ww_e^{(i,k)} \cdot (\uu_e^{(i,k)})^2\Big)} \\
\leq &~ \sqrt{e^{\delta}\Phi(\ww^{(i,k)}) \cdot \Psi(\rrbar^{(i,k)})} \\
\leq &~  e^{\epsilon+\delta} \cdot \Phi(\ww^{(i,k)}),
\end{align*}
where the third step follows from Lemma~\ref{lem:PsiPhi}. Next, from Part 3 of Lemma~\ref{lem:hat_u_bounds}, it holds that with probability $1-1/n^4$, $| \sum_e \ww_e^{(i,k)} \cdot (\wh{\uu}^{(i,k)}_e - \uu^{(i,k)}_e)| \leq \frac{C_2}{\sqrt{b}} \cdot \Phi(\ww^{(i,k)})^{1/2} \cdot \Psi(\rrbar^{(i,k)})^{1/2}$. So we have,

\begin{align}\label{eq:potential_increase_1}
\sum_e \ww_e^{(i,k)} \cdot \wh{\uu}_e^{(i,k)} \leq &~ \sum_e \ww_e^{(i,k)} \cdot \uu_e^{(i,k)} + |\sum_e \ww_e^{(i,k)} \cdot (\wh{\uu}_e^{(i,k)} - \uu_e^{(i,k)})| \notag \\
\leq &~ e^{\epsilon +\delta} \cdot \Phi(\ww^{(i,k)}) + \frac{C_2}{\sqrt{b}} \cdot \Phi(\ww^{(i,k)})^{1/2} \cdot \Psi(\rrbar^{(i,k)})^{1/2} \notag \\
\leq &~ e^{\epsilon +\delta} \cdot (1 + \frac{C_2}{\sqrt{b}}) \cdot \Phi(\ww^{(i,k)}).
\end{align}
In the last step we used that $\Psi(\bar{\rr})\leq e^{\delta}\Psi(\rr) \leq e^{\epsilon+2\delta}\Phi(\ww).$

For the second term $\sum_{e} \ww^{(i,k)}_e \cdot (\wh{\uu}^{(i,k)}_e)^2$ in Eq.~\eqref{eq:potential_increase}, we first bound $\sum_{e} \ww^{(i,k)}_e \cdot (\uu^{(i,k)}_e)^2$. We have

\begin{align*}
\sum_{e} \ww^{(i,k)}_e \cdot (\uu^{(i,k)}_e)^2 \leq &~ e^{\delta}\sum_{e} \rrbar^{(i,k)}_e \cdot (\uu^{(i,k)}_e)^2 
= e^{\delta}\Psi(\rrbar^{(i,k)}) \leq e^{\epsilon + 2\delta} \Phi(\ww^{(i,k)}).
\end{align*}

Next again using Part 3 of Lemma~\ref{lem:hat_u_bounds} that with probability $1-1/n^4$, $|\wh{\uu}^{(i,k)}_e - \uu^{(i,k)}_e| \leq \frac{C_2}{\sqrt{b}} \cdot \frac{\sqrt{n}}{\sqrt{\epsilon}} \cdot \Phi(\ww^{(i,k)})^{-1/2} \cdot \Psi(\rrbar^{(i,k)})^{1/2}$, we get

\begin{align}\label{eq:potential_increase_2}
\sum_{e} \ww^{(i,k)}_e \cdot (\wh{\uu}^{(i,k)}_e)^2 \leq &~ 2 \sum_{e} \ww^{(i,k)}_e \cdot (\uu^{(i,k)}_e)^2 + 2 \sum_{e} \ww^{(i,k)}_e \cdot (\uu^{(i,k)}_e - \wh{\uu}^{(i,k)}_e)^2 \notag \\
\leq &~ 2 e^{\epsilon+2\delta} \Phi(\ww^{(i,k)}) + 2 \sum_{e} \ww^{(i,k)}_e \cdot \frac{C_2^2 \cdot n}{b \epsilon} \cdot \Phi(\ww^{(i,k)})^{-1} \cdot \Psi(\rrbar^{(i,k)}) \notag \\
\leq &~ 2 e^{\epsilon+2\delta} \left(1 + \frac{C_2^2 \cdot n}{b \epsilon}\right) \cdot \Phi(\ww^{(i,k)}).
\end{align}

Plugging Eq.~\eqref{eq:potential_increase_1} and \eqref{eq:potential_increase_2} into Eq.~\eqref{eq:potential_increase}, we have
\begin{align*}
\Phi(\ww^{(i+1,k)}) \leq &~ \Phi(\ww^{(i,k)}) + \epsilon \alpha \cdot e^{\epsilon+\delta} \cdot \left(1 + \frac{C_2}{\sqrt{b}}\right) \cdot \Phi(\ww^{(i,k)}) + 2\epsilon^2 \alpha^2 \cdot e^{\epsilon+2\delta} \left(1 + \frac{C_2^2 \cdot n}{b \epsilon}\right) \cdot \Phi(\ww^{(i,k)}) \\
= &~ \Phi(\ww^{(i,k)}) \cdot \left( 1 + e^{\epsilon+\delta}\epsilon \alpha  \cdot \left(1 + \frac{C_2}{\sqrt{b}}\right) + 2e^{\epsilon+2\delta} \epsilon^2 \alpha^2 \left(1 + \frac{C_2^2 \cdot n}{b \epsilon}\right) \right).
\end{align*}
The decrease of the potential follows from a similar proof.

\paragraph*{Width Reduction Step.}
The proof is the same as that of Lemma~\ref{lem:ChangePhiStab} since the algorithms are the same.
\end{proof}

\subsubsection*{Change in $\Psi$}
We recall the definitions of $\textsc{LastWidth}(i,e)$, $\textsc{Last}(i,e)$, $S_i \subseteq [2n]$, $\ell_{i,e}$, $\textsc{Size}(k)$, and $L \leq \frac{1}{100 (\log^4 n) \epsilon \alpha \rho}$ from Section~\ref{sec:NonMonStab}. We will use these in the proof of the following lemma. 

\begin{lemma}[Change in $\Psi$ for Algorithm~\ref{alg:non_monotone_accel_robust}]\label{lem:ChangePsiRobust}
For any integer $c \geq 0$, after $L$ primal steps from $(c-1)L$ to $cL$, if $\rho^2 \tau^{-1} \geq 0.1$, the potential $\Psi$ is bounded as follows:
\begin{align*}
\Psi(\ov{\rr}^{(cL,k_{cL})}) 
\geq &~ \Psi(\ov{\rr}^{((c-1)L,k_{(c-1)L})}) \cdot \Big(1 -  \wt{O}(\epsilon \alpha \rho L) \Big) \cdot \prod_{k=k_{(c-1)L}}^{k_{cL}} \left( 1 + O\Big(\frac{\epsilon^{4/3} \rho^{2/3} \cdot \textsc{Size}(k)^{1/3}}{n^{1/3} \cdot \log^{2/3}(\frac{n}{\epsilon \rho})}\Big) \right).
\end{align*}
\end{lemma}
\begin{proof}
{\bf Primal steps.}
For a primal step, $\ww_e^{(i+1,k_i)} = \ww_e^{(i,k_i)}(1+\epsilon\overrightarrow{\alpha}^{(i,k_i)}_e \wh{\uu}^{(i,k)}_e))$, we have
\begin{align*}
\rr^{(i+1,k_i)}_e - \rr^{(i,k_i)}_e = &~ \ww^{(i+1,k_i)}_e - \ww^{(i,k_i)}_e + \frac{\epsilon}{n} \cdot (\Phi(\ww^{(i+1,k_i)}) - \Phi(\ww^{(i,k_i)})) \notag \\
= &~ \epsilon \overrightarrow{\alpha}_e^{(i,k_i)}  \wh{\uu}^{(i,k_i)}_e \ww^{(i,k_i)}_e + \frac{\epsilon}{n} \cdot (\Phi(\ww^{(i+1,k_i)}) - \Phi(\ww^{(i,k_i)})) \notag
\end{align*}
Since from Lemma~\ref{lem:ChangePhiRobust}, $\Phi(\ww^{(i,k_i)}) \cdot (1-\wt{\Omega}(\epsilon\alpha)) \leq \Phi(\ww^{(i+1,k_i)}) \leq \Phi(\ww^{(i,k_i)}) \cdot (1+\wt{O}(\epsilon\alpha))$, and $\rr_e^{(i,k_i)} \geq \frac{\epsilon}{2n} \Phi(\ww^{(i,k_i)})$, we have
\begin{align}\label{eq:r_i_diff_bounds}
\epsilon \overrightarrow{\alpha}_e^{(i,k_i)}  \wh{\uu}^{(i,k_i)}_e \ww^{(i,k_i)}_e - \wt{\Omega}(\epsilon\alpha) \cdot \rr_e^{(i,k_i)} \leq \rr^{(i+1,k_i)}_e - \rr^{(i,k_i)}_e \leq \epsilon \overrightarrow{\alpha}_e^{(i,k_i)}  \wh{\uu}^{(i,k_i)}_e \ww^{(i,k_i)}_e + \wt{O}(\epsilon\alpha) \cdot \rr_e^{(i,k_i)}.
\end{align}
Since $|\overrightarrow{\alpha}^{(i,k_i)} \wh{\uu}^{(i,k_i)}| \leq 1 / 10$ by Lemma~\ref{lem:positivity_weights}, and $\ww_e^{(i,k_i)} \leq \rr_e^{(i,k_i)}$ we also have 
\begin{align}\label{eq:r_i_diff_bounds_2}
(1 - \epsilon) \rr_e^{(i,k_i)} \leq \rr_e^{(i+1,k_i)} \leq (1 + \epsilon) \rr_e^{(i,k_i)}.
\end{align}

Next we consider the two cases that could happen to the coordinate $e$ in the $i$-th primal step. From now on, when it's clear from the context, we will use $\rr_e^{(i)}$ to refer to $\rr_e^{(i,k_{i-1})}$, and similarly $\rrbar_e^{(i)}$ to refer to $\rrbar_e^{(i,k_{i-1})}$, so that this is consistent with the notations used in \textsc{SelectVector}.
\begin{enumerate}
\item If \textsc{SelectVector} doesn't update $\ov{\rr}_e$ on the $i$-th iteration, then we have $\ov{\rr}^{(i+1,k_i)}_e = \ov{\rr}^{(i,k_i)}_e$. 
\item  If \textsc{SelectVector} updates $\ov{\rr}_e$ on the $i$-th iteration, i.e., $e \in S_i$, then we make the same definitions as the proof of Lemma~\ref{lem:ChangePsiStab}: we define $j_{i,e} := \max\{\textsc{LastWidth}(i,e), \textsc{Last}(i,e)\}$, i.e., $j_{i,e} \leq i$ is the last primal iterate during which the algorithm updates $\ww_e$. We also define $\ell_{i,e}$ to be the smallest integer $\ell$ such that $i+1 \equiv 0 \pmod{2^{\ell}}$ and $|\ln(\frac{\rr_e^{(i+1)}}{\rr_e^{(i+1-2^\ell)}})| \geq \frac{\delta}{2 \log n}$. 
By definition we have $\ov{\rr}^{(i+1,k_i)}_e = \rr^{(i+1,k_i)}_e$, and $\ov{\rr}^{(i,k_i)}_e = \rr^{(j_{i,e},k_{j_{i,e}})}_e$. 

Then by the same argument as Lemma~\ref{lem:ChangePsiStab}, in this case we have the following properties:
\begin{itemize}
\item For all $j \in [j_{i,e}+1, i]$, the value of $\rrbar_e$ remains the same for all width reduction steps between the $(j-1)^{th}$ and $j^{th}$ primal steps, i.e., $\rrbar^{(j,k_{j-1})}_e = \rrbar^{(j,k_{j})}_e$. 
\item For all $j \in [j_{i,e}+1, i]$, $\rr_e^{(j,k_j)} \approx_{\delta} \ov{\rr}_e^{(j,k_j)} = \rr_e^{(j_{i,e},k_{j_{i,e}})}$. 
\item If $i+1-2^{\ell_{i,e}} > j_{i,e}$, then $|\rr_e^{(i+1-2^{\ell_{i,e}})} - \rr_e^{(j_{i,e},k_{j_{i,e}})}| \leq 5 \log n \cdot |\rr_e^{(i+1)} - \rr_e^{(i+1-2^{\ell_{i,e}})}|$.
\end{itemize}
\end{enumerate}
And same as Lemma~\ref{lem:ChangePsiStab}, we can without loss of generality assume that $j_{i,e} \geq i+1-2^{\ell_{i,e}}$, since otherwise we can upper bound $|\ov{\rr}^{(i+1,k_i)}_e - \ov{\rr}^{(i ,k_i)}_e|$ by $\wt{O}(|\rr^{(i+1,k_i)}_e - \rr^{(i+1-2^{\ell_{i,e}})}_e|)$ instead of $|\rr_e^{(i+1,k_i)} - \rr_e^{(j_{i,e},k_{j_{i,e}})}|$.

Abbreviating $\wt{\CC}\Delta^{(i,k_i)}-\wt{\dd}$ as $\uu^{(i,k_i)}$, now we have

\begin{align*}
\Psi(\ov{\rr}^{(i+1,k_i)}) 
\substack{(i)\\ \geq} &~ \Psi(\ov{\rr}^{(i,k_i)}) - \sum_e \Big(\frac{\rrbar_e^{(i+1,k_i)}-\ov{\rr}^{(i,k_i)}_e}{\ov{\rr}_e^{(i+1,k_i)}} \Big) \ov{\rr}^{(i,k_i)}_e (\uu^{(i,k_i)}_e)^2\\
\substack{(ii)\\ =}&~ \Psi(\ov{\rr}^{(i,k_i)}) - \sum_{e \in S_i} \Big(\frac{\rr_e^{(i+1,k_i)} - \rr_e^{(j_{i,e},k_{j_{i,e}})}}{\rr_e^{(i+1,k_i)}} \Big) \ov{\rr}^{(i,k_i)}_e (\uu^{(i,k_i)}_e)^2 \\
\substack{(iii)\\ =} &~ \Psi(\ov{\rr}^{(i,k_i)}) - \sum_{e \in S_i} \sum_{j=j_{i,e}}^{i} \frac{ (\rr_e^{(j+1,k_j)} - \rr_e^{(j,k_j)}) }{\rr_e^{(i+1,k_i)}} \ov{\rr}^{(i,k_i)}_e (\uu^{(i,k_i)}_e)^2 \\
\substack{(iv)\\ \geq} &~ \Psi(\ov{\rr}^{(i,k_i)}) - \sum_{e \in S_i} \sum_{j=j_{i,e}}^{i} \frac{ \epsilon \overrightarrow{\alpha}_e^{(j,k_j)}  \wh{\uu}^{(j,k_j)}_e \ww^{(j,k_j)}_e + \wt{O}(\epsilon\alpha) \cdot \rr_e^{(j,k_j)}}{\rr_e^{(i+1,k_i)}} \ov{\rr}^{(i,k_i)}_e (\uu^{(i,k_i)}_e)^2 \\
\substack{(v)\\ =} &~ \Psi(\ov{\rr}^{(i,k_i)}) - \wt{O}(\epsilon \alpha) \cdot \sum_{e \in S_i} \sum_{j=j_{i,e}}^{i} \ov{\rr}^{(i,k_i)}_e (\uu^{(i,k_i)}_e)^2  \\
&~ - \sum_{e \in S_i} \sum_{j=j_{i,e}}^{i} \frac{ \epsilon \overrightarrow{\alpha}_e^{(j,k_j)}  \wh{\uu}^{(j,k_j)}_e \ww^{(j,k_j)}_e}{\rr_e^{(i+1,k_i)}} \ov{\rr}^{(i,k_i)}_e (\uu^{(i,k_i)}_e)^2\\
\substack{(vi)\\ \geq} &~ \Psi(\ov{\rr}^{(i,k_i)}) - \wt{O}(\epsilon \alpha) \cdot \sum_{e \in S_i} \sum_{j=j_{i,e}}^{i} \ov{\rr}^{(i,k_i)}_e (\uu^{(i,k_i)}_e)^2 \\
&~ - \sum_{e \in S_i} \sum_{j=j_{i,e}}^{i} \frac{ \epsilon \overrightarrow{\alpha}_e^{(j,k_j)} \ww^{(j,k_j)}_e}{\rr_e^{(i+1,k_i)}} \ov{\rr}^{(i,k_i)}_e |\uu^{(j,k_j)}_e|(\uu^{(i,k_i)}_e)^2 \\
&~ - \epsilon \sum_{e \in S_i}  |\sum_{j=j_{i,e}}^{i} \overrightarrow{\alpha}_e^{(j,k_j)}  (\wh{\uu}^{(j,k_j)}_e - \uu^{(j,k_j)}_e) \ww^{(j,k_j)}_e| \frac{\ov{\rr}^{(i,k_i)}_e}{\rr_e^{(i+1,k_i)}} (\uu^{(i,k_i)}_e)^2.
\end{align*}

Here $(i)$ follows from Lemma~\ref{lem:PsiChange}, $(ii)$ follows from $\ov{\rr}^{(i+1,k_i)}_e = \rr^{(i+1,k_i)}_e$, and $\ov{\rr}^{(i,k_i)}_e = \rr^{(j_{i,e},k_{j_{i,e}})}_e$ for $e \in S_i$, $(iv)$ follows from Eq.~\eqref{eq:r_i_diff_bounds}. To obtain(v) we use, $\rr_e^{(j+1,k_j)} \approx_{\delta} \rr_e^{(j_{i,e},k_{j_{i,e}})}$, $\rr_e^{(i,k_i)} \approx_{\delta} \rr_e^{(j_{i,e},k_{j_{i,e}})}$ as we argued above, and $\rr_e^{(i+1,k_i)} \approx_{\epsilon} \rr_e^{(i,k)}$ from Eq.~\eqref{eq:r_i_diff_bounds_2}. Combining these we have $\rr_e^{(j,k_j)} \approx_{3\epsilon} \rr_e^{(i+1,k)}$. Further, $(vi)$ follows from splitting $\wh{\uu}^{(j,k_j)}_e = \uu^{(j,k_j)}_e + (\wh{\uu}^{(j,k_j)}_e - \uu^{(j,k_j)}_e)$.

Now, we can bound the third term using AM-GM inequality --  $|\wt{\CC}\Delta^{(j,k_j)}-\wt{\dd}|_e (\wt{\CC}\Delta^{(i,k_i)}-\wt{\dd})_e^2 \leq \frac{1}{3} \cdot |\wt{\CC}\Delta^{(j,k_j)}-\wt{\dd}|_e^3 + \frac{2}{3} \cdot |\wt{\CC}\Delta^{(i,k_i)}-\wt{\dd}|_e^3$, use that $\rr_e^{(j,k_j)} \approx_{\delta} \rr_e^{(j_{i,e},k_{j_{i,e}})} \approx_{\delta} \ov{\rr}_e^{(i,k_i)}$ for all $e \in S_i$ and $j \in [j_{i,e}, i]$, and bound the fourth term by Lemma~\ref{lem:sum_hat_u},
\[
|\sum_{j=j_{i,e}}^i \overrightarrow{\alpha}_e^{(j,k_j)}  (\wh{\uu}^{(j,k)}_e - \uu^{(j,k)}_e) \ww_e^{(j,k_j)}| \leq \Otil(\alpha)\sqrt{\frac{n (i + 1 - j_{i,e})}{b\epsilon}} \cdot \rr_e^{(i+1,k_i)},
\]
to get,

\begin{align*}
\Psi(\ov{\rr}^{(i+1,k_i)}) 
\geq &~ \Psi(\ov{\rr}^{(i,k_i)}) - O(\epsilon \alpha) \cdot \sum_{e \in S_i} (i + 1 - j_{i,e}) \cdot \ov{\rr}^{(i,k_i)}_e (\uu^{(i,k_i)}_e)^2 \\
&~ - O(\epsilon \alpha) \cdot \sum_{e \in S_i} (i + 1 - j_{i,e}) \cdot \ov{\rr}^{(i,k_i)}_e |\uu^{(i,k_i)}_e|^3 
 - O(\epsilon \alpha) \cdot \sum_{e \in S_i} \sum_{j=j_{i,e}}^{i-1} \ov{\rr}^{(j,k_j)}_e |\uu^{(i,k_i)}_e|^3 \\
&~ - \Otil(\epsilon\alpha) \cdot \sum_{e \in S_i}\sqrt{\frac{n (i + 1 - j_{i,e})}{b\epsilon}} \cdot \ov{\rr}^{(i,k_i)}_e (\uu^{(i,k_i)}_e)^2,
\end{align*}

Similar to the proof of Lemma~\ref{lem:ChangePsiStab}, abusing the notation, let $\ell_i$ denote the largest integer such that $i+1 \equiv 0 \pmod{2^{\ell_i}}$. Since we assumed that $j_{i,e} \geq i+1-2^{\ell_{i,e}}$ for all $e \in S_i$, we have $i+1 - j_{i,e} \leq 2^{\ell_{i,e}} \leq 2^{\ell_i}$. Also note that in primal steps we have $\sum_e \rrbar_e^{(i,k)} |\wt{\CC}\Delta^{(i,k)}-\wt{\dd}|_e^3\leq 2\rho\Psi(\rrbar^{(i,k)})$, and further note that $\sqrt{\frac{n 2^{\ell_i}}{b}} \leq \rho \cdot 2^{\ell_i}$ since $b = \wt{\Theta}(n^{1/2+\eta})$ and $\rho = \wt{\Theta}(n^{1/2-3\eta})$ and $\eta \leq 1/10$, so the above equation becomes, for some $C_1 = \wt{O}(1)$
\begin{align}
\Psi(\ov{\rr}^{(i+1,k_i)}) \geq \Big( 1 - C_1 \cdot \epsilon \alpha \rho \cdot 2^{\ell_i} \Big) \cdot \Psi(\ov{\rr}^{(i,k_i)}) - C_1 \cdot \epsilon \alpha \cdot \sum_{e \in S_i} \sum_{j=j_{i,e}}^{i-1} \ov{\rr}^{(j,k_j)}_e |\wt{\CC}\Delta^{(j,k_j)}-\wt{\dd}|_e^3.
\end{align}

The remaining proof is the same as in the proof of Lemma~\ref{lem:ChangePsiStab}.

\paragraph*{Width Reduction Step.}
The proof is the same as that of Lemma~\ref{lem:ChangePhiStab} since the algorithms are the same.
\end{proof}

\subsubsection*{Proof of Theorem~\ref{thm:RobustAccMWU}}
\begin{proof}
Let $\xxhat = \frac{\xx^{(T)}}{T}$ be the solution returned by Algorithm \ref{alg:non_monotone_accel_robust}. We will first prove that $\Phi(\ww^{(T,K)}) \leq n^{\wt{O}(1/\epsilon)}$. Then, from Lemma~\ref{lem:number_of_width_reduction_steps}, the number of width reduction steps are bounded by $\wt{O}\left(\frac{n^{1/3} \rho^{1/3} }{\epsilon^{10/3}}\right) \leq \Otil(\tau + \rho^{2})$. We will then show how to bound the objective value at $\wh{\xx}$.

Since Algorithm~\ref{alg:non_monotone_accel_robust} has at most $\alpha^{-1}\log n/\epsilon^2$ primal steps, and suppose the algorithm has at most $\Otil(\tau + \rho^{2})$ width reduction steps\footnote{Similar to Lemma~\ref{lem:number_of_width_reduction_steps}, it is sufficient to argue that this is a sufficient upper bound on the number of width reduction steps}, from Lemma~\ref{lem:ChangePhiRobust}, 
\[
\Phi\left(\ww^{(T,K)}\right) \le  2n\cdot  e^{(1+\wt{O}(\epsilon\alpha)) \cdot \alpha^{-1}\frac{\log n}{\epsilon^2} 
 + \frac{K}{(\tau + \rho^2)}}\leq   n^{\wt{O}\left(\frac{1}{\epsilon} \right)}.
\]
We can now follow the same argument as in Lemma~\ref{lem:error_non_monotone_acc_stab} up to Eq,~\eqref{eq:finalBound} to get,

\begin{align*}
\left|\sum_{i=0}^{T-1} \wh{\uu}^{(i)}_e \right| = \left|\sum_{i=0}^{T-1} (\CC \Delta^{(i,k_i)} - \dd)_e\right| \leq \frac{\ln(\Phi(\ww^{(T,K)}))}{\epsilon(1-\epsilon)\alpha}.
\end{align*}

Additionally, from Lemma~\ref{lem:sum_hat_u} we have that with probability $1 - 1/n^2$, for all $e \in [n]$
\[
\left|\sum_{i=0}^{T-1} (\wh{\uu}^{(i)}_e - \uu^{(i)}_e)\right| \leq \frac{10 (C_1 + C_2) \log n \cdot \sqrt{nT}}{\sqrt{b \epsilon}}.
\]
Now,

\begin{align*}
\left|\sum_{i=0}^{T-1} \uu^{(i)}_e \right|&\leq \left|\sum_{i=0}^{T-1} \wh{\uu}^{(i)}_e \right| 
 + \left|\sum_{i=0}^{T-1} (\wh{\uu}^{(i)}_e - \uu^{(i)}_e)\right|\\
 &\leq \frac{\ln(\Phi(\ww^{(T,K)}))}{\epsilon(1-\epsilon)\alpha} + \frac{10 (C_1 + C_2) \log n \cdot \sqrt{nT}}{\sqrt{b \epsilon}}.
\end{align*}
So we have for $b \geq \frac{n \alpha \log n}{\epsilon},$
\begin{align*}
\|\CC \xxhat - \dd\|_{\infty} = &~ \frac{1}{T} \max_e \left|\sum_{i=0}^{T-1} \uu^{(i)}_e\right| \\
\leq &~ \frac{\ln(\Phi(\ww^{(T,K)}))}{\epsilon(1-\epsilon)\alpha T} + \frac{10 (C_1 + C_2) \log n \cdot \sqrt{n}}{\sqrt{b \epsilon T}}\\
\leq &~ 1+10\epsilon + O(\epsilon)\\
\leq &~ 1 + O(\epsilon).
\end{align*}

Therefore, the total number of iterations for $\eta = 1/10$ is,
\[
T + K \leq \Otil(1) \left( \alpha^{-1}\epsilon^{-2}  + \frac{n^{1/3}\rho^{1/3}}{\epsilon^{10/3}}\right)= \tilde{O}(1) n^{2/5}\epsilon^{-4}.
\]

Similar to the proof of Theorem~\ref{thm:StableAccMWU}, we also have $|H| \leq \Otil(n^{1/2+\eta})$.
\end{proof}

%% file: MWUStability.tex
\section{Stability Guarantees of Algorithms~\ref{alg:MWU} and~\ref{alg:non_monotone_accel_robust}}\label{sec:StabilityMWU}
In this section we prove the stability guarantees of Algorithms~\ref{alg:MWU} and~\ref{alg:non_monotone_accel_robust}. 

For primal steps, we will prove that the $\ell_2$ or $\ell_3$ norm of the relative changes in resistances are bounded, and this means we can maintain a coordinate-wise approximation of the resistances under a low-rank update scheme. 

For width reduction steps, we directly prove that the number of coordinates updated in each iteration follow the same low-rank update scheme.

In Section~\ref{sec:stability_algo_1}, we prove the stability guarantees of the iterates of the multiplicative weights update algorithm with monotone weights given in Algorithm~\ref{alg:MWU}. Both primal and width reduction steps of this algorithm satisfy the stronger $\ell_3$-stability guarantee, i.e., the $\ell_3$ norm of the relative changes in resistances is bounded.

In Section~\ref{sec:stability_algo_warm_up}, we prove that the primal steps of Algorithm~\ref{alg:non_monotone_accel_stab} satisfy the weaker $\ell_2$-stability guarantee, and we also prove that its width reduction steps follow a low-rank update scheme.

In Section~\ref{sec:stability_algo_2}, we prove a robust $\ell_2$-stability guarantee for the primal steps of Algorithm~\ref{alg:non_monotone_accel_robust}, and the width reduction steps are the same as that of Algorithm~\ref{alg:non_monotone_accel_stab}.

\subsection{Stability Guarantees of Algorithm~\ref{alg:MWU}}\label{sec:stability_algo_1}


\MonotonePrimal*

\begin{proof}
Note that from Lemma~\ref{lem:PsiChange}, for $k = k_i$ in each primal iteration the potential is increased by at least
\begin{align}\label{eq:Stable1}
\Psi(\rrbar^{(i+1,k)}) - \Psi(\rrbar^{(i,k)}) \geq &~ \sum_e \left(1-\frac{\rrbar^{(i,k)}_e}{\rrbar^{(i+1,k)}_e}\right)\rrbar^{(i,k)}_e (\CC\Delta^{(i,k)}-\dd)_e^2 \nonumber\\
= &~ \sum_e \left(\frac{\rrbar_e^{\left(i + 1,k\right)}-\rrbar_e^{\left(i,k \right)}}{\rrbar_e^{\left(i,k \right)}} \right) \frac{(\rrbar_e^{(i,k)})^2}{\rrbar^{(i+1,k)}_e} (\CC\Delta^{(i,k)}-\dd)_e^2.
\end{align}
We know from the algorithm that for a primal step, $\ww^{(i+1,k)} = \ww^{(i,k)}(1+\epsilon\alpha|\CC\Delta^{(i,k)}-\dd|)$. We will use this to compute the relative change in resistances. Let $\rr' = \rr^{(i+1,k)}$, $\rr = \rr^{(i,k)}$ and $\Delta = \Delta^{(i,k)}$.

\begin{align*}
    \frac{\rr_e'-\rr_e}{\rr_e}& = \frac{\ww'_e - \ww_e}{\rr_e} + \frac{\epsilon}{n}\frac{\Phi(\ww')-\Phi(\ww)}{\rr_e}\\
    & \leq \epsilon\alpha |\CC\Delta-\dd|_e + \frac{\epsilon}{n}\frac{\Phi(\ww')-\Phi(\ww)}{\rr_e}.
\end{align*}
From the above, we note that,
\begin{align*}
|\CC\Delta-\dd|_e^2 \geq &~ \left(\frac{1}{\epsilon\alpha}\frac{\rr_e'-\rr_e}{\rr_e} - \frac{1}{\alpha n}\frac{\Phi(\ww')-\Phi(\ww)}{\rr_e}\right)^2 \\
\geq &~ \frac{1}{\epsilon^2\alpha^2}\left(\frac{\rr_e'-\rr_e}{\rr_e}\right)^2 - 2\frac{1}{\epsilon\alpha^2 n}\frac{\left(\Phi(\ww')-\Phi(\ww)\right)\left(\rr_e'-\rr_e\right)}{\rr_e^2}.
\end{align*}
We know that for a primal step, $\Phi(\ww')-\Phi(\ww)\leq (1+\epsilon)\epsilon \alpha \Phi(\ww)$. Using this and $e \in S_i$, the above becomes, 
\[
|\CC\Delta-\dd|_e^2 \geq \frac{1}{3\epsilon^2\alpha^2}\left(\frac{\rr_e'-\rr_e}{\rr_e}\right)^2
\]

Now, for $e$, let $i'_e$ denote the last primal iterate where $e$ was updated via a width reduction step. Now, since the $\rr$'s are increasing,
\[
\frac{\rrbar_e^{(i+1,k_i)}-\rrbar^{(i,k_i)}_e}{\rrbar^{(i,k_i)}_e}=\sum_{j=i'}^i\frac{\rrbar_e^{(j+1,k_j)}-\rrbar^{(j,k_j)}_e}{\rrbar^{(j,k_j)}_e}\geq \frac{\rr_e^{(i+1,k_i)}-\rr^{(i,k_i)}_e}{\rr^{(i,k_i)}_e}.
\]
Using these bounds in Eq.~\eqref{eq:Stable1},
\begin{align*}
  \Psi(\rrbar^{(i+1,k)}) & \geq  \Psi(\rrbar^{(i,k)})  + \sum_e \left(\frac{\rrbar_e^{\left(i + 1,k\right)}-\rrbar_e^{\left(i,k \right)}}{\rrbar_e^{\left(i,k \right)}} \right) \frac{(\rrbar_e^{(i,k)})^2}{\rrbar^{(i+1,k)}_e} (\CC\Delta^{(i,k)}-\dd)_e^2\\
  & \geq \Psi(\rrbar^{(i,k)}) + \frac{(1-\delta)}{\epsilon^2\alpha^2}\sum_e \left(\frac{\rr_e^{(i+1,k_i)}-\rr^{(i,k_i)}_e}{\rr^{(i,k_i)}_e}\right)^3\rrbar_e^{(i,k)}\\
  & \geq \Psi(\rrbar^{(i,k)}) + \frac{(1-\delta)^3}{\epsilon\alpha^2 n}\sum_e \left(\frac{\rr_e^{(i+1,k_i)}-\rr^{(i,k_i)}_e}{\rr^{(i,k_i)}_e}\right)^3\Psi(\rrbar_e^{(i,k)})\\
  & \geq \Psi(\rrbar^{(i,k)})\left(1 + \frac{1}{10\epsilon\alpha^2 n}\sum_e \left(\frac{\rr_e^{(i+1,k_i)}-\rr^{(i,k_i)}_e}{\rr^{(i,k_i)}_e}\right)^3\right)
\end{align*}

We can now recurse on the above and get,
\begin{equation}\label{eq:Stable4}
  \Psi(\rr^{(T,K)}) \geq \Psi(\rr^{(0,0)})\Pi_{t\geq 0}\left(1 + \frac{1}{\epsilon\alpha^2n} \sum_{e\in S_t}\left(\frac{\rr^{(t+1,k_t)}_e}{\rr^{(t,k_t)}_e} - 1\right)^3 \right).
\end{equation}
Taking logs,
\[
\log \frac{\Psi(\rr^{(T,K)})}{\Psi(\rr^{(0,0)})} \geq \frac{1}{\epsilon(1+\epsilon)\alpha^2n}\sum_{t\geq 0}\sum_{e\in S_t}\left(\frac{\rr^{(t+1,k_t)}_e}{\rr^{(t,k_t)}_e} - 1\right)^3 .
\]
Since $\Psi(\rr^{(T,K)}) \leq \Phi(\ww^{(T,K)})\leq n^{O(1/\epsilon)}$, and $\Psi(\rr^{(0,0)})\geq L$,
\[
\sum_{t\geq 0}\sum_{e\in S_t}\left(\frac{\rr^{(t+1,k_t)}_e}{\rr^{(t,k_t)}_e} - 1\right)^3  \leq \Otil(\alpha^2 n).
\]

\end{proof}

\MonotoneWidth*

\begin{proof}
Again from Lemma~\ref{lem:PsiChange},
\[
\Psi(\rrbar^{(i,k+1)})\geq \Psi(\rrbar^{(i,k)}) + \sum_e \left(\frac{\rrbar_e^{\left(i,k+1\right)}-\rrbar_e^{\left(i,k \right)}}{\rrbar_e^{\left(i,k\right)}} \right) \frac{(\rrbar_e^{(i,k)})^2}{\rrbar^{(i,k+1)}_e} (\CC\Delta^{(i,k)}-\dd)_e^2.
\]
Now, $\rrbar_e^{(i, k+1)}\leq \rrbar_e^{(i, k)}(1+\epsilon)(1+\delta)$ and for $e\in H_k$, we know that $|\CC\Delta^{(i,k)}-\dd|_e\geq \tau.$ Also, 
\[
\frac{\rrbar_e^{\left(i,k+1\right)}-\rrbar_e^{\left(i,k \right)}}{\rrbar_e^{\left(i,k \right)}} \geq \frac{\rr_e^{\left(i,k+1\right)}-\rr_e^{\left(i,k \right)}}{\rr_e^{\left(i,k \right)}} -\frac{2\delta}{1-\delta}\geq \frac{1}{10}\frac{\rr_e^{\left(i,k+1\right)}-\rr_e^{\left(i,k \right)}}{\rr_e^{\left(i,k \right)}}.
\]

This gives us,
\[
\Psi(\rrbar^{(i,k+1)})\geq \Psi(\rrbar^{(i,k)}) + \sum_e \left(\frac{\rr_e^{\left(i,k+1\right)}-\rr_e^{\left(i,k \right)}}{\rr_e^{\left(i,k\right)}} \right)^3 \frac{\rr_e^{(i,k)}}{1+\epsilon} \frac{\tau^2}{\epsilon^2}.
\]

Further, using that $\rr_e^{(i,k)}\geq \epsilon\Psi(\rr_e^{(i,k)})/n$, and $\Psi(\rr^{(T,K)})\leq (1+\epsilon)\Phi(\ww^{(T,K)}) \leq n^{O(1/\epsilon)}$,
\[
n^{O(1/\epsilon)}\geq L\Pi_{k\geq 0}\left(1+\sum_e \left(\frac{\rr_e^{\left(i_k,k+1\right)}-\rr_e^{\left(i_k,k \right)}}{\rr_e^{\left(i_k,k\right)}} \right)^3 \frac{1}{(1+\epsilon)\epsilon} \frac{\tau^2}{n}\right)
\]

In other words,
\[
\sum_k\sum_e \left(\frac{\rr_e^{\left(i,k+1\right)}-\rr_e^{\left(i,k \right)}}{\rr_e^{\left(i,k\right)}} \right)^3 \leq \widetilde{O} (n^{1/3}).
\]
\end{proof}

\subsection{Algorithm Warm Up: Low-Rank Update Scheme}\label{sec:stability_algo_warm_up}
\begin{lemma}[Low-rank update scheme of width reduction steps]\label{lem:low_rank_update_width_reduction_steps}
Let $\eta = 1/10$. For every $\ell = 0,1, 2, \cdots,$ $ \log(\frac{10 n^{2/5}}{\tau^{1/2} \epsilon^{1/2}})$, in Algorithm~\ref{alg:non_monotone_accel_stab} there are at most $\frac{T}{2^{\ell}}$ number of width reductions steps in which $\ov{\rr}$ receives an update of rank $O\left(n^{1/5} 2^{2 \ell} \cdot (\log n)^{28/3} \log(\frac{n}{\Psi_0})^{4/3} \epsilon^{-1}\right)$. 
\end{lemma}
\begin{proof}
First note that from Lemma~\ref{lem:size_H}, we have that for any width reduction step, the size of $H$ satisfies $|H| \leq \frac{n}{\tau \epsilon} \cdot e^{\epsilon + 2 \delta}$, and hence the update to $\ov{\rr}$ in any width reduction step has size at most $\frac{n}{\tau \epsilon} \cdot e^{\epsilon + 2 \delta} + 1 \leq n^{1/5} 2^{2\ell_{\max}}$ where $\ell_{\max} := \log(\frac{10 n^{2/5}}{\tau^{1/2} \epsilon^{1/2}})$.

Consider any fixed integer $c \in [1:\frac{T}{L}]$. Consider any integer $\ell \in [0:\ell_{\max}]$, and let $K_{c,\ell}$ denote the number of width steps between primal iterations $(c-1)L$ and $(cL)$ such that $\textsc{Size}(k) \in [n^{1/5} 2^{2\ell}, 4 \cdot n^{1/5} 2^{2\ell}]$. Using Lemma~\ref{lem:ChangePsiStab} and using a proof similar to that of Lemma~\ref{lem:number_of_width_reduction_steps}, we have
\begin{align*}
\left( 1 + C_3 \frac{\epsilon^{4/3} \cdot \rho^{2/3} \cdot (n^{1/5} 2^{2 \ell})^{1/3}}{n^{1/3} \log^{2/3}(\frac{n}{\epsilon\rho})} \right)^{K_{c,\ell}} \leq &~ \frac{2 \Psi\left(\ov{\rr}^{(cL,k_{cL})}\right)}{\Psi(\ov{\rr}^{((c-1)L,k_{(c-1)L})})}  
\leq \frac{4 \Psi\left(\ov{\rr}^{(cL,k_{cL})}\right)}{\Psi_0} 
\leq \frac{10 n^{3 \log n / \epsilon}}{\Psi_0} \\
\implies K_{c,\ell} \leq &~ O\left(\frac{n^{4/15} \log n}{\epsilon^{4/3} \cdot \rho^{2/3} \cdot 2^{2\ell/3}} \cdot \frac{\log n}{\epsilon} \cdot \log\big(\frac{n}{\Psi_0}\big)\right),
\end{align*}
where the first step follows from Lemma~\ref{lem:ChangePsiStab} and that the $C_2 \epsilon \alpha \rho \cdot L$ factor of Lemma~\ref{lem:ChangePsiStab} is upper bounded by $1/2$ since $L \leq \frac{1}{100 (\log^4 n) \epsilon \alpha \rho}$, the second step follows from the same proof as Lemma~\ref{lem:lower_bound_Psi} that $\Psi(\ov{\rr}^{((c-1)L,k_{(c-1)L})})\geq \frac{1}{1+2\epsilon}\cdot\Psi(\ov{\rr}^{(0,0)})$ and Lemma~\ref{lem:lower_bound_Psi_0} that $\Psi(\ov{\rr}^{(0,0)}) \geq \Psi_0$, the third step follows from $\Psi(\rrbar^{(cL,k_{cL})}) \leq e^{\epsilon + \delta} \Phi(\ww^{(cL,k_{cL})}) \leq e^{\epsilon + \delta} n^{3 \log n / \epsilon}$ by Lemma~\ref{lem:PsiPhi} and Eq.~\eqref{eq:upper_bound_Phi} of Lemma~\ref{lem:number_of_width_reduction_steps}.

So over $T = \alpha^{-1}\epsilon^{-2}\log n$ primal steps, and since $L = \Theta(\frac{1}{\epsilon \alpha \rho \log^4 n})$, the total number of width reduction steps with update size in $[n^{1/5} 2^{2\ell}, 4 \cdot n^{1/5} 2^{2\ell}]$ is upper bounded by
\begin{align}\label{eq:UB_width_steps_ell}
\frac{T}{L} \cdot K_{c,\ell} \leq &~ T \cdot O((\log^4 n) \epsilon \alpha \rho) \cdot O\left(\frac{n^{4/15} \log n}{\epsilon^{4/3} \cdot \rho^{2/3} \cdot 2^{2\ell/3}} \cdot \frac{\log n}{\epsilon} \cdot \log\big(\frac{n}{\Psi_0}\big)\right) \notag \\
= &~ T \cdot O\left(\frac{n^{4/15} \alpha \rho^{1/3} (\log n)^6}{\epsilon^{4/3} \cdot 2^{2\ell/3}} \cdot \log\big(\frac{n}{\Psi_0}\big)\right) \notag \\
= &~ T \cdot O\left(\frac{n^{-1/15} (\log n)^6}{\epsilon \cdot 2^{2\ell/3}} \cdot \log\big(\frac{n}{\Psi_0}\big)\right),
\end{align}
where the third step follows from $\alpha = O(n^{-1/2+\eta} \cdot \epsilon \cdot \log(n)^{-4/3} \log(\frac{n}{\Psi_0})^{-1/3})$ and $\rho = n^{1/2-3\eta} \cdot \epsilon^{-2} \cdot \log(n)^4 \log(\frac{n}{\Psi_0})$.

Since we only consider $\ell \leq \ell_{\max} = \log(\frac{10 n^{2/5}}{\tau^{1/2} \epsilon^{1/2}})$, we have
\begin{align*}
2^{\ell} \leq O\left(\frac{n^{2/5}}{\tau^{1/2} \epsilon^{1/2}}\right)
= O\left(\frac{n^{2/5}}{n^{1/4-\eta/2} \cdot \epsilon^{-3/2} \cdot \log(n)^4 \log(\frac{n}{\Psi_0})}\right) 
= O\left(\frac{n^{\eta/2+3/20} \epsilon^{3/2}}{\log(n)^4 \log(\frac{n}{\Psi_0})}\right),
\end{align*}
where the second step follows from $\tau = n^{1/2-\eta} \cdot \epsilon^{-4} \cdot \log(n)^8 \log(\frac{n}{\Psi_0})^2$. So we have
\[
O\left(\frac{n^{\eta/6+1/20} \epsilon^{1/2}}{\log(n)^{4/3} \log(\frac{n}{\Psi_0})^{1/3}}\right) \cdot \frac{1}{2^{\ell/3}} \geq 1,
\]
so we can multiply this factor to the upper bound of Eq.~\eqref{eq:UB_width_steps_ell}, and we have that the total number of width reduction steps with update size in $[n^{1/5} 2^{2\ell}, 4 \cdot n^{1/5} 2^{2\ell}]$ is upper bounded by
\begin{align*}
&~ T \cdot O\left(\frac{n^{-1/15} (\log n)^6}{\epsilon \cdot 2^{2\ell/3}} \cdot \log\big(\frac{n}{\Psi_0}\big)\right) \cdot O\left(\frac{n^{\eta/6+1/20} \epsilon^{1/2}}{\log(n)^{4/3} \log(\frac{n}{\Psi_0})^{1/3}}\right) \cdot \frac{1}{2^{\ell/3}} \\
\leq &~ T \cdot O\left(\frac{1}{2^{\ell}} \cdot \frac{(\log n)^{14/3} \log(\frac{n}{\Psi_0})^{2/3}}{\epsilon^{1/2}}\right),
\end{align*}
where the second step follows from $\eta = 1/10$.

Letting $\ell' = \ell - \log\left(O\left(\frac{(\log n)^{14/3} \log(\frac{n}{\Psi_0})^{2/3}}{\epsilon^{1/2}}\right)\right)$, we have that there are at most $T \cdot \frac{1}{2^{\ell'}}$ width reduction steps with update size $O\left(n^{1/5} 2^{2 \ell'} \cdot (\log n)^{28/3} \log(\frac{n}{\Psi_0})^{4/3} \epsilon^{-1}\right)$.
\end{proof}

\begin{lemma}[$\ell_2$ stability of primal steps]
Algorithm~\ref{alg:non_monotone_accel_stab} satisfies that for all primal steps $i$,
\[
\sum_e \left(\log(\rr_e^{(i+1,k_i)})-\log (\rr_e^{(i,k_i)})\right)^2 \leq \Otil(n^{2\eta}\epsilon^{3}).
\]
\end{lemma}
\begin{proof}
We will first compute the ratio $\frac{\rr_e^{(i+1,k)}}{\rr_e^{(i,k)}}$.
    \begin{align*}
   \frac{\rr_e^{(i+1,k)}}{\rr_e^{(i,k)}} &\leq \frac{\ww_e^{(i,k)}(1+\epsilon \overrightarrow{\alpha}^{(i,k)} |\wt{\CC}\Delta^{(i,k)}-\wt{\dd}|_e)+ \frac{\epsilon}{2n}\Phi(\ww^{(i+1,k)})}{\rr_e^{(i,k)}}\\
    & \leq \frac{\ww_e^{(i,k)} + \epsilon \alpha\ww^{(i,k)}_e|\wt{\CC}\Delta^{(i,k)}-\wt{\dd}|_e +\frac{\epsilon}{2n}(1+\epsilon\alpha e^{\epsilon + \delta})\Phi(\ww^{(i,k)})}{\rr^{(i,k)}_e}\\
    & \leq 1 + \epsilon \alpha |\wt{\CC}\Delta^{(i,k)}-\wt{\dd}|_e + \epsilon\alpha e^{\epsilon + \delta},
    \end{align*}
where the second step follows from Lemma~\ref{lem:ChangePhiStab}.
    
    Now taking $\log$ and using that $\log(1+x) \leq x$ for all $x\geq -1$,
    \[
     \log \left(\frac{\rr_e^{(i+1,k)}}{\rr_e^{(i,k)}}\right) \leq \epsilon \alpha |\wt{\CC}\Delta^{(i,k)}-\wt{\dd}|_e + \epsilon\alpha e^{\epsilon + \delta}.
    \]
    Similarly we also have $\log \left(\frac{\rr_e^{(i+1,k)}}{\rr_e^{(i,k)}}\right) \geq -\epsilon \alpha |\wt{\CC}\Delta^{(i,k)}-\wt{\dd}|_e - \epsilon\alpha e^{\epsilon + \delta}$.
    
    Squaring and summing over all $e$,
    \begin{align*}
    \sum_e \log \left(\frac{\rr_e^{(i+1,k)}}{\rr_e^{(i,k)}}\right)^2 & \leq 2\epsilon^2 \alpha^2 \sum_e |\wt{\CC}\Delta^{(i,k)}-\wt{\dd}|^2_e + 2 \epsilon^2 \alpha^2 e^{2(\epsilon + \delta)} n\\
    &\leq 4\epsilon \alpha^2 n (1+\delta) \sum_e \frac{\rrbar^{(i,k)}_e |\wt{\CC}\Delta^{(i,k)}-\wt{\dd}|^2_e}{\Phi(\ww^{(i,k)})} + 2 e^{2(\epsilon + \delta)}\epsilon^2 \alpha^2 n\\
    & = O(\epsilon\alpha^2 n) = O\left(n^{2\eta} \cdot \epsilon^3 \cdot \log(n)^{-8/3} \log(\frac{n}{\Psi_0})^{-2/3}\right). 
    \end{align*}
\end{proof}

From the above lemma and Lemma~\ref{lem:LowRankL2}, we directly have the following corollary.
\begin{corollary}[Low-rank update scheme of primal steps]
For every $\ell = 0,1, \cdots, \log T$, in Algorithm~\ref{alg:non_monotone_accel_stab} there are at most $\frac{T}{2^{\ell}}$ number of primal steps in which $\ov{\rr}$ receives an update of rank $O\left(n^{2\eta} 2^{2\ell}\cdot \epsilon \cdot \log(n)^{-2/3} \log(\frac{n}{\Psi_0})^{-2/3} \right)$. 
\end{corollary}

\subsection{Algorithm~\ref{alg:non_monotone_accel_robust}}\label{sec:stability_algo_2}

\paragraph{Low-rank update scheme}
First note that the width reduction steps of Algorithm~\ref{alg:non_monotone_accel_robust} are the same as Algorithm~\ref{alg:non_monotone_accel_stab}, so they follow the same low-rank update scheme as Lemma~\ref{lem:low_rank_update_width_reduction_steps}.

Next we prove the robust $\ell_2$ stability guarantees of the primal steps of Algorithm~\ref{alg:non_monotone_accel_robust}, and this combined with Lemma~\ref{lem:LowRankL2Robust} will give us the desired low-rank update scheme. 

\begin{lemma}[Robust $\ell_2$ stability of primal steps]
For every primal step $(i,k)$ of Algorithm~\ref{alg:non_monotone_accel_robust}, define a ``fake'' weight:
\begin{align}\label{eq:def_tilde_r}
\wt{\rr}^{(i+1,k)} = \rr^{(i+1,k)} - \ww^{(i,k)} \epsilon \alpha \cdot ( \wh{\uu}^{(i,k)} - \uu^{(i,k)}) \cdot \frac{(1 + \overrightarrow{\alpha}^{(i,k)} \wh{\uu}^{(i,k)})}{(1 + \alpha \wh{\uu}^{(i,k)})}.
\end{align}

Every primal step of Algorithm~\ref{alg:non_monotone_accel_robust} satisfies the following robust $\ell_2$ stability property if $b \geq \frac{n \alpha \log^4 n}{\epsilon^3}$: 
\begin{enumerate}
    \item 
    \[
    \sum_e \ln\left(\frac{\wt{\rr}_e^{(i+1,k)}}{\rr_e^{(i,k)}}\right)^2 \leq O(n^{2\eta}). 
    \]
    \item $\forall t \in [T]$, $\forall e$, with probability $1 - 1/n^4$,
    \[
    \left|\sum_{i = t'-t}^{t'} \ln\left(\frac{\wt{\rr}_e^{(i,k)}}{\rr_e^{(i,k)}}\right)\right| \leq O(\epsilon). 
    \]
\end{enumerate}
\end{lemma}
\begin{proof}

{\bf Part 1 (primal steps: $\ell_2$ norm of tilde version).} 
We first note that by the definition of $\wt{\rr}^{(i+1,k)}$ we have
\begin{align*}
&~ \wt{\rr}^{(i+1,k)} - \rr^{(i,k)} \\
= &~ \rr^{(i+1,k)} - \rr^{(i,k)} - \epsilon \alpha \cdot \ww^{(i,k)} \cdot (\wh{\uu}^{(i,k)} - \uu^{(i,k)}) \cdot \frac{(1 + \overrightarrow{\alpha}^{(i,k)} \wh{\uu}^{(i,k)})}{(1 + \alpha \wh{\uu}^{(i,k)})} \\
= &~ \ww^{(i,k)} \cdot \epsilon \overrightarrow{\alpha}^{(i,k)} \wh{\uu}^{(i,k)} + \frac{\epsilon}{n} \Big(\Phi(\ww^{(i+1,k)}) - \Phi(\ww^{(i,k)})\Big) \\
&~ - \epsilon \alpha \cdot \ww^{(i,k)} \cdot (\wh{\uu}^{(i,k)} - \uu^{(i,k)}) \cdot \frac{(1 + \overrightarrow{\alpha}^{(i,k)} \wh{\uu}^{(i,k)})}{(1 + \alpha \wh{\uu}^{(i,k)})} \\
= &~ \epsilon \overrightarrow{\alpha}^{(i,k)} \cdot \ww^{(i,k)} \cdot \uu^{(i,k)} + \epsilon (\overrightarrow{\alpha}^{(i,k)} - \alpha) \cdot \ww^{(i,k)} \cdot (\wh{\uu}^{(i,k)} - \uu^{(i,k)}) \\
&~ + \left(1 - \frac{(1 + \overrightarrow{\alpha}^{(i,k)} \wh{\uu}^{(i,k)})}{(1 + \alpha \wh{\uu}^{(i,k)})} \right) \cdot \epsilon \alpha \cdot \ww^{(i,k)} \cdot (\wh{\uu}^{(i,k)} - \uu^{(i,k)}) + \frac{\epsilon}{n} \Big(\Phi(\ww^{(i+1,k)}) - \Phi(\ww^{(i,k)})\Big)
\end{align*}
where we used $\rr^{(i,k)} = \ww^{(i,k)} + \frac{\epsilon}{m} \Phi(\ww^{(i,k)})$ and $\ww^{(i+1,k)} = \ww^{(i,k)} (1 + \epsilon \overrightarrow{\alpha}^{(i,k)} \wh{\uu}^{(i,k)})$ in the second step.

Next we provide upper bounds for the terms appearing in the previous equation. By Part~3 of Lemma~\ref{lem:hat_u_bounds}, we have that with probability $1-1/n^3$, for all $e$ we have $|\wh{\uu}^{(i,k)}_e - \uu^{(i,k)}_e| \leq \frac{C_2}{\sqrt{b}} \cdot \frac{\sqrt{n}}{\sqrt{\epsilon}}$. Then using Lemma~\ref{lem:positivity_weights} and the definition of $\overrightarrow{\alpha}^{(i,k)}$ we have $\overrightarrow{\alpha}^{(i,k)} \leq (1 + \epsilon) \alpha$ and $|\overrightarrow{\alpha}^{(i,k)} \wh{\uu}^{(i,k)}| \leq 0.1$. By Lemma~\ref{lem:ChangePhiRobust} we have with probability at least $1-1/n^3$, $\Phi(\ww^{(i+1,k)}) \leq \Phi(\ww^{(i,k)}) \cdot (1 + 10 \epsilon \alpha)$. By the definition that $\overrightarrow{\alpha}^{(i,k)}_e = \begin{cases}
\alpha \cdot (1 + \epsilon\alpha \wh{\uu}^{(i,k)}_e) & \text{ if } \wh{\uu}^{(i,k)}_e \geq 0\\
\alpha / (1 - \epsilon\alpha \wh{\uu}^{(i,k)}_e) & \text{ else }
\end{cases}$, we have 
\begin{align*}
|\overrightarrow{\alpha}^{(i,k)}_e - \alpha| = &~ \begin{cases}
\epsilon \alpha^2 \wh{\uu}^{(i,k)}_e & \text{ if } \wh{\uu}^{(i,k)}_e \geq 0\\
\epsilon \alpha^2 |\wh{\uu}^{(i,k)}_e| \cdot
(1 - \epsilon \alpha \wh{\uu}^{(i,k)}_e)^{-1} & \text{ else }
\end{cases} \\
\leq &~ (1 + \epsilon) \epsilon \alpha^2 |\wh{\uu}^{(i,k)}_e|.
\end{align*}

Plugging these upper bounds into the previous equation, we have that with probability $1 - 1/n^3$,
\begin{align*}
&~ |\wt{\rr}^{(i+1,k)} - \rr^{(i,k)}| \\
\leq &~ (1 + \epsilon) \epsilon \alpha \cdot \ww^{(i,k)} \cdot |\uu^{(i,k)}| + (1 + \epsilon) \epsilon^2 \alpha^2 \cdot \ww^{(i,k)} \cdot |\wh{\uu}^{(i,k)}| \cdot |\wh{\uu}^{(i,k)} - \uu^{(i,k)}| \\ 
&~ + 2 \epsilon \alpha^2 |\wh{\uu}^{(i,k)}|^2 \cdot \epsilon \alpha \cdot \ww^{(i,k)} \cdot |\wh{\uu}^{(i,k)} - \uu^{(i,k)}|  + 10 \epsilon \alpha \cdot \frac{\epsilon}{n} \cdot \Phi(\ww^{(i,k)}) \\
\leq &~ (1 + \epsilon) \epsilon \alpha \cdot \ww^{(i,k)} \cdot |\uu^{(i,k)}| + 2 \epsilon^2 \alpha^2 \cdot \ww^{(i,k)} \cdot |\wh{\uu}^{(i,k)}| \cdot |\wh{\uu}^{(i,k)} - \uu^{(i,k)}| + 10 \epsilon \alpha \cdot \frac{\epsilon}{n} \cdot \Phi(\ww^{(i,k)}) \\ 
\leq &~ (1 + \epsilon) \epsilon \alpha (1 + \frac{2 C_2 \sqrt{n \epsilon} \alpha}{\sqrt{b}}) \cdot \ww^{(i,k)} \cdot |\uu^{(i,k)}| + (1 + \epsilon) \epsilon \alpha^2 \cdot \ww^{(i,k)} \cdot \frac{C_2^2 n}{b} + 10 \epsilon \alpha \cdot \frac{\epsilon}{n} \cdot \Phi(\ww^{(i,k)}).
\end{align*}

Since $\rr^{(i,k)} = \ww^{(i,k)} + \frac{\epsilon}{n} \Phi(\ww^{(i,k)})$,
\begin{align*}
|\frac{\wt{\rr}^{(i+1,k)} - \rr^{(i,k)}}{\rr^{(i,k)}}| \leq &~ (1 + \epsilon) \epsilon \alpha (1 + \frac{2 C_2 \sqrt{n \epsilon} \alpha}{\sqrt{b}}) \cdot |\uu^{(i,k)}| + (1 + \epsilon) \epsilon \alpha^2 \cdot \frac{C_2^2 n}{b} + 10 \epsilon \alpha \\
\leq &~ 2 \epsilon \alpha \cdot |\uu^{(i,k)}| + O(\epsilon \alpha),
\end{align*}
where the second step follows from $b \geq \frac{n\alpha\log n}{\epsilon}$.

Now taking $\ln$ and using that $|\ln(1+x)| \leq 2 |x|$ for all $|x| \geq -0.5$, we have
\begin{align*}
\left|\ln\left(\frac{\wt{\rr}^{(i+1,k)}_e}{\rr^{(i,k)}_e}\right)\right| \leq &~ 2 \left|\frac{\wt{\rr}_e^{(i+1,k)} - \rr_e^{(i,k)}}{\rr_e^{(i,k)}}\right| \leq 4 \epsilon \alpha \cdot |\uu_e^{(i,k)}| + O(\epsilon \alpha).
\end{align*}
Squaring and summing over all $e$, we have
\begin{align*}
\sum_e \ln\left(\frac{\wt{\rr}^{(i+1,k)}_e}{\rr^{(i,k)}_e}\right)^2 \leq &~ 32 \epsilon^2 \alpha^2 \sum_e (\uu^{(i,k)}_e)^2 + O(\epsilon^2 \alpha^2 n) \\
\leq &~ 32 \epsilon \alpha^2 n \sum_e \frac{\rr_e^{(i,k)}(\uu^{(i,k)}_e)^2}{\Psi(\rr^{(i,k)})} + O(\epsilon^2 \alpha^2 n) \\
\leq &~ O(\epsilon \alpha^2 n) = O(n^{2 \eta} ),
\end{align*}
where the second step follows from $\rr^{(i,k)} = \ww^{(i,k)} + \frac{\epsilon}{n} \Phi(\ww^{(i,k)}) \geq \frac{\epsilon}{n} \Phi(\ww^{(i,k)})$, and the last step follows from $\alpha = \wt{\Theta}(n^{-1/2 + \eta} \epsilon)$.

{\bf Part 2 (primal steps: error of tilde version over all iterations).}
Consider any fixed coordinate $e$, and consider the $t$ iterations between $t'-t$ to $t'$. We have
\begin{align}\label{eq:robust_stability}
&~ \sum_{i = t'-t}^{t'} \ln\left(\frac{\wt{\rr}_e^{(i,k)}}{\rr_e^{(i,k)}}\right) \notag \\
= &~ \sum_{i = t'-t}^{t'} \ln\left( 1 - \frac{\ww_e^{(i-1,k)} \cdot \epsilon \alpha \cdot ( \wh{\uu}_e^{(i-1,k)} - \uu_e^{(i-1,k)})}{\rr_e^{(i,k)}} \cdot \frac{(1 + \overrightarrow{\alpha}_e^{(i-1,k)} \wh{\uu}_e^{(i-1,k)})}{(1 + \alpha \wh{\uu}_e^{(i-1,k)})} \right) \notag \\
\leq &~ \epsilon \alpha \cdot \sum_{i = t'-t}^{t'} \frac{\ww_e^{(i-1,k)} \cdot ( \uu_e^{(i-1,k)} - \wh{\uu}_e^{(i-1,k)})}{\rr_e^{(i,k)}} \cdot \frac{(1 + \overrightarrow{\alpha}_e^{(i-1,k)} \wh{\uu}_e^{(i-1,k)})}{(1 + \alpha \wh{\uu}_e^{(i-1,k)})} \notag \\
= &~ \epsilon \alpha \cdot \sum_{i = t'-t}^{t'} \frac{(\rr_e^{(i,k)} - \frac{\epsilon}{n} \Phi(\ww^{(i,k)})) \cdot ( \uu_e^{(i-1,k)} - \wh{\uu}_e^{(i-1,k)})}{\rr_e^{(i,k)} (1 + \alpha \wh{\uu}_e^{(i-1,k)})} \notag \\
\leq &~ \epsilon \alpha \cdot \sum_{i = t'-t}^{t'} \frac{(\rr_e^{(i,k)} - \frac{\epsilon}{n} \Phi(\ww^{(i,k)})) \cdot ( \uu_e^{(i-1,k)} - \wh{\uu}_e^{(i-1,k)})}{\rr_e^{(i,k)} (1 + \alpha \uu_e^{(i-1,k)})}  \notag \\
&~ + 2 \epsilon \alpha^2 \cdot \sum_{i = t'-t}^{t'} \frac{| \rr_e^{(i,k)} - \frac{\epsilon}{n} \Phi(\ww^{(i,k)})|}{\rr_e^{(i,k)}} \cdot |  \uu_e^{(i-1,k)} - \wh{\uu}_e^{(i-1,k)}|^2 \notag \\
\leq &~ \epsilon \alpha \cdot \left|\sum_{i = t'-t}^{t'} \frac{\uu_e^{(i-1,k)} - \wh{\uu}_e^{(i-1,k)}}{1 + \alpha \uu_e^{(i-1,k)}}\right| + \epsilon \alpha \cdot \left|\sum_{i = t'-t}^{t'} \frac{\frac{\epsilon}{n} \Phi(\ww^{(i,k)}) \cdot ( \uu_e^{(i-1,k)} - \wh{\uu}_e^{(i-1,k)})}{\rr_e^{(i,k)} (1 + \alpha \uu_e^{(i-1,k)})}\right| \notag \\
&~ + 2 \epsilon \alpha^2 \sum_{i = t'-t}^{t'} | \uu_e^{(i-1,k)} - \wh{\uu}_e^{(i-1,k)}|^2,
\end{align}
where the first step follows from the definition \eqref{eq:def_tilde_r} that $\wt{\rr}^{(i,k)} = \rr^{(i,k)} - \ww^{(i-1,k)} \epsilon \alpha \cdot ( \wh{\uu}^{(i-1,k)} - \uu^{(i-1,k)}) \cdot \frac{(1 + \overrightarrow{\alpha}^{(i-1,k)} \wh{\uu}^{(i-1,k)})}{(1 + \alpha \wh{\uu}^{(i-1,k)})}$,
the second step follows from $\ln(1+x) \leq x$ for all $x$, the third step follows from $\ww_e^{(i,k)} = \ww_e^{(i-1,k)} \cdot (1 + \overrightarrow{\alpha}^{(i-1,k)} \wh{\uu}^{(i-1,k)})$ and $\rr_e^{(i,k)} = \ww_e^{(i,k)} + \frac{\epsilon}{n} \Phi(\ww^{(i,k)})$ and so $\ww_e^{(i-1,k)} = \frac{\rr_e^{(i,k)} - \frac{\epsilon}{n} \Phi(\ww^{(i,k)})}{1 + \overrightarrow{\alpha}^{(i-1,k)} \wh{\uu}^{(i-1,k)}}$, the fourth step follows from $|\alpha \wh{\uu}_e^{(i-1,k)}| \leq 0.1$ with probability $1-1/n^4$ and $|\alpha \uu_e^{(i-1,k)}| \leq 0.1$, and hence 
\begin{align*}
&~ \left|\frac{1}{1 + \alpha \uu_e^{(i-1,k)}} - \frac{1}{{1 + \alpha \uu_e^{(i-1,k)} + \alpha (\wh{\uu}_e^{(i-1,k)} - \uu_e^{(i-1,k)})}}\right| \\
= &~ \left|\frac{\alpha (\wh{\uu}_e^{(i-1,k)} - \uu_e^{(i-1,k)})}{\left(1 + \alpha \uu_e^{(i-1,k)}\right) \cdot \left(1 + \alpha \uu_e^{(i-1,k)} + \alpha (\wh{\uu}_e^{(i-1,k)} - \uu_e^{(i-1,k)}) \right)}\right| 
\leq  2 \alpha |\wh{\uu}_e^{(i-1,k)} - \uu_e^{(i-1,k)}|.
\end{align*}
Next we bound the three terms in Eq.~\eqref{eq:robust_stability} one by one.

For the first term, since $|\alpha \uu_e^{(i,k)}| \leq 0.1$ and each $\uu_e^{(i,k)}$ only depends on the randomness of $\SS^{(0)}, \cdots, \SS^{(i-1)}$, using Lemma~\ref{lem:sum_hat_u} we have
\begin{align}\label{eq:robust_stability_1}
\epsilon \alpha \cdot \left|\sum_{i = t'-t}^{t'} \frac{\uu_e^{(i-1,k)} - \wh{\uu}_e^{(i-1,k)}}{1 + \alpha \uu_e^{(i-1,k)}}\right| \leq \epsilon \alpha \cdot O( \frac{\sqrt{n t}}{\sqrt{b \epsilon}}) = O( \frac{\sqrt{n t \epsilon} \alpha}{\sqrt{b}}). 
\end{align}

For the third term, since with probability $1-1/n^3$ we have $|\wh{\uu}^{(i,k)}_e - \uu^{(i,k)}_e| \leq \frac{C_2}{\sqrt{b}} \cdot \frac{\sqrt{n}}{\sqrt{\epsilon}}$ for all $i$, we have
\begin{align}\label{eq:robust_stability_3}
2 \epsilon \alpha^2 \sum_{i = t'-t}^{t'} | \uu_e^{(i-1,k)} - \wh{\uu}_e^{(i-1,k)}|^2 \leq &~ \epsilon \alpha^2 \cdot O\left(\frac{n t}{b \epsilon}\right) = O\left(\frac{n t \alpha^2}{b}\right).
\end{align}

It's more complicated to upper bound the second term, and we start by upper bounding the following term:
\begin{align}\label{eq:robust_stability_2_1}
&~ \left|\sum_{i = t'-t}^{t'} \frac{\frac{\epsilon}{n} \Phi(\ww^{(i-1,k)}) \cdot ( \uu_e^{(i-1,k)} - \wh{\uu}_e^{(i-1,k)})}{\rr_e^{(i,k)} (1 + \alpha \uu_e^{(i-1,k)})}\right| \notag \\
\leq &~ \left|\sum_{i = t'-t}^{t'} \frac{\frac{\epsilon}{n} \Phi(\ww^{(i-1,k)}) \cdot ( \uu_e^{(i-1,k)} - \wh{\uu}_e^{(i-1,k)})}{\rr_e^{(i-1,k)} (1 + \alpha \uu_e^{(i-1,k)})}\right|\notag \\
&~ + \left|\sum_{i = t'-t}^{t'}\Big(\frac{1}{\rr_e^{(i,k)}} - \frac{1}{\rr_e^{(i-1,k)}}\Big) \frac{\frac{\epsilon}{n} \Phi(\ww^{(i-1,k)}) \cdot ( \uu_e^{(i-1,k)} - \wh{\uu}_e^{(i-1,k)})}{ (1 + \alpha \uu_e^{(i-1,k)})}\right| \notag \\
\leq &~ O(\frac{\sqrt{nt}}{\sqrt{b \epsilon}}) + \left|\sum_{i = t'-t}^{t'}  \frac{\overrightarrow{\alpha}_e^{(i-1,k)} \wh{\uu}_e^{(i-1,k)} \ww_e^{(i-1,k)}}{\rr_e^{(i-1,k)}} \cdot \frac{\frac{\epsilon}{n} \Phi(\ww^{(i-1,k)}) \cdot (\uu_e^{(i-1,k)} - \wh{\uu}_e^{(i-1,k)})}{\rr_e^{(i,k)} (1 + \alpha \uu_e^{(i-1,k)})} \right| \notag \\ 
&~ + \left|\sum_{i = t'-t}^{t'}  \frac{\frac{\epsilon}{n} (\Phi(\ww^{(i,k)}) - \Phi(\ww^{(i-1,k)})) \cdot \frac{\epsilon}{n} \Phi(\ww^{(i-1,k)})}{\rr_e^{(i,k)} \rr_e^{(i-1,k)}} \cdot \frac{(\uu_e^{(i-1,k)} - \wh{\uu}_e^{(i-1,k)})}{(1 + \alpha \uu_e^{(i-1,k)})} \right| \notag \\
\leq &~ O(\frac{\sqrt{nt}}{\sqrt{b \epsilon}} + \frac{\alpha t \sqrt{\epsilon n}}{\sqrt{b}}) + \left|\sum_{i = t'-t}^{t'}  \frac{\overrightarrow{\alpha}_e^{(i-1,k)} \wh{\uu}_e^{(i-1,k)} \ww_e^{(i-1,k)}}{\rr_e^{(i-1,k)}} \cdot \frac{\frac{\epsilon}{n} \Phi(\ww^{(i-1,k)}) \cdot (\uu_e^{(i-1,k)} - \wh{\uu}_e^{(i-1,k)})}{\rr_e^{(i,k)} (1 + \alpha \uu_e^{(i-1,k)})} \right|,
\end{align}
where in the second step we upper bound the first term using Lemma~\ref{lem:sum_hat_u} since $\Phi(\ww^{(i-1,k)})$, $\rr_e^{(i-1,k)}$, $1+\alpha\uu_e^{(i-1,k)}$ are all independent of $\wh{\uu}_e^{(i-1,k)}$, and we upper bound the second term by $\rr_e^{(i,k)} = \rr_e^{(i-1,k)} + \overrightarrow{\alpha}_e^{(i-1,k)} \wh{\uu}_e^{(i-1,k)} \ww_e^{(i-1,k)} + \frac{\epsilon}{n} (\Phi(\ww^{(i,k)}) - \Phi(\ww^{(i-1,k)}))$ and hence
\begin{align*}
\frac{1}{\rr_e^{(i,k)}} - \frac{1}{\rr_e^{(i-1,k)}}
= &~ \frac{\overrightarrow{\alpha}_e^{(i-1,k)} \wh{\uu}_e^{(i-1,k)} \ww_e^{(i-1,k)} + \frac{\epsilon}{n} (\Phi(\ww^{(i,k)}) - \Phi(\ww^{(i-1,k)}))}{\rr_e^{(i,k)} \rr_e^{(i-1,k)}},
\end{align*}
and the third step follows from $\Phi(\ww^{(i,k)}) \leq \Phi(\ww^{(i-1,k)}) \cdot (1 + O(\epsilon \alpha))$, and $\rr_e^{(i,k)} \geq 0.9 \rr_e^{(i-1,k)} \geq 0.9 \frac{\epsilon}{n} \Phi(\ww^{(i-1,k)})$, and with probability $1-1/n^3$ we have $|\wh{\uu}^{(i,k)}_e - \uu^{(i,k)}_e| \leq \frac{C_2}{\sqrt{b}} \cdot \frac{\sqrt{n}}{\sqrt{\epsilon}}$ for all $i$.

Note that the term $\frac{\overrightarrow{\alpha}_e^{(i-1,k)} \wh{\uu}_e^{(i-1,k)} \ww_e^{(i-1,k)}}{\rr_e^{(i-1,k)}}$ is independent of $\wh{\uu}_e^{(i-1,k)}$ and it's at most $0.1$ with probability $1-1/n^3$. So we can continue the same upper bound of Eq.~\eqref{eq:robust_stability_2_1} for $2 \log n$ times and we have
\begin{align}\label{eq:robust_stability_2_2}
&~ \left|\sum_{i = t'-t}^{t'} \frac{\frac{\epsilon}{n} \Phi(\ww^{(i-1,k)}) \cdot ( \uu_e^{(i-1,k)} - \wh{\uu}_e^{(i-1,k)})}{\rr_e^{(i,k)} (1 + \alpha \uu_e^{(i-1,k)})}\right| \notag \\
\leq &~ O\Big((\frac{\sqrt{nt}}{\sqrt{b \epsilon}} + \frac{\alpha t \sqrt{\epsilon n}}{\sqrt{b}})\cdot \log n \Big) \notag \\ 
&~ + \left|\sum_{i = t'-t}^{t'}  \Big(\frac{\overrightarrow{\alpha}_e^{(i-1,k)} \wh{\uu}_e^{(i-1,k)} \ww_e^{(i-1,k)}}{\rr_e^{(i-1,k)}}\Big)^{2 \log n} \cdot \frac{\frac{\epsilon}{n} \Phi(\ww^{(i-1,k)}) \cdot (\uu_e^{(i-1,k)} - \wh{\uu}_e^{(i-1,k)})}{\rr_e^{(i,k)} (1 + \alpha \uu_e^{(i-1,k)})} \right| \notag \\
\leq &~ O\Big((\frac{\sqrt{nt}}{\sqrt{b \epsilon}} + \frac{\alpha t \sqrt{\epsilon n}}{\sqrt{b}})\cdot \log n \Big).
\end{align}
We are finally ready to upper bound the second term of Eq.~\eqref{eq:robust_stability}. We have
\begin{align}\label{eq:robust_stability_2}
&~ \left|\sum_{i = t'-t}^{t'} \frac{\frac{\epsilon}{n} \Phi(\ww^{(i,k)}) \cdot ( \uu_e^{(i-1,k)} - \wh{\uu}_e^{(i-1,k)})}{\rr_e^{(i,k)} (1 + \alpha \uu_e^{(i-1,k)})}\right| \notag \\
\leq &~ \left|\sum_{i = t'-t}^{t'} \frac{\frac{\epsilon}{n} \Phi(\ww^{(i-1,k)}) }{\rr_e^{(i,k)}} \cdot \frac{(\uu_e^{(i-1,k)} - \wh{\uu}_e^{(i-1,k)})}{(1 + \alpha \uu_e^{(i-1,k)})} \right|  \notag \\
&~ + \left|\sum_{i = t'-t}^{t'} \Big(\frac{\frac{\epsilon}{n} \Phi(\ww^{(i,k)})}{\rr_e^{(i,k)}} - \frac{\frac{\epsilon}{n} \Phi(\ww^{(i-1,k)})}{\rr_e^{(i,k)}}\Big) \cdot \frac{(\uu_e^{(i-1,k)} - \wh{\uu}_e^{(i-1,k)})}{(1 + \alpha \uu_e^{(i-1,k)})} \right| \notag \\
\leq &~ O\Big((\frac{\sqrt{nt}}{\sqrt{b \epsilon}} + \frac{\alpha t \sqrt{\epsilon n}}{\sqrt{b}})\cdot \log n \Big), 
\end{align}
where in the second step we upper bound the first term by Eq.~\eqref{eq:robust_stability_2_2}, and we upper bound the second term by Lemma~\ref{lem:ChangePhiRobust} that $\Phi(\ww^{(i,k)}) \leq \Phi(\ww^{(i-1,k)}) \cdot (1 + O(\epsilon \alpha))$ and $\rr_e^{(i,k)} \geq 0.9 \rr_e^{(i-1,k)} \geq 0.9 \frac{\epsilon}{n} \Phi(\ww^{(i-1,k)})$, and with probability $1-1/n^3$ we have $|\wh{\uu}^{(i,k)}_e - \uu^{(i,k)}_e| \leq \frac{C_2}{\sqrt{b}} \cdot \frac{\sqrt{n}}{\sqrt{\epsilon}}$ for all $i$.

Plugging the three upper bounds Eq.~\eqref{eq:robust_stability_1}, \eqref{eq:robust_stability_2}, and \eqref{eq:robust_stability_3} into Eq.~\eqref{eq:robust_stability}, we have
\begin{align*}
\sum_{i = t'-t}^{t'} \ln(\frac{\wt{\rr}_e^{(i,k)}}{\rr_e^{(i,k)}}) \leq &~ O( \frac{\sqrt{n t \epsilon} \alpha}{\sqrt{b}}) + \epsilon \alpha \cdot O\Big((\frac{\sqrt{nt}}{\sqrt{b \epsilon}} + \frac{\alpha t \sqrt{\epsilon n}}{\sqrt{b}})\cdot \log n \Big) + O(\frac{n t \alpha^2}{b}) 
\leq  O(\epsilon),
\end{align*}
where the last step follows from $b = \frac{n \alpha \log^4 n}{\epsilon^3}$, and $t \leq T = \alpha^{-1}\epsilon^{-2}\ln n$.

We can similarly prove an upper bound of $-\sum_{i = t'-t}^{t'} \ln(\frac{\wt{\rr}_e^{(i,k)}}{\rr_e^{(i,k)}})$. Similar to Eq.~\eqref{eq:robust_stability}, we have
\begin{align*}
-\sum_{i = t'-t}^{t'} \ln(\frac{\wt{\rr}_e^{(i,k)}}{\rr_e^{(i,k)}}) 
= &~ -\sum_{i = t'-t}^{t'} \ln\left( 1 - \frac{\ww_e^{(i-1,k)} \cdot \epsilon \alpha \cdot ( \wh{\uu}_e^{(i-1,k)} - \uu_e^{(i-1,k)})}{\rr_e^{(i,k)}} \cdot \frac{(1 + \overrightarrow{\alpha}_e^{(i-1,k)} \wh{\uu}_e^{(i-1,k)})}{(1 + \alpha \wh{\uu}_e^{(i-1,k)})} \right)  \\
\leq &~ \epsilon \alpha \cdot \left|\sum_{i = t'-t}^{t'} \frac{\ww_e^{(i-1,k)} \cdot ( \uu_e^{(i-1,k)} - \wh{\uu}_e^{(i-1,k)})}{\rr_e^{(i,k)}} \cdot \frac{(1 + \overrightarrow{\alpha}_e^{(i-1,k)} \wh{\uu}_e^{(i-1,k)})}{(1 + \alpha \wh{\uu}_e^{(i-1,k)})} \right| \\
&~ + \epsilon^2 \alpha^2 \cdot \sum_{i = t'-t}^{t'} \left( \frac{\ww_e^{(i-1,k)} \cdot ( \uu_e^{(i-1,k)} - \wh{\uu}_e^{(i-1,k)})}{\rr_e^{(i,k)}} \cdot \frac{(1 + \overrightarrow{\alpha}_e^{(i-1,k)} \wh{\uu}_e^{(i-1,k)})}{(1 + \alpha \wh{\uu}_e^{(i-1,k)})} \right)^2 \\
\leq &~ O(\epsilon) + O(\frac{\epsilon nt \alpha^2}{b}) \leq O(\epsilon),
\end{align*}
where the second step follows from $\ln(1 + x) \geq x - x^2$ for all $|x| \leq 0.5$, and in the third step we bound the first term in the same way as Eq.~\eqref{eq:robust_stability} and \eqref{eq:robust_stability_1}, \eqref{eq:robust_stability_2}, \eqref{eq:robust_stability_3}, and we bound the second term by the property that with probability $1-1/n^3$ we have $|\wh{\uu}^{(i,k)}_e - \uu^{(i,k)}_e| \leq \frac{C_2}{\sqrt{b}} \cdot \frac{\sqrt{n}}{\sqrt{\epsilon}}$ for all $i$, and $\overrightarrow{\alpha}_e^{(i-1,k)} \wh{\uu}_e^{(i-1,k)}) \leq 0.1$, $\alpha \wh{\uu}_e^{(i-1,k)}) \leq 0.1$, and $\rr_e^{(i,k)} \geq 0.9 \rr_e^{(i-1,k)} \geq 0.9 \ww_e^{(i-1,k)}$.
\end{proof}

Combining the above lemma and Lemma~\ref{lem:LowRankL2Robust}, we directly have the following corollary.
\begin{corollary}[Low-rank update scheme of primal steps]
For every $\ell = 0,1, \cdots, \log T$, in Algorithm~\ref{alg:non_monotone_accel_robust} there are at most $\frac{T}{2^{\ell}}$ number of primal steps in which $\ov{\rr}$ receives an update of rank $O((\frac{\log n}{\delta})^2 \cdot n^{2 \eta} \cdot 2^{2\ell})$.
\end{corollary}

%% file: Appendix.tex
\section{Missing Proofs}\label{sec:missing_proofs}
\subsection*{Proof of Lemma~\ref{lem:PsiPhi}} 
\begin{proof}
We first prove that $\Psi(\rrbar) \approx_{\delta} \Psi(\rr)$. Let $\Delta^* := \arg \min_{\Delta\in \rea^d }\sum_e \rrbar_e (\CC\Delta-\dd)^2_e$. We have
\begin{align*}
\Psi(\rr) = \min_{\Delta\in \rea^d } \sum_e \rr_e (\CC\Delta-\dd)^2_e \leq \sum_e \rr_e (\CC\Delta^*-\dd)^2_e \leq e^{\delta} \cdot \sum_e \rrbar_e (\CC\Delta^*-\dd)^2_e = e^{\delta} \cdot  \Psi(\rrbar).
\end{align*}
Similarly, we can also show $\Psi(\rrbar) \leq e^{\delta} \cdot \Psi(\rr)$, and so $\Psi(\rrbar) \approx_{\delta} \Psi(\rr)$.

Next we prove $\Psi(\rrbar) \leq e^{\epsilon+\delta} \cdot \Phi(\ww)$. The following inequalities follow from Hölder's inequality and the optimum $\xx^{\star}$ of the problem satisfies $\|\CC\xx^{\star}-\dd\|_{\infty}\leq 1$:
\begin{align*}
\Psi(\rrbar) = &~ \min_{\Delta\in \rea^d} \sum_{e=1}^n\rrbar_e(\CC\Delta-\dd)_e^2 \\
\leq &~ \sum_e\rrbar_e(\CC\xx^{\star}-\dd)_e^2
 \leq e^{\delta} \cdot \|\CC\xx^{\star}-\dd\|^2_{\infty}\left(\|\ww\|_1 + \epsilon \|\ww\|_1\right)\leq e^{\epsilon + \delta} \cdot \Phi(\ww). \qedhere
\end{align*}
\end{proof}

\subsection*{Proof of Lemma~\ref{lem:lower_bound_Psi_0}} 
\begin{proof}
Note that $\rr_e^{(0,0)} = 1+\epsilon$ for all $e$. Therefore,
\[
\Psi(\rr^{(0,0)}) = (1+\epsilon)\min_{\Delta}\|\CC\Delta-\dd\|_2^2.
\]
Using KKT conditions and the solution of the above problem that $\Delta^* = (\CC^{\top} \CC)^{-1} \CC^{\top} \dd$, we get the required lower bound.
\end{proof}

\subsection*{Proof of Lemma~\ref{lem:PsiChange}} 
This is similar to the proof of Lemma 5.11 of \cite{adil2019iterative}.

\begin{proof}
    We have defined,
    \[
    \Psi(\rr) = \min_{\Delta}\sum_e \rr_e (\CC\Delta-\dd)_e^2.
    \]
    This is the same as,
    \[
    \Psi(\rr) = \min_{\rho,\Delta}\sum_e \rr_e \rho_e^2, \text{ s.t. } \rho = \CC\Delta-\dd.
    \] 
    From Lagrange duality and using strong duality,
    \begin{align*}
    \Psi(\rr) & = \min_{\rho,\Delta} \max_{\yy}\sum_e \rr_e \rho_e^2 + 2\yy^{\top}(\rho - \CC\Delta +\dd) \\
    & = \max_{\yy}\min_{\rho,\Delta}\sum_e \rr_e \rho_e^2 + 2\yy^{\top}(\rho - \CC\Delta +\dd).
    \end{align*}
    Using optimality conditions for $\Delta$ and $\rho$, we get that the optimizers $\rho^{\star}$ and $\Delta^{\star}$ satisfy,
    \[
    \RR \rho^{\star} = - \yy, \quad \CC^{\top}\yy = 0.
    \]
    Using these values back in our program, we get,
    \begin{equation}\label{eq:Prog1}
    \Psi(\rr) = \max_{\yy: \CC^{\top}\yy = 0} 2\yy^{\top}\dd - \yy^{\top}\RR^{-1}\yy.
    \end{equation}
    We now note that, the following program has minimizer $(\zz^{\star},\theta^{\star}) = (\yy^{\star},1)$ and achieves the same optimum value as Program~\ref{eq:Prog1}.
    \begin{equation}\label{eq:Prog2}
        \Psi(\rr) = \max_{\zz,\theta: \CC^{\top}\zz = 0} 2\theta\zz^{\top}\dd - \theta^2 \zz^{\top}\RR^{-1}\zz.
    \end{equation}
    Since $\zz^{\star} = \yy^{\star}$ and $\theta^{\star} = 1$ is an optimum for~\eqref{eq:Prog2},
    \[
    \left[\frac{d}{d\theta}\left(2\theta{\yy^{\star}}^{\top}\dd - \theta^2 {\yy^{\star}}^{\top}\RR^{-1}\yy^{\star}\right)\right]_{\theta=1} = 0.
    \]
    As a result, we get that,
    \[
    \dd^{\top}\yy^{\star} = {\yy^{\star}}^{\top}\RR^{-1}\yy^{\star},
    \]
    and $\Psi(\rr) = \dd^{\top}\yy^{\star}$. Using this, we claim that the following program has objective value $2 - 1/\Psi(\rr)$:
    \[
    2 - \frac{1}{\Psi(\rr)} = \max_{\substack{\CC^{\top}\yy = 0\\ \dd^{\top}\yy = 1}} 2 \dd^{\top}\yy - \yy^{\top}\RR^{-1}\yy.
    \]
    The reason is that for any feasible $\yy$ of \eqref{eq:Prog1}, the vector $\yy' := \yy/\dd^{\top} \yy$ is a feasible solution of this new program, and in particular the optimal solution $\yy^*$ of \eqref{eq:Prog1} induces a solution $\yy^* / \dd^{\top} \yy^* = \yy^* / \Psi(\rr)$ for this new program, and it gives an objective value of $2 - \frac{1}{\Psi(\rr)}$. Now it suffices to prove that $\yy^* / \dd^{\top} \yy^*$ is the optimal solution of this new program. Suppose by contrary that there is another vector $\yy'$ that is the optimal solution of this new program, then we have $\dd^{\top} \yy' =1$ and $\yy'^{\top} \RR^{-1} \yy' \leq 1/\Psi(\rr)$. But then the vector $\yy' \cdot \Psi(\rr)$ would be a solution for \eqref{eq:Prog1} which gives a larger objective value than $\Psi(\rr)$, and this contradicts with the definition that $\Psi(\rr)$ is the optimal objective value of \eqref{eq:Prog1}.
    
    We can now write the following problem,
    \begin{equation}\label{eq:Prog3}
    \frac{1}{\Psi(\rr)} = \min_{\substack{\CC^{\top}\yy = 0\\ \dd^{\top}\yy = 1}} \yy^{\top}\RR^{-1}\yy,
    \end{equation}
    with optimum value at $\yytil = \frac{\yy^{\star}}{\Psi(\rr)}$. Let us now consider Program~\ref{eq:Prog3} at $\rr'$, i.e., $\Psi(\rr')$. We note that $\yytil=\frac{\yy^{\star}}{\Psi(\rr)}$ is a feasible solution for the new program as well. Therefore,
    \[
    \frac{1}{\Psi(\rr')} = \min_{\substack{\CC^{\top}\yy = 0\\ \dd^{\top}\yy = 1}} \yy^{\top}{\RR'}^{-1}\yy \leq \frac{1}{\Psi(\rr)^2} {\yy^{\star}}^{\top}{\RR'}^{-1}\yy^{\star} =\frac{1}{\Psi(\rr)^2} \sum_e \frac{\rr_e}{\rr_e'} \rr_e (\CC\Delta^{\star}-\dd)_e^2.
    \]
    Rearranging the above,
    \[
    \frac{1}{\Psi(\rr')} \leq \frac{1}{\Psi(\rr)}\left(1 - \frac{ \sum_e \left(\frac{\rr_e' - \rr_e}{\rr_e'}\right) \rr_e (\CC\Delta^{\star}-\dd)_e^2}{\Psi(\rr)}\right).
    \]
    And since we assumed that $|\rr_e'-\rr_e| \leq \rr'_e$, we have that $\sum_e \left(\frac{\rr_e' - \rr_e}{\rr_e'}\right) \rr_e (\CC\Delta^{\star}-\dd)_e^2 \leq \Psi(\rr)$, and hence $\left(1 - \frac{ \sum_e \left(\frac{\rr_e' - \rr_e}{\rr_e'}\right) \rr_e (\CC\Delta^{\star}-\dd)_e^2}{\Psi(\rr)}\right)^{-1} \geq \left(1 + \frac{ \sum_e \left(\frac{\rr_e' - \rr_e}{\rr_e'}\right) \rr_e (\CC\Delta^{\star}-\dd)_e^2}{\Psi(\rr)}\right)$, and so we have
    \[
    \Psi(\rr') \geq \Psi(\rr)+  \sum_e \left(\frac{\rr_e' - \rr_e}{\rr_e'}\right) \rr_e (\CC\Delta^{\star}-\dd)_e^2.
    \]
\end{proof}

\subsection*{Proof of Fact~\ref{fact:FMM_convex}}

\begin{proof}
Let $x = \alpha \cdot p + (1-\alpha) \cdot q$ for $\alpha \in (0,1)$. For notational simplicity, we assume that $n^{\beta}$, $n^p$, $n^q$, and $n^{\alpha}$ are all integers. 
Consider two rectangular matrices of dimensions $n \times n^x$ and $n^x \times n^{\beta}$. Since $\alpha  p \leq x$, we can tile the $n \times n^x$ rectangular matrix with matrices of dimensions $n^{\alpha} \times n^{\alpha  p}$, and tile the $n^x \times n^{\beta}$ rectangular matrix with matrices of dimensions $n^{\alpha p} \times n^{\alpha  \beta}$. Then, the  product of the two tiled matrices can be obtained by viewing it as a multiplication of a matrix of dimensions $n/n^{\alpha} \times n^{x}/n^{\alpha  p}$ with a matrix of dimensions $n^{x}/n^{\alpha  p} \times n^{\beta} / n^{\alpha \beta}$, where each ``element'' of the first matrix is itself a matrix of dimensions $n^{\alpha} \times n^{\alpha  p}$, and each ``element'' of the second matrix is itself a matrix of dimensions $n^{\alpha p} \times n^{\alpha  \beta}$. With this recursion in tow, we obtain the following upper bound. 
\begin{align*}
\Tmat(n, n^x, n^{\beta}) 
\leq & \Tmat(n^{\alpha}, n^{\alpha p}, n^{\alpha \beta}) \cdot \Tmat(n/n^\alpha, n^x/n^{\alpha  p}, n^{\beta} / n^{\alpha \beta}) \\
= & \Tmat(n^{\alpha}, n^{\alpha  p}, n^{\alpha \beta}) \cdot \Tmat(n^{(1-\alpha)}, n^{(1-\alpha) q}, n^{(1-\alpha) \beta})\\
\leq & n^{\alpha \cdot \omega_{\beta}(p) + o(1)} \cdot n^{(1-\alpha) \cdot \omega_{\beta}(q) + o (1)},
\end{align*} 
where the last step follows from denoting $m_1 = n^\alpha$ and  observing that by the definition of $\omega_{\beta}(x)$, multiplying matrices of dimensions $n^\alpha \times n^{\alpha p}$ and $n^{\alpha p} \times n^{\alpha \beta}$ costs $m_1^{\omega_{\beta}(p) + o(1)}$, which is exactly $n^{\alpha (\omega_{\beta}(p) + o(1))}$, and similarly denoting $m_2 = n^{1-\alpha}$, we have that multiplying matrices of dimensions $n^{1-\alpha} \times n^{(1-\alpha) q}$ and $n^{(1-\alpha) q} \times n^{(1-\alpha) \beta}$ costs $m_2^{\omega_{\beta}(q) + o(1)}$, which is exactly $n^{(1-\alpha) (\omega_{\beta}(q) + o(1))}$. Comparing exponents, this implies that 
\begin{align*}
\omega_{\beta}(x) &\leq \alpha\cdot\omega_{\beta}(p) + (1-\alpha)\cdot \omega_{\beta}(q) ,
\end{align*} 
which proves the convexity of the function $\omega_{\beta}(x)$. 
\end{proof}

\subsection*{Missing Proofs of inverse mainenance data structures}
\subsection*{Proof of Lemma~\ref{lem:InverseMaintenanceOneLevel}}

\begin{proof}
The initialization and update operations are straightforward. Next we prove the time complexity of the reset and query operations.

Consider the $t$-th iteration. We define $\Delta := \Delta^{(t_0+1)} + \cdots + \Delta^{(t)}$, and note that $k = \nnz(\Delta^{(t_0+1)}) + \cdots + \nnz(\Delta^{(t)}) \geq \nnz(\Delta)$. W.l.o.g.~we assume $k = \nnz(\Delta)$. We write the decomposition $\Delta = \UU \CC \VV^{\top}$, where $\UU, \VV \in \R^{n \times k}$ each consists of $k$ columns of the identity matrix, and $\CC \in \R^{k \times k}$ consists of the non-zero entries of $\Delta$.

{\bf Reset.} Using Woodbury identity, we have 
\begin{align*}
(\MM^{(t)})^{-1} = &~ (\MM^{(t_0)} + \Delta)^{-1} \\
= &~ (\MM^{(t_0)} + \UU \CC \VV^{\top})^{-1} \\
= &~ (\MM^{(t_0)})^{-1} - (\MM^{(t_0)})^{-1} \UU \big(\CC^{-1} + \VV^{\top} (\MM^{(t_0)})^{-1} \UU\big)^{-1} \VV^{\top} (\MM^{(t_0)})^{-1} \\
= &~ \NN - \NN \UU (\CC^{-1} + \VV^{\top} \NN \UU )^{-1} \VV^{\top} \NN.
\end{align*}
The matrices $\NN \UU$ and $\VV^{\top} \NN$ can be directly read off from the maintained inverse $\NN$. The dominating term to compute this inverse is to multiply the $n \times k$ matrix $\NN \UU$ with the $k \times n$ matrix $(\CC^{-1} + \VV^{\top} \NN \UU)^{-1} \VV^{\top} \NN$, and it takes $\Tmat(n,n,k)$ time.

{\bf Query.} Again using Woodbury identity, we have 
\begin{align*}
\big((\MM^{(t)})^{-1} \big)_{J_r, J_c} = &~ \NN_{J_r, J_c} - \NN_{J_r, :} \UU (\CC^{-1} + \VV^{\top} \NN \UU )^{-1} \VV^{\top} \NN_{:, J_c}.
\end{align*}
W.l.o.g.~we assume $\ell_c \leq \ell_r$. The running time has the following parts:
\begin{itemize}
\item Computing the inverse $(\CC^{-1} + \VV^{\top} \NN \UU )^{-1}$ takes $k^{\omega}$ time.
\item Computing $(\CC^{-1} + \VV^{\top} \NN \UU )^{-1} \cdot (\VV^{\top} \NN_{:, J_c})$ takes $\Tmat(k, k, \ell_c)$ time.
\item Computing $(\NN_{J_r, :} \UU) \cdot (\CC^{-1} + \VV^{\top} \NN \UU )^{-1} \VV^{\top} \NN_{:, J_c}$ takes $\Tmat(\ell_r, k, \ell_c)$ time.
\end{itemize}
Since if $\ell_c \leq k$, then $\Tmat(k, k, \ell_c) \leq k^{\omega}$, and otherwise if $\ell_c > k$, then $\Tmat(k, k, \ell_c) \leq \Tmat(\ell_r, k, \ell_c)$, so the total time is
\[
O\big(k^{\omega} + \Tmat(\ell_r, k, \ell_c) \big).
\]
\end{proof}

\subsection*{Proof of Lemma~\ref{lem:InverseMaintenanceTwoLevel}}

\begin{proof}
The initialization and update operations are straightforward. Next we prove the time complexity of the reset, partial reset, and query operations.

Consider the $t$-th iteration. We define $\Delta_0 := \Delta^{(t_0+1)} + \cdots + \Delta^{(t)}$, and note that $k_0 = \nnz(\Delta^{(t_0+1)}) + \cdots + \nnz(\Delta^{(t)}) \geq \nnz(\Delta_0)$. Similarly define $\Delta_1 := \Delta^{(t_1+1)} + \cdots + \Delta^{(t)}$ and note that $k_1 \geq \nnz(\Delta_1)$.

{\bf Reset.} We compute $(\MM^{(t)})^{-1} = (\MM^{(t_0)} + \Delta_0)^{-1}$ using Woodbury identity. Similar to the reset operation of Lemma~\ref{lem:InverseMaintenanceOneLevel}, it takes $O(\Tmat(n,n,k_0))$ time.

{\bf Partial reset.} Let $J^{\new} \subseteq [n]$ denote the indexes of the non-zero columns of $\Delta^{(t_0+1)} + \cdots + \Delta^{(t)}$ and let $\BB^{\new} = (\TT^{(t_0,t)}_{J^{\new},J^{\new}})^{-1}$ and $\EE^{\new} = (\TT^{(t_0,t)}_{J^{\new},J^{\new}})^{-1} \cdot \NN_{J^{\new},:}$ denote the matrices that we want to obtain. Next we show how to compute these two matrices efficiently.

First note that we have
\begin{align*}
\TT^{(t_0,t)} = &~ \II + (\MM^{(t_0)})^{-1} \cdot (\MM^{(t)} - \MM^{(t_0)}) \\
= &~ \II + (\MM^{(t_0)})^{-1} \cdot (\MM^{(t_1)} - \MM^{(t_0)}) + (\MM^{(t_0)})^{-1} \cdot (\MM^{(t)} - \MM^{(t_1)}) \\
= &~ \TT^{(t_0,t_1)} + \NN \cdot \Delta_1.
\end{align*}

Note that $\TT^{(t_0,t_1)}$ is identity matrix plus some non-zero entries on the columns in $J$, so we have
\begin{align*}
\TT^{(t_0,t)}_{J^{\new},J^{\new}} = &~
\begin{bmatrix}
\TT^{(t_0,t_1)}_{J,J} & 0 \\
0 & \II
\end{bmatrix} + 
\begin{bmatrix}
0 & 0 \\
\TT^{(t_0,t_1)}_{J^{\new} \backslash J, J} & 0
\end{bmatrix} + (\NN \cdot \Delta_1)_{J^{\new}, J^{\new}}
\end{align*}
Since $\nnz(\MM^{(t_1)}-\MM^{(t_0)}) \leq k_0$ and $|J^{\new} \backslash J| \leq k_1$, we can compute $\TT^{(t_0,t_1)}_{J^{\new} \backslash J, J} = \big(\NN \cdot (\MM^{(t_1)}-\MM^{(t_0)})\big)_{J^{\new} \backslash J, J}$ in $O(k_0 k_1)$ time. Since $\nnz(\Delta_1) \leq k_1$, let $J' \subseteq $ denote the row indexes of the non-zero entries of $\Delta_1$, we can write $(\NN \cdot \Delta_1)_{J^{\new}, J^{\new}} = (\NN_{J^{\new}, J'}) \cdot (\Delta_1)_{J', J^{\new}}$ in $O(k_0 k_1)$ time. 
So in conclusion, we have
\begin{align*}
\TT^{(t_0,t)}_{J^{\new},J^{\new}} = &~ \begin{bmatrix}
\TT^{(t_0,t_1)}_{J,J} & 0 \\
0 & \II
\end{bmatrix} + (\II_{J^{\new}, J^{\new}\backslash J}) \cdot (\TT^{(t_0,t_1)}_{J^{\new} \backslash J, J^{\new}}) + (\NN_{J^{\new}, J'}) \cdot (\Delta_1)_{J', J^{\new}} \\
= &~ \begin{bmatrix}
\TT^{(t_0,t_1)}_{J,J} & 0 \\
0 & \II
\end{bmatrix} + \UU \VV^{\top},
\end{align*}
where both $\UU$ and $\VV$ have size at most $k_0 \times O(k_1)$, and we can compute $\UU$ and $\VV$ in $O(k_0 k_1)$ time. Furthermore, because $\UU = [(\II_{J^{\new}, J^{\new}\backslash J}), (\NN_{J^{\new}, J'})]$ and since we already maintain $\EE = \BB \cdot \NN_{J,:}$, we directly read off entries from the matrix $\begin{bmatrix}
\BB & 0 \\
0 & \II
\end{bmatrix} \cdot \UU$.

Then using Woodbury identity and since we already have $\BB = (\TT^{(t_0,t_1)}_{J,J})^{-1}$, we can compute the updated $\BB^{\new} = (\TT^{(t_0,t)}_{J^{\new},J^{\new}})^{-1}$ in $O(\Tmat(k_0, k_0, k_1))$ time. 

Next we compute the updated $\EE^{\new}$ using the maintained $\EE = (\TT^{(t_0,t_1)}_{J,J})^{-1} \cdot \NN_{J,:}$. Using Woodbury identity, we have
\begin{align}\label{eq:E_new}
\EE^{\new} = &~ \left( \begin{bmatrix}
\TT^{(t_0,t_1)}_{J,J} & 0 \\
0 & \II
\end{bmatrix} + \UU \VV^{\top} \right)^{-1} \cdot \NN_{J^{\new},:} \notag \\
= &~ \left( \begin{bmatrix}
\BB & 0 \\
0 & \II
\end{bmatrix} - \begin{bmatrix}
\BB & 0 \\
0 & \II
\end{bmatrix} \UU (\II + \VV^{\top} \begin{bmatrix}
\BB & 0 \\
0 & \II
\end{bmatrix} \UU)^{-1} \VV^{\top} \begin{bmatrix}
\BB & 0 \\
0 & \II
\end{bmatrix} \right) \cdot \NN_{J^{\new},:} \notag \\
= &~ \begin{bmatrix}
\EE \\
\NN_{J^{\new} \backslash J,:}
\end{bmatrix} 
- \begin{bmatrix}
\BB & 0 \\
0 & \II
\end{bmatrix} \UU (\II + \VV^{\top} \begin{bmatrix}
\BB & 0 \\
0 & \II
\end{bmatrix} \UU)^{-1} \VV^{\top} \begin{bmatrix}
\EE \\
\NN_{J^{\new}\backslash J,:}
\end{bmatrix}.
\end{align}

The dominating terms are to compute $\VV^{\top} \cdot \begin{bmatrix}
\EE \\
\NN_{J^{\new}\backslash J,:}
\end{bmatrix}$ and to multiply $\begin{bmatrix}
\BB & 0 \\
0 & \II
\end{bmatrix} \UU$ with the matrix 
$(\II + \VV^{\top} \begin{bmatrix}
\BB & 0 \\
0 & \II
\end{bmatrix} \UU)^{-1} \VV^{\top} \begin{bmatrix}
\EE \\
\NN_{J^{\new}\backslash J,:}
\end{bmatrix}$, and both of these two steps take $O(\Tmat(n,k_0,k_1))$ time.

In summary, the total time of the partial reset operation is $O(\Tmat(n, k_0, k_1))$.

{\bf Query.} 
By the definition of the transformation matrix, we have $\MM^{(t)} = \MM^{(t_0)} \cdot \TT^{(t_0, t)}$. Again let $J^{\new} \subseteq [n]$ denote the indexes of the non-zero columns of $\Delta^{(t_0+1)}, \cdots, \Delta^{(t)}$. And recall that $J \subseteq [n]$ denotes the indexes of the non-zero columns of $\Delta^{(t_0+1)}, \cdots, \Delta^{(t_1)}$. 

Using a similar computation as the partial update operation, we can write $\TT^{(t_0,t)}_{J^{\new},J^{\new}} = \begin{bmatrix}
\TT^{(t_0,t_1)}_{J,J} & 0 \\
0 & \II
\end{bmatrix} + \UU \VV^{\top}$, where both $\UU$ and $\VV$ have size $k_0 \times O(k_1)$, and we can compute them in $O(k_0 k_1)$ time. Again denote $\EE^{\new} = (\TT^{(t_0,t)}_{J^{\new},J^{\new}})^{-1} \cdot \NN_{J^{\new},:}$, and it can be written as Eq.~\eqref{eq:E_new}. In the query operation we don't compute the matrix $\EE^{\new}$ explicitly, instead, we note that by using Eq.~\eqref{eq:E_new} and since the matrix $\begin{bmatrix}
\BB & 0 \\
0 & \II
\end{bmatrix} \cdot \UU$ is already maintained, we can compute any $\ell_r \times \ell_c$ submatrix of $\EE^{\new}$ in $O(\Tmat(k_0, k_1, k_1) + \Tmat(k_0, k_1, \ell_c) + \Tmat(\ell_r, k_1, \ell_c))$ time.

Next note that since $\MM^{(t)}$ only differs from $\MM^{(t_0)}$ on columns in set $J^{\new}$, denoting $\ov{J^{\new}} = [n]\backslash J^{\new}$, we can write $\TT^{(t_0, t)}$ as
\[
\TT^{(t_0, t)} = 
\begin{bmatrix}
\TT^{(t_0, t)}_{J^{\new},J^{\new}} & 0 \\
\TT^{(t_0, t)}_{\ov{J^{\new}},J^{\new}} & \II
\end{bmatrix}.
\]
So its inverse is $\begin{bmatrix}
(\TT^{(t_0, t)}_{J^{\new},J^{\new}})^{-1} & 0 \\
- \TT^{(t_0, t)}_{\overline{J^{\new}},J^{\new}} \cdot (\TT^{(t_0, t)}_{J^{\new},J^{\new}})^{-1} & \II
\end{bmatrix}$, and we have
\begin{align}\label{eq:M_t_inv}
(\MM^{(t)})^{-1} = &~ (\TT^{(t_0, t)})^{-1} \cdot (\MM^{(t_0)})^{-1} \notag \\
= &~ \begin{bmatrix}
(\TT^{(t_0, t)}_{J^{\new},J^{\new}})^{-1} & 0 \\
- \TT^{(t_0, t)}_{\overline{J^{\new}},J^{\new}} \cdot (\TT^{(t_0, t)}_{J^{\new},J^{\new}})^{-1} & \II
\end{bmatrix} 
\cdot \NN \notag \\
= &~ \begin{bmatrix}
(\TT^{(t_0, t)}_{J^{\new},J^{\new}})^{-1} \cdot \NN_{J^{\new}, :} \\
- (\II + \NN \cdot \Delta_0)_{\overline{J^{\new}},J^{\new}} \cdot (\TT^{(t_0, t)}_{J^{\new},J^{\new}})^{-1} \cdot \NN_{J^{\new}, :} + \NN_{\ov{J^{\new}}, :}
\end{bmatrix}.
\end{align}
To compute a $\ell_r \times \ell_c$ submatrix of $(\MM^{(t)})^{-1}$, we first compute $\ell_c$ columns of $(\TT^{(t_0,t)}_{J^{\new},J^{\new}})^{-1} \cdot \NN_{J^{\new},:}$, and as we just proved, this takes $O(\Tmat(k_0, k_1, k_1) + \Tmat(k_0, k_1, \ell_c))$ time. Next we compute the multiplication of $\ell_r$ rows of $(\NN \cdot \Delta_0)_{\overline{J^{\new}},J^{\new}}$ with $\ell_c$ columns of $(\TT^{(t_0,t)}_{J^{\new},J^{\new}})^{-1} \cdot \NN_{J^{\new},:}$, and this takes $O(\ell_r, k_0, \ell_c)$.

In summary, the total time of the query operation is $O(\Tmat(k_0, k_1, k_1) + \Tmat(k_0, k_1, \ell_c) + \Tmat(k_0, \ell_r, \ell_c))$.
\end{proof}

\subsection*{Proof of Lemma~\ref{lem:implicit_inverse_maintenance}}
\begin{proof}
In the algorithm we always maintain the following invariants:
\begin{align*}
J = &~ \text{set of indexes of the non-zero columns of } \Delta^{(t_0+1)} + \cdots + \Delta^{(t_1)} \\
\NN = &~ (\MM^{(t_0)})^{-1}  \\
\BB = &~ (\TT^{(t_0,t_1)}_{J,J})^{-1} \\
\EE = &~ (\TT^{(t_0,t_1)}_{J,J})^{-1} \cdot \NN_{J,:} \\
\uu_0 + \NN \uu_1 + \begin{bmatrix}
\NN_{J,:} \cdot \uu_2 \\ 0
\end{bmatrix} = &~ \sum_{i=0}^t (\MM^{(i)})^{-1} \vv.
\end{align*}
Apart from these invariants, the data structure also maintains the vectors $\NN \cdot \vv$, $\BB \cdot \vv_{J}$, and $\EE \cdot \vv$.

Next we describe each operation and bound its time complexity.

{\bf Initialize.} Initially we let $t_0 = t_1 = 0$. Let $\NN = (\MM^{(0)})^{-1} \in \R^{n \times n}$, and let $\BB$ and $\EE$ be empty matrices. We also pre-compute and maintain $\NN \cdot \vv$. We let $\uu_0 = (\MM^{(0)})^{-1} \cdot \vv$ and $\uu_1 = \uu_2 = 0$. Initialization takes $O(n^{\omega})$ time.

{\bf Query sum.} Since we maintain the invariants, we simply output $\uu_0 + \NN \uu_1 + \begin{bmatrix}
\NN_{J,:} \cdot \uu_2 \\ 0
\end{bmatrix} = \sum_{i=0}^t (\MM^{(i)})^{-1} \vv$, and this takes $O(n^2)$ time.

For the reset, partial reset, and update operations, we define the following notations. Consider the $t$-th iteration. We define $\Delta_0 := \Delta^{(t_0+1)} + \cdots + \Delta^{(t)}$, and note that $k_0 = \nnz(\Delta^{(t_0+1)}) + \cdots + \nnz(\Delta^{(t)}) \geq \nnz(\Delta_0)$. Similarly define $\Delta_1 := \Delta^{(t_1+1)} + \cdots + \Delta^{(t)}$ and note that $k_1 \geq \nnz(\Delta_1)$.

{\bf Reset.} We first update $\uu_0 \gets \uu_0 + \NN \uu_1 + \begin{bmatrix}
\NN_{J,:} \cdot \uu_2 \\ 0
\end{bmatrix}$ and $\uu_1, \uu_2 \gets 0$. 

Next, we update $\NN = (\MM^{(t)})^{-1} = (\MM^{(t_0)} + \Delta_0)^{-1}$ using Woodbury identity in th exact same way as the reset operation of Lemma~\ref{lem:InverseMaintenanceOneLevel} and \ref{lem:InverseMaintenanceTwoLevel}, and this operation takes $O(\Tmat(n,n,k_0))$ time. 

Finally we also recompute the vector $\NN \cdot \vv$ in $O(n^2)$ time.

{\bf Partial reset.} We first update $\uu_0 \gets \uu_0 + \begin{bmatrix}
\NN_{J,:} \cdot \uu_2 \\ 0
\end{bmatrix}$ and $\uu_2 \gets 0$. This takes $O(n \cdot k_0)$ time. 

Next let $J^{\new} \subseteq [n]$ denote the indexes of the non-zero columns of $\Delta^{(t_0+1)} + \cdots + \Delta^{(t)}$ and let $\BB^{\new} = (\TT^{(t_0,t)}_{J^{\new},J^{\new}})^{-1}$ and $\EE^{\new} = (\TT^{(t_0,t)}_{J^{\new},J^{\new}})^{-1} \cdot \NN_{J^{\new},:}$ denote the matrices that we want to obtain. We compute $\BB^{\new}$ and $\EE^{\new}$ in the exact same way as Lemma~\ref{lem:InverseMaintenanceTwoLevel} in $O(\Tmat(n, k_0, k_1))$ time. 

Finally we also recompute the vectors $\BB^{\new} \cdot \vv_{J^{\new}}$ and $\EE^{\new} \cdot \vv$ in $O(n \cdot k_0)$ time.

{\bf Update.} Again let $J^{\new} \subseteq [n]$ denote the indexes of the non-zero columns of $\Delta^{(t_0+1)} + \cdots + \Delta^{(t)}$ and let $\BB^{\new} = (\TT^{(t_0,t)}_{J^{\new},J^{\new}})^{-1}$ and $\EE^{\new} = (\TT^{(t_0,t)}_{J^{\new},J^{\new}})^{-1} \cdot \NN_{J^{\new},:}$. Similar as the proof of Lemma~\ref{lem:InverseMaintenanceTwoLevel}, we can compute the decomposition $\TT^{(t_0,t)}_{J^{\new},J^{\new}} = \begin{bmatrix}
\TT^{(t_0,t_1)}_{J,J} & 0 \\
0 & \II
\end{bmatrix} + \UU \VV^{\top}$ and $\begin{bmatrix}
\BB & 0 \\
0 & \II
\end{bmatrix} \cdot \UU$ in $O(k_0 k_1)$ time, where both $\UU$ and $\VV$ have size $k_0 \times O(k_1)$. And we still have Eq.~\eqref{eq:M_t_inv} and \eqref{eq:E_new} that
\begin{align*}
(\MM^{(t)})^{-1} = &~ \begin{bmatrix}
\EE^{\new} \\
- (\II + \NN \cdot \Delta_0)_{\overline{J^{\new}},J^{\new}} \cdot \EE^{\new} + \NN_{\ov{J^{\new}}, :} 
\end{bmatrix}, \\
\EE^{\new} = &~ \begin{bmatrix}
\EE \\
\NN_{J^{\new} \backslash J,:}
\end{bmatrix} 
- \begin{bmatrix}
\BB & 0 \\
0 & \II
\end{bmatrix} \UU (\II + \VV^{\top} \begin{bmatrix}
\BB & 0 \\
0 & \II
\end{bmatrix} \UU)^{-1} \VV^{\top} \begin{bmatrix}
\EE \\
\NN_{J^{\new}\backslash J,:}
\end{bmatrix}.
\end{align*}
So we have
\begin{align*}
\EE^{\new} \vv = &~ \begin{bmatrix}
\EE \vv \\
\NN_{J^{\new} \backslash J,:} \vv
\end{bmatrix} 
- \begin{bmatrix}
\BB & 0 \\
0 & \II
\end{bmatrix} \UU (\II + \VV^{\top} \begin{bmatrix}
\BB & 0 \\
0 & \II
\end{bmatrix} \UU)^{-1} \VV^{\top} \begin{bmatrix}
\EE \vv \\
\NN_{J^{\new}\backslash J,:} \vv
\end{bmatrix}
\end{align*}
We can compute this vector in the following steps:
\begin{itemize}
\item Since the data structure maintains $\EE \vv$ and $\NN \vv$, we first directly read off the vector $\begin{bmatrix}
\EE \vv \\
\NN_{J^{\new}\backslash J,:} \vv
\end{bmatrix}$ from the maintained vectors.
\item Compute $(\II + \VV^{\top} \begin{bmatrix}
\BB & 0 \\
0 & \II
\end{bmatrix} \UU)^{-1}$ in $O(\Tmat(k_1, k_0, k_1))$ time. 
\item Compute the matrix vector products from right to left as \[
\left(\begin{bmatrix}
\BB & 0 \\
0 & \II
\end{bmatrix} \UU \right) \cdot (\II + \VV^{\top} \begin{bmatrix}
\BB & 0 \\
0 & \II
\end{bmatrix} \UU)^{-1} \cdot \VV^{\top} \cdot \begin{bmatrix}
\EE \vv \\
\NN_{J^{\new}\backslash J,:} \vv
\end{bmatrix},
\]
and this takes $O(k_0 k_1)$ time.
\end{itemize}
In summary we can compute $\EE^{\new} \vv$ in $O(\Tmat(k_0, k_1, k_1))$ time. Next we compute $(\MM^{(t)})^{-1} \cdot \vv$:
\begin{align*}
(\MM^{(t)})^{-1} \cdot \vv = &~ \begin{bmatrix}
\EE^{\new} \vv \\
- (\II + \NN \cdot \Delta_0)_{\overline{J^{\new}},J^{\new}} \cdot \EE^{\new} \vv + \NN_{\ov{J^{\new}}, :} \vv
\end{bmatrix} \\
= &~ \begin{bmatrix}
\EE^{\new} \vv \\
-\II_{\overline{J^{\new}},J^{\new}} \cdot \EE^{\new} \vv + \NN_{\overline{J^{\new}},:} \cdot (\vv - \ww)
\end{bmatrix} \\
= &~ \begin{bmatrix}
\EE^{\new} \vv - \NN_{J^{\new}, :} \cdot (\vv - \ww)\\
-\II_{\overline{J^{\new}},J^{\new}} \cdot \EE^{\new} \vv
\end{bmatrix} + \NN \cdot (\vv - \ww) \\
= &~ \begin{bmatrix}
\EE^{\new} \vv - \ww' \\
-\II_{\overline{J^{\new}},J^{\new}} \cdot \EE^{\new} \vv
\end{bmatrix} - \begin{bmatrix}
\NN_{J,:} \cdot (\vv - \ww) \\ 0
\end{bmatrix} + \NN \cdot (\vv - \ww),
\end{align*}
where in the second step we define a vector $\ww \in \R^n$ such that its entries in $J^{\new}$ are $(\Delta_0)_{J^{\new}, J^{\new}} \cdot \EE^{\new} \vv$ and its rest entries are all zero, and we can compute it in $O(k_0)$ time since $\nnz(\Delta_0) \leq k_0$, in the fourth step we define a vector $\ww' \in \R^{|J^{\new}|}$ such that its entries in $J^{\new} \backslash J$ are $\NN_{J^{\new}\backslash J, :} \cdot (\vv - \ww)$ and its rest entries are all zero, and we can compute it in $O(k_0 k_1)$ time since $|J^{\new}\backslash J| \leq k_1$.

So we update the vectors as
\begin{align*}
\uu_0 \gets \uu_0 + \begin{bmatrix}
\EE^{\new} \vv - \ww' \\
-\II_{\overline{J^{\new}},J^{\new}} \cdot \EE^{\new} \vv
\end{bmatrix}, ~~~
\uu_1 \gets \uu_1 + \vv - \ww, ~~~
\uu_w \gets \uu_1 - \vv + \ww.
\end{align*}
In this way we still maintain the invariant of the three vectors.
\end{proof}

%% file: NonMonotoneAcceleration.tex
\section{MWU with Non-Monotone Weights and \texorpdfstring{$n^{1/3}$}{TEXT} Iterations (Optional)}\label{sec:NonMonMWU}

In this section, we will prove Theorem~\ref{thm:NonMonotoneMWUAcc}. Our analysis follows a similar structure to that of Algorithm~\ref{alg:MWU}. We would track the same potentials as defined in Equations~\eqref{eq:defPhi} and \eqref{eq:defPsi}. 

In this section, we again use the notation $k_i$ to denote the total number of width reduction steps taken before the $i^{th}$ primal step is being executed, and we use the notation $i_k$ to denote the number of primal steps taken by the algorithm when the $k^{th}$ width step is being executed. We show that for input $\wt{\CC} = \begin{bmatrix}\CC\\ -\CC    \end{bmatrix}$ and $\wt{\dd} = \begin{bmatrix}
   \dd\\ -\dd
\end{bmatrix}$, the algorithm returns $\wt{\xx}$ such that $\|\CC\wt{\xx}-\dd\|_{\infty}\leq 1+ O(\epsilon)$.

\subsection*{Convergence Analysis}
We begin by showing that $\ww^{(i,k)}>0$ for all $i$ and $k$. Observe that this was always true for Algorithm~\ref{alg:MWU}.

\begin{lemma}\label{lem:positiveWopt}
 If $\alpha \rho^{1/3}\leq \frac{\epsilon^{1/3}}{10 n^{1/3}}$ and $\delta \leq 1/2$, then for every iteration $i$ and $k$ and every coordinate $e$, $\ww^{(i,k)}_e > 0$. Furthermore, every primal update satisfies, $|\overrightarrow{\alpha}^{(i,k)} (\wt{\CC}\Delta^{(i,k)}-\wt{\dd})|\leq 1/10.$
\end{lemma}
\begin{proof}
    The weights will always increase during a width reduction step and can only decrease during a primal step. Therefore, we prove that after a primal step the weights can never be negative, i.e., for every $i,k$, $1 + \epsilon \overrightarrow{\alpha}(\wt{\CC}\Delta^{(i,k)}-\wt{\dd}) \geq \frac{1}{2}$.

    Observe that, when we do a primal step, i.e., the condition on Line~\ref{algline:CheckPrimalNM} of Algorithm~\ref{alg:non_monotone_accel_opt} is true, since for all $e$, $\rr^{(i,k)}_e \geq \frac{\epsilon}{2n}\Phi(\ww^{(i,k)})\geq \frac{\epsilon}{2n}\Psi(\rr^{(i,k)})$, 
    \[
    \frac{\epsilon}{2n}\Psi(\rr^{(i,k)})\|\wt{\CC}\Delta^{(i,k)}-\wt{\dd}\|_3^3 \leq 2\rho \Psi(\rr^{(i,k)}).
    \]
    This implies that,
    \[
    \|\wt{\CC}\Delta^{(i,k)}-\wt{\dd}\|_{\infty}
\leq \|\wt{\CC}\Delta^{(i,k)}-\wt{\dd}\|_3 \leq \frac{4^{1/3} n^{1/3}\rho^{1/3}}{\epsilon^{1/3}} \leq \frac{2 n^{1/3}\rho^{1/3}}{\epsilon^{1/3}}.
    \]
    When $(\wt{\CC}\Delta^{(i,k)}-\wt{\dd})_e \geq 0$, our weights can only increase, therefore we consider the case when it is negative. The multiplicative change to the weights now becomes,
    \[
    1 + \epsilon \overrightarrow{\alpha}(\wt{\CC}\Delta^{(i,k)}-\wt{\dd}) \geq 1 - \epsilon \alpha_-\|\wt{\CC}\Delta^{(i,k)}-\wt{\dd}\|_{\infty} \geq 1 - \frac{\epsilon}{1+2\epsilon}\alpha \cdot \frac{4 n^{1/3}\rho^{1/3}}{\epsilon^{1/3}} \geq \frac{1}{2}.
    \]
    The second part just follows by using the values of $\overrightarrow{\alpha}^{(i,k)} $ and the bound on $\|\wt{\CC}\Delta^{(i,k)}-\wt{\dd}\|_{\infty}$.
    {\bf We would like to remark that this proof also works for the further robust algorithm described in Section~\ref{sec:NonMonStab}}
\end{proof}

We now begin our analysis. The next two lemmas show how our potentials change with every iteration of the algorithm.

\subsubsection*{Change in $\Phi$}
\begin{restatable}{lemma}{ChangePhiNMOpt}
  \label{lem:ChangePhiNMOpt}
  After $i$ primal steps, and $k$ width-reduction steps,
  the potential $\Phi$ is bounded as follows:
  \begin{align*}
\Phi\left(\ww^{(i,k)}\right) \leq \left(\Phi(\ww^{(0,0)})\right)\left(1+\epsilon\alpha_+ e^{\epsilon/2}\right)^i\left(1 + \frac{2\epsilon e^{\epsilon}}{\tau}\right)^k.
 \end{align*}
 Furthermore, after every primal step, the potential can decrease by at most,
 \[
 \Phi\left(\ww^{(i+1,k)}\right) \geq  \Phi\left(\ww^{(i,k)}\right)\left(1-\epsilon\alpha_+ e^{\epsilon/2}\right)^i.  
 \]
\end{restatable}
\begin{proof}
  We prove this claim by induction. Initially, $i = k = 0$, and the claim holds trivially. Assume that the claim holds for some $i,k \ge 0.$
We will use $\Phi$ as an abbreviated notation for $\Phi(\ww^{(i,k)})$, $\overrightarrow{\alpha}$ for $\overrightarrow{\alpha}^{(i,k)} $, and $\ww$ to denote $\ww^{(i,k)}$.
\paragraph*{Primal Step.} If the next step is a \emph{primal} step, 

\begin{equation}\label{eq:helpPhi1NoDelta}
\Phi\left(\ww^{( i+1,k)} \right) = \norm{ \ww+ \epsilon \overrightarrow{\alpha} (\wt{\CC}\Delta-\wt{\dd})\ww}_1 = \|\ww\|_1 + \epsilon \sum_e \ww_e \overrightarrow{\alpha}_e (\wt{\CC}\Delta-\wt{\dd})_e
\end{equation}

We first bound $\sum_e \ww_e \cdot \overrightarrow{\alpha}_e \cdot (\wt{\CC}\Delta-\wt{\dd})_e$. Using Cauchy-Schwarz inequality, we have
\begin{align*}
\sum_e \ww_e \cdot \overrightarrow{\alpha}_e \cdot |\wt{\CC}\Delta-\wt{\dd}|_e \leq &~ \sqrt{\Big(\sum_e \ww_e (\overrightarrow{\alpha}_e)^2\Big) \cdot \Big(\sum_e \ww_e \cdot (\wt{\CC}\Delta-\wt{\dd})_e^2\Big)} \\
\leq &~ \alpha_+ \cdot \sqrt{\Phi(\ww) \cdot \Psi(\rr)} \\
\leq &~ \alpha_+ \cdot e^{\epsilon/2} \cdot \Phi(\ww),
\end{align*}
where the second step follows from $\alpha_- < \alpha _+$, the third step follows from Lemma~\ref{lem:PsiPhi} that $\Psi(\rr) \leq e^{\epsilon} \cdot \Phi(\ww)$.

Now, from Equation~\eqref{eq:helpPhi1NoDelta}, and the fact that $\ww_e > 0$ from Lemma~\ref{lem:positiveWopt},
\[
\Phi(\ww) - \epsilon \sum_e \ww_e \overrightarrow{\alpha}_e |\wt{\CC}\Delta-\wt{\dd}|_e \leq \Phi\left(\ww^{( i+1,k)} \right)  \leq \Phi(\ww) + \epsilon \sum_e \ww_e \overrightarrow{\alpha}_e |\wt{\CC}\Delta-\wt{\dd}|_e. 
\]
Therefore, we get our bounds,
\[
\Phi\left(\ww^{( i,k)} \right)(1-\epsilon\alpha_+e^{\epsilon/2}) \leq  \Phi\left(\ww^{( i+1,k)} \right) \leq  \Phi\left(\ww^{( i,k)} \right)(1+\epsilon\alpha_+e^{\epsilon/2}) 
\]

\paragraph*{Width Reduction Step.}
Let $\Delta$ be the solution returned in Line~\ref{algline:linsysNM} of Algorithm~\ref{alg:non_monotone_accel_opt}. 
 We have the following:
\begin{align*}
\Phi(\ww^{(i,k+1)}) & = \sum_{e \notin H\cup \{\bar{e}\}}  \ww_e^{(i,k)} + \sum_{e \in H}  \Big((1+\epsilon) \ww_e^{(i,k)} +\frac{\epsilon^2}{n}\Phi(\ww^{(i,k)})\Big)+  \Big((1+\epsilon\gamma)\ww_{\bar{e}}^{(i,k)} +\frac{\epsilon^2\gamma}{n}\Phi(\ww^{(i,k)})\Big)\\
& = \Phi(\ww^{(i,k)}) + \epsilon \sum_{e \in H}  \rr_e^{(i,k)} + \epsilon \gamma  \rr^{(i,k)}_{\bar{e}}\\
& \leq\Phi(\ww^{(i,k)})  + \frac{\epsilon}{\tau} \Psi(\rr^{(i,k)}) + \epsilon \gamma \rr^{(i,k)}_{\bar{e}},
\end{align*}
where the last step follows from $\sum_{e\in H}\rr^{(i,k)}_e\leq \tau^{-1}\Psi(\rr^{(i,k)})$ as guaranteed by Line~\ref{algline:sum_r_in_H} of Algorithm~\ref{alg:non_monotone_accel_opt}.

Now, if $\gamma = 1$, this means that $\rr^{(i,k)}_{\bar{e}} \leq \tau^{-1}\Psi(\rr^{(i,k)})$, and when $\gamma \neq 1$, then $\gamma\rr^{(i,k)}_{\bar{e}}  = \tau^{-1}\Psi(\rr^{(i,k)})$. Therefore, combining both cases,

\begin{align*}
\Phi(\ww^{(i,k+1)}) & \leq \Phi(\ww^{(i,k)})  + \epsilon \tau^{-1}\Psi(\rr^{(i,k)}) + \epsilon \tau^{-1}\Psi(\rr^{(i,k)})\\
& \leq \Phi(\ww^{(i,k)}) \left(1 + \frac{2\epsilon e^{\epsilon}}{\tau}\right).
\end{align*}
Also note that for a width reduction step, $\Phi(\ww^{(i,k+1)}) \geq \Phi(\ww^{(i,k)})$.
\end{proof}

\subsubsection*{Change in $\Psi$}

We now prove how our potential $\Psi$ changes with a primal and width reduction step.

\begin{restatable}{lemma}{ChangePsiNMOpt}
\label{lem:ChangePsiNMOpt}
If the parameters satisfy $\rho^2 \geq \tau \epsilon^{1/3} n^{-2\eta}$, $\rho\geq n^{1/2-3\eta}$, then after $i$ primal and $k$ width reduction steps, the potential $\Psi(\rr^{(i,k)})$ satisfies,
\[
\Psi(\rr^{(i,k)})\geq \Psi(\rr^{(0,0)})\left(1-10\epsilon\alpha\rho\right)^i\left(1+\frac{\epsilon^{4/3}n^{-2\eta}}{10}\right)^k
\]
\end{restatable}
\begin{proof}
We first bound the change in $\Psi$ for a width reduction step. We will use Lemma~\ref{lem:PsiChange}:
\begin{equation}\label{eq:PsiChangeOpt}
   \Psi(\rr')\geq \Psi(\rr) + \sum_e \left(1-\frac{\rr_e}{\rr_e'}\right)\rr_e(\wt{\CC}\Delta-\wt{\dd})_e^2, \text{ where } \Delta = \arg \min_{\Delta'} \sum_e \rr_e (\wt{\CC}\Delta'-\wt{\dd})_e^2. 
\end{equation}

\paragraph{Width Steps.}
Suppose we have had $i$ primal steps.  

\begin{itemize}
    \item We first consider the case when $\sum_{e\in S}\rr_e^{(i,k)} \geq \tau^{-1}\Psi(\rr^{(i,k)})$, and in this case we perturb all edges in $H$ and one extra edge $\bar{e} \in S \backslash H$. For edge $\bar{e}$, $\ww^{(i,k+1)}_{\bar{e}} = \ww^{(i,k)}_{\bar{e}} (1+\epsilon\gamma) +\frac{\epsilon^2\gamma}{n}\Phi(\ww^{(i,k)})$ and as a result,
\[
\frac{\rr^{(i,k+1)}_{\bar{e}}-\rr^{(i,k)}_{\bar{e}}}{\rr^{(i,k+1)}_{\bar{e}}} \geq \frac{\ww^{(i,k+1)}_{\bar{e}}-\ww^{(i,k)}_{\bar{e}} }{\rr^{(i,k+1)}_{\bar{e}}}  \geq \frac{\epsilon\gamma\ww_{\bar{e}}^{(i,k)} + \frac{\epsilon^2\gamma}{n}\Phi(\ww^{(i,k)})}{\rr_{\bar{e}}^{(i,k+1)}} = \epsilon\gamma\frac{\rr_{\bar{e}}^{(i,k)}}{\rr_{\bar{e}}^{(i,k+1)}} \geq \frac{\epsilon\gamma}{1+2\epsilon},
\]
where the last inequality follows from
\begin{align*}
\rr^{(i,k+1)}_{\bar{e}} = &~ \ww_{\bar{e}}^{(i,k)}(1+\gamma\epsilon) +\frac{\epsilon^2\gamma}{2n}\Phi(\ww^{(i,k)}) + \frac{\epsilon}{2n}
\Phi(\ww^{(i,k+1)}) \\
\leq &~ \ww_{\bar{e}}^{(i,k)} + \epsilon\rr^{(i,k)} + \frac{\epsilon}{2n}\Phi(\ww^{(i,k)})(1+2\epsilon)\leq (1+3\epsilon)\rr^{(i,k)}.
\end{align*}
Similarly for the other edges $e$ in $H$,
\[
\frac{\rr^{(i,k+1)}_{e}-\rr^{(i,k)}_e}{\rr^{(i,k+1)}_e} \geq \frac{\epsilon}{1+3\epsilon}.
\]

We will use these bounds in Equation~\eqref{eq:PsiChangeOpt}. Let us first consider the case when $\gamma = 1$.
\begin{align*}
\Psi(\rr^{(i,k+1)}) & \geq \Psi(\rr^{(i,k)}) + \sum_{e \in H \cup \{\ov{e}\}} (1 - \frac{\rr^{(i,k)}_e}{\rr^{(i,k+1)}_e}) \cdot \rr^{(i,k)}_e \cdot (\wt{\CC} \Delta^{(i,k)} - \wt{\dd})_e^2 \\
& \geq \Psi(\rr^{(i,k)}) +  \sum_{e\in H\cup \{\bar{e}\}} \frac{\epsilon}{(1+3\epsilon)} \rr_{e}^{(i,k)} \rho^2\\
& \geq \Psi(\rr^{(i,k)}) + \frac{\epsilon\rho^2}{(1+3\epsilon)} \tau^{-1} \Psi(\rr^{(i,k)}) \\  
& = \Psi(\rrbar^{(i,k)}) \left(1 + \frac{\epsilon \rho^2}{(1+3\epsilon)\tau}\right),
\end{align*}
where the second step follows from $|\wt{\CC} \Delta^{(i,k)} - \wt{\dd}|_e \geq \rho$ for $e \in H$, and the third step follows from the definition of $H$ that $H \subseteq S$ is maximal subset such that $\sum_{e\in H}\rr^{(i,k)}_e\leq \tau^{-1}\Psi(\rr^{(i,k)})$, and so  $\sum_{e\in H \cup \{\ov{e}\}}\rr^{(i,k)}_e\geq \tau^{-1}\Psi(\rr^{(i,k)})$.

Now, in the case when $\gamma\neq 1$, we have $\gamma = \frac{\tau^{-1}}{ \rr^{(i,k)}_{\bar{e}}}\Psi(\rr^{(i,k)}) < 1$, 

\begin{align*}
\Psi(\rr^{(i,k+1)})&\geq \Psi(\rr^{(i,k)}) +   \left(\frac{\epsilon\gamma}{1+3\epsilon} \right)\rr_{\bar{e}}^{(i,k)} \rho^2\\
&\geq \Psi(\rr^{(i,k)}) + \frac{\epsilon \rho^2}{(1+3\epsilon) \tau}\Psi(\rr^{(i,k)}) \\
&\geq \Psi(\rr^{(i,k)}) \left(1 + \frac{\epsilon \rho^2}{10 \tau}\right).
\end{align*}
Therefore, when $\sum_{e\in S}\rr_e^{(i,k)}\geq \tau^{-1}\Psi(\rr^{(i,k)})$,
\begin{equation*}
    \Psi(\rr^{(i,k+1)})\geq\Psi(\rr^{(i,k)}) \left(1 + \frac{\epsilon \rho^2}{10 \tau}\right).
\end{equation*}
\item  Now, in the case when $\sum_{e\in S} \rr^{(i,k)}_e < \tau^{-1}\Psi(\rr^{(i,k)})$, i.e., $H = S$,
\[
\sum_{e \notin H}\rr^{(i,k)}_e|\CC\Delta^{(i,k)}-\dd|_e^3 \leq \max_{e\notin H}\{|\wt{\CC} \Delta^{(i,k)} - \wt{\dd}|_e\}\sum_{e \notin H}\rr^{(i,k)}_e (\wt{\CC} \Delta^{(i,k)} - \wt{\dd})_e^2 \leq \rho\Psi(\rr^{(i,k)}).
\]
Since this is a width reduction step, we know that $\sum_{e }\rr^{(i,k)}_e|\wt{\CC} \Delta^{(i,k)} - \wt{\dd}|_e^3 \geq 2 \rho \Psi(\rr^{(i,k)})$, and therefore we must have,
\[
\sum_{e\in H }\rr^{(i,k)}_e|\wt{\CC} \Delta^{(i,k)} - \wt{\dd}|_e^3 \geq \rho\Psi(\rr^{(i,k)}).
\]
Further, we can assume that for all $e\in H$, $|\wt{\CC} \Delta^{(i,k)} - \wt{\dd}|_e \leq n^{1/2-\eta}\epsilon^{-1/3}$ since otherwise if there were one such edge, then using \eqref{eq:PsiChangeOpt},  $\rr\geq \frac{\epsilon}{n}\Psi(\rr),$ 
\[
\Psi(\rr^{(i,k+1)})\geq \Psi(\rr^{(i,k)}) +  \frac{\epsilon^{1-2/3}}{(1+2\epsilon)} n^{1-2\eta}\frac{\epsilon}{n}\Psi(\rr^{(i,k)}) \geq \Psi(\rr^{(i,k)})\left(1 + \frac{\epsilon^{4/3} n^{-2\eta}}{(1+\epsilon)} \right),
\]
which gives us the required bound.
 Now, 
\[
\sum_{e\in H} \rr^{(i,k)}_e (\wt{\CC} \Delta^{(i,k)} - \wt{\dd})_e^2 \geq \frac{\sum_{e\in H}\rr^{(i,k)}_e |\wt{\CC} \Delta^{(i,k)} - \wt{\dd}|_e^3}{\max_{e\in H} |\wt{\CC} \Delta^{(i,k)} - \wt{\dd}|_e} \geq \frac{\epsilon^{1/3}\rho}{n^{1/2-\eta}} \Psi(\rr^{(i,k)}).
\]
Again, using this with \eqref{eq:PsiChangeOpt} gives us, 
\[
\Psi(\rr^{(i,k+1)}) \geq \Psi(\rr^{(i,k)}) \left(1 + \frac{\epsilon^{4/3} \rho}{(1+3\epsilon)n^{1/2-\eta}}\right).
\]
\end{itemize}
For the values of $\rho $ and $\tau,$ such that $\rho^2/\tau \geq \epsilon^{1/3} n^{-2\eta}$, after every width reduction steps, we get,
\[
\Psi(\rr^{(i_k,k+1)})\geq \Psi(\rr^{(i_k,k)})\left(1 + \frac{\epsilon^{4/3} n^{-2\eta}}{10}\right).
\]

\paragraph{Primal Step.}
\noindent We next look at a primal step. For a primal step, $\ww_e^{(i+1,k)} = \ww_e^{(i,k)}(1+\epsilon\overrightarrow{\alpha} (\wt{\CC} \Delta^{(i,k)} - \wt{\dd}))$. 
Therefore,
\begin{align*}
\rr_e^{(i+1,k)} & = \ww_e^{(i+1,k)} + \frac{\epsilon}{2n}\Phi(\ww^{(i+1,k)}) \\
&\geq \ww_e^{(i,k)} -  \epsilon\alpha_+ \ww^{(i,k)}_e |\wt{\CC} \Delta^{(i,k)} - \wt{\dd}|_e + \frac{\epsilon}{2n} \Phi(\ww^{(i,k)}) (1-\epsilon e^{\epsilon/2} \alpha_+)\\
& \geq \ww^{(i,k)}_e (1-\epsilon) + \frac{\epsilon}{2n} \Phi(\ww^{(i,k)}) (1-\epsilon) = \rr^{(i,k)}_e (1-\epsilon).
\end{align*}
We also have,
\begin{align*}
    \abs{\frac{\rr_e^{(i+1,k)}-\rr_e^{(i,k)}}{\rr_e^{(i+1,k)}}} & \leq \frac{|\ww_e^{(i+1,k)}-\ww_e^{(i,k)}| + \frac{\epsilon}{2n}\abs{\Phi(\ww^{(i+1,k)})-\Phi(\ww^{(i,k)})}}{\rr_e^{(i,k)} (1-\epsilon)}\\
    & \leq \frac{\epsilon \alpha_+ \ww_e^{(i,k)}|\wt{\CC} \Delta^{(i,k)} - \wt{\dd}|_e + \frac{\epsilon}{2n}e^{\epsilon}\epsilon\alpha_+\Phi(\ww^{(i,k)})}{\rr^{(i,k)}_e (1-\epsilon)}\\
    & \leq \frac{\epsilon \alpha_+ \rr^{(i,k)}_e|\wt{\CC} \Delta^{(i,k)} - \wt{\dd}|_e + e^{\epsilon}\epsilon\alpha_+\rr^{(i,k)}_e}{\rr^{(i,k)}_e (1-\epsilon)}\\
    & = e^{\epsilon}\epsilon\alpha_+ |\wt{\CC} \Delta^{(i,k)} - \wt{\dd}|_e + e^{2\epsilon}\epsilon\alpha_+. 
\end{align*}

So we have,
\begin{align*}
\Psi(\rr^{(i+1,k_i)}) & \geq \Psi(\rr^{(i,k_i)}) - \sum_e \abs{\frac{\rr_e^{(i+1,k_i)}-\rr^{(i,k_i)}_e}{\rr_e^{(i+1,k_i)}}}\rr^{(i,k_i)}_e (\wt{\CC} \Delta^{(i,k)} - \wt{\dd})_e^2\\
& \geq \Psi(\rr^{(i,k_i)}) - e^{2\epsilon}\epsilon \alpha \Psi(\rr^{(i,k_i)}) - e^{\epsilon}\epsilon\alpha \sum_{e\in S_i}  |\wt{\CC} \Delta^{(i,k)} - \wt{\dd}|^3_e \rr^{(i,k_i)}_e \\
& \geq \Psi(\rr^{(i,k_i)}) - e^{2\epsilon}\epsilon \alpha \Psi(\rr^{(i,k_i)}) - 2e^{\epsilon}\epsilon\alpha \rho\sum_{e\in S_i}  \Psi(\rr^{(i,k_i)})\\
& \geq \Psi(\rr^{(i,k_i)})\left(1 -10\epsilon\alpha \rho \right).
\end{align*}
In the second last step we used the condition from Line~\ref{algline:CheckPrimalNM}.
\end{proof}

We will now combine the changes in the two potentials similar to the proof of Theorem~\ref{thm:InfRegMainMonotone}.
\subsubsection*{Proof of Theorem~\ref{thm:NonMonotoneMWUAcc}}
\begin{proof}
Let $\xxhat = \frac{\xx}{T}$ be the solution returned by Algorithm \ref{alg:non_monotone_accel_opt}. We would bound the objective value at $\xxhat$. Suppose the algorithm terminates in $T = \alpha^{-1}\epsilon^{-2}\ln n$ primal steps and $K \leq \tau/\epsilon^2$ width reduction steps. We can now apply Lemma \ref{lem:ChangePhiNMOpt} to get,
\[
\Phi\left(\ww^{(T,K)}\right) \le  n\cdot  e^{e^{\epsilon/2}\epsilon\alpha T} e^{2\epsilon e^{\epsilon}\tau K} \leq   n^{O\left(\frac{1}{\epsilon} \right)}.
\]

We bound the $\ell_{\infty}$ norm of $\CC \xxhat - \dd = \frac{1}{T} \cdot \sum_{i=0}^{T-1} (\CC \Delta^{(i,k_i)} - \dd)$ using the upper bound of the potential. Since $\wt{\CC} = \begin{bmatrix}
    \CC \\-\CC
\end{bmatrix}$, and $\wt{\dd}= \begin{bmatrix}
    \dd \\-\dd
\end{bmatrix}$, we have that the weights $\ww \in \mathbb{R}^{2n}$. Therefore, for $\ww_+\in \mathbb{R}^n$ and $\ww_{-}\in \mathbb{R}^n$, we can write $\ww^{(i,k)} = \begin{bmatrix}
    \ww_+^{(i,k)}\\ \ww_{-}^{(i,k)}
\end{bmatrix}$, and we have that $\Phi(\ww) = \sum_{e\in [n]}{\ww_+}_e + {\ww_{-}}_e$. We can similarly define $\rr_+$ and $\rr_{-}$ such that $\rr = \begin{bmatrix}
    \rr_+\\ \rr_{-}
\end{bmatrix}$. Since $\Delta^{(i,k)}$ is obtained by solving,
\[
\Delta^{(i,k)} = \arg\min_{\Delta}\sum_{e\in [2n]}\rr^{(i,k)}_e (\wt{\CC}\Delta-\wt{\dd})_e^2 = \sum_{e\in [n]}(\rr_+^{(i,k)} + \rr_{-})^{(i,k)}_e (\CC\Delta-\dd)_e^2, 
\]
the update rule $\ww^{(i+1,k)} = \ww^{(i,k)} \cdot \big(1 + \overrightarrow{\alpha}^{(i,k)} (\wt{\CC} \Delta^{(i,k)} - \wt{\dd})\big)$ implies that in every primal step,
\[
(\ww_+)_e^{(i+1,k)}=(\ww_+)_e^{(i,k)} \cdot \big(1 + \overrightarrow{\alpha}^{(i,k)} (\wt{\CC} \Delta^{(i,k)} - \wt{\dd})\big), \quad (\ww_-)_e^{(i+1,k)}=(\ww_-)_e^{(i,k)} \cdot \big(1 - \overrightarrow{\alpha}^{(i,k)} (\wt{\CC} \Delta^{(i,k)} - \wt{\dd})\big).
\]
Now,
\begin{align*}
(\ww_+)^{(T,K)}_e = &~ \ww^{(0)}_e \cdot \prod_{i=0}^{T-1} \Big(1 + \epsilon \overrightarrow{\alpha}^{(i,k_i)}_e (\CC \Delta^{(i,k_i)} - \dd)_e\Big) \\
= &~ \prod_{i: (\CC \Delta^{(i,k_i)} - \dd)_e \geq 0} (1 + \epsilon \alpha_+ (\CC \Delta^{(i,k_i)} - \dd)_e) \cdot \prod_{i: (\CC \Delta^{(i,k_i)} - \dd)_e < 0} (1 + \epsilon \alpha_- (\CC \Delta^{(i,k_i)} - \dd)_e) \\
\geq &~ \exp\left(\epsilon(1-\epsilon)\alpha \cdot \sum_{i=0}^{T-1} (\CC \Delta^{(i,k_i)} - \dd)_e\right),
\end{align*}
where the second step follows from $\ww^{(0)} = 1_n$, and $\overrightarrow{\alpha}_e^{(i,k_i)} = \alpha_+$ if $(\CC \Delta^{(i,k_i)} - \dd)_e \geq 0$ and $\overrightarrow{\alpha}_e^{(i,k_i)} = \alpha_-$ otherwise, the third step follows from $1 + \epsilon x \geq \exp(\epsilon (1-\epsilon) x)$ for all $0 \leq x \leq 1$ and $1 + \epsilon x \geq \exp(\epsilon (1 + \epsilon) x)$ for all $-1 \leq x \leq 0$, and we have that $|\overrightarrow{\alpha}^{(i,k)} \cdot (\CC \Delta^{(i,k)} - \dd)| \leq  \frac{1}{10}$ by Lemma~\ref{lem:positiveWopt}. Similarly, we also get,
\[
(\ww_{-})^{(T,K)}_e \geq \exp\left(\epsilon(1-\epsilon)\alpha \cdot \sum_{i=0}^{T-1} - (\CC \Delta^{(i,k_i)} - \dd)_e\right).
\]

This implies that
\begin{align*}
\abs{\sum_{i=0}^{T-1} (\CC \Delta^{(i,k_i)} - \dd)_e} \leq \frac{\ln\left((\ww_+)^{(T,K)}_e +(\ww_{-})_e^{(T,K)}\right)}{\epsilon(1-\epsilon)\alpha} \leq \frac{\ln(\Phi(\ww^{(T,K)}))}{\epsilon(1-\epsilon)\alpha}.
\end{align*}
So we have
\begin{align*}
\|\CC \xxhat - \dd\|_{\infty} = &~ \frac{1}{T} \max_e \abs{\sum_{i=0}^{T-1} (\CC \Delta^{(i,k_i)} - \dd)_e)} \\
\leq &~ \frac{\ln(\Phi(\ww^{(T,K)}))}{\alpha T} \\
\leq &~ \frac{\ln n + (1+\epsilon)\epsilon \alpha T + (1+\epsilon)}{\epsilon(1-\epsilon)\alpha T} \\
\leq &~ 1 + 10\epsilon,
\end{align*}

We have shown that if the number of width reduction steps is bounded by $K$ then our algorithm returns the required solution. We will next prove that we cannot have more than $K$ width reduction steps.

We first show that if $\eta \leq 1/6$, the number of width steps $K$ must be at most $\tilde{O}(1)\cdot T$. 
for the values of $\tau, \rho$ and $\alpha$, the guarantee of Lemma~\ref{lem:ChangePsiNMOpt} becomes,
\begin{align*}
\Psi(\rr^{(i,k)})& \geq \Psi(\rr^{(0,0)})\left(1 - 10\epsilon\alpha \rho \right)^T\left(1 +  \frac{\epsilon^{4/3}n^{-2\eta}}{10}\right)^K\\
& = \exp \left\{\epsilon^{4/3} n^{-2\eta} K/20 - 20\epsilon\alpha \rho T\right\}
\end{align*}
Since $\Psi(\rr)\leq (1+\epsilon)\Phi \leq n^{O(1/\epsilon)}$, we must have,
\[
n^{O(1/\epsilon)}\geq L \exp\left\{\epsilon^{4/3} n^{-2\eta} K/20 - 20\epsilon\alpha \rho T\right\},
\]
or,
\[
K \leq O(T) + \Otil(1) n^{2\eta}/\epsilon^{7/3}.
\]
Since $\eta\leq 1/6$, and $T = \alpha^{-1}\epsilon^{-2}\ln n$,  $T + \Otil(1) n^{2\eta}/\epsilon^{7/3} \leq \tau\epsilon^{-2}\ln n$ as required. Therefore the total number of iterations is at most,
\[
T + K \leq  \alpha^{-1}\epsilon^{-2}\log n+  n^{2\eta}/\epsilon^{7/3} = \tilde{O}\left((n^{1/2-\eta} + n^{2\eta})/\epsilon^{7/3}\right).
\]
\end{proof}

\subsection*{Lower Bound on \texorpdfstring{$\Psi$}{}}

\begin{lemma}\label{lem:lower_bound_Psi}
    For all $i,k$ and $\rr^{(i,k)}$ as defined in Algorithm~\ref{alg:non_monotone_accel_opt}, $\Psi(\rr^{(i,k)})\geq \frac{1}{1+2\epsilon}\cdot\Psi(\rr^{(0,0)})$.
\end{lemma}
\begin{proof}
    We first note that, $\Psi(\rr^{(0,0)}) =  2(1+2\epsilon) \dd^{\top}\left(\II - \CC^{\top}(\CC^{\top}\CC)^{-1}\CC\right)\dd$. Now, from Lemma~\ref{lem:PsiChange}, we know that for any $\rr'\geq \rr$, $\min_{\Delta}\sum_e \rr'_e (\CC\Delta-\dd)_e^2 \geq \min_{\Delta}\sum_e \rr_e (\CC\Delta-\dd)_e^2$. Let $\rr = \begin{bmatrix}
        \rr_+\\ \rr_{-}
    \end{bmatrix}$. We also know that,
    \[
    \Psi(\rr^{(i,k)}) = \min_{\Delta}\sum_e \left(\rr_+^{(i,k)} + \rr_{-}^{(i,k)}\right)_e (\CC\Delta-\dd)_e^2.
    \]
    To prove our result, it is sufficient to prove that $\rr_+^{(i,k)} + \rr_{-}^{(i,k)}\geq 2 $ because, then 
    \[
     \Psi(\rr^{(i,k)}) \geq 2 \min_{\Delta}\sum_e (\CC\Delta-\dd)_e^2 = 2 \dd^{\top}\left(\II - \CC^{\top}(\CC^{\top}\CC)^{-1}\CC\right)\dd = \frac{1}{1+2\epsilon}\Psi(\rr^{(0,0)}).    \]
     Now,
     \[
     \rr_+^{(i,k)} + \rr_{-}^{(i,k)} = \ww_+^{(i,k)} + \ww_{-}^{(i,k)} + \frac{2\epsilon}{n}\Phi(\ww^{(i,k)})\geq \ww_+^{(i,k)} + \ww_{-}^{(i,k)}.
     \]
     Since the weights only increase during a width reduction step, we always have that,
     \begin{align*}
    (\ww_+)^{(i,k)}_e \geq &~ \ww^{(0,0)}_e \cdot \prod_{j=0}^{i-1} \Big(1 + \epsilon \overrightarrow{\alpha}^{(j,k_j)}_e (\CC \Delta^{(j,k_j)} - \dd)_e\Big) \\
= &~ \prod_{j: (\CC \Delta^{(j,k_j)} - \dd)_e \geq 0} (1 + \epsilon \alpha_+ (\CC \Delta^{(j,k_j)} - \dd)_e) \cdot \prod_{j: (\CC \Delta^{(j,k_j)} - \dd)_e < 0} (1 + \epsilon \alpha_- (\CC \Delta^{(j,k_j)} - \dd)_e) \\
\geq &~ \exp\left(\epsilon(1-\epsilon)\alpha \cdot \sum_{j=0}^{i-1} (\CC \Delta^{(j,k_j)} - \dd)_e\right),
\end{align*}
and similarly,
\[
(\ww_-)^{(i,k)}_e \geq \exp\left(-\epsilon(1-\epsilon)\alpha \cdot \sum_{j=0}^{i-1} (\CC \Delta^{(j,k_j)} - \dd)_e\right).
\]
The function $e^x + e^{-x}\geq 2$ for all $x\in \mathbb{R}$, since it is minimized at $x = 0$. Therefore, for all $e$ 
\[
\left(\rr_+^{(i,k)} + \rr_{-}^{(i,k)}\right)_e \geq \left(\ww_+^{(i,k)} + \ww_{-}^{(i,k)}\right)_e \geq 2,
\]
concluding the proof of our result.
\end{proof}